\documentclass[a4paper,11pt]{article}

\usepackage{geometry}
\newgeometry{left=1in, right=1in, top=1in, bottom=1in}

\usepackage[T1]{fontenc}
\usepackage[utf8]{inputenc}
\usepackage{authblk}
\usepackage{commath}
\usepackage{braket}
\usepackage{multicol}
\usepackage{enumerate}
\usepackage{mathrsfs}
\usepackage{amsmath,amsthm,amssymb,xspace}
\usepackage{cite}
\usepackage{txfonts}
\usepackage{graphics}
\usepackage{float}
\usepackage{epstopdf}
\usepackage{epsfig}
\usepackage{caption}
\usepackage[labelformat=simple]{subcaption}

\usepackage[colorlinks,bookmarksopen,bookmarksnumbered,citecolor=blue, linkcolor=purple, urlcolor=black]{hyperref}
\usepackage[dvipsnames]{xcolor}
\definecolor{lgray}{gray}{0.88}
\definecolor{llgray}{gray}{0.93}
\definecolor{llyellow}{RGB}{255,255,224}
\definecolor{lyellow}{RGB}{238,238,209}
\definecolor{llblue}{RGB}{224,255,255}
\definecolor{lblue}{RGB}{176,224,230}
\definecolor{llgreen}{RGB}{193,255,193}
\definecolor{lgreen}{RGB}{180,238,180}
\usepackage{url}
\usepackage{pgfplots}
\usepgfplotslibrary{polar}
\usepgflibrary{shapes.geometric}
\usetikzlibrary{calc}

\usepackage{booktabs}
\usepackage{commath}
\usepackage{tikz}
\usepackage{tikz-3dplot}
\usepackage[qm]{qcircuit}
\usepackage{makecell}
\usepackage[ruled,vlined,linesnumbered]{algorithm2e}
\usepackage{multirow}

\newcommand{\B}{\mbox{$\{0,1\}$}}
\newcommand{\Bn}{\mbox{$\{0,1\}^n$}}

\newcommand\defeq{\stackrel{\mathrm{\scriptsize def}}{=}}

\newcommand\Path{{\tt Path}}
\newcommand\Star{{\tt Star}}
\newcommand\Tree{{\tt Tree}}
\newcommand\Grid{{\tt Grid}}
\newcommand\diam{{\tt diam}}
\newcommand\brickwall{{\tt Brickwall}}
\newcommand\Cnot{{\sf CNOT}}
\newcommand\Expander{{\tt Expander}}
\newcommand\qsp{{\tt QSP}}

\newtheorem{theorem}{Theorem}
\newtheorem{corollary}{Corollary}
\newtheorem{lemma}{Lemma}
\newtheorem{proposition}{Proposition}
\newtheorem{definition}{Definition}

\usepackage{thmtools}
\usepackage{thm-restate}

\title{Does qubit connectivity impact quantum circuit complexity?}
 \author[1]{Pei Yuan \footnote{Email: peiyuan@tencent.com} }
 \author[2]{Jonathan Allcock \footnote{Email: jonallcock@tencent.com}}
\author[1]{Shengyu Zhang \footnote{Email: shengyzhang@tencent.com}}
 \affil[1]{Tencent Quantum Laboratory, Tencent, Shenzhen, Guangdong 518057, China}
 \affil[2]{Tencent Quantum Laboratory, Tencent, Hong Kong, China}
\date{}

\begin{document}





\maketitle

\begin{abstract}
Some physical implementation schemes of quantum computing can apply two-qubit gates only on certain pairs of qubits. These connectivity constraints are commonly viewed as a significant disadvantage. {For example, compiling an unrestricted $n$-qubit quantum circuit to one with poor qubit connectivity, such as a 1D chain, usually results in a blowup of depth by $O(n^2)$ and size by $O(n)$. It is appealing to conjecture that this overhead is unavoidable---a random circuit on $n$ qubits has $\Theta(n)$ two-qubit gates in each layer and a constant fraction of them act on qubits separated by distance $\Theta(n)$. }

While it is known that almost all $n$-qubit unitary operations need quantum circuits of $\Omega(4^n/n)$ depth and $\Omega(4^n)$ size to realize with all-to-all qubit connectivity, in this paper, we show that \textit{all} $n$-qubit unitary operations can be implemented by quantum circuits of $O(4^n/n)$ depth and $O(4^n)$ size even under {1D chain}  qubit connectivity constraint. 

We extend this result and investigate qubit connectivity in three directions. First, we consider more general connectivity graphs and show that the circuit size can always be made $O(4^n)$ as long as the graph is connected. For circuit depth, we study $d$-dimensional grids, complete $d$-ary trees and expander graphs, and show results similar to the 1D chain. 
Second, we consider the case when ancillary qubits are available. We show that, with ancilla, the circuit depth can be made polynomial, and the space-depth trade-off is not impaired by connectivity constraints unless we have exponentially many ancillary qubits. 
Third, we obtain nearly optimal results on  special families of unitaries, including diagonal unitaries, 2-by-2 block diagonal unitaries, and Quantum State Preparation (QSP) unitaries, the last being a fundamental task used in many quantum algorithms for machine learning and linear algebra. 
\end{abstract}

\section{Introduction}
\label{sec:intro}
Quantum computation has shown advantages over classical computation in solving some intractable computational problems, based on the unique properties of quantum mechanics. In recent years, tremendous advances have been made in quantum technologies, in both theory and experiment, and hundreds of quantum algorithms have been proposed with rigorous mathematical proofs of speedup over the best possible or best-known classical counterparts \cite{quantumalgorithm}. When these algorithms are realized in quantum circuits consisting of 1-qubit and 2-qubit gates, however, qubit connectivity often comes as a constraint. Some leading implementation schemes such as superconducting qubits~\cite{IBMQ,arute2019quantum,gong2021quantum} and quantum dots \cite{ciorga2000addition,elzerman2003few,petta2004manipulation,schroer2007electrostatically,zajac2016scalable}, and cold atoms \cite{bloch2008quantum,buluta2011natural,bernien2017probing,graham2019rydberg},
can only apply 2-qubit gates on certain pairs of qubits, while other schemes such as trapped ion \cite{leibfried2003quantum,schindler2013quantum,pogorelov2021compact,blatt2012quantum}
and photonic quantum computers \cite{wang201818,zhong2020quantum,madsen2022quantum}
may not be subject to the same constraints. While this connectivity constraint is typically viewed as a considerable disadvantage,  the extent of this disadvantage seems yet to be systematically studied. 
This paper aims to address the central question:
\begin{quote}
    {\it How does qubit connectivity affect quantum circuit complexity?}
\end{quote}

We study this question in terms of circuit depth and size. 
Let us start with a motivating example. Some early-stage superconducting quantum systems have qubits arranged in a 1D chain and only allow nearest neighbor interactions \cite{IBMQ, kelly2015state}, which we refer to as being under \textit{path constraint}. 
The 1D chain has very poor connectivity by almost all graph-theoretic measures, such as diameter, average degree, number of edges, vertex or edge expansion, etc. Compiling a quantum circuit based on all-to-all qubit connectivity to one compatible with 1D chain connectivity usually results in a blowup of depth by $O(n^2)$ and of size by $\Theta(n)$. Indeed, each layer generally has $\Theta(n)$ two-qubit gates and many of these gates act on two qubits that are $\Theta(n)$ apart on the chain. In this regard, it is even appealing to conjecture that these overheads in depth and size are unavoidable for generic quantum circuits. However, this intuition turns out to be wrong, as the following result shows.

\begin{theorem}\label{thm:US_noancilla_path_intro}
    Any $n$-qubit unitary can be implemented by a quantum circuit of depth $O(4^n/n)$ and size $O(4^n)$ under path constraint.
\end{theorem}
Note that these bounds are tight: even \textit{without any connectivity restrictions}, almost all $n$-qubit unitary circuits need depth $\Omega(4^n/n)$ and size $\Omega(4^n)$ to implement \cite{sun2021asymptotically}. Therefore, the above theorem implies that the qubit connectivity constraint does not increase the depth and size complexity (by more than a constant factor) for almost all $n$-qubit unitaries. 

This somewhat counter-intuitive example calls for more systematic studies of the central question in specific settings. In this paper, we investigate three aspects of this topic: 
\begin{enumerate}
    \item Graphs: What constraint graphs affect circuit complexity and by how much? Is there a simple graph property such as diameter, vertex degree, or expansion constant that characterizes the impact the graph has on circuit depth and size? 
    
    \item Space: Recent studies show that ancillary qubits can be used to reduce quantum circuit depth. How does the qubit connectivity constraint affect this? 
    
    \item Unitaries: What can we say about specific sets of unitary operations, in terms of worst-case and average-case complexity? 
\end{enumerate}

Our main results are described below. The results involving ancillary qubits are easiest to state and will be used subsequently, so we begin with those.

\paragraph{Ancillary qubits and depth-space trade-offs} A number of recent results have shown that one can reduce circuit depth by utilizing ancillary qubits \cite{low2018trading,wu2019optimization,sun2021asymptotically,yuan2022optimal,rosenthal2021query}.
When connectivity constraints are taken into consideration, for example, when all $n+m$ qubits are arranged in a 1D chain, can we still trade ancilla for depth\footnote{Technically speaking, one should specify where the $n$ non-ancilla qubits are located in the chain, e.g., if they are located at the two ends, with the $m$ ancilla in a contiguous block in the middle, then one requires at least $O(n+m)$ depth to let them ``reach''  each other. Here we consider the case where the ancilla and non-ancilla qubits form two contiguous blocks, a scenario more natural for downstream applications.}? We show:
\begin{theorem}
    For all $m\le O(2^{n/2})$, any $n$-qubit unitary
    can be implemented by a quantum circuit of depth $O(4^n/(n+m))$ and size $O(4^n)$ under the $(n+m)$-long path constraint, using $m$ ancillary qubits. These bounds are tight. 
\end{theorem}

That is, when at most $O(2^{n/2})$ ancilla are available, 1D chain connectivity does not affect either the worst case or generic circuit depth or size. On the other hand, we show circuit depth upper and lower bounds of $O\big(2^{3n/2} + \frac{4^n}{n+m}\big)$ and $\Omega\big(2^{n} + \frac{4^n}{n+m}\big)$, respectively, with $m > O(2^{n/2})$ ancilla. Comparing this with the depth upper bound of $O\big(n2^{n/2}+\frac{n^{1/2}2^{3n/2}}{m^{1/2}}\big)$ in the unrestricted case \cite{yuan2022optimal}, we see that the effect of connectivity on circuit complexity can be sensitive to the number of ancilla.


\paragraph{The effect of graph constraints on connectivity}

Qubit connectivity can be modelled by an undirected, connected {\it constraint graph} $G=(V,E)$, with vertices $v\in V$ corresponding to qubits, and edges $(u,v)\in E$ corresponding to pairs of qubits on which one can apply 2-qubit gates.  The case where $G=K_n$, i.e., the complete graph on $n$ vertices, describes an $n$-qubit circuit with all-to-all connectivity (or, equivalently, no connectivity constraints). 


Current superconducting quantum processors have qubit connectivity constraint corresponding to a wide range of constraint graphs. 1D chain is the common qubit layout used in many early-stage chips. In addition to that, bilinear chains~\cite{IBMQ,ye2019propagation}, 2D grids~\cite{arute2019quantum,gong2021quantum}, brick-wall graphs~\cite{IBMQ} and trees~\cite{IBMQ} have also been realized, and 3D grids may potentially be utilized by multi-layer chips in the future.

We study three families of graphs: (i) $d$-dimensional grids, (ii) $d$-ary trees, and (iii) expanders. In each family, we can see the dependence of depth overhead on some key parameter ($d$  or expansion). We start from the grids. The following result concerns grid graphs $n^{1/d}\times \cdots \times n^{1/d}$.

\begin{theorem}
     For all $m\le O(2^{\frac{dn}{d+1}}/d)$, any $n$-qubit unitary
    can be implemented by a quantum circuit of $O(4^n/(n+m))$ depth and $O(4^n)$ size under the $(n+m)^{1/d}\times \cdots \times (n+m)^{1/d}$-grid constraint using $m$ ancillary qubits, and these bounds are tight. When no ancillary qubits are used, the required circuit depth is $O(4^n/n)$, the same as for unrestricted circuits. 
\end{theorem}
We make several remarks.
%
First, in later sections, we give circuit constructions for $d$-dimensional grids of general sizes $n_1\times \cdots \times n_d$, which include bilinear chains as a special case. 
Of particular importance are the cases $d=2$ and $d=3$, which correspond to practical implementations of superconducting processors. 
Second, some graphs, such as the brick-wall graph found in some IBM processors, do not fall into this family, but we shall show how it reduces to the 2D grid with a similar (and tight) bound. 
Third, for $m$ larger than $O(2^{dn/(d+1)}/d)$, upper and lower bounds are also given. 

The second family of graphs are the complete $d$-ary trees.
\begin{theorem}
    For all $m\ge 0$, any $n$-qubit unitary 
    can be realized by a quantum circuit of depth \[\Tilde{O}\left(dn2^n + \frac{(n+d)4^n}{n+m}\right)\] and size $O(4^n)$ under complete $d$-ary tree (with $n+m$ vertices) constraint, using $m$ ancillary qubits. In particular, when no ancillary qubits are available, the required circuit depth is $O(4^n)$, and this is optimal up to a factor of $O(n/d)$.
\end{theorem}

As qubit connectivity in real devices can vary greatly (see \cite{IBMQ} for a few examples), we also study circuit size under general graph constraints. We show: 

\begin{theorem}\label{thm:US_size_upper_intro}
Any $n$-qubit unitary matrix can be implemented by a quantum circuit of size $O(4^n)$ under arbitrary connected graph constraints.
\end{theorem}
This result is tight, as the circuit size lower bound is $\Omega(4^n)$ even assuming all-to-all connectivity \cite{shende2004minimal}. 
This implies that for almost all unitary operations, arbitrary graph constraints do not impact the required circuit size. 

The results above, along with others summarized in Table \ref{tab:US_graph}, relate to the challenge of General Unitary Synthesis (GUS), i.e., the implementation of general $n$-qubit unitary operations. Similar to size complexity, our circuit constructions apply to the worst case (i.e. work for all unitary operations), and our lower bounds hold for generic (i.e., almost all) unitaries, which make our results stronger.

\begin{table}[hbt]
    \centering
    \caption{Circuit depth bounds for $n$-qubit general unitary synthesis (GUS) under graph constraints, using $m$ ancillary qubits. All graphs have $n+m$ vertices. The $(n_1,\ldots,n_d)$-Grid is a $d$-dimensional grid of size $n_1\times n_2\times \cdots \times n_d$ with $n_1\ge n_2\ge\cdots \ge n_d\ge 1$. In the (complete) $d$-ary tree, every non-leaf node has exactly $d$ children. $d$ can be $2$ (a binary tree) and $n+m-1$ (a Star). The last column gives ranges of $m$ where our upper and lower bounds match.}
    \resizebox{\textwidth}{!}{\begin{tabular}{c|cc|c}
    \hline
      Graph  &  Depth upper bounds / $O(\cdot)$  & Depth lower bounds / $\Omega(\cdot)$   &Optimal range of $m$  \\
     \hline
     \hline
     \multirow{2}*{Path}   & $4^{3n/4}+\frac{4^n}{n+m}$ & $4^{n/2}+\frac{4^n}{n+m}$ & \multirow{2}*{$0\le m\le O(2^{n/2})$}  \\
     & [Thm. \ref{thm:US_path_grid}]& [Thm. \ref{thm:lower_bound_grid_k_US}] &
\\      \hline
     
     \multirow{2}*{$(n_1,n_2)$-Grid}   & $4^{2n/3}+\frac{4^{3n/4}}{(n_2)^{1/2}}+\frac{4^n}{n+m}$ & $\max\left\{4^{n/3}, \frac{4^{n/2}}{(n_2)^{1/2}},\frac{4^n}{n+m}\right\}$ & \multirow{2}*{$0\le m\le O\big(\frac{2^n}{2^{n/3}+\frac{2^{n/2}}{(n_2)^{1/2}})}\big)$}  \\
     & [Thm. \ref{thm:US_path_grid}] &[Thm. \ref{thm:lower_bound_grid_k_US}] & \\
      \hline
     
     \multirow{2}*{$(n_1,\ldots,n_d)$-Grid}   & $n^22^n+d4^{\frac{(d+2)n}{2(d+1)}}+\max\limits_{j\in\{2,\ldots,d\}}\big\{\frac{d4^{(j+1)n/(2j)}}{(\Pi_{i=j}^d n_i)^{1/j}}\big\}+\frac{4^n}{n+m}$ & $ n+4^{\frac{n}{d+1}}+\max\limits_{j\in [d]}\big\{\frac{4^{n/j}}{(\Pi_{i=j}^d n_i)^{1/j}}\big\}$& \multirow{2}*{$0\le m\le O\big(\frac{2^n}{n^2+d2^{\frac{n}{d+1}}+\max\limits_{j\in\{2,\ldots,d\}}\big\{\frac{d2^{n/j}}{(\Pi_{i=j}^d n_i)^{1/j}}\big\}}\big)$}  \\
     &[Thm. \ref{thm:US_path_grid}] & [Thm. \ref{thm:lower_bound_grid_k_US}] & \\
      \hline
      
      \multirow{2}*{Binary Tree} & $n^2\log(n)2^n+\frac{\log(n)4^n}{n+m}$& $\max\left\{n,\frac{4^n}{n+m}\right\}$ & \multirow{2}*{optimal up to $\log (n)$ when $m \le O(2^n/n^2)$}\\
      & [Thm. \ref{thm:US_tree}] & [Thm. \ref{thm:lower_bound_tree_US}] & \\
      \hline
      
      \multirow{2}*{$d$-ary Tree} & $n2^nd\log_d (n+m)\log_d(n+d)+\frac{(n+d)\log_d(n+d) 4^{n}}{n+m}$& $\max\left\{n,\frac{d4^n}{n+m}\right\}$& \multirow{2}*{optimal up to $n\log (n)$ when $m \le O(2^n/d\log (n))$}  \\
      & [Thm. \ref{thm:US_tree}] & [Thm. \ref{thm:lower_bound_tree_US}] &\\
     \hline
     
     \multirow{2}*{Star}  & $4^n$  &  $4^n$ & \multirow{2}*{$m\ge 0$}  \\
     &[Thm. \ref{thm:US_tree}] & [Thm. \ref{thm:lower_bound_tree_US}] &\\
     \hline
     
     \multirow{2}*{Expander} & $n^22^n+\frac{\log(m)4^n}{n+m}$ & $\max\left\{n,\frac{4^n}{n+m}\right\}$ & \multirow{2}*{optimal up to $n$ when $m \le O(2^n/n)$}\\
     &[Thm. \ref{thm:US_expander_graph}] &[Thm. \ref{thm:lower_exapnder_US}] &\\
     \hline
    \end{tabular}}
    \label{tab:US_graph}
\end{table}

\paragraph{Circuit complexity for special families of unitaries}

While the above results for GUS tell us what we can hope for in a generic solution for all $n$-qubit unitaries, special families of unitary operations warrant further study. Firstly, by utilizing the structure of particular unitaries, one may design better constructions (in particular we desire  $poly(n)$-depth circuits where possible). Secondly, by focusing on special tasks, one may derive tighter circuit complexity bounds, which may elucidate the effects of connectivity constraints. We study three special families of unitary operations: 
\begin{enumerate}
    \item Diagonal unitaries.
    \item 2-by-2 block diagonal unitaries.
    \item Quantum state preparation (QSP) unitaries.
\end{enumerate}

These three families are closely related and have all been extensively studied in quantum circuit theory. 
For brevity, here we discuss QSP only (for diagonal or 2-by-2 block diagonal unitaries, refer to \cite{bergholm2005quantum,mottonen2005decompositions,plesch2011quantum}). QSP is an important subroutine in many quantum machine learning algorithms \cite{lloyd2014quantum,kerenidis2017quantum,rebentrost2018quantum,kerenidis2020quantum, harrow2009quantum,wossnig2018quantum,kerenidis2019q,rebentrost2014quantum} 
and Hamiltonian simulation algorithms \cite{low2017optimal,berry2015hamiltonian,low2019hamiltonian}, and has been the subject of increasing attention \cite{zhang2021low, sun2021asymptotically,yuan2022optimal,rosenthal2021query,johri2021nearest}, culminating at \cite{yuan2022optimal} achieving the optimal depth for any number of ancillary qubits. 

For QSP, we can again consider circuit size under general graph constraints, and circuit depth for grids and complete $d$-ary tree graphs. We have the following results.
\begin{theorem}
    An $n$-qubit QSP unitary can be implemented by a quantum circuit of size $O(2^n)$ under any graph constraint. 
\end{theorem}
This bound is tight, as QSP needs $\Omega(2^n)$ size even without any connectivity constraints \cite{plesch2011quantum}, and the presence of constraints does not increase the required circuit size. 

For $d$-dimensional grids, we prove asymptotically optimal circuit depth requirements for any constant $d$, and almost optimal results for larger $d$:
\begin{theorem}
    An $n$-qubit QSP unitary can be implemented by a quantum circuit of depth $O\left(2^{n/2} + \frac{2^n}{n+m}\right)$ under 1D chain constraint, depth $O\left(2^{n/3} + \frac{2^n}{n+m}\right)$ under 2D grid constraint, and depth $O\left(n^3+d2^{\frac{n}{d+1}}+\frac{2^n}{n+m}\right)$ under $d$-dimensional grid of size $(n+m)^{1/d}\times \cdots \times (n+m)^{1/d}$ constraint, using $m\ge 0$ ancillary qubits. These bounds are tight for any constant $d$, and off by at most a factor of $d$ for $d(n)=\omega(1)$.
\end{theorem}
For trees, we give circuit constructions whose depth is optimal if $m$ is not too large.
\begin{theorem}
    An $n$-qubit QSP unitary can be implemented by a quantum circuit of depth $\tilde O\left(n^2 2^{n} + 4^n/(n+m)\right)$ under complete binary tree constraint, depth $\tilde O\left(d n 2^n + (n+d)4^{n}/(n+m)\right)$ on complete $d$-ary tree constraint, and depth $O\left(4^n\right)$ under star graph constraint, using $m\ge 0$ ancillary qubits. The bound for the star graph is tight, and the bound for general complete $d$-ary trees is tight for $m = O(2^n/n^2 d)$.
\end{theorem}

\begin{table}[]
    \centering
    \caption{Circuit depth bounds for $n$-qubit quantum state preparation (QSP) under graph constraints, using $m$ ancillary qubits. All graphs have $n+m$ vertices. The $(n_1,\ldots,n_d)$-Grid is a $d$-dimensional grid of size $n_1\times n_2\times \cdots \times n_d$ with $n_1\ge n_2\ge\cdots \ge n_d\ge 1$. In the (complete) $d$-ary tree, every non-leaf node has exactly $d$ children. $d$ can be $2$ (a binary tree) and $n+m-1$ (a Star). The last column gives ranges of $m$ where our upper and lower bounds match.
    }
    \resizebox{\textwidth}{!}{\begin{tabular}{c|c c |c}
    \hline
       Graph  &  Depth upper bounds / $O(\cdot)$  & Depth lower bounds/ $\Omega(\cdot)$  & Optimal range of $m$  \\
     \hline
     \hline
      \multirow{2}*{Path}   & $2^{n/2}+\frac{2^n}{n+m}$  & $2^{n/2}+\frac{2^n}{n+m}$  & \multirow{2}*{$m\ge 0$}  \\
      &[Thm. \ref{thm:QSP_grid_main}] &[Thm. \ref{thm:lower_bound_grid_k_QSP}] & \\
      \hline
      
     \multirow{2}*{$(n_1,n_2)$-Grid}   & $2^{n/3}+\frac{2^{n/2}}{(n_2)^{1/2}}+\frac{2^n}{n+m}$  & $\max\big\{2^{n/3}, \frac{2^{n/2}}{(n_2)^{1/2}},\frac{2^n}{n+m}\big\}$  & \multirow{2}*{$m\ge 0$}  \\
     &[Thm. \ref{thm:QSP_grid_main}]] &[Thm. \ref{thm:lower_bound_grid_k_QSP}]  & \\
       \hline
       
     \multirow{2}*{$(n_1,\ldots,n_d)$-Grid }  & $n^3+d2^{\frac{n}{d+1}}+\max\limits_{j\in\{2,\ldots,d\}}\Big\{\frac{d2^{n/j}}{(\Pi_{i=j}^d n_i)^{1/j}}\Big\}+\frac{2^n}{n+m}$  & \multirow{2}*{$ n+2^{\frac{n}{d+1}}+\max\limits_{j\in [d]}\big\{\frac{2^{n/j}}{(\Pi_{i=j}^d n_i)^{1/j}}\big\}$ } & if $d$ is a constant, $m\ge 0$; \\
    
     &[Thm. \ref{thm:QSP_grid_main}] &[Thm. \ref{thm:lower_bound_grid_k_QSP}]   & otherwise, $0\le m\le O\Big(\frac{2^n}{n^3+d2^{\frac{n}{d+1}}+\max\limits_{j\in\{2,\ldots,d\}}\big\{\frac{d2^{n/j}}{(\Pi_{i=j}^d n_i)^{1/j}}\big\}}\Big)$\\
      \hline
      
      \multirow{2}*{Binary Tree}   & $n^3\log(n)+\frac{\log(n)2^n}{n+m}$  & $\max\left\{n,\frac{2^n}{n+m}\right\}$ &  \multirow{2}*{optimal up to $\log (n)$ when $m \le O(2^n/n^3)$} \\
      &[Thm. \ref{thm:QSP_tree}] & [Thm. \ref{thm:lower_bound_tree_QSP}] & \\
      \hline
      
      \multirow{2}*{$d$-ary Tree}  & $n^2d\log_d (n+m)\log_d(n+d)+\frac{(n+d)\log_d(n+d) 2^{n}}{n+m}$  &  $\max\left\{n,\frac{d2^n}{n+m}\right\}$ & optimal up to $n\log (n)$ when $m \le O(2^n/nd\log n)$  \\
      & [Thm. \ref{thm:QSP_tree}] &[Thm. \ref{thm:lower_bound_tree_QSP}] & \\
      \hline
      
      \multirow{2}*{Star}  & $2^n$  & $2^n$& \multirow{2}*{$m\ge 0$}  \\
      & [Thm. \ref{thm:QSP_tree}] & [Thm. \ref{thm:lower_bound_tree_QSP}] & \\
      \hline
      
      \multirow{2}*{Expander}  & $n^3+\frac{\log(m)2^n}{n+m}$   & $\max\left\{n,\frac{2^n}{n+m}\right\}$  & \multirow{2}*{optimal up to $n$ when $m \le O(2^n/n^2)$} \\
      &[Thm. \ref{thm:QSP_expander}] &[Thm. \ref{thm:lower_exapnder_QSP}] &\\
     \hline
    \end{tabular}}
    \label{tab:QSP_graph}
\end{table}

Our results for QSP are summarized in Table \ref{tab:QSP_graph}. Now, we examine the effect of connectivity constraints on QSP circuits.

First, connectivity constraints make it harder to trade space for depth. Without connectivity constraints, tight bounds for QSP are known for any number $m$ of ancillary qubits~\cite{yuan2022optimal}: The optimal circuit depth is $O(n+2^n/(n+m))$ and the optimal size is $O(2^n)$. In particular, QSP circuit depth is polynomial (in fact, linear) in $n$ when sufficiently many ancilla are available. However, both constant-dimensional grid and $d$-ary tree constraints cause the required circuit depth to become exponential in $n$, regardless of the number of ancillary qubits. 

Second, more connectivity generally implies smaller depth, with the quantitative characterization depending on graphs. In both $d$-dimensional grids and $d$-ary trees, as $d$ grows larger (with the number of vertices roughly fixed), the diameter decreases, and the degree and expansion increase--- intuitively, the graph gets more connected. For grids, our results show that the circuit depth decreases with $d$, consistent with the intuition that greater connectivity enables shallower circuits. However, for $d$-ary trees, the required circuit depth increases slightly with $d$, reaching a maximum when $d$ takes its largest possible value (i.e, a star graph). This is because the size of a maximum matching also plays an important role in circuit depth---if the constraint graph does not contain a large matching, it limits how many two-qubit gates can be applied in parallel. Thus, it seems difficult to use one simple measure of graph connectivity to characterize its effect on circuit complexity.

\paragraph{Related work} 

\begin{table}[]
    \centering
    \caption{Previous circuit depths for $n$-qubit GUS and QSP under no qubit connectivity constraints.}
    \begin{tabular}{c|c|c|c}
       \hline
       Problem & Circuit depth & Number of ancilla $m$ & References\\
       \hline
      \multirow{2}*{QSP}  &$O\left(n+\frac{2^n}{n+m}\right)$& $m\ge 0$ &\cite{sun2021asymptotically,yuan2022optimal} \\
       & $\Omega\left(n+\frac{2^n}{n+m}\right)$& $m\ge 0$ &\cite{sun2021asymptotically}\\
       \hline
      \multirow{4}*{GUS}  &  $O\left(n2^n+\frac{4^n}{n+m}\right)$ & $m\ge 0$ & \cite{sun2021asymptotically}\\
             &  $O(n2^{n/2})$ & $m=\Theta(n4^n)$& \cite{rosenthal2021query}\\
            & $O\big(n2^{n/2}+\frac{n^{1/2}2^{3n/2}}{m^{1/2}}\big)$& $\Omega(2^n)\le m\le O(4^n/n)$ &\cite{yuan2022optimal}\\
             & $\Omega\left(n+\frac{4^n}{n+m}\right)$& $m\ge 0$ &\cite{sun2021asymptotically}\\
      \hline
    \end{tabular}
    \label{tab:related-work}
\end{table}
%
 The circuit sizes of $n$-qubit QSP and GUS are $\Theta(2^n)$ \cite{plesch2011quantum,bergholm2005quantum} and $\Theta(4^n)$ \cite{mottonen2005decompositions,shende2004minimal} in the absence of graph constraints, respectively. Circuit Depth for GUS and QSP in the absence of graph constraints has been widely investigated (see Table~\ref{tab:related-work}).
 There are some known circuit constructions for QSP and specific unitary synthesis under the path constraint. In \cite{mottonen2005decompositions}, the circuit size of any $n$-qubit uniformly controlled gate (UCG) and QSP circuit can be optimized to $O(2^n)$ under path constraint.
 { Ref.\cite{rosenbaum2013optimal} showed that the depth and size required for a general $n$-qubit-controlled 1-qubit gate are $\Theta(n^{1/k})$ and $\Theta(n)$, respectively, under $n^{1/k}\times \cdots \times n^{1/k}$ grid constraint. The paper also shows the same bounds for the Fanout operation with $n$ target qubits. Ref. \cite{herbert2018depth}  showed that there exist $n$-qubit circuits such that a multiplicative overhead of $\Omega(\log(n))$ on depth is needed under certain  constant-degree graph constraints, and there exist constant-degree graphs $G$ that such a logarithmic depth overhead is sufficient for any circuit on $G$.}
 
\paragraph{Organization} 
The rest of this paper is organized as follows. In Section \ref{sec:preliminaries}, we introduce notation and review some previous results. 
In Sections \ref{sec:diag_without_ancilla} and \ref{sec:diag_with_ancilla} we give circuit constructions for diagonal unitary matrices under various graph constraints, which are used in subsequent sections. We prove circuit depth and size upper bounds for QSP and GUS under various graph constraints in Section~\ref{sec:QSP_US_graph}, and prove corresponding lower bounds in Section \ref{sec:QSP_US_lowerbound}. We conclude in Section~\ref{sec:conlusion}.

\section{Preliminaries}
\label{sec:preliminaries}
\subsection{Notation}
Let $[n]$ denote the set $\{1,2,\cdots,n\}$. All logarithms $\log(\cdot)$ are taken base 2. Let $\mathbb{I}_n\in\mathbb{R}^{2^n\times 2^n}$ be the $n$-qubit identity operator. For any $x=x_1\cdots x_s\in\{0,1\}^s$, $y=y_1\cdots y_t\in\{0,1\}^t$, $xy$ denotes the $(s+t)$-bit string $x_1\cdots x_s y_1\cdots y_t\in\{0,1\}^{s+t}$. For $x=x_1\cdots x_n$, $ y=y_1y_2\cdots y_n$, the inner product of $x,y$ is $\langle x,y\rangle:=\oplus_{i=1}^n x_i\cdot y_i$, where  addition $\oplus$ and multiplication $\cdot$ are over the field $\mathbb{F}_2$. We use $x\oplus y$ to denote the bit-wise XOR of $x$ and $y$. For any quantum state $\ket{\psi}$ and qubit set $S$, $\ket{\psi}_S$ denotes the reduced quantum state corresponding to qubits in $S$. If $S=\{i\}$, we simply write $\ket{\psi}_{i}$ for $\ket{\psi}_{\{i\}}$. For sets $S$ and $T$, define $S-T:=\{x: x\in S \text{~and~} x\notin T\}$. 

An $n$-qubit quantum circuit implements a $2^n\times 2^n$ unitary transformation by a sequence of gates. The set of all single qubit gates and the 2-qubit CNOT gate can implement any unitary transformation, and is therefore said to be universal for quantum computation. We refer to circuits consisting of only these gates as \textit{standard quantum circuits}. All circuits in this paper are standard quantum circuits.

\subsection{Graph constraints}
Some implementation schemes of  real quantum computers have a notion of \textit{connectivity}. That is, two-qubit gates may only be implementable between certain pairs of qubits. This can be modelled by a graph $G = (V,E)$ with vertex and edge sets $V$ and $E$, respectively, where a two-qubit gate can be applied to qubits $(i,j)$ if and only if $(i,j)\in E$. We refer to $G$ as the \textit{constraint graph} of the circuit, and the corresponding circuit is said to be \textit{under $G$ constraint}. 
For any graph $G$,  $d_G(u,v)$ denotes the distance between vertices $u$ and $v$ in $G$, i.e, the number of edges on the shortest path from $u$ to $v$. The subscript $G$ is dropped when no confusion is caused. The diameter of $G$ is defined to be $\diam(G)\defeq \max_{u,v\in V}d(u,v)$. 


 A key question we consider is: 
What properties of the constraint graph influence quantum circuit complexity the most?
 To study this, we investigate three families of graphs: (i) grids, (ii) trees, and (iii) expanders. We also consider the general case with $m$ ancillary qubits available.
\paragraph{$d$-dimensional grids} These are graphs with vertex and edge sets:
\begin{align*}
V&=\{v_{i_1,i_2,\ldots,i_d}:\forall i_k\in[n_k],\forall k\in[d]\}, \\
E&=\left\{(v_{i_1,i_2,\ldots,i_d},v_{i_1+1,i_2,\ldots,i_d}), (v_{i_1,i_2,\ldots,i_d},v_{i_1,i_2+1,\ldots,i_d}),\ldots,(v_{i_1,i_2,\ldots,i_d},v_{i_1,i_2,\ldots,i_d+1}): \forall i_k\in[n_k-1], \forall k\in[d]\right\}.
\end{align*}


\begin{definition}
A quantum circuit on $n$ qubits will be said to be under $\Grid^{n_1,n_2, \ldots, n_d}_{n}$ constraint if the the constraint graph is a $d$-dimensional grid with $\prod_{k=1}^d n_k = n$. Without loss of generality, we assume that $n_1 \ge n_2 \ge \cdots\ge n_d$. We will refer to the case $d=1$ as $\Path_{n}$ (see Fig.~\ref{fig:n-path}). \end{definition}

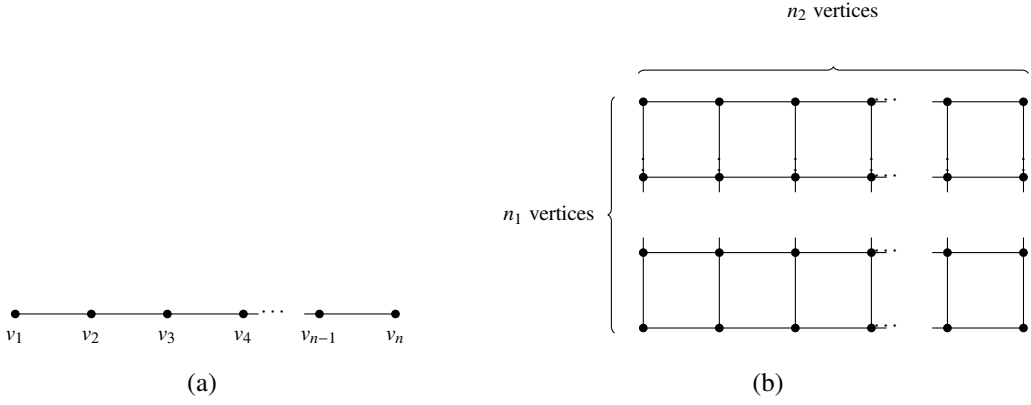
\begin{figure}[!t]
\centering
\subfloat[]{\centering
    \begin{tikzpicture}
        \draw (0,0) -- (3.2,0)
              (3.8,0) -- (5,0);
        
        \draw [fill=black] (0,0) circle (0.05)
                           (1,0) circle (0.05)
                           (2,0) circle (0.05)
                           (3,0) circle (0.05)
                           (4,0) circle (0.05)
                           (5,0) circle (0.05);
                           
         \draw (3.1,0) node[anchor=west]{\scriptsize $\cdots$};
         \draw (0,-0.3) node{\scriptsize $v_1$}
               (1,-0.3) node{\scriptsize $v_2$}
               (2,-0.3) node{\scriptsize $v_3$}
               (3,-0.3) node{\scriptsize $v_4$}
               (4,-0.3) node{\scriptsize $v_{n-1}$}
               (5,-0.3) node{\scriptsize $v_{n}$};
        \end{tikzpicture}%
\label{fig:n-path}}
\hfil
\subfloat[]{\centering
    \begin{tikzpicture}
     \node (a) at (-0.2,3.2) {};
     \node (b) at (5.2,3.2) {};
     \draw[decorate,decoration={brace,raise=5pt}] (a) -- (b);
     
     \node (a) at (-0.2,-0.2) {};
     \node (b) at (-0.2,3.2) {};
     \draw[decorate,decoration={brace,raise=5pt}] (a) -- (b);
    
    \draw (0,0) -- (3.2,0) (3.8,0) -- (5,0)
     (0,1) -- (3.2,1) (3.8,1) -- (5,1)
     (0,2) -- (3.2,2) (3.8,2) -- (5,2)
     (0,3) -- (3.2,3) (3.8,3) -- (5,3);
    
    \draw [fill=black] (0,0) circle (0.05) (1,0) circle (0.05) (2,0) circle (0.05)  (3,0) circle (0.05) (4,0) circle (0.05) (5,0) circle (0.05);
    \draw [fill=black] (0,1) circle (0.05) (1,1) circle (0.05) (2,1) circle (0.05)  (3,1) circle (0.05) (4,1) circle (0.05) (5,1) circle (0.05);
    \draw [fill=black] (0,2) circle (0.05) (1,2) circle (0.05) (2,2) circle (0.05)  (3,2) circle (0.05) (4,2) circle (0.05) (5,2) circle (0.05);
    \draw [fill=black] (0,3) circle (0.05) (1,3) circle (0.05) (2,3) circle (0.05)  (3,3) circle (0.05) (4,3) circle (0.05) (5,3) circle (0.05);
    
     \draw (2.9,0) node[anchor=west]{\scriptsize $\cdots$} (2.9,1) node[anchor=west]{\scriptsize $\cdots$} (2.9,2) node[anchor=west]{\scriptsize $\cdots$} (2.9,3)  node[anchor=west]{\scriptsize $\cdots$};
     
     \draw (0,0) -- (0,1.2) (0,1.8) -- (0,3) (1,0) -- (1,1.2) (1,1.8) -- (1,3) (2,0) -- (2,1.2) (2,1.8) -- (2,3) (3,0) -- (3,1.2) (3,1.8) -- (3,3) (4,0) -- (4,1.2) (4,1.8) -- (4,3) (5,0) -- (5,1.2) (5,1.8) -- (5,3);
     
     \draw (0,2.6) node[anchor=north]{\scriptsize $\vdots$} (1,2.6) node[anchor=north]{\scriptsize $\vdots$} (2,2.6) node[anchor=north]{\scriptsize $\vdots$} (3,2.6) node[anchor=north]{\scriptsize $\vdots$} (4,2.6) node[anchor=north]{\scriptsize $\vdots$} (5,2.6) node[anchor=north]{\scriptsize $\vdots$};
     
    \draw (2.5,4.2) node{\scriptsize $n_2$ vertices} (-1.2,1.5) node{\scriptsize $n_1$ vertices~~};
    \end{tikzpicture}%
\label{fig:m1-m_2-grid}}
\caption{Examples of $d$-dimensional grids with $d=1$ and $2$. (a) The $1$-dimensional path $\Path_{n}$. (b) The 2-dimensional grid $\Grid_{n}^{n_1,n_2}$. }
    \label{fig:d-grids}
\end{figure}

\paragraph{$d$-ary trees} The complete $d$-ary tree is a tree in which every non-leaf node has exactly $d$ children (see Fig.~\ref{fig:tree}). 

\begin{definition}
A quantum circuit on $n$ qubits will be said to be under $\Tree_{n}(d)$ constraint if the constraint graph is a $d$-ary tree with $\sum_{i=0}^h d^i=n$. $\Tree_{n}(2)$ corresponds to a binary tree and the case $d = n -1$ will be denoted $\Star_{n}$.
\end{definition}
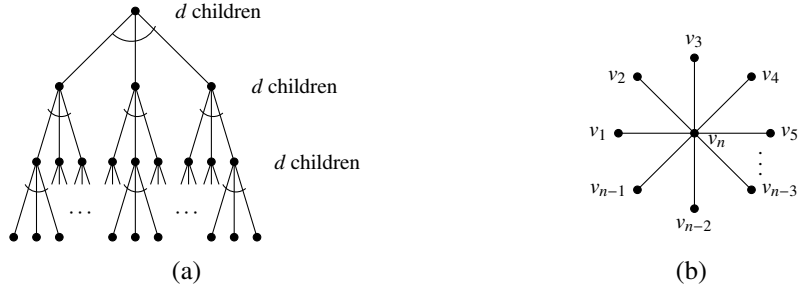
\begin{figure}[!t]
\centering
\subfloat[]{\centering
    \begin{tikzpicture}
       \draw [fill=black] (0,0) circle (0.05) (-1,-1) circle (0.05) (0,-1) circle (0.05) (1,-1) circle (0.05);
       \draw [fill=black] (-1.3,-2) circle (0.05) (-1,-2) circle (0.05) (-0.7,-2) circle (0.05) (-0.3,-2) circle (0.05) (0,-2) circle (0.05) (0.3,-2) circle (0.05) (0.7,-2) circle (0.05) (1,-2) circle (0.05) (1.3,-2) circle (0.05);
       \draw [fill=black] (-1.6,-3) circle (0.05) (-1.3,-3) circle (0.05) (-1,-3) circle (0.05) (-0.3,-3) circle (0.05) (0,-3) circle (0.05) (0.3,-3) circle (0.05) (1.6,-3) circle (0.05) (1.3,-3) circle (0.05) (1,-3) circle (0.05);
       \draw (0,0)--(-1,-1) (0,0)--(0,-1) (0,0)--(1,-1);
       \draw (-1,-1)--(-1.3,-2) (-1,-1)--(-1,-2) (-1,-1)--(-0.7,-2);
       \draw (0,-1)--(-0.3,-2) (0,-1)--(0,-2) (0,-1)--(0.3,-2);
       \draw (1,-1)--(0.7,-2) (1,-1)--(1,-2) (1,-1)--(1.3,-2);
       \draw (-1.3,-2)--(-1.6,-3) (-1.3,-2)--(-1.3,-3) (-1.3,-2)--(-1,-3);
       \draw (0,-2)--(-0.3,-3) (0,-2)--(0,-3) (0,-2)--(0.3,-3);
       \draw (1.3,-2)--(1.6,-3) (1.3,-2)--(1.3,-3) (1.3,-2)--(1,-3);
       \draw (-1,-2)--(-1.1,-2.3) (-1,-2)--(-1,-2.3) (-1,-2)--(-0.9,-2.3) ;
       \draw (-0.7,-2)--(-0.8,-2.3) (-0.7,-2)--(-0.7,-2.3) (-0.7,-2)--(-0.6,-2.3) ;
       \draw (-0.3,-2)--(-0.4,-2.3) (-0.3,-2)--(-0.3,-2.3) (-0.3,-2)--(-0.2,-2.3);
       \draw (0.3,-2)--(0.2,-2.3) (0.3,-2)--(0.3,-2.3) (0.3,-2)--(0.4,-2.3);
       \draw (0.7,-2)--(0.6,-2.3) (0.7,-2)--(0.7,-2.3) (0.7,-2)--(0.8,-2.3);
       \draw (1,-2)--(0.9,-2.3) (1,-2)--(1,-2.3) (1,-2)--(1.1,-2.3);
        \draw (-0.7,-2.7) node{\scriptsize $\cdots$} (0.7,-2.7) node{\scriptsize $\cdots$};
        \draw (-0.3,-0.3) arc(-135:-25:0.35);
        \draw (-0.15,-1.3) arc(-150:-50:0.2);
        \draw (-1.15,-1.3) arc(-150:-50:0.2);
        \draw (0.85,-1.3) arc(-150:-50:0.2);
        \draw (-1.45,-2.3) arc(-150:-50:0.2);
        \draw (1.15,-2.3) arc(-150:-50:0.2);
        \draw (-0.15,-2.3) arc(-150:-50:0.2);
         \draw (1.1,0) node{\scriptsize $d$ children};
         \draw (2.1,-1) node{\scriptsize $d$ children};
        \draw (2.4,-2) node{\scriptsize $d$ children};
    \end{tikzpicture}
\label{fig:darytree}}
\hfil
\subfloat[]{\centering
    \begin{tikzpicture}
        \draw (-1,0) -- (1,0) (0,-1) -- (0,1) (0.75, 0.75) -- (-0.75, -0.75) (-0.75, 0.75) -- (0.75, -0.75);
        \draw [fill=black] (0,0) circle (0.05) (-1,0) circle (0.05) (1,0) circle (0.05) (0,-1) circle (0.05) (0,1) circle (0.05) (0.75, 0.75) circle (0.05) (-0.75, -0.75) circle (0.05)  (-0.75, 0.75) circle (0.05)  (0.75, -0.75) circle (0.05);
        \draw (0.3,-0.1) node{\scriptsize $v_{n}$};
        \draw (-1,0) [anchor=east] node{\scriptsize $v_1$};
        \draw (-0.75, -0.75) [anchor=east] node{\scriptsize $v_{n-1}$};
        \draw (-0.75, 0.75) [anchor=east] node{\scriptsize $v_{2}$};
        \draw (0,1) [anchor=south] node{\scriptsize $v_{3}$};
        \draw (1,0) [anchor=west] node{\scriptsize $v_5$};
        \draw (0.75, 0.75) [anchor=west] node{\scriptsize $v_{4}$};
        \draw (0.75, -0.75) [anchor=west] node{\scriptsize $v_{n-3}$};
        \draw (0,-1) [anchor=north] node{\scriptsize $v_{n-2}$};
        \draw (0.7, -0.3) [anchor=west] node{\scriptsize $\vdots$};
   \end{tikzpicture}
\label{fig:star_n}}
\caption{Examples of $d$-ary trees. (a) The general $d$-ary tree $\Tree_{n}(d)$. (b) The $(n)$-star graph $\Star_{n}$. }
    \label{fig:tree}
\end{figure}

\paragraph{Expander graphs}
\begin{definition}[Vertex expansion]\label{def:expansion}
The vertex expansion of $G$ is defined as
\[
h_{out}(G):=\min_{S\subseteq V,~0<|S|<|V|/2}|\partial_{out}(S)|/|S|,
\]
where $\partial_{out}(S):=\{v\in V-S : \exists u\in S, \text{~s.t.~} (u,v)\in E\}$.
\end{definition}

An expander is a graph $G$ such that $h_{out}(G)\ge c$ for some constant $c>0$. 
\begin{definition}
A quantum circuit on $n$ qubits will be said to be under $\Expander_{n}$ constraint if the constraint graph is an $n$-vertex expander.
\end{definition}

\paragraph{Examples of constraint graphs}
Connectivity for a number of superconducting processors can be expressed in terms of these graphs: 

\begin{itemize}
    \item $\Path_n$: IBM's Falcon r5.11L chip \cite{IBMQ}, 9-qubit chip \cite{kelly2015state}.
    \item $\Grid^{2, n}_{2n}$ (i.e., bilinear chain): IBM's Melbourne chip \cite{IBMQ}, USTC's 24-qubit chip \cite{ye2019propagation}.
    \item $\Grid_{n_1n_2}^{n_1, n_2}$: Google's Sycamore chip \cite{arute2019quantum,acharya2022suppressing} 
    , USTC's Zuchongzhi chip \cite{gong2021quantum}. 
     \item $\Tree_n(2)$ :  IBM's Falcon r5.11H chips \cite{IBMQ}.
\end{itemize}
In addition to these, several other constraint graphs are also encountered in practice:
\begin{itemize}
    \item Brick-wall: IBM's Falcon r8/Falcon r4/Falcon r5.10/Falcon r5.11/
Hummingbird r3/
Eagle r1 chips \cite{IBMQ}. 
    \item T-shape: IBM's Falcon r4T chips \cite{IBMQ}.
\end{itemize}
Of particular note is the brick-wall structure, which we briefly describe below.

\paragraph{Brick-walls} For integers $n_1,n_2\ge 1$, $b_1\ge 2$, $b_2\ge 3$ and $b_2$ odd, the $(n_1,n_2,b_1,b_2)$-brick-wall $\brickwall_{n}^{n_1,n_2,b_1,b_2}$ graph is divided into $n_1$ layers, with each layer containing $n_2$ `bricks', and each brick a rectangle containing $b_1$ vertices on `vertical' edges and $b_2$ vertices on `horizontal' edges (see Fig.~\ref{fig:brickwall}). Brick-wall $\brickwall_{n}^{n_1,n_2,b_1,b_2}$ contains $n$ vertices. In IBM's brick-wall chips, $b_1=3$ and $b_2=5$. 
\begin{figure}[!hbt]
    \centering
    \begin{tikzpicture}
    \filldraw[fill=llgray] (4,0)--(4,0.5)--(6,0.5)--(6,0)--cycle;
    
       \draw [fill=black] (0,0) circle (0.05) (0.5,0) circle (0.05) (1,0) circle (0.05) (1.5,0) circle (0.05) (2,0) circle (0.05) (2.5,0) circle (0.05) (3,0) circle (0.05) (3.5,0) circle (0.05) (4,0) circle (0.05) (4.5,0) circle (0.05) (5,0) circle (0.05) (5.5,0) circle (0.05) (6,0) circle (0.05);
        \draw [fill=black] (0,0.25) circle (0.05) (2,0.25) circle (0.05)  (4,0.25) circle (0.05) (6,0.25) circle (0.05);
        \draw [fill=black] (0,0.5) circle (0.05) (0.5,0.5) circle (0.05) (1,0.5) circle (0.05) (1.5,0.5) circle (0.05) (2,0.5) circle (0.05) (2.5,0.5) circle (0.05) (3,0.5) circle (0.05) (3.5,0.5) circle (0.05) (4,0.5) circle (0.05) (4.5,0.5) circle (0.05) (5,0.5) circle (0.05) (5.5,0.5) circle (0.05) (6,0.5) circle (0.05);
     \draw [fill=black] (0,1) circle (0.05) (0.5,1) circle (0.05) (1,1) circle (0.05) (1.5,1) circle (0.05) (2,1) circle (0.05) (2.5,1) circle (0.05) (3,1) circle (0.05) (3.5,1) circle (0.05) (4,1) circle (0.05) (4.5,1) circle (0.05) (5,1) circle (0.05) (5.5,1) circle (0.05) (6,1) circle (0.05);
     \draw [fill=black] (1,0.75) circle (0.05) (3,0.75) circle (0.05) (5,0.75) circle (0.05);
    \draw [fill=black] (0,1.25) circle (0.05) (2,1.25) circle (0.05)  (4,1.25) circle (0.05) (6,1.25) circle (0.05);
     \draw [fill=black] (0,1.5) circle (0.05) (0.5,1.5) circle (0.05) (1,1.5) circle (0.05) (1.5,1.5) circle (0.05) (2,1.5) circle (0.05) (2.5,1.5) circle (0.05) (3,1.5) circle (0.05) (3.5,1.5) circle (0.05) (4,1.5) circle (0.05) (4.5,1.5) circle (0.05) (5,1.5) circle (0.05) (5.5,1.5) circle (0.05) (6,1.5) circle (0.05);
    \draw [fill=black] (1,1.75) circle (0.05) (3,1.75) circle (0.05) (5,1.75) circle (0.05);
     \draw [fill=black] (0,2) circle (0.05) (0.5,2) circle (0.05) (1,2) circle (0.05) (1.5,2) circle (0.05) (2,2) circle (0.05) (2.5,2) circle (0.05) (3,2) circle (0.05) (3.5,2) circle (0.05) (4,2) circle (0.05) (4.5,2) circle (0.05) (5,2) circle (0.05) (5.5,2) circle (0.05) (6,2) circle (0.05);
        \draw [fill=black] (0,2.25) circle (0.05) (2,2.25) circle (0.05)  (4,2.25) circle (0.05) (6,2.25) circle (0.05);
     \draw [fill=black] (0,2.5) circle (0.05) (0.5,2.5) circle (0.05) (1,2.5) circle (0.05) (1.5,2.5) circle (0.05) (2,2.5) circle (0.05) (2.5,2.5) circle (0.05) (3,2.5) circle (0.05) (3.5,2.5) circle (0.05) (4,2.5) circle (0.05) (4.5,2.5) circle (0.05) (5,2.5) circle (0.05) (5.5,2.5) circle (0.05) (6,2.5) circle (0.05);
     \draw (0,0)--(6,0) (0,0.5)--(6,0.5) (0,1)--(6,1) (0,1.5)--(6,1.5) (0,2)--(6,2) (0,2.5)--(6,2.5);
     \draw (0,0)--(0,0.5) (0,1)--(0,1.5) (0,2)--(0,2.5) (2,0)--(2,0.5) (2,1)--(2,1.5) (2,2)--(2,2.5) (4,0)--(4,0.5) (4,1)--(4,1.5) (4,2)--(4,2.5) (6,0)--(6,0.5) (6,1)--(6,1.5) (6,2)--(6,2.5);
     \draw (1,0.5)--(1,1) (1,1.5)--(1,2) (3,0.5)--(3,1) (3,1.5)--(3,2) (5,0.5)--(5,1) (5,1.5)--(5,2);

     \node (a) at (-0.2,2.5) {};
     \node (b) at (6.2,2.5) {};
     \draw[decorate,decoration={brace,raise=5pt}] (a) -- (b);
      \draw (3,3.2) node{\scriptsize $n_2$ bricks in each layer};
    \node (a) at (0,-0.2) {};
     \node (b) at (0,2.7) {};
     \draw[decorate,decoration={brace,raise=5pt}] (a) -- (b);
      \draw (-2.1,1.25) node{\scriptsize $n_1$ layers of bricks};
    \node (a) at (3.8,-0.05) {};
     \node (b) at (6.2,-0.05) {};
     \draw[decorate,decoration={brace,raise=5pt,mirror}] (a) -- (b);
      \draw (5,-0.7) node{\scriptsize $b_2$ vertices};
          \node (a) at (6,-0.2) {};
     \node (b) at (6,0.7) {};
     \draw[decorate,decoration={brace,raise=5pt,mirror}] (a) -- (b);
      \draw (7,0.25) node{\scriptsize $b_1$ vertices};
    \end{tikzpicture}
   \caption{The brick-wall graph $\brickwall_{n}^{n_1,n_2,b_1,b_2}$.}
    \label{fig:brickwall}
\end{figure}
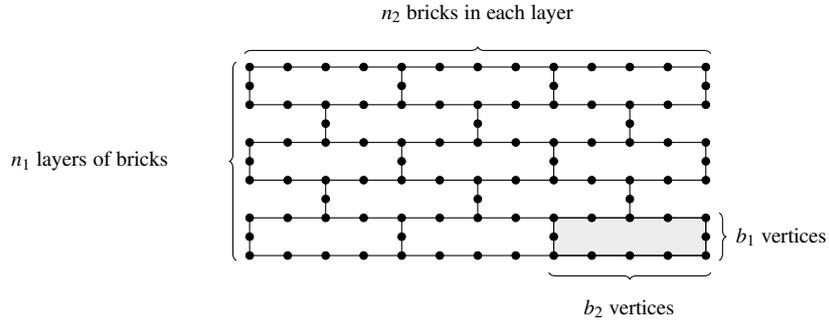

\begin{definition}A quantum circuit on $n$ qubits will be said to be under $\brickwall^{n_1, n_2, b_1, b_2}_{n}$ constraint if the the constraint graph is an $(n_1, n_2, b_1, b_2)$-brick-wall.
\end{definition}

While brick-walls lie outside the families of graphs we consider, in Section~\ref{sec:circuit_transformation_sub} we show that our results for the $2$-dimensional grid can be used to construct a circuit for brick-wall graphs with similar bounds.

\subsection{Gray codes}\label{subsec:gray}
An $n$-bit Gray code is an ordering of all $2^n$ $n$-bit strings such that any two successive strings differ in exactly one bit, as do the first and the last strings. 
An explicit construction uses the ruler function $\zeta(n)=\max\{k: 2^{k-1}|n \}$ as follows. 
It is not hard to verify that for all $k\in[n]$, there are $2^{n-k}$ elements $i \in[2^n-1]$ such that $\zeta(i)=k$. For all $i\in[n]$ and $j\in[2^n]$, define $h_{ij}$ as
\begin{equation}\label{eq:index}
    h_{ij} = (\zeta(j-1)+i-2 \mod n )+1, \text{ with } \zeta(0)\defeq 0.
\end{equation}
It is straightforward to show that
\begin{align}\label{eq:h1j}
 &h_{i1} = \begin{cases} n & \text{if } i=1, \\ i-1 & \text{if } 2\le i\le n, \end{cases} ~ \text{and} \quad 
h_{1j} = \begin{cases} n & \text{if } j=1, \\ \zeta(j-1) & \text{if } 2\le j\le n. \end{cases}   
\end{align}
For each $i\in [k]$, one can make use of $h_{ij}$ to construct an $n$-bit Gray code, as follows.
\begin{lemma}[\cite{frank1953pulse,savage1997survey,gilbert1958gray}]\label{lem:GrayCode}
For any $i\in[n]$, construct $n$-bit strings $c^i_1,c^i_2,\cdots,c^i_{2^{n}-1}, c^i_{2^{n}}$ as follows: Let $c_1^{i}=0^{n}$, and for each $j = 2, 3, \ldots, 2^{n}$, string $c_j^i$ is obtained by flipping the $h_{ij}$-th bit of $c_{j-1}^i$. The following properties hold.
\begin{enumerate}
    \item $c^i_1,c^i_2,\cdots,c^i_{2^{n}-1}, c^i_{2^{n}}$ are all distinct and form an $n$-bit Gray code: for $j\ge 2$, each $c^i_j$ differs from $c^i_{j-1}$ in the $h_{ij}$-th bit, and $c_1^i$ and $c_{2^{n}}^i$ differ in the $h_{i1}$-th bit. 
    \item For each $k\in [n]$, there are $2^{n-k}$ elements $j\in \{2,3,\ldots,2^{n}\}$ such that $h_{ij} = (k+i-2 \mod n)+1$. In particular, there are $2^{n-k}$ elements $j\in \{2,3,\ldots,2^{n}\}$ such that $h_{1j} = k$.
\end{enumerate} 
\end{lemma}
We refer to this ordered sequence $c_1^i,c_2^i,\cdots, c_{2^{n}}^i$ as an $(n,i)$-Gray code, or simply an $i$-Gray code if $n$ is clear from context. 

\subsection{Quantum gates and circuits} 
\label{subsec:qgates}
For arbitrary $ \theta\in\mathbb{R}$, single-qubit gates $R(\theta)$ are defined as 
\begin{equation}\label{eq:rotation}
R(\theta)=\left(\begin{array}{cc}
   1  & 0 \\
   0  & e^{i\theta}
\end{array}\right).
\end{equation}
Two special cases that will be used later are the phase gate $S$ and the Hadamard gate $H$
\[~S=\left(\begin{array}{cc}
  1   &  \\
    &  i
\end{array}\right), \quad \quad 
~H=\frac{1}{\sqrt{2}}\left(\begin{array}{cc}
  1   & 1 \\
   1  &  -1
\end{array}\right).
\] 
A CNOT gate on qubits $u$ and $v$, denoted $\Cnot^u_v$, effects the transformation \[\Cnot^u_v\ket{x}_u\ket{y}_v=\ket{x}_u\ket{x\oplus y}_v\] for $x,y\in \B$. Here, $u$ is referred to as the \textit{control} and $v$ the \textit{target}. 


\begin{restatable}{lem}{cnotpath}
\label{lem:cnot_path_constraint}
${\sf CNOT}_v^u$ can be implemented by a CNOT circuit of depth and size $O(d(u,v))$ under arbitrary graph constraint, where $d(u,v)$ is the minimum distance between vertices $u$ and $v$ in $G$.
\end{restatable}
The proof of Lemma \ref{lem:cnot_path_constraint} is given in Appendix \ref{append:basic_circuit}.
We call a quantum circuit consisting of only CNOT gates a \textit{CNOT circuit}. An $n$-qubit invertible linear transformation over $\mathbb{F}_2$ can be implemented by an efficient $n$-qubit CNOT circuit:

\begin{lemma}[\cite{wu2019optimization}]\label{lem:cnot_circuit}
Let $G_\delta$ be a connected graph with $n$ vertices and  minimum degree $\delta$. Any $n$-qubit invertible linear transformation can be implemented in circuit depth and size $O(n^2/\log (\delta))$ under $G_\delta$ constraint.
\end{lemma}



\subsection{Quantum state preparation and general unitary synthesis}
Two key tasks addressed in this paper are:

\paragraph{Quantum state preparation (QSP)} Given a vector $v=(v_x)_{x\in\{0,1\}^n}\in\mathbb{C}^{2^n}$ where $\sqrt{\sum_{x\in\{0,1\}^n}|v_x|^2}=1$, prepare the corresponding $n$-qubit quantum state 
    \[\ket{\psi_v}=\sum_{x\in\{0,1\}^n}v_x\ket{x}\]
    by a standard quantum circuit, starting from initial state $\ket{0^n}$. We shall refer to such a circuit as a \textit{QSP circuit}.

\paragraph{General unitary synthesis (GUS)} Given an $n$-qubit unitary $U=[u_{xy}]_{x,y\in\{0,1\}^n}\in\mathbb{C}^{2^n\times 2^n}$, construct a standard quantum circuit for $U$. We call such a circuit a \textit{GUS circuit}.

QSP and GUS circuits may make use of ancilla. In this case, we say that:
\begin{enumerate}
    \item A circuit $C_{\rm QSP}$ with $m$ ancillary qubits solves the QSP problem if 
    \[C_{\rm QSP}\ket{0^n}\ket{0^m}=\ket{\psi_v}\ket{0^m}.\]
    \item A circuit $C_{\rm GUS}$ with $m$ ancillary qubits solves the GUS problem if
    \[C_{\rm GUS}\ket{x}\ket{0^m}=(U\ket{x})\ket{0^m},\quad\forall x\in\Bn.\]
\end{enumerate}
    
    

\subsection{Uniformly controlled gates and diagonal unitary matrices} \label{sec:ucg}
Given single-qubit unitary matrices $U_1$, $U_2$, $\ldots$, $U_{2^{n-1}-1}$, $U_{2^{n-1}} \in\mathbb{C}^{2\times 2}$, an $n$-qubit \textit{uniformly controlled gate} (UCG) $V_n$ is a block diagonal matrix given by
\begin{equation}\label{eq:UCG}
V_n=\left(\begin{array}{cccc}
     U_1 &  & &\\
     &  U_2& &\\
     &  & \ddots &\\
     &  & & U_{2^{n-1}}
\end{array}\right)\in\mathbb{C}^{2^n\times 2^n}.
\end{equation}
That is, conditioned on the state of the first $n-1$ qubits, $V_n$ applies the corresponding $U_i$ operation to the $n$-th qubit.
Any $n$-qubit \textit{diagonal unitary matrix} $\Lambda_n$ can be expressed as
\begin{equation*}
\Lambda_n=\left(\begin{array}{cccc}
     1&  & &\\
     &  e^{i\theta_1}& &\\
     &  & \ddots &\\
     &  & & e^{i\theta_{2^n-1}}
\end{array}\right)\in\mathbb{C}^{2^n\times 2^n},
\end{equation*}
where $\theta_1,\ldots, \theta_{2^n-1} \in\mathbb{R}$. As quantum states that differ only by a global phase are indistinguishable, without loss of generality we may set the first entry to 1.
{The task of implementing UCGs can be reduced to that of implementing diagonal unitary matrices: 
\begin{lemma}[\cite{sun2021asymptotically}]\label{lem:UCG_decomposition}
Any $n$-qubit UCG $V_n$ can be decomposed as
$V_n=\Lambda_n'''(\mathbb{I}_{n-1}\otimes (S H))\Lambda_n'' (\mathbb{I}_{n-1}\otimes (HS^\dagger))\Lambda_n',$
where $\Lambda_n',\Lambda_n'',\Lambda_n'''\in\mathbb{C}^{2^n\times 2^n}$ are $n$-qubit diagonal unitary matrices. 
\end{lemma}}


To implement $\Lambda_n$, which can be represented as
\begin{equation}\label{eq:diag}
  \ket{x}\to e^{i\theta(x)}\ket{x},~\forall x\in \Bn-\{0^n\},  
\end{equation}
in a quantum circuit, it suffices to accomplish the following two tasks.
\begin{enumerate}
    \item For every $s\in \{0,1\}^n-\{0^n\}$, effect a phase shift of $\alpha_s$ on each basis $\ket{x}$ with $\langle s,x\rangle = 1$, i.e. 
    \begin{equation}\label{eq:task1}
     \ket{x} \to  e^{i\alpha_s\langle s,x\rangle } \ket{x}.
    \end{equation}
    
    \item Find $\{\alpha_s:s\in \Bn-\{0^n\}\}$ s.t. \begin{equation}\label{eq:alpha}
        \sum_{s\in \{0,1\}^n-\{0^n\}}\alpha_s\langle x,s\rangle = \theta(x), \quad \forall x\in \{0,1\}^n-\{0^n\}.
    \end{equation}
\end{enumerate}
Combining the two gives
\[\ket{x} \to \prod_{s\in \{0,1\}^n-\{0^n\}} e^{i\alpha_s\langle s,x\rangle } \ket{x} = e^{i\sum_s\alpha_s \langle s,x\rangle} \ket{x} = e^{i\theta(x) } \ket{x},\]  
as required. For notational convenience, define $\alpha_{0^n}=0$.
For any $x\in\Bn$, if we generate a state $\ket{\langle s,x\rangle}$ on a qubit, apply $R(\alpha_s)$ on it, and restore this qubit, then Task 1 is implemented. We  call the process of generating $\ket{\langle s,x\rangle}$ \textit{generating $s$}. Given $\{\theta(x):x\in \B-\{0^n\}\}$, the values $\{\alpha_s:s\in \Bn-\{0^n\}\}$ in Task 2 can be efficiently found by the Walsh-Hadamard transform \cite{sun2021asymptotically}. After we generate all $s\in\Bn-\{0^n\}$, apply $R(\alpha_s)$ on $\ket{\langle s,x\rangle}$ and restore the qubits, we have implemented $\Lambda_n$ (by Eq. \eqref{eq:alpha}).

\section{Circuits for diagonal unitary matrices under qubit connectivity constraints, without ancillary qubits}
\label{sec:diag_without_ancilla}
Here we present circuits for diagonal unitary matrices $\Lambda_n$ under graph constraints without ancillary qubits, which are used in Section \ref{sec:QSP_US_graph} to construct QSP and GUS circuits.

\subsection{Circuit framework}
\label{sec:diag_without_ancilla_framework_main}


We construct a circuit based on the framework of~\cite{sun2021asymptotically}, modified to minimize additional overhead costs when graph constraints are imposed. 
Additional details and omitted proofs from this section are given in Appendix \ref{sec:diag_without_ancilla_framework}.

\paragraph{Old method}
In~\cite{sun2021asymptotically}, an $n$-bit string $s$ is divided into two parts: an $r_c$-bit prefix and an $r_t$-bit suffix, where $r_c=\lceil n/2\rceil$ and $r_t=\lfloor n/2\rfloor$. The suffix set $\B^{r_t}-\{0^{r_t}\}$ is itself divided into $\ell~ \le \frac{2^{r_t+2}}{r_t+1}-1$ sets $T^{(1)}, \ldots, T^{(\ell)}$, each of size  $r_t$. The process of generating all $s\in\Bn-\{0^n\}$ consists of $\ell$ phases, where the $i$-th phase generates all $n$-bit strings with suffixes in $T^{(i)}$. For bit strings ending with the $j$-th suffix in $T^{(i)}$, prefixes are enumerated in the order of a $j$-Gray code. The prefixes are implemented by CNOT gates where the control qubit lies in the first $r_c$ qubits and the target qubit lies in the last $r_t$ qubits. Bit strings with suffix $0^{r_t}$ need special treatment for technical reasons and are handled by recursive generation.

Unfortunately, this framework is inefficient under qubit connectivity constraints. More specifically, when we generate $r_t$ prefixes simultaneously by Gray code, we apply $r_t$ CNOT gates which cannot be implemented in parallel and impose an overhead of $O(n^2)$ to the circuit depth. To resolve this issue, we choose different lengths of prefixes (and suffixes) under different constraint graphs and rearrange the positions of the control and target qubits to minimize the number of controlled operations that involve distant qubits on the graph. 

\paragraph{New method}
Our circuit framework for $\Lambda_n$ is shown in Fig. \ref{fig:diag_without_ancilla_framwork}. The $n$ input qubits are labelled $1,2,\cdots,n$, and are divided into two registers: control register $\textsf{C}$ and target register $\textsf{T}$, with sizes $r_c$ and $r_t:=n-r_c$, respectively, where $[n]=\textsf{C}\cup \textsf{T}$. Compared to \cite{sun2021asymptotically}, our construction differs in two main ways:
\begin{enumerate}
    \item The design of registers $\textsf{C}$ and $\textsf{T}$. In~\cite{sun2021asymptotically}, $\textsf{C}$ and $\textsf{T}$ are specified as the first $\lceil n/2\rceil$ and the last $\lfloor n/2\rfloor$ qubits, respectively. In this work, $\textsf{C}$ and $\textsf{T}$ depend on the constraint graph (details specified in the following sections): they do not always have sizes $r_c = r_t = n/2$, and their positions are determined by a transformation $\Pi$ which permutes the first $r_c$ and the last $r_t$ qubits. 
    \item The choice of Gray codes. The implementation of the $C_i$ operators involves choosing $r_t$ Gray codes, specified by integers $j_1, \ldots, j_{r_t}$. The choice of Gray codes determines the sequence of qubits which act as the control for CNOT operations required to implement $C_k$. If two or more integers $j_i$ are the same, this corresponds to the same control qubit used for CNOT operations acting on different target qubits. In~\cite{sun2021asymptotically} the Gray codes used correspond to choosing $j(i)=i$. Here, by carefully choosing the $j_i$, accounting for the graph connectivity and the choice of $\sf C$ and $\sf T$, we can achieve a reduction in circuit depth.
\end{enumerate}


\begin{figure}[]
    
    \centerline{
  \Qcircuit @C=1em @R=0.5em {
  \lstick{\scriptstyle\ket{x_1}} &\multigate{5}{\scriptstyle\Pi} &\multigate{5}{\scriptstyle C_1} &\multigate{5}{\scriptstyle C_2} &\qw & {\scriptstyle\cdots~~~} & \multigate{5}{\scriptstyle C_\ell}   & \multigate{5}{\scriptstyle \Pi^{\dagger}} & \multigate{2}{\scriptstyle \Lambda_{r_c}} &\qw  \\
  \lstick{\scriptstyle\vdots~} & \ghost{\scriptstyle\Pi} &\ghost{\scriptstyle C_1} &\ghost{\scriptstyle C_2} &\qw & {\scriptstyle\cdots~~~} & \ghost{\scriptstyle C_\ell}  & \ghost{\scriptstyle \Pi^{\dagger}} & \ghost{\scriptstyle \Lambda_{r_c}} &\qw \\
  \lstick{\scriptstyle\ket{x_{r_c}}} & \ghost{\scriptstyle\Pi} &\ghost{\scriptstyle C_1} &\ghost{\scriptstyle C_2} &\qw & {\scriptstyle\cdots~~~} & \ghost{\scriptstyle C_\ell}  & \ghost{\scriptstyle \Pi^{\dagger}} & \ghost{\scriptstyle \Lambda_{r_c}} &\qw\\
  \lstick{\scriptstyle\ket{x_{r_c+1}}} & \ghost{\scriptstyle\Pi}&\ghost{\scriptstyle C_1} &\ghost{\scriptstyle C_2}&\qw & {\scriptstyle\cdots~~~} &  \ghost{\scriptstyle C_\ell}  & \ghost{\scriptstyle \Pi^{\dagger}} & \multigate{2}{\scriptstyle \mathcal{R}} &\qw \\
  \lstick{\scriptstyle\vdots~} & \ghost{\scriptstyle\Pi} &\ghost{\scriptstyle C_1} &\ghost{\scriptstyle C_2}&\qw & {\scriptstyle\cdots~~~}&  \ghost{\scriptstyle C_\ell}  & \ghost{\scriptstyle \Pi^{\dagger}} & \ghost{\scriptstyle \mathcal{R}} &\qw \\
  \lstick{\scriptstyle\ket{x_n}} & \ghost{\scriptstyle\Pi} &\ghost{\scriptstyle C_1} &\ghost{\scriptstyle C_2} &\qw & {\scriptstyle\cdots~~~} & \ghost{\scriptstyle C_\ell} & \ghost{\scriptstyle \Pi^{ \dagger}} & \ghost{\scriptstyle \mathcal{R}} &\qw  \\
  }
  }
    \caption{Circuit framework for implementing diagonal unitaries $\Lambda_n$ under graph constraints, without ancilla. Control and target register sizes, $r_c$ and $r_t$, respectively, depend on the graph constraint.  Operations $\Pi$ and $\Pi^\dagger$ are used to move control and target qubits close together. Our implementation of each $C_i$ differs from that in \cite{sun2021asymptotically}, and needs to be adapted to different constraint graphs. $\ell\le \frac{2^{r_t+2}}{r_t+1}-1$. }
    \label{fig:diag_without_ancilla_framwork}
\end{figure}
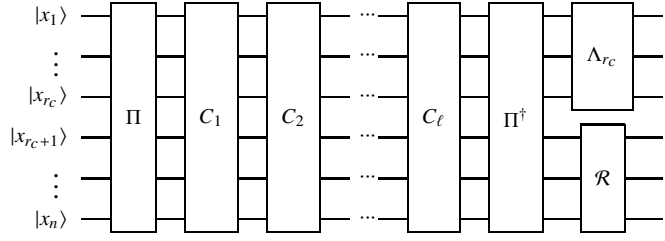

To describe the operators $\Pi$, $C_1,\cdots, C_\ell$, $\Lambda_{r_c}$ and $\mathcal{R}$ in Fig. \ref{fig:diag_without_ancilla_framwork},  recall the following result in \cite{sun2021asymptotically}. 
For some integer $\ell \le \frac{2^{r_t+2}}{r_t+1}-1$, there exists sets $T^{(1)},T^{(2)},\cdots,T^{(\ell)}$ of size $r_t$, such that
the Boolean vectors in $T^{(i)}=\{{t^{(i)}_1},{t^{(i)}_2},\cdots,{t^{(i)}_{r_t}}\}\subseteq \{0,1\}^{r_t}-\{0^{r_t}\}$ are linearly independent over $\mathbb{F}_2$ and $\bigcup_{i\in[\ell]} T^{(i)}= \{0,1\}^{r_t} - \{0^{r_t}\} $.



For each $k\in[\ell]\cup\{0\}$, define an $r_t$-qubit state in register $\textsf{T}$:
\begin{equation} \label{eq:yk}
\ket {y^{(k)}}_\textsf{T} = \ket{y_1^{(k)}\cdots y_{r_t}^{(k)}}_\textsf{T}, 
    y_j^{(k)} =\left\{\begin{array}{ll}
       x_{r_c+j}  & \text{if~} k=0, \\
       \langle {0^{r_c}t_j^{(k)}},x\rangle  &  \text{if~} k\in[\ell].
    \end{array}\right.
\end{equation}
Next, define disjoint sets $F_1,\ldots,F_\ell$ from $T^{(1)},\ldots, T^{(\ell)}$ by removing duplicates. 
\begin{equation}\label{eq:F_k}
\left\{\begin{array}{ll}
     F_1=\{ct:\ t\in T^{(1)},c\in\{0,1\}^{r_c}\}, & \\
         F_k=\{ct:\ t\in T^{(k)},c\in\{0,1\}^{r_c}\}-\bigcup_{d\in[k-1]}F_{d}, & 2\le k\le \ell.
\end{array}\right.
\end{equation}
These satisfy
$F_i\cap F_j =\emptyset$, for all $i\neq j \in[\ell]$, and 
\begin{equation}\label{eq:set_eq}
  \bigcup_{k\in [\ell]} F_k = \B^{r_c}\times \bigcup_{k\in [\ell]} T^{(k)} = \Bn-\{c0^{r_t}:c\in\{0,1\}^{r_c}\}.
\end{equation}
We are now in a position to define the unitary operators $\Pi$, $C_k$, $\mathcal{R}$ and $\Lambda_{r_c}$. 
\begin{enumerate}
    \item $\Pi$ is an $n$-qubit unitary defined by
    \begin{equation}\label{eq:pi}
        \Pi\ket{x_1x_2\cdots x_n}_{[n]}=\ket{x_1x_2\cdots x_{r_c}}_{\textsf{C}}\ket{x_{r_c+1}x_{r_c+2}\cdots x_n}_{\textsf{T}} \defeq\ket{x_{control}}_{\textsf{C}}\ket{x_{target}}_{\textsf{T}}.
    \end{equation}
    That is, $\Pi$ moves the content of the first $r_c$ qubits to register $\textsf{C}$ and the remaining qubits to register $\textsf{T}$. Note that $\Pi$ can be implemented by a sequence of SWAPs, and is thus an invertible linear transformation over $\mathbb{F}_2$.
    
    \item For $k\in[\ell]$, 
    \begin{equation}\label{eq:Ck}
        C_k\ket{x_{control}}_{\textsf{C}}\ket{y^{(k-1)}}_{\textsf{T}}=e^{i\sum_{s \in F_k}\langle s,x\rangle \alpha_s } \ket{x_{control}}_{\textsf{C}}\ket{y^{(k)}}_{\textsf{T}},
    \end{equation}
    where $\alpha_s$ is determined by Eq. \eqref{eq:alpha}, i.e., $C_k$ introduces a phase and updates stage $k-1$ to stage $k$. 
    
    \item $\mathcal{R}$ acts on qubit set $[n]-[r_c]$ and resets the suffix state as follows
    \begin{equation}\label{eq:reset}
        \mathcal{R}\ket{y^{(\ell)}}_{[n]-[r_c]}=\ket{y^{(0)}}_{[n]-[r_c]}.
    \end{equation}
    $\mathcal R$ is an invertible linear transformation over $\mathbb{F}_2$.
    
    \item $\Lambda_{r_c}$ is an $r_c$-qubit diagonal matrix acting on qubit set $[r_c]$, which satisfies 
    \begin{equation}\label{eq:Lambda_rc}
        \Lambda_{r_c}\ket{x_{control}}_{[r_c]}= e^{i\sum_{c\in \{0,1\}^{r_c}-\{0^{r_c}\}}\langle c0^{r_t},x\rangle\alpha_{c0^{r_t}}}\ket{x_{control}}_{[r_c]}.
    \end{equation}
\end{enumerate}

We now present circuit constructions for $C_k$, $\mathcal{R}$ and $\Lambda_{r_c}$ under general graph constraints.
\paragraph{Circuit construction for $C_k$}

Let $G=(V,E)$ denote a connected graph with vertex set $V={\sf C}\cup {\sf T}$. 
For all $k\in [\ell]$, $C_k$ is constructed in two stages:
\begin{align}
    \ket{x_{control}}_{\textsf{C}}\ket{y^{(k-1)}}_{\textsf{T}}&\xrightarrow{U_{Gen}^{(k)}}\ket{x_{control}}_{\textsf{C}}\ket{y^{(k)}}_{\textsf{T}}, \label{eq:Ugen_Graph}\\
    &\xrightarrow{U_{Gray}^{(k)}}e^{i\sum_{s \in F_k}\langle s,x\rangle \alpha_s}\ket{x_{control}}_{\textsf{C}}\ket{y^{(k)}}_{\textsf{T}}. \label{eq:Ugray_Graph}
\end{align}
$U^{(k)}_{Gen}$ is a linear transformation (over $\mathbb{F}_2$) on register $\textsf{T}$, and updates $\ket{y^{(k-1)}}_{\sf T}\rightarrow \ket{y^{(k)}}_{\sf T}$. 

$U^{(k)}_{Gray}$ is parameterized by $r_t$ integers $j_1, j_2, \ldots, j_{r_t}\in[r_c]$, each of which specifies an $(r_c,j_i)$-Gray code.  These Gray codes are used to update each qubit in the target register in a sequence of steps. More precisely, $U^{(k)}_{Gray}$ is carried out in $2^{r_c}+1$ phases, with each phase $p$ implementing a unitary $U_p$:

\begin{enumerate}
    \item Phase 1. For all $i\in[r_t]$,  $U_1$ applies $R(\alpha_{0^{r_c}t_i^{(k)}})$ (see Eqs.~\eqref{eq:rotation} and~\eqref{eq:alpha}) to the $i$-th qubit in $\sf T$ if $0^{r_c}t_i^{(k)} \in F_k$.
    \item Phases $2\le p\le 2^{r_c}$. $U_{p}$ consists of two steps:\begin{enumerate}
        \item Step $p.1$: Apply a unitary transformation $C_{p.1}$ satisfying, $\forall x\in\Bn$,
        \begin{align}\label{eq:step_p1}
           &\ket{x_{control}}_{\sf C}\ket{\langle c_{p-1}^{j_1 }t_1^{(k)},x\rangle,\langle c_{p-1}^{j_2 }t_2^{(k)},x\rangle,\cdots,\langle c_{p-1}^{j_{r_t} }t_{r_t}^{(k)},x\rangle}_{\sf T}\nonumber\\ \xrightarrow{C_{p.1}}& \ket{x_{control}}_{\sf C}\ket{\langle c_{p}^{j_1 }t_1^{(k)},x\rangle,\langle c_{p}^{j_2 }t_2^{(k)},x\rangle,\cdots,\langle c_{p}^{j_{r_t} }t_{r_t}^{(k)},x\rangle}_{\sf T}, \nonumber \\
          =&\ket{x_{control}}_{\sf C}|\langle c_{p-1}^{j_1 }t_1^{(k)},x \rangle \oplus x_{h_{j_1 p}},\langle c_{p-1}^{j_2 }t_2^{(k)},x\rangle \oplus x_{h_{j_2 p}},\cdots,\langle c_{p-1}^{j_{r_t} }t_{r_t}^{(k)},x\rangle \oplus x_{h_{j_{r_t}p}}\rangle_{\sf T}.
        \end{align}
        Note that each update $\langle c^{j_i}_{p-1} t_1^{(k)},x \rangle\rightarrow \langle c^{j_i}_{p} t_1^{(k)},x \rangle$ changes the prefix from $c^{j_i}_{p-1}$ to $c^{j_i}_{p}$, and can be implemented by a $\mathsf{CNOT}$ with control qubit $\ket{x_{h_{j_i p}}}$ and target being the $i$-th qubit in ${\sf T}$. 
        \item Step $p.2$: For all $i\in[r_t]$, apply $R(\alpha_{c^{j_i }_pt_{i}^{(k)}})$ to the $i$-th qubit in $\sf T$ if $c^i_pt_{i}^{(k)}\in F_k$, where $\alpha_{c_p^{j_i }t_i^{(k)}}$ is defined in Eq. \eqref{eq:alpha}.
    \end{enumerate}
    \item Phase $2^{r_c}+1$. $U_{2^{r_c}+1}$ carries out a transformation satisfying, $\forall x\in\Bn$
    ,
            \begin{align}\label{eq:phase_2rc+1}
           &\ket{x_{control}}_{\sf C}\ket{\langle c_{2^{r_c}}^{j_1 }t_1^{(k)},x\rangle,\langle c_{2^{r_c}}^{j_2 }t_2^{(k)},x\rangle,\cdots,\langle c_{2^{r_c}}^{j_{r_t} }t_{r_t}^{(k)},x\rangle}_{\sf T}\nonumber\\ 
           \xrightarrow{U_{2^{r_c}+1}}& \ket{x_{control}}_{\sf C}\ket{\langle c_{1}^{j_1 }t_1^{(k)},x\rangle,\langle c_{1}^{j_2 }t_2^{(k)},x\rangle,\cdots,\langle c_{1}^{j_{r_t} }t_{r_t}^{(k)},x\rangle}_{\sf T}, \nonumber\\
           =&\ket{x_{control}}_{\sf C}|\langle c_{2^{r_c}}^{j_1 }t_1^{(k)},x \rangle \oplus x_{h_{j_1 1}},\langle c_{2^{r_c}}^{j_2 }t_2^{(k)},x\rangle \oplus x_{h_{j_2 1}},\cdots,\langle c_{2^{r_c}}^{j_{r_t} }t_{r_t}^{(k)},x\rangle\oplus x_{h_{j_{r_c} 1}}\rangle_{\sf T}.
        \end{align}
Each update $\langle c_{2^{r_c}}^{j_i }t_1^{(k)},x\rangle\rightarrow \langle c_{1}^{j_i }t_1^{(k)},x\rangle$ changes the last prefix  $c_{2^{r_c}}^{j_i }$ to the first one $c_{1}^{j_i }$, which can be implemented by a CNOT with the $i$-th qubit in ${\sf T}$ as the target, controlled by $\ket{x_{h_{j_1 1}}}$.
\end{enumerate}

Let $\mathcal{D}(C_{p.1})$ and $\mathcal{S}(C_{p.1})$ denote the circuit depth and size, respectively, required to implement $C_{p.1}$ (Eq. \eqref{eq:step_p1}) under arbitrary graph constraint. Then the circuit depth and size for $C_k$ are shown as follows.
\begin{restatable}{lem}{Ck}\label{lem:Ck}
For all $k\in[\ell]$, the circuit $C_k$  in Eq. \eqref{eq:Ck} can be implemented by a quantum circuit of depth $O(n^2+2^{r_c}+\sum_{p=2}^{2^{r_c}}\mathcal{D}(C_{p.1}))$ and size $O(n^2+r_t2^{r_c}+\sum_{p=2}^{2^{r_c}}\mathcal{S}(C_{p.1}))$ under arbitrary graph constraint.
\end{restatable}


 


\paragraph{Circuit construction for $\mathcal{R}$} 
\begin{lemma}\label{lem:reset}
Unitary transformation $\mathcal{R}$ (Eq. \eqref{eq:reset}) can be implemented by a quantum circuit of depth and size $O(n^2)$ under arbitrary graph constraint.
\end{lemma}
\begin{proof}
$\mathcal{R}$ is an invertible linear transformation over $\mathbb{F}_2$ acting on qubits $[n]-[r_c]$. The result follows from Lemma~\ref{lem:cnot_circuit}. 
\end{proof}

\paragraph{Circuit construction for $\Lambda_{r_c}$} 

In \cite{sun2021asymptotically}, unitary $\Lambda_{r_c}$ is implemented recursively in depth $O(2^{r_c}/r_c)$. Under a constraint graph, however, the $r_c$ qubits of $\Lambda_{r_c}$ are not necessarily connected, and we therefore cannot implement $\Lambda_{r_c}$ recursively as before.
Fortunately, we can still realize $\Lambda_{r_c}$ with only a modest $O(n)$ overhead.
\begin{restatable}{lem}{Lambdarc}
\label{lem:Lambda_rc}
The $r_c$-qubit diagonal unitary matrix $\Lambda_{r_c}$ (Eq.\eqref{eq:Lambda_rc}) can be implemented by a quantum circuit of depth and size $O(n2^{r_c})$ under arbitrary graph constraint.
\end{restatable}



\subsection{Efficient circuits: general framework}
\label{sec:general_framework_noancilla_main}

An $O(2^n/n)$-depth and $O(2^n)$-size circuit construction for general $n$-qubit diagonal unitary $\Lambda_n$ under no graph constraints is given in~\cite{sun2021asymptotically}, using no ancillary qubits.  From these, it is straightforward, via Lemma~\ref{lem:cnot_path_constraint}, to obtain upper bounds on the circuit depth required under various graph constraints (see Table~\ref{tab:lambda-bounds}). These bounds lead to an increase in circuit depth by a factor of $n\cdot \diam(G)$, which may seem unavoidable. However, we show that this is not the case, and savings can be had by the constructions we give in the remainder of this section. Note that for $\Path_n$, $\Tree_n(2)$ and $\Expander_n$, our constructed circuits have depth either $O(2^n/n)$ or $O(\log (n)\cdot 2^n/n)$, which are almost tight as a lower bound of $\Omega(2^n/n)$ is known for QSP (or diagonal unitary operations) even without graph constraints. For general graphs, our constructed circuit has depth $O(2^n)$.

\begin{table}[h!]
    \centering
    \caption{Circuit depth upper bounds (ub) required to implement $\Lambda_n$ under various graph constraints. Trivial bounds are based on the unconstrained construction from~\cite{sun2021asymptotically} and Lemma~\ref{lem:cnot_path_constraint}, which implies that the required circuit depth is $O(n\cdot \diam(G)\cdot 2^n/n)$. Big O notation is implied.
    } 
    \begin{tabular}{c|cccc}
         &  $\Path_n$ & $\Tree_n(2)$ & $\Expander_n$ & General $G$\\ \hline
            $\diam(G)$  & $n$ & $\log n$ & $\log n$ & $n$\\ \hline
    Depth (ub, trival)     & $n2^n$ & $\log(n)2^n$ & $\log(n)2^n$ &$n2^n$  \\
    \hline
    \multirow{2}{*}{Depth (ub, new)}    & $2^n/n$   & $\log(n) 2^n/n$ & $\log(n)2^n/n$  &  $2^n$ \\
     & [Lem.~\ref{lem:diag_path_withoutancilla_main}]   & [Lem.~\ref{lem:diag_tree_withoutancilla}]&  [Lem.~\ref{lem:diag_expander_withoutancilla}] & [Lem. \ref{lem:diag_graph_withoutancilla}]
    \end{tabular}
    \label{tab:lambda-bounds}
\end{table}


To achieve the improved results in Table~\ref{tab:lambda-bounds} we make two design choices for each constraint graph type:
\begin{enumerate}
    \item The choice of control and target registers $\sf C$ and $\sf T$. 
    \item The choice of Gray codes, as specified by the integers $j_1, j_2, \ldots, j_{r_t}$ used to implement the $C_k$ operators.
\end{enumerate}
We adopt the following general strategy. As per Lemma~\ref{lem:cnot_path_constraint}, a graph $G$ constraint leads to an overhead cost $d(u,v)$ when implementing CNOT gates between vertices $u,v\in G$.  To minimize the increase in circuit depth, we aim to choose $\sf C$ and $\sf T$
such that the control and target qubits are close for as many CNOT gates as possible. Ideally, one desires that all CNOT gates have control and target qubits $O(1)$-distance apart. As this does not appear to be possible, we instead design $\sf C$ and $\sf T$ and the circuits $C_i$ in such a way that the number of CNOT gates acting across distance $d$ decays exponentially with $d$, leading to only a small ($O(\log n)$) or even constant overall overhead. 
As mentioned in Section \ref{sec:diag_without_ancilla_framework_main}, we also adapt the choice of Gray codes to account for the graph constraints, and the choices of control and target registers (see Table~\ref{tab:choice-of-gray}).


\begin{table}[!ht]
    \centering
    \caption{Integers $j_i$ ($i=1, \ldots, r_t$) which specify the $r_t$ Gray codes used in the implementation of $C_k$ operators. $A(i) = (i-1)(2^{a+1}-2)$, with $ a=\lceil \log(2\log n)\rceil$. $K_n$ is the complete graph on $n$ vertices.}
    \begin{tabular}{c|c|cccc}
       &  $K_n$~\cite{sun2021asymptotically}  & $\Path_n$  & $\Tree_n(2)$ & $\Expander_n$ & General $G$\\ \hline
    $j_i$    & $i$ &$ i$ &  $1+A(i)$  & $1$ & $1$ 
    \end{tabular}
    
    \label{tab:choice-of-gray}
\end{table}

\subsection{Efficient circuits under path and $d$-dimensional grid constraints}
\label{sec:diag_without_ancilla_path_main}
Here we present the circuit depth and size required for diagonal unitary matrices under path and grid constraints, using no ancillary qubits. Omitted proofs are given in Appendix \ref{sec:diag_without_ancilla_path}.
\paragraph{Choice of $\sf C$ and $\sf T$}
If $n-2\lceil \log n\rceil$ is even, let $\tau=2\lceil \log n\rceil$; if $n-2\lceil \log n\rceil$ is odd, let $\tau=2\lceil  \log n\rceil+1$. Then $2\lceil  \log n\rceil\le \tau\le 2\lceil  \log n\rceil+1$ and $n-\tau$ is even.
The control and target registers are taken to be ${\sf C}=\left\{2i-1:\forall i\in[r_t]\right\}\cup \left\{2r_t+j:\forall j\in[\tau]\right\}$ and ${\sf T}=\left\{2i:\forall i\in[r_t]\right\}$
respectively (see Fig. \ref{fig:pi_n^k_main}, lower part), where $r_c=\frac{n+\tau}{2}$ and $r_t=\frac{n-\tau}{2}$. 

\paragraph{Implementation of $\Pi$}
In this section, the unitary $\Pi$ of Eq. \eqref{eq:pi} is denoted $\Pi^{path}$, and implemented as follows. 
\begin{restatable}{lem}{pipath}
\label{lem:pi_path}
The transformation $\Pi^{path}$, defined by 
\begin{align*}
\bigotimes_{i=1}^n\ket{x_i}_{i}\xrightarrow{\Pi^{path}} &\bigotimes_{i=1}^{r_t
    }\ket{x_i}_{{2i-1}}
    \bigotimes_{i=1}^{\tau}\ket{x_{r_t
    +i}}_{{n-\tau+i}}
    \bigotimes_{i=1}^{r_t
    }\ket{x_{i+
    r_c
    }}_{{2i}} \defeq \ket{x_{control}}_{\sf C}\ket{x_{target}}_{\sf T},
\end{align*}
i.e, which moves the last $r_t$ qubits to the first $r_t$ even positions, can be implemented by a CNOT circuit of depth O(n) and size $O(n^2)$ under $\Path_n$ constraint.
\end{restatable} 
The effect of $\Pi^{path}$ is shown in Fig. \ref{fig:pi_n^k_main}.
    \begin{figure}[ht]
        \centering
    \begin{tikzpicture}
    \draw (-2,0) -- (1.2,0) (1.8,0) -- (5.2,0) (5.8,0) -- (7,0);
    
    \draw [fill=black] (-2,0) circle (0.05)
                       (-1,0) circle (0.05)
                       (0,0) circle (0.05)
                      (1,0) circle (0.05)
                      (2,0) circle (0.05)
                      (3,0) circle (0.05);
    \draw [fill=black] (4,0) circle (0.05)
                      (5,0) circle (0.05)
                      (6,0) circle (0.05)
                      (7,0) circle (0.05);
    \draw (-3,0) node{\scriptsize qubit};
    \draw (-3,-0.8) node{\scriptsize  state};
    \draw (-3,-2) node{\scriptsize qubit};
    \draw (-3,-2.8) node{\scriptsize state};
                       
     \draw (1.2,0) node[anchor=west]{\scriptsize $\cdots$} (5.2,0) node[anchor=west]{\scriptsize $\cdots$};
          \draw (-2,-0.3) node{\scriptsize $1$} (-2,-0.8) node{\scriptsize $\ket{x_1}$};
         \draw (-1,-0.3) node{\scriptsize $2$} (-1,-0.8) node{\scriptsize $\ket{x_2}$};
     \draw (0,-0.3) node{\scriptsize $3$} (0,-0.8) node{\scriptsize $\ket{x_3}$}
          (1,-0.3) node{\scriptsize $4$} (1,-0.8) node{\scriptsize $\ket{x_4}$}
          (2,-0.3) node{\scriptsize ${r_c-1}$} (2,-0.8) node{\scriptsize $\ket{x_{r_c-1}}$}
          (3,-0.3) node{\scriptsize ${r_c}$} (3,-0.8) node{\scriptsize $\ket{x_{r_c}}$}
          (4,-0.3) node{\scriptsize ${r_c+1}$} (4,-0.8) node{\scriptsize  $\underline{\ket{x_{r_c+1}}}$}
          (5,-0.3) node{\scriptsize ${r_c+2}$} (5,-0.8) node{\scriptsize  $\underline{\ket{x_{r_c+2}}}$}
          (6,-0.3) node{\scriptsize ${n-1}$} (6,-0.8) node{\scriptsize  $\underline{\ket{x_{n-1}}}$}
          (7,-0.3) node{\scriptsize ${n}$} (7,-0.8) node{\scriptsize  $\underline{\ket{x_n}}$};
          
     \draw (1.2,-2) node[anchor=west]{\scriptsize $\cdots$} (5.2,-2) node[anchor=west]{\scriptsize $\cdots$};
        \draw (-2,-2) -- (1.2, -2) (1.8,-2) -- (5.2,-2) (5.8,-2) -- (7,-2);
    
    \draw [fill=black] (-2,-2) circle (0.05) (0,-2) circle (0.05)
                      (2,-2) circle (0.05) (4,-2) circle (0.05) (5,-2) circle (0.05) (6,-2) circle (0.05) (7,-2) circle (0.05);
    \draw [fill=white, draw=black] (-1,-2) circle (0.05) (1,-2) circle (0.05) (3,-2) circle (0.05);
            \draw (-2,-2.3) node{\scriptsize $1$} (-2,-2.8) node{\scriptsize $\ket{x_1}$};
        \draw (-1,-2.3) node{\scriptsize $2$} (-1,-2.8) node{\scriptsize $\underline{\ket{x_{r_c+1}}}$};
     \draw (0,-2.3) node{\scriptsize $3$} (0,-2.8) node{\scriptsize $\ket{x_2}$}
          (1,-2.3) node{\scriptsize $4$} (1,-2.8) node{\scriptsize  $\underline{\ket{x_{r_c+2}}}$}
          (2,-2.3) node{\scriptsize ${n-\tau-1}$} (2,-2.8) node{\scriptsize $\ket{x_{r_t}}$}
          (3,-2.3) node{\scriptsize ${n-\tau}$} (3,-2.8) node{\scriptsize $\underline{\ket{x_{n}}}$}
          (4,-2.3) node{\scriptsize ${n-\tau+1}$} (4,-2.8) node{\scriptsize  $\ket{x_{r_t+1}}$}
          (5,-2.3) node{\scriptsize ${n-\tau+2}$} (5,-2.8) node{\scriptsize  $\ket{x_{r_t+2}}$}
          (6,-2.3) node{\scriptsize ${n-1}$} (6,-2.8) node{\scriptsize $\ket{x_{r_c-1}}$}
          (7,-2.3) node{\scriptsize ${n}$} (7,-2.8) node{\scriptsize  $\ket{x_{r_c}}$};
       \draw (2.5,-1.5) node{$\downarrow~\Pi^{path}$};           
    
    \end{tikzpicture}
        \caption{The effect of $\Pi^{path}$. In the lower figure, the qubits in hollow circles (called white qubits) form register $\sf T$, and those in filled circles (called white qubits) form register $\sf C$.}
        \label{fig:pi_n^k_main}
     \end{figure}
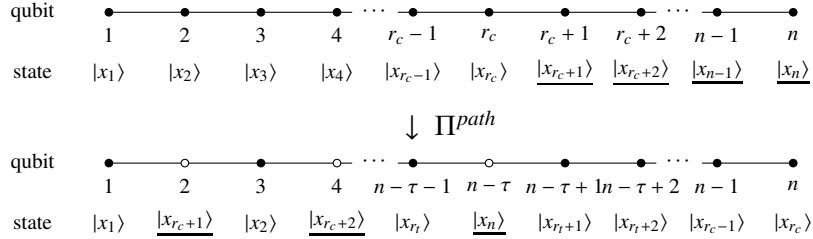

\paragraph{Implementation of $C_k$}
\begin{restatable}{lem}{Ukwithoutancilla}\label{lem:U(k)_without_ancilla}
Let $x=x_1\cdots x_{r_c}\in \B^{r_c}$ and $y=y_1\cdots y_{r_t}\in\B^{r_t}$. For all $k\le r_c$, define the unitary $U^{(k)}$ (see Fig.~\ref{fig:U(k)_without_ancilla_main}) by Eq. \eqref{eq:U(k)_without_ancilla_main}.

\begin{align}\label{eq:U(k)_without_ancilla_main}
&\bigotimes_{i=1}^{r_t}\ket{x_iy_i}_{\{2i-1,2i\}}\bigotimes_{i=r_t+1}^{r_c}\ket{x_{i}}_{r_t+i} \nonumber\\
\xrightarrow{U^{(k)}} &\left\{
\begin{array}{ll}
      \bigotimes_{i=1}^{r_t}\ket{x_i(y_i\oplus x_{i+k-1})}_{\{2i-1,2i\}}\bigotimes_{i=r_t+1}^{r_c}\ket{x_{i}}_{r_t+i}, &\text{if~}k\le \tau+1, \\
      \bigotimes_{i=1}^{r_c-k+1}\ket{x_i(y_i\oplus x_{i+k-1})}_{\{2i-1,2i\}}\bigotimes_{i=r_c-k+2}^{r_t}\ket{x_i(y_i
     \oplus x_{i-r_c+k-1})}_{\{2i-1,2i\}}\bigotimes_{i=r_t+1}^{r_c}\ket{x_{i}}_{r_t+i},&\text{if~}k\ge \tau+2. \\
\end{array}\right.
\end{align}

Under $\Path_n$ constraint, $U^{(k)}$ can be implemented by circuit of depth $O(k)$ and size $O(r_t k)$ for $k\le \tau+1$, and depth and size $O(r_tk)$ for $k\ge \tau+2$. 
\end{restatable}

\begin{figure}[hbt]
    \centering
    \begin{tikzpicture}
        \draw (0.5,2) parabola bend (1.25,2.25) (2,2) (2.5,2) parabola bend (3.25,2.25) (4,2) (4.5,2) parabola bend (5.25,2.25) (6,2) (6.5,2) parabola bend (7.25,2.25) (8,2) (1.5,2) parabola bend (2.25,2.25) (3,2) (3.5,2) parabola bend (4.25,2.25) (5,2) (5.5,2) parabola bend (6.25,2.25) (7,2) (7.5,2) parabola bend (8,2.25) (8.5,2);
     
     \draw (0.5,1) parabola bend (2.75,1.4) (5,1) (1.5,1) parabola bend (3.75,1.4) (6,1) (2.5,1) parabola bend (4.75,1.4) (7,1) (3.5,1) parabola bend (5.75,1.4) (8,1) (4.5,1) parabola bend (6.5,1.4) (8.5,1) (5.5,1) parabola bend (7.25,1.4) (9,1) (0,1) parabola bend (3.25,0.3) (6.5,1)  (1,1) parabola bend (4.25,0.3) (7.5,1);
    
        \draw (0,1) -- (9,1);
        \draw (0,2) -- (9,2);
        \draw [fill=black] (0,1) circle (0.05) (1,1) circle (0.05) (2,1) circle (0.05)  (3,1) circle (0.05) (4,1) circle (0.05) (5,1) circle (0.05) (6,1) circle (0.05) (7,1) circle (0.05);
        \draw [fill=black] (0,2) circle (0.05) (1,2) circle (0.05) (2,2) circle (0.05)  (3,2) circle (0.05) (4,2) circle (0.05) (5,2) circle (0.05) (6,2) circle (0.05) (7,2) circle (0.05);       
        \draw [fill=white] (0.5,1) circle (0.05) (1.5,1) circle (0.05) (2.5,1) circle (0.05)  (3.5,1) circle (0.05) (4.5,1) circle (0.05) (5.5,1) circle (0.05) (6.5,1) circle (0.05) (7.5,1) circle (0.05);
        \draw [fill=white] (0.5,2) circle (0.05) (1.5,2) circle (0.05) (2.5,2) circle (0.05)  (3.5,2) circle (0.05) (4.5,2) circle (0.05) (5.5,2) circle (0.05) (6.5,2) circle (0.05) (7.5,2) circle (0.05);
        \draw [fill=black,draw=black] (8,1) circle (0.05) (8.5,1) circle (0.05) (9,1) circle (0.05) (8,2) circle (0.05) (8.5,2) circle (0.05) (9,2) circle (0.05);
        
        \draw [fill=black] (0,1)  node[anchor=north]{\scriptsize $x_1$} (1,1) node[anchor=north]{\scriptsize $x_2$} (2,1) node[anchor=north]{\scriptsize $x_3$}  (3,1) node[anchor=north]{\scriptsize $x_4$} (4,1) node[anchor=north]{\scriptsize $x_5$} (5,1) node[anchor=north]{\scriptsize $x_{r_t-2}$} (6,1) node[anchor=north]{\scriptsize $x_{r_t-1}$} (7,1) node[anchor=north]{\scriptsize $x_{r_t}$};
        \draw [fill=black] (0,2)  node[anchor=north]{\scriptsize $x_1$} (1,2) node[anchor=north]{\scriptsize $x_2$} (2,2) node[anchor=north]{\scriptsize $x_3$}  (3,2) node[anchor=north]{\scriptsize $x_4$} (4,2) node[anchor=north]{\scriptsize $\cdots$} (5,2) node[anchor=north]{\scriptsize $x_{r_t-2}$} (6,2) node[anchor=north]{\scriptsize $x_{r_t-1}$} (7,2) node[anchor=north]{\scriptsize $x_{r_t}$};

        \draw [fill=black] (0.5,1)  node[anchor=north]{\scriptsize $y_1$} (1.5,1) node[anchor=north]{\scriptsize $y_2$} (2.5,1) node[anchor=north]{\scriptsize $y_3$}  (3.5,1) node[anchor=north]{\scriptsize $y_4$} (4.5,1) node[anchor=north]{\scriptsize $\cdots$} (5.5,1) node[anchor=north]{\scriptsize $y_{r_t-2}$} (6.5,1) node[anchor=north]{\scriptsize $y_{r_t-1}$} (7.5,1) node[anchor=north]{\scriptsize $y_{r_t}$};
        
        \draw [fill=black] (0.5,2)  node[anchor=north]{\scriptsize $y_1$} (1.5,2) node[anchor=north]{\scriptsize $y_2$} (2.5,2) node[anchor=north]{\scriptsize $y_3$}  (3.5,2) node[anchor=north]{\scriptsize $y_4$} (4.5,2) node[anchor=north]{\scriptsize $\cdots$} (5.5,2) node[anchor=north]{\scriptsize $y_{r_t-2}$} (6.5,2) node[anchor=north]{\scriptsize $y_{r_t-1}$} (7.5,2) node[anchor=north]{\scriptsize $y_{r_t}$};
        
        \draw [fill=black] (8,1) node[anchor=north]{\scriptsize $x_{r_t+1}$} (8.5,1)node[anchor=north]{\scriptsize $\cdots$} (9,1) node[anchor=north]{\scriptsize $x_{r_c}$} (8,2) node[anchor=north]{\scriptsize $x_{r_t+1}$} (8.5,2) node[anchor=north]{\scriptsize $\cdots $} (9,2) node[anchor=north]{\scriptsize $x_{r_c}$};
    
    \draw (-0.7,1) node{\scriptsize $k\ge \tau+2$} (-0.7,2) node{\scriptsize $k\le \tau+1$};
    \end{tikzpicture}
    \caption{The effect of $U^{(k)}$. Curves indicate that the value of a black qubit has been added (xor-ed) to the corresponding white qubit.}
    \label{fig:U(k)_without_ancilla_main}
\end{figure}
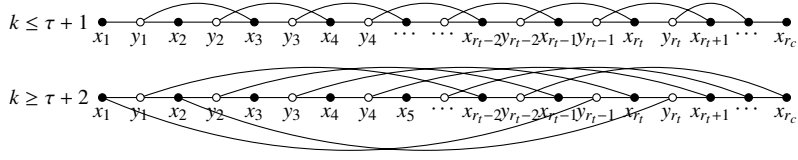
The above lemma is used to prove the circuit complexity required to implement $C_k$ under $\Path_n$ constraint:

\begin{lemma}\label{lem:Ck_path}
For all $k\in[\ell]$, $C_k$ (Eq.~\eqref{eq:Ck}) can be implemented by a quantum circuit of depth $O(2^{r_c})$ and size $O(n2^{r_c})$ under $\Path_n$ constraint.
\end{lemma}
\begin{proof}
First, we construct quantum circuits for $C_{p.1}$ (Eq. \eqref{eq:step_p1})  for all $p\in\{2,3,\ldots,2^{r_c}\}$. For every $i\in[r_t]$, choose integers $j_i=i $. The strings $c_{p-1}^i$ and $c_{p}^i$ in the $(r_c,i)$-Gray code differ in the $h_{ip}$-th bit for all $p\in\{2,3,\ldots,2^{r_c}\}$.

Recall that $C_{p.1}$ transforms prefix $c_{p-1}^{j_i}$ to $c_p^{j_i}$ in the $i$-th qubit of target register $\sf T$, i.e., for any $x=x_1x_2\cdots x_n\in \Bn$,
\begin{align*}
&\bigotimes_{i=1}^{r_t}\ket{x_i,\langle c^i_{p-1}t_i^{(k)},x\rangle}_{\{2i-1,2i\}}\bigotimes_{i=r_t+1}^{r_c}\ket{x_{i}}_{r_t+i}\\
\xrightarrow{C_{p.1}} &\bigotimes_{i=1}^{r_t}\ket{x_i,\langle c^i_{p}t_i^{(k)},x\rangle}_{\{2i-1,2i\}}\bigotimes_{i=r_t+1}^{r_c}\ket{x_{i}}_{r_t+i}
=\bigotimes_{i=1}^{r_t}\ket{x_i,\langle c^i_{p-1}t_i^{(k)},x\rangle\oplus x_{h_{ip}}}_{\{2i-1,2i\}}\bigotimes_{i=r_t+1}^{r_c}\ket{x_{i}}_{r_t+i},
\end{align*}
We now show that $C_{p.1}$ is equivalent to $U^{(h_{1p})}$ where $U^{(\cdot)}$ is defined in Eq.~\eqref{eq:U(k)_without_ancilla_main}. We first define $ k' \defeq h_{1p} = \zeta(p-1)$ (see Eq.~\eqref{eq:h1j}).
We consider two cases:
\begin{enumerate}
    \item $k'\le \tau+1$: For all $i\in [r_t]$, we have
        \[0  \le \zeta(p-1)+i-2 = k' + i -2 
        \le \tau + i -1
        = r_c - r_t + i -1 
        \le r_c -1\]
    where the second equality is because of  $r_c = r_t + \tau$. This implies that
        \[h_{ip} = (\zeta(p-1)+i-2\bmod r_c)+1
        = \zeta(p-1) + i -1 
        = k'+i-1.\]
    Thus, the qubits $\ket{x_{k'}}$, $\ket{x_{k'+1}}$, $\ldots$, $\ket{x_{r_t+k'-1}}$ are exactly $\ket{x_{h_{1p}}}$, $\ket{x_{h_{2p}}}$, $\ldots$, $\ket{x_{h_{r_t p}}}$, which implies $U^{(h_{1p})}$ corresponds to the first case of Eq. \eqref{eq:U(k)_without_ancilla_main}.
    \item $k'\ge \tau+2$: In this case, 
    \[0  \le \zeta(p-1)+i-2 = k' + i -2 \begin{cases}
        \le r_c -1, &\text{ if } i \le r_c - k' + 1,\\
        \ge r_c, &\text{ if } i \ge r_c - k' +2.
        \end{cases}
    \]
    It follows that $ h_{ip} = (\zeta(p-1)+i-2\bmod r_c)+1$, which is equal to $ \zeta(p-1) + i -1 = k' + i -1$ if $i \le r_c - k' + 1$; and $\left(\zeta(p-1) + i -2\right) -r_c +1 = k' +i -r_c -1$ if $i \ge r_c - k' +2$ since $i\le r_t$ and $\zeta(p-1)\le r_c$.
     Therefore, the qubits $\ket{x_{k'}}$, $\ket{x_{k'+1}}$, $\ldots$, $\ket{x_{r_c}}$, $\ket{x_1}$, $\ldots$, $\ket{x_{k'-\tau-1}}$ are exactly $\ket{x_{h_{1p}}}$, $\ket{x_{h_{2p}}}$, $\ldots$, $\ket{x_{h_{r_t p}}}$. Namely, $U^{(h_{1p})}$ is the same as defined in the second case of Eq. \eqref{eq:U(k)_without_ancilla_main}.
\end{enumerate}

We now analyze the circuit depth of $C_k$.
{By Lemma \ref{lem:U(k)_without_ancilla}, the depth and size of $U^{(k')}$ are $O(k')$ and $O(r_tk')$ if $k'\le \tau+1$; the depth and size are both $O(r_tk')$ if $\tau+2\le k'\le r_c$.} 
Recall that for every $k''\in[r_c]$, there are $2^{r_c-k''}$ many $p\in \{2,3,\ldots,2^{r_c}\}$ satisfying $h_{1p} = k''$ (Lemma \ref{lem:GrayCode}). Thus, $U^{(k')}$ appears $2^{r_c-k'}$ times in Step $p$.1 ($C_{p.1}$) when we run all iterations $p\in\{2,3,\ldots,2^{r_c}\}$. With $\mathcal{D}(C_{p.1})$ and $\mathcal{S}(C_{p.1})$ denoting the circuit depth and size for $C_{p.1}$, 
by Lemma \ref{lem:Ck}, $C_k$ has  circuit depth
\begin{equation}\label{eq:ck-depth}
O(n^2+2^{r_c}+\sum_{p=2}^{2^{r_c}}\mathcal{D}(C_{p.1}))=O(n^2+2^{r_c})+\sum_{k'=1}^{\tau+1} O(k')2^{r_c-k'} +\sum_{k'=\tau+2}^{r_c}O(r_tk')2^{r_c-k'}= O(2^{r_c}),
\end{equation}
and circuit size \[ O(n^2+r_t2^{r_c}+\sum_{p=2}^{2^{r_c}}\mathcal{S}(C_{p.1}))=O(n^2+2^{r_c})+\sum_{k'=1}^{r_c}2^{r_c-k'}O(r_tk')= O(r_t2^{r_c}),\]
where we use the fact that $2\lceil\log(n)\rceil\le \tau\le 2\lceil\log(n)\rceil+1$.
\end{proof}

\paragraph{Remark} The reason we choose $\tau = 2\lceil \log n\rceil$ is the following. The series $\sum_{j=1}^n j \cdot 2^{-j} \le 2$, with the first $2\log(n)$ terms contributing the majority of the sum, i.e., if $\tau = 2\lceil\log n\rceil$ then
\[
    \sum_{j=\tau}^n j\cdot 2^{-j} = O(1/n).
\]
In Eq.~\eqref{eq:ck-depth}, the circuit depth contains contributions from the terms 
    \[2^{r_c}\sum_{k'=1}^{\tau+1}O(k')2^{-k'} + 2^{r_c}r_t\sum_{k'=\tau+2}^{r_c}O(k')2^{-k'},\]
which, for each $k'$, can be understood roughly as $2^{r_c-k'}$ CNOT circuits, in which each CNOT gate acts on qubits separated by distance $k'$.  Noting that $\tau = 2\lceil\log n\rceil$, $r_t = (n-\tau)/2 \approx n/2$ and $r_c= (n+\tau)/2\approx n/2$, second term has the factor of $r_t$ cancelled by the factor of $1/n$ that comes from the series summation. The number of CNOT circuits with CNOT gates acting on qubits separated by distances $d$ greater than $2\log n$ is exponentially reduced, and the cost of implementing those gates is suppressed by $1/n$.  We take a similar approach with other graph constraints.

\paragraph{Implementation of $\Lambda_n$}
Now putting everything together, we can obtain the complexity for $\Lambda_n$.
\begin{lemma}\label{lem:diag_path_withoutancilla_main}
Any $n$-qubit diagonal unitary matrix $\Lambda_n$ can be realized by a quantum circuit of depth $O(2^n/n)$ and size $O(2^n)$, under $\Path_n$ constraint without ancillary qubits.
\end{lemma}

This result can be extended to $d$-dimensional grid.

\begin{lemma}\label{lem:diag_grid_withoutancilla_main}
    Any $n$-qubit diagonal unitary matrix $\Lambda_n$ can be realized by a quantum circuit of depth $O(2^n/n)$ and size $O(2^n)$ under $\Grid_{n}^{n_1,n_2,\ldots,n_d}$ constraint. 
\end{lemma}

\subsection{Circuit implementation under $d$-ary tree, expander graph and arbitrary graph constraints}
\label{sec:diag_without_ancilla_tree_expander_graph}
We can similarly bound the circuit depth and size for diagonal unitary matrices under different graph constraints. Proofs are given in Appendices \ref{sec:diag_without_ancilla_binarytree}, \ref{sec:diag_without_ancilla_expander} and \ref{sec:diag_without_ancilla_graph}.
\begin{restatable}{lem}{diagtreenoancilla}\label{lem:diag_tree_withoutancilla}
Any $n$-qubit diagonal unitary matrix $\Lambda_n$ can be realized by a quantum circuit of depth 
\begin{enumerate}
    \item $O(\log(n)2^n/n)$ under $\Tree_n(2)$ constraint.
    \item $O(\log_d(n))$ under $\Tree_n(d)$ constraint for $d<n$.
    \item $O(2^n)$ under $\Star_n$ constraint.
\end{enumerate}
\end{restatable}

\begin{restatable}{lem}{expandernoancilla}
\label{lem:diag_expander_withoutancilla}
Any $n$-qubit diagonal unitary matrix $\Lambda_n$ can be realized by a quantum circuit of depth $O(\log(n)2^n/n)$ under $\Expander_n$ constraint, using no ancillary qubits.
\end{restatable}

\begin{restatable}{lem}{graphnoancilla}
\label{lem:diag_graph_withoutancilla}
Any $n$-qubit diagonal unitary matrix $\Lambda_n$ can be realized by a standard quantum circuit of size $O(2^n)$  under arbitrary graph constraint, using no ancillary qubits.
\end{restatable}

 \section{Circuits for diagonal unitary matrices under qubit connectivity constraints, with ancillary qubits }
\label{sec:diag_with_ancilla}
Here we present circuit constructions for $\Lambda_n$, using $m$ ancillary qubits, under various graph constraints.
These results will be used to construct QSP and GUS circuits in Section \ref{sec:QSP_US_graph}. Note that the constructions of Section \ref{sec:diag_without_ancilla}, which do not use ancilla, are not simply special cases (corresponding to $m=0$) of the constructions in this section.  The constructions here are fundamentally different and require $m\ge \Omega(n)$.  
Additional details and omitted proofs are given in Appendix \ref{append:diag_with_ancilla}.
\subsection{Circuit framework}

Our circuit framework for $\Lambda_n$ using $m$ ancillary qubits generalizes the ancilla-based framework of~\cite{sun2021asymptotically}. Both circuit frameworks consist of 5 stages: suffix copy, Gray initial, prefix copy, Gray cycle, and inverse stage. In this section, we show the main ideas of both the original and our new circuit frameworks and highlight the differences between them.
More details are given in Appendix \ref{append:diag_with_ancilla}.

In the original framework, input $x\in\{0,1\}^n$ is divided into a prefix $x_{pre}$ and suffix $x_{suf}$ of lengths $n-\log(m/2)$ and $\log(m/2)$, respectively, and similarly for $s$ with the same cutoff point. The ancillary qubits are divided into an $m/2$-qubit copy register and an $m/2$-qubit target register, the former used for storing copies of $x_{pre}$ and $x_{suf}$, to increase the degree to which cycling through Gray codes can be done in parallel. 
The $m/2$ qubits in the target register are each responsible for enumerating a suffix of $s$, and different layers of circuits are used to enumerate all prefixes. The procedure then consists of five stages, which are similar to the five stages we use in our new procedure. 

In our approach, the $n+m$ qubits are divided into 4 registers:
\begin{itemize}
    \item ${\sf R}_{\rm inp}$: an $n$-qubit input register used to hold the input state $\ket{x}$, with $x\in\{0,1\}^n$ divided into an  $(n-p)$-bit prefix $x_{pre}=x_1x_2\ldots x_{n-p}$ and a $p$-bit suffix $x_{suf}=x_{n-p+1}\ldots x_n$. The first $\tau$ bits of $x_{pre}$ (with $\tau$ dependent on the constraint graph) are referred to as $x_{aux}$, i.e., $x_{aux}=x_1x_2\ldots x_{\tau}$ and hold frequently used content, to be copied close to the target qubits in order to reduce the circuit depth of the Gray cycle stage.
    \item The $m$ ancillary qubits are divided into three registers:
    \begin{itemize}
        \item ${\sf R}_{\rm copy}$: the copy register of size $\lambda_{copy} \ge n$
        \item ${\sf R}_{\rm targ}$: the target register of size $\lambda_{targ} = 2^p \ge n$
        \item ${\sf R}_{\rm aux}$: the auxiliary register  of size $\lambda_{aux} \ge n$ 
    \end{itemize}
\end{itemize}

The circuit itself consists of $5$ stages:
\begin{enumerate}
    \item Suffix Copy: makes $O(\lambda_{copy}/p)$ copies of $\ket{x_{suf}}$ in ${\sf R}_{\rm copy}$.
    
    \item Gray Initial: prepares the state \[\ket{\langle c_1^{\ell_1}t_1, x\rangle}\otimes\cdots\otimes\ket{\langle c_1^{\ell_{2^p}}t_{2^p},x\rangle}=\ket{\langle t_1, x_{suf}\rangle}\otimes\cdots\otimes\ket{\langle t_{2^p},x_{suf}\rangle}\] in ${\sf R}_{\rm targ}$, where $\ell_k$ (for $k\in[2^p]$) are integers specifying $2^p$ $(n-p, \ell_k)$-Gray codes $\{c^{\ell_k}_1,c^{\ell_k}_2,\ldots c^{\ell_k}_{2^{n-p}}\}$, $\{t_1, \ldots, t_{2^p}\} = \{0,1\}^{p}$, and $c_1^i=0^{n-p}$ and $t_i$ are the prefix and suffix of $s$ (see Eq.~\eqref{eq:task1}). 
    
    \item Prefix Copy: makes $O(\lambda_{aux}/\tau)$ copies of $\ket{x_{aux}}$ in ${\sf R}_{\rm aux}$, and replaces the copies of $\ket{x_{suf}}$ in ${\sf R}_{\rm copy}$ with $O\left(\lambda_{copy}/(n-p)\right)$ copies of $\ket{x_{pre}}$.
    
    \item Gray Cycle:
    This stage enumerates all $2^{n-p}$ prefixes of $s$ by going along a Gray code---each qubit $k$ uses $(n-p,\ell_k)$-Gray code, which consists of $2^{n-p}$ steps, with each step $j$ responsible for (i) updating prefix, and (ii) implementing a phase shift.
    
    \item Inverse: restores all ancillary qubits to zero.
\end{enumerate}

Compared to~\cite{sun2021asymptotically}, which did not consider connectivity constraints, our construction differs in:
\begin{enumerate}
    \item The design of the registers. In \cite{sun2021asymptotically}, the value of $p$ (which specifies the division of $x$ into $x_{pre}$ and $x_{suf}$) is fixed at $p =\log(m/2)$, and the ancillary qubits are divided into 2 registers only, with the first $m/2$ qubits forming ${\sf R}_{\rm copy}$ and the second $m/2$ qubits forming ${\sf R}_{\rm targ}$. In this work, $p$ is chosen dependent on the constraint graph, and we add a new register ${\sf R}_{\rm aux}$. The positions and sizes of the ancillary registers now also depend on the constraint graph.
    \item The prefix copy stage. In \cite{sun2021asymptotically}, prefix copy is responsible for making copies of $\ket{x_{pre}}$ in the copy register. Here, it also makes copies of $\ket{x_{aux}}$ in the auxiliary register: when we generate prefixes by Gray code, we apply CNOT gates where the control qubits are in a copy of $\ket{x_{pre}}$ or $\ket{x_{aux}}$. This imposes an overhead of $O(n^2)$ to the circuit depth since the distances between control and target qubits are at most $O(n)$. To resolve this issue, we make copies of $\ket{x_{aux}}$ and arrange them close to the qubits in the target register. If the distance between control qubits in $\ket{x_{pre}}$ and the target qubits is too large, we use qubits in $\ket{x_{aux}}$ as the control qubits instead.  
    \item The choice of Gray codes. The Gray Cycle stage involves choosing $2^{p}$ Gray codes, specified by the integers $\ell_k$. In \cite{sun2021asymptotically}, these are chosen as $\ell_k = (k-1)\mod (n-p)+1$ for every $k\in[2^p]$. Here we choose $\ell_k$ dependent on the constraint graph.
\end{enumerate}
These changes were made to address the fact that the framework of~\cite{sun2021asymptotically} does not perform well under connectivity constraints. 
In particular, the generation of the $2^p$ prefixes by Gray codes (during the Gray Initial and Gray Cycle stages) involves $2^p$ CNOT gates which may not be implementable in parallel under connectivity constraints and may impose an overhead of $O(m^2)$ to the circuit depth.


\subsection{Efficient circuits under path, $d$-dimensional grid, $d$-ary tree and expander graph constraints}
Using our new circuit framework, we bound the circuit depth required for diagonal unitary matrices under path, $d$-dimensional grid, $d$-ary tree and expander graph constraints. 
\begin{lemma}\label{lem:diag_d_grid_ancilla}
    Any $n$-qubit diagonal unitary matrix $\Lambda_n$ can be implemented by a  quantum circuit using $m\ge \Omega(n)$ ancillary qubits, of depth
    \begin{enumerate}
        \item $O\left(2^{n/2}+\frac{2^n}{m}\right)$ under $\Path_{n+m}$ constraint.
        \item $O\big(2^{n/3}+\frac{2^{n/2}}{\sqrt{n_2}}+\frac{2^n}{n+m}\big)$ under $\Grid^{n_1,n_2}_{n+m}$ constraint.
        \item $O\big(n^2+d2^{\frac{n}{d+1}}+\max\limits_{k\in\{2,\ldots,d\}}\big\{\frac{d2^{n/k}}{(\Pi_{i=k}^d n_i)^{1/k}}\big\}+\frac{2^n}{n+m}\big)$ under $\Grid^{n_1,n_2,\ldots,n_d}_{n+m}$ constraint. If $n_1=n_2=\cdots=n_d$, the circuit depth is $O\left(n^2+d2^{\frac{n}{d+1}}+\frac{2^n}{n+m}\right)$.
    \end{enumerate}
\end{lemma}

\begin{lemma}\label{lem:diag_tree_ancilla}
    Any $n$-qubit diagonal unitary matrix $\Lambda_n$ can be implemented by a  quantum circuit using $m\ge \Omega(n)$ ancillary qubits, of depth 
    \begin{enumerate}
        \item $O\left(n^2\log n+\frac{\log(n)2^n}{m}\right)$ under $\Tree_{n+m}(2)$ constraint.
        \item $O\big(nd\log_d (n+m)\log_d(n+d)+\frac{(n+d)\log_d(n+d) 2^{n}}{n+m}\big)$ under $\Tree_{n+m}(d)$ constraint for $d<n+m$.
        \item $O(2^n)$ under $\Star_{n+m}$ constraint.
    \end{enumerate}
\end{lemma}

\begin{restatable}{lem}{diagexpanderancilla}\label{lem:diag_expander_ancilla}
Any $n$-qubit diagonal unitary matrix $\Lambda_n$ can be realized by a quantum circuit of depth $O\big(n^2+\frac{\log(m)2^n}{m}\big)$
under $\Expander_{n+m}$ constraint, using $m\ge \Omega(n)$ ancillary qubits.
\end{restatable}

\section{Circuits for QSP and GUS under qubit connectivity constraints}
\label{sec:QSP_US_graph}
In this section, we bound the circuit size and depth for QSP and GUS, based on the circuit constructions for diagonal unitary matrices in Sections \ref{sec:diag_without_ancilla} and \ref{sec:diag_with_ancilla}. In Sections \ref{sec:QSP_graph_sub} and \ref{sec:US_graph_sub}, we present QSP and GUS circuits under path, $d$-dimensional grid, binary tree, expander graph and general graph constraints. In Section \ref{sec:circuit_transformation_sub}, we present a transformation between circuits under different graph constraints, which we use to upper bound the circuit depth for QSP and GUS under brick-wall constraint. Additional details and proofs are given in Appendix \ref{append:QSP_US_graph}.

\subsection{Circuit depth and size upper bounds for QSP}
\label{sec:QSP_graph_sub}
Our results for QSP in this section are based on the following lemma:
\begin{lemma}[\cite{grover2002creating,kerenidis2017quantum}]\label{lem:QSP_UCGs}
The QSP problem can be solved by $n$ UCGs $V_1,V_2,\ldots, V_n$ acting on $1,2,\ldots, n$ qubits, respectively.
\end{lemma}
Combining this with the fact that every $j$-qubit UCG $V_j$ can be decomposed into 3 $j$-qubit diagonal unitary matrices and 4 single-qubit gates (Lemma~\ref{lem:UCG_decomposition}), and using the results for diagonal unitary matrices from Sections \ref{sec:diag_without_ancilla} and \ref{sec:diag_with_ancilla}, we obtain the following circuit depth bounds for QSP.


\begin{theorem}\label{thm:QSP_grid_main}
Any $n$-qubit quantum state can be prepared by a circuit with $m\ge 0$ ancillary qubits, of depth 
\begin{enumerate}
    \item $O\left(2^{n/2}+\frac{2^n}{n+m}\right)$ under $\Path_{n+m}$ constraint.
    \item $O\left(2^{n/3}+\frac{2^{n/2}}{(n_2)^{1/2}}+\frac{2^n}{n+m}\right)$ under $\Grid^{n_1,n_2}_{n+m}$ constraint.
    \item $
   O\big(n^3+d2^{\frac{n}{d+1}}+\max\limits_{j\in\{2,\ldots,d\}}\big\{\frac{d2^{n/j}}{(\Pi_{i=j}^d n_i)^{1/j}}\big\}+\frac{2^n}{n+m}\big)$ under $\Grid_{n+m}^{n_1,n_2,\ldots,n_d}$ constraint. If $n_1=n_2=\cdots=n_d$, the depth is $O\left(n^3+d2^{\frac{n}{d+1}}+\frac{2^n}{n+m}\right)$.
\end{enumerate}
\end{theorem}

\begin{theorem}\label{thm:QSP_tree}
Any $n$-qubit quantum state can be prepared by a circuit with $m\ge 0$ ancillary qubits, of depth 
\begin{enumerate}
    \item $O\left(n^3\log(n)+\frac{\log(n)2^n}{n+m}\right)$ under $\Tree_{n+m}(2)$ constraint.
    \item $O\big(n^2d\log_d (n+m)\log_d(n+d)+\frac{(n+d)\log_d(n+d) 2^{n}}{n+m}\big)$, under $\Tree_{n+m}(d)$ constraint for $d<n+m$.
    \item $O(2^n)$ under $\Star_{n+m}$ constraint.
\end{enumerate}
\end{theorem}

\begin{restatable}{theorem}{qspexpander}\label{thm:QSP_expander}
Any $n$-qubit quantum state can be prepared by a quantum circuit of depth
$O\big(n^3+\frac{\log(n+m)2^n}{n+m}\big)$
under $\Expander_{n+m}$ constraint, using $m\ge 0$ ancillary qubits.
\end{restatable}

Under general graph constraints we have the following:
\begin{restatable}{theorem}{qspgraph}\label{thm:QSP_graph}
Any $n$-qubit quantum state can be prepared by a  quantum circuit of size and depth $O(2^n)$ under arbitrary graph $G$ constraint, using no ancillary qubits.
\end{restatable}

\subsection{Circuit depth and size upper bounds for GUS}
\label{sec:US_graph_sub}
Our results for GUS are based on the following lemma:
\begin{lemma}[\cite{mottonen2005decompositions}]\label{lem:unitary_UCGs}
Any $n$-qubit unitary matrix $U\in\mathbb{C}^{2^{n}\times 2^n}$ can be decomposed into $2^{n}-1$ $n$-qubit UCGs.
\end{lemma}
Combining this with Lemma~\ref{lem:UCG_decomposition} and the results for diagonal unitary matrices of Sections~\ref{sec:diag_without_ancilla} and \ref{sec:diag_with_ancilla} give the following circuit depth bounds for GUS:


%
\begin{theorem}\label{thm:US_path_grid}
Any $n$-qubit unitary can be realized by a quantum circuit with $m\ge 0$ ancillary qubits, of depth
\begin{enumerate}
    \item $O\big(4^{3n/4}+\frac{4^n}{n+m}\big)$ under $\Path_{n+m}$ constraint.
    \item $O\big(4^{2n/3}+\frac{4^{3n/4}}{(n_2)^{1/2}}+\frac{4^n}{n+m}\big)$ under $\Grid^{n_1,n_2}_{n+m}$ constraint.
    \item $O\big(n^22^n+d4^{\frac{(d+2)n}{2(d+1)}}+\max\limits_{j\in\{2,\ldots,d\}}\big\{\frac{d4^{(j+1)n/(2j)}}{(\Pi_{i=j}^d n_i)^{1/j}}\big\}+\frac{4^n}{n+m}\big)$ under $\Grid^{n_1,n_2,\ldots,n_d}_{n+m}$ constraint.  When $n_1=n_2=\cdots=n_d$, the depth is $O\big(n^22^n+d4^{\frac{(d+2)n}{2(d+1)}}+\frac{4^n}{n+m}\big)$.
\end{enumerate}
\end{theorem}

\begin{theorem}\label{thm:US_tree}
 Any $n$-qubit unitary can be realized by a quantum circuit with $m\ge 0$ ancillary qubits, of depth
\begin{enumerate}
\item $O\left(n^2\log(n)2^n+\frac{\log(n)4^n}{n+m}\right)$ under $\Tree_{n+m}(2)$ constraint.

\item $O\big(n2^nd\log_d (n+m)\log_d(n+d)+\frac{(n+d)\log_d(n+d) 4^{n}}{n+m}\big)$
under $\Tree_{n+m}(d)$ constraint for $d<n+m$.

\item $O\left(2^n\right)$ under $\Star_{n+m}$ constraint.
\end{enumerate}
\end{theorem}

\begin{restatable}{theorem}{usexpander}
\label{thm:US_expander_graph}
Any $n$-qubit unitary matrix can be realized by a quantum circuit of depth $O\big(n^22^n+\frac{\log(n+m)4^n}{n+m}\big),$ 
under $\Expander_{n+m}$ constraint, using $m\ge 0$ ancillary qubits.
\end{restatable}


\begin{restatable}{theorem}{usgraph}\label{thm:US_graph}
Any $n$-qubit unitary can be realized by a quantum circuit of size and depth $O(4^n)$ under general graph $G$ constraint.
\end{restatable}
\subsection{Circuit transformation between different graph constraints}
\label{sec:circuit_transformation_sub}
In this section, we first give a transformation between circuits under different graph constraints.  We then use this transformation to obtain QSP and GUS circuits under brick-wall constraint, by reduction to a 2-dimensional grid.
\begin{restatable}{lem}{circuittrans}\label{lem:circuit_trans}
Let $G=(V,E)$ and $G'=(V,E')$ be two graphs with common vertex set $V$, and with edge sets $E\subseteq E'$, $\abs{E'\setminus E} := \bigcup_{i=1}^c E_i$ such that, 
\begin{enumerate}
    \item $E_i \cap E_j = \emptyset$.
    \item  For each $i\in [c]$,  there are vertex disjoint paths $P_{st}$ in $G$ of length at most $c'$, connecting all edges $(v_s,v_t)\in E_i$.
    
\end{enumerate}
If $\mathcal{C}'$ is a circuit of depth $d$ and size $s$ under $G'$ constraint, there exists a circuit $\mathcal{C}$ implementing the same transformation of depth $O(cc'd)$ and size $O(c's)$, under $G$ constraint.
\end{restatable}
 
\begin{restatable}{coro}{qspbrickwall}\label{coro:QSP_brickwall}
Any $n$-qubit quantum state can be prepared by a circuit of depth \[O\big(2^{n/3}+\frac{2^{n/2}}{\sqrt{\min\{n_1,n_2\}}}+\frac{2^n}{n+m}\big)\] under $\brickwall^{n_1,n_2,b_1,b_2}_{n+m}$ constraint with $b_1, b_2=O(1)$, using $m\ge 0$ ancillary qubits. 
\end{restatable}

\begin{restatable}{coro}{usbrickwall}\label{coro:US_brickwall}
Any $n$-qubit unitary matrix can be implemented by a circuit of depth \[O\big(4^{2n/3}+\frac{4^{3n/4}}{\sqrt{\min\{n_1,n_2\}}}+\frac{4^n}{n+m}\big)\]
under $\brickwall^{n_1,n_2,b_1,b_2}_{n+m}$ constraint with $b_1, b_2=O(1)$, using $m\ge 0$ ancillary qubits.
\end{restatable}

\section{Circuit size and depth lower bounds under graph constraints}
\label{sec:QSP_US_lowerbound}
In this section, we give circuit depth and size lower bounds for QSP and GUS under graph constraints. Omitted proofs are given in Appendix \ref{append:QSP_US_lowerbound}.
\subsection{Circuit size and depth lower bounds under general graph constraints}
We first present circuit size and depth lower bounds for QSP and GUS under general graph constraints.

\begin{theorem}\label{thm:size_lowerbound_QSP_graph}
For an arbitrary graph $G=(V,E)$, there exist $n$-qubit states which can only be prepared by quantum circuits of size at least $\Omega(2^n)$ under $G$ constraint.
\end{theorem}

\begin{theorem}\label{thm:size_lowerbound_US_graph}
For an arbitrary graph $G=(V,E)$, there exist $n$-qubit unitaries which can be only prepared by quantum circuits of size at least $\Omega(4^n)$ under $G$ constraint.
\end{theorem}


To prove these lower bounds, we first associate a quantum circuit with a directed graph. (See Fig.~\ref{fig:example_directed_graph} for an example.)
\begin{definition}[Directed graphs for quantum circuits]\label{def:circuit-digraph_main}
Let $C$ be a quantum circuit on $n$ input and $m$ ancillary qubits consisting of $d$ depth-1 layers, with odd layers consisting only of single-qubit gates, even layers consisting only of CNOT gates, and any two (non-identity) single-qubit gates acting on the same qubit must be separated by at least one CNOT gate acting on that qubit.  Let $L_1,L_2,\cdots,L_d$ denote the $d$ layers of this circuit, i.e., $C=L_dL_{d-1}\cdots L_1$.  Define the directed graph $H=(V_C,E_C)$ associated with $C$ as follows. 
\begin{enumerate}
    \item Vertex set $V_C$: For each $i\in[d+1]$, define $S_{i}:=\{v_i^j:j\in[n+m]\}$, where $v_i^j$ is a label corresponding to the $j$-th qubit. Then, $V_C:=\bigcup_{i=1}^{d+1} S_i$. 
    \item Edge set $E_C$: For all $i\in [d]$:
    \begin{enumerate}
        \item 
        If there is a single-qubit gate acting on the $j$-th qubit in layer $L_i$ then, for all $i \le i' \le d$ there exists a directed edge $(v_{i'+1}^j,v_{i'}^j)$.
        
        \item If there is a CNOT gate acting on qubits $j_1$ and $j_2$ in layer $L_i$, then there exist $4$ directed edges $(v_{i+1}^{j_1},v_i^{j_1})$, $(v_{i+1}^{j_2}, v_i^{j_1})$, $(v_{i+1}^{j_1},v_i^{j_2})$ and $(v_{i+1}^{j_2},v_i^{j_2})$.
    \end{enumerate}
  Note that edges are directed from $S_{i+1}$ to $S_i$. 
\end{enumerate}
\end{definition}

\begin{figure}[!t]
\centering
\subfloat[]{
\begin{minipage}[]{0.4\textwidth}
\centerline{
\Qcircuit @C=0.8em @R=0.8em {
\lstick{\scriptstyle\ket{x_1}}&\gate{ } &\ctrl{1} & \qw & \qw & \gate{ } &\ctrl{1} &\gate{ } & \qw\\
\lstick{\scriptstyle\ket{x_2}}&\gate{ } & \targ & \gate{ } &\ctrl{1} & \gate{ } &\targ &\gate{ }& \qw\\
\lstick{\scriptstyle\ket{x_3}}&\gate{ } &\ctrl{1} &\gate{ } &\targ &\qw &\qw &\gate{ }& \qw\\
\lstick{\scriptstyle\ket{0}}&\gate{ } & \targ & \qw &\qw  &\gate{ } &\ctrl{1} &\gate{ }& \qw &\rstick{\scriptstyle \ket{0}}\\
\lstick{\scriptstyle\ket{0}}&\gate{ } &\ctrl{1} &\gate{ } &\ctrl{1} & \gate{ }&\targ &\gate{ }& \qw &\rstick{\scriptstyle \ket{0}}\\
\lstick{\scriptstyle\ket{0}}&\gate{ } &\targ &\gate{ } & \targ &\qw &\qw &\gate{ }& \qw &\rstick{\scriptstyle \ket{0}}\\
&{\scriptstyle L_1} & {\scriptstyle L_2} & {\scriptstyle L_3}  & {\scriptstyle L_4}  & {\scriptstyle L_5} & {\scriptstyle L_6}  & {\scriptstyle L_7} & \\
&&&&&&&&\\
}}%
\end{minipage}
}
\hfil
\subfloat[
]{ \begin{minipage}[]{0.4\textwidth}
    \begin{tikzpicture}
     \centering
     \draw [->] (3.5,1)--(3,1); \draw [->] (3.5,0.5)--(3,0.5); \draw [->] (3.5,0)--(3,0); \draw [->](3.5,-0.5)--(3,-0.5); \draw [->] (3.5,-1)--(3,-1);\draw [->] (3.5,-1.5)--(3,-1.5);
     \draw [->] (3,1)--(2.5,1); \draw [->] (3,1)--(2.5,0.5); \draw [->] (3,0.5)--(2.5,1);\draw [->] (3,0.5)--(2.5,0.5);\draw [->] (3,-0.5)--(2.5,-0.5);\draw [->] (3,-0.5)--(2.5,-1);\draw [->] (3,-1)--(2.5,-0.5);\draw [->] (3,-1)--(2.5,-1);
     \draw [->] (2.5,1)--(2,1); \draw [->] (2.5,0.5)--(2,0.5);\draw [->] (2.5,-0.5)--(2,-0.5);\draw [->] (2.5,-1)--(2,-1);
     \draw [->] (2,0.5)--(1.5,0.5); \draw [->](2,0)--(1.5,0);\draw [->] (2,0.5)--(1.5,0);\draw [->] (2,0)--(1.5,0.5);\draw [->] (2,-1)--(1.5,-1);\draw [->] (2,-1.5)--(1.5,-1.5);\draw [->] (2,-1)--(1.5,-1.5); \draw [->] (2,-1.5)--(1.5,-1);
     \draw [->] (1.5,0.5)--(1,0.5); \draw [->] (1.5,0)--(1,0); \draw [->] (1.5,-1)--(1,-1); \draw [->] (1.5,-1.5)--(1,-1.5);
     \draw [->] (1,1)--(0.5,1); \draw [->](1,0.5)--(0.5,0.5); \draw [->] (1,1)--(0.5,0.5); \draw [->] (1,0.5)--(0.5,1); \draw [->] (1,0)--(0.5,0); \draw [->] (1,-0.5)--(0.5,-0.5); \draw [->] (1,0)--(0.5,-0.5);\draw [->] (1,-0.5)--(0.5,0);\draw [->] (1,-1)--(0.5,-1);\draw [->] (1,-1.5)--(0.5,-1.5); \draw [->](1,-1)--(0.5,-1.5);\draw [->] (1,-1.5)--(0.5,-1);
     \draw [->] (0.5,1)--(0,1); \draw [->] (0.5,0.5)--(0,0.5);\draw [->] (0.5,0)--(0,0);\draw [->] (0.5,-0.5)--(0,-0.5);\draw [->] (0.5,-1)--(0,-1);\draw [->] (0.5,-1.5)--(0,-1.5);
     \draw (0,-2) node{\scriptsize $S_1$} (0.5,-2) node{\scriptsize $S_2$} (1,-2) node{\scriptsize $S_3$} (1.5,-2) node{\scriptsize $S_4$} (2,-2) node{\scriptsize $S_5$} (2.5,-2) node{\scriptsize $S_6$} (3,-2) node{\scriptsize $S_7$} (3.5,-2) node{\scriptsize $S_8$};
     \draw  (-0.1,1.1)--(0.1,1.1)--(0.1,-0.6)--(-0.1,-0.6)--cycle (0.4,1.1)--(0.6,1.1)--(0.6,-0.6)--(0.4,-0.6)--cycle (0.9,0.6)--(1.1,0.6)--(1.1,-0.1)--(0.9,-0.1)--cycle (1.4,0.6)--(1.6,0.6)--(1.6,-0.1)--(1.4,-0.1)--cycle
     (1.9,1.1)--(2.1,1.1)--(2.1,0.4)--(1.9,0.4)--cycle
     (2.4,1.1)--(2.6,1.1)--(2.6,0.4)--(2.4,0.4)--cycle
     (2.9,1.1)--(3.1,1.1)--(3.1,-0.1)--(2.9,-0.1)--cycle
     (3.4,1.1)--(3.6,1.1)--(3.6,-0.1)--(3.4,-0.1)--cycle;
     
     \draw[->,densely dotted,thick] (2.5,-1.5)--(2,-1.5);
     \draw[->,densely dotted,thick] (3,-1.5)--(2.5,-1.5);
     \draw[->,densely dotted,thick] (3,0)--(2.5,-0);
     \draw[->,densely dotted,thick] (2.5,0)--(2,-0);
     \draw[->,densely dotted,thick] (2,-0.5)--(1.5,-0.5);
     \draw[->,densely dotted,thick] (1.5,-0.5)--(1,-0.5);
     \draw[->,densely dotted,thick] (2,1)--(1.5,1);
     \draw[->,densely dotted,thick] (1.5,1)--(1,1);
      \draw [fill=black] (0,0) circle (0.05) (0.5,0) circle (0.05) (1,0) circle (0.05) (1.5,0) circle (0.05) (2,0) circle (0.05) (2.5,0) circle (0.05) (3,0) circle (0.05) (3.5,0) circle (0.05);
      \draw [fill=black] (0,0.5) circle (0.05) (0.5,0.5) circle (0.05) (1,0.5) circle (0.05) (1.5,0.5) circle (0.05) (2,0.5) circle (0.05) (2.5,0.5) circle (0.05) (3,0.5) circle (0.05) (3.5,0.5) circle (0.05);
     \draw [fill=black] (0,1) circle (0.05) (0.5,1) circle (0.05) (1,1) circle (0.05) (1.5,1) circle (0.05) (2,1) circle (0.05) (2.5,1) circle (0.05) (3,1) circle (0.05) (3.5,1) circle (0.05);
      \draw [fill=white] (0,-0.5) circle (0.05) (0.5,-0.5) circle (0.05) (1,-0.5) circle (0.05) (1.5,-0.5) circle (0.05) (2,-0.5) circle (0.05) (2.5,-0.5) circle (0.05) (3,-0.5) circle (0.05) (3.5,-0.5) circle (0.05);
     \draw [fill=white] (0,-1) circle (0.05) (0.5,-1) circle (0.05) (1,-1) circle (0.05) (1.5,-1) circle (0.05) (2,-1) circle (0.05) (2.5,-1) circle (0.05) (3,-1) circle (0.05) (3.5,-1) circle (0.05);
     \draw [fill=white] (0,-1.5) circle (0.05) (0.5,-1.5) circle (0.05) (1,-1.5) circle (0.05) (1.5,-1.5) circle (0.05) (2,-1.5) circle (0.05) (2.5,-1.5) circle (0.05) (3,-1.5) circle (0.05) (3.5,-1.5) circle (0.05);
     
    \draw (0,1.5) node{\scriptsize $S'_1$} (0.5,1.5) node{\scriptsize $S'_2$} (1,1.5) node{\scriptsize $S'_3$} (1.5,1.5) node{\scriptsize $S'_4$} (2,1.5) node{\scriptsize $S'_5$} (2.5,1.5) node{\scriptsize $S'_6$} (3,1.5) node{\scriptsize $S'_7$} (3.5,1.5) node{\scriptsize $S'_8$};
    \end{tikzpicture}
    \end{minipage}}
\caption{A quantum circuit $C$ and its corresponding directed graph $H=(V_C, E_C)$ with all arrows pointing from right to left. (a) A depth $d=7$ circuit $C=L_7L_6\cdots L_1$ on $n=3$ input and $m=3$ ancillary qubits. (b) The directed graph corresponding to $C$: $n+m$ vertices in  each layer $S_i$, for all $1\le i\le d+1$, with black (white) vertices corresponding to input (ancillary) qubits. $S'_i\subseteq S_i$ (boxes) denotes the reachable subset in layer $i$ in $H$ (see Def.~\ref{def:reachable_main}). A dotted arrow indicates no gate operation on the qubit in that layer.
}
\label{fig:example_directed_graph}
\end{figure}
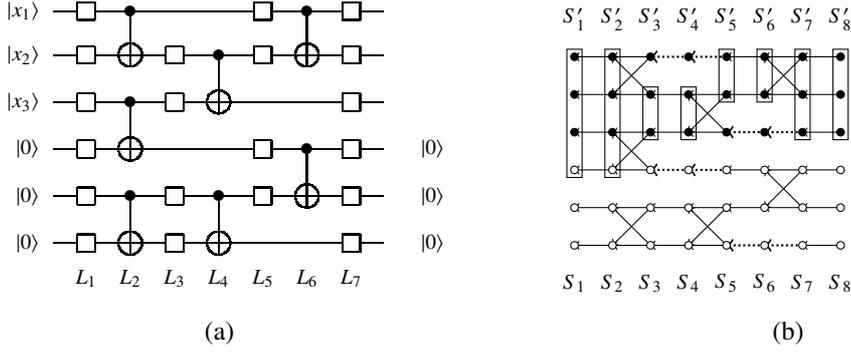

We then define reachable subsets in the directed graph.
\begin{definition}[Reachable subsets]\label{def:reachable_main} Let $H=(V_C,E_C)$ be the directed graph associated with quantum circuit $C$ of depth $d$, with vertex set $V_C = \bigcup_{i=1}^{d+1} S_i$.  For each $i\in[d+1]$ define the reachable subsets $S'_{i}$ of $H$ as follows:
\begin{itemize}
    \item $S'_{d+1} = \{v^j_{d+1} : j\in[n]\}$, i.e., the subset of $n$ vertices in $S_{d+1}$ corresponding to the $n$ input qubits.
    \item For $i\in[d]$, $S'_{i}\subseteq S_i$ is the subset of vertices $v^j_i$ in $S_i$ which are (i) reachable by a directed path from vertices in $S'_{d+1}$, and (ii) there is a quantum gate acting on qubit $j$ in circuit layer $L_{i}$. 
\end{itemize}
\end{definition}

Using reachable subsets, we finally present circuit depth lower bounds under arbitrary graph constraints.
\begin{restatable}{theorem}{reachablesetbounds}\label{thm:reachable-set-bounds}
Let $H=(V_C,E_C)$ be the directed graph associated with quantum circuit $C$, of depth $d$, acting on $n$ input and $m$ ancillary qubits. Let $S'_1, \ldots S'_{d+1}$ be the reachable subsets of $H$. 
\begin{enumerate}
    \item If $C$ is a circuit for any $n$-qubit quantum state, then $O(\sum_{i=1}^d |S_i'|)\ge 2^{n}-1$;
    \item If $C$ is a circuit for any $n$-qubit diagonal unitary matrix, then $O(\sum_{i=1}^d |S_i'|)\ge 2^{n}-1$;
    \item If $C$ is a circuit for any $n$-qubit general unitary matrix, then $O(\sum_{i=1}^d |S_i'|)\ge 4^{n}-1$.
\end{enumerate}
\end{restatable}

By constructing reachable subsets for the constraint graph, we present circuit depth lower bounds based on Theorem \ref{thm:reachable-set-bounds}.
\begin{theorem}\label{thm:depth_lowerbound_QSP_graph}
    Let $G=(V,E)$ denote an arbitrary connected graph with $n+m$ vertices, with $\nu$ the size of a maximum matching of $G$. There exist $n$-qubit quantum states which can only be prepared by circuits of depth at least $\Omega\left(\max\{n,2^n/\nu\}\right)$  under $G$ constraint, using $m\ge 0$ ancillary qubits.
\end{theorem}

\begin{theorem}\label{thm:depth_lowerbound_US_graph}
Let $G=(V,E)$ denote an arbitrary connected graph with $n+m$ vertices, with $\nu$ the size of a maximum matching of $G$. There exist $n$-qubit unitaries which can only be implemented by circuits of depth at least $\Omega\left(\max\{n,4^n/\nu\}\right)$  under $G$ constraint, using $m\ge 0$ ancillary qubits.
\end{theorem}
\subsection{Circuit depth lower bounds under specific graph constraints}
Theorems \ref{thm:reachable-set-bounds}, \ref{thm:depth_lowerbound_QSP_graph} and \ref{thm:depth_lowerbound_US_graph} can be used to obtain circuit depth lower bounds under specific graph constraints.
\begin{restatable}{theorem}{lowerboundgrid}\label{thm:lower_bound_grid_k_QSP}
There exists an $n$-qubit quantum state that requires a circuit using $m\ge 0$ ancillary qubits, of depth 
\begin{enumerate}
    \item $\Omega\big(\max\big\{2^{n/2},\frac{2^n}{n+m}\big\}\big)$ under $\Path_{n+m}$ constraint.
    \item $\Omega\big(\max\big\{ 2^{n/3},\frac{2^{n/2}}{\sqrt{n_2}},\frac{2^n}{n+m}\big\}\big)$ under $\Grid^{n_1,n_2}_{n+m}$ constraint.
    \item $\Omega\big( \max\limits_{j\in [d]}\big\{n,2^{\frac{n}{d+1}},\frac{2^{n/j}}{(\Pi_{i=j}^d n_i)^{1/j}}\big\}\big)$ under $\Grid^{n_1,\cdots,n_d}_{n+m}$ constraint.
\end{enumerate}
\end{restatable}

\begin{restatable}{theorem}{lowerboundtreeus}\label{thm:lower_bound_tree_QSP}
There exists an $n$-qubit quantum state that requires a circuit using $m\ge 0$ ancillary qubits, of depth 
\begin{enumerate}
    \item $\Omega\big(\max\big\{n,\frac{2^n}{n+m}\big\}\big)$  under $\Tree_{n+m}(2)$ constraint.
    \item $\Omega\big(\max\big\{n,\frac{d2^n}{n+m}\big\}\big)$ under $\Tree_{n+m}(d)$ constraint for $d<m+n$.
    \item $\Omega(2^n)$ under $\Star_{n+m}$ constraint.
\end{enumerate}
\end{restatable}

\begin{theorem}\label{thm:lower_exapnder_QSP}
There exists an $n$-qubit quantum state that requires a quantum circuit of depth $\Omega\big(\max\{n, \frac{2^n}{n+m}\}\big)$ under $\Expander_{n+m}$ constraint,  using $m\ge 0$ ancillary qubits.
\end{theorem}

\begin{restatable}{theorem}{qspbrickwalllowerbound}
\label{thm:lower_bound_brickwall_QSP}
There exists an $n$-qubit quantum state that requires a quantum circuit of depth \[\Omega\big(\max\big\{2^{n/3},\frac{2^{n/2}}{\sqrt{\min\{n_1,n_2\}}},\frac{2^n}{n+m}\big\}\big)\] under $\brickwall_{n+m}^{n_1,n_2,b_1,b_2}$ constraint, using $m\ge 0$ ancillary qubits.
\end{restatable}

\begin{restatable}{theorem}{lowerboundgridus}
\label{thm:lower_bound_grid_k_US}
There exists an $n$-qubit unitary that requires a quantum circuit using $m\ge 0$ ancillary qubits, of depth 
\begin{enumerate}
    \item $\Omega\big(\max\big\{4^{n/2},\frac{4^n}{n+m}\big\}\big)$ under $\Path_{n+m}$ constraint.
    \item $\Omega\big(\max\big\{4^{n/3},\frac{4^{n/2}}{\sqrt{n_2}},\frac{4^n}{n+m}\big\}\big)$ under $\Grid^{n_1,n_2}_{n+m}$ constraint.
    \item $\Omega\big( \max\limits_{j\in [d]}\big\{n,4^{\frac{n}{d+1}},\frac{4^{n/j}}{(\Pi_{i=j}^d n_i)^{1/j}}\big\}\big)$ under $\Grid^{n_1,\cdots,n_d}_{n+m}$ constraint.
\end{enumerate}
\end{restatable}

\begin{theorem}\label{thm:lower_bound_tree_US}
There exists an $n$-qubit unitary that requires a quantum circuit using $m\ge 0$ ancillary qubits, of depth 
\begin{enumerate}
    \item $\Omega\big(\max\big\{n,\frac{4^n}{n+m}\big\}\big)$  under $\Tree_{n+m}(2)$ constraint.
    \item $\Omega\big(\max\big\{n,\frac{d4^n}{n+m}\big\}\big)$ under $\Tree_{n+m}(d)$ constraint for $d<m+n$.
        \item $\Omega(4^n)$ under $\Star_{n+m}$ constraint.
\end{enumerate}
\end{theorem}

\begin{theorem}\label{thm:lower_exapnder_US}
There exists an $n$-qubit unitary that requires a quantum circuit of depth $\Omega\big(\max\{n, \frac{2^n}{n+m}\}\big)$ under $\Expander_{n+m}$ constraint,  using $m\ge 0$ ancillary qubits.
\end{theorem}

\begin{restatable}{theorem}{usbrickwalllowerbound}\label{thm:lower_bound_brickwall_US}
There exists an $n$-qubit unitary that requires a quantum circuit of depth \[\Omega\big(\max\big\{4^{n/3},\frac{4^{n/2}}{\sqrt{\min\{n_1,n_2\}}},\frac{4^n}{n+m}\big\}\big)\] under $\brickwall_{n+m}^{n_1,n_2,b_1,b_2}$ constraints  using $m\ge 0$ ancillary qubits.
\end{restatable}

\section{Conclusions}
\label{sec:conlusion}
We have investigated the effects of qubit connectivity on quantum circuit size and depth complexity. We have shown that, somewhat surprisingly, connectivity constraints do not increase the order of circuit size required for implementing almost all unitaries, as well as for quantum state preparation. 

The circuit depth complexity is more subtle. We have shown that connectivity constraints do not increase the order of the circuit depth required for implementing almost all unitary operations, even for the very restricted case of 1D chains with nearest neighbor connectivity, and this remains true when $m$ ancilla are available unless $m$ is exponentially large. However, compared with the unrestricted case, qubit connectivity does hinder space-depth trade-offs: it makes it harder to use a large number of ancilla qubits to achieve smaller depth. 

We have investigated various constraint graphs, including $d$-dimensional grids, complete $d$-ary trees, expander graphs, and general graphs. We have found that common measures for graph connectivity such as graph diameter, vertex degree, and graph expansion, as well as less prominent measures such as the size of a maximum matching, all seem to have some impact on the required circuit depth. 

These results combine analytic bounds with explicit circuit constructions, which hopefully have practical applications for circuit design as well. A number of interesting related research directions warrant futher study:

\begin{enumerate}
    \item Better bounds. Gaps remain between upper and lower bounds for GUS in the {$d$-dimensional grid and $d$-ary tree} cases {when the number of ancillary qubits is large}. It would be technically interesting to close them in these settings.
    \item More graph properties. What other graph properties have an important impact on quantum circuit depth complexity for certain natural families of unitaries? 
    \item More unitary families. We cannot hope to have an efficient algorithm to optimize the circuit complexity for any given unitary as it is QMA-hard \cite{Janzing2005identity}, but it would be interesting to have more circuit constructions for specific unitaries. Can we study some other families of unitaries which have structures that can be exploited to give efficient circuit constructions? 
    
    \item Small scale quantum circuits. Though our designs aim at achieving optimal asymptotic bounds, the constant factor hidden in the big-O notation is not large, and we hope our constructions may inspire efficient constructions for small scale quantum circuits, such as those on $10^2 \sim 10^5$ qubits. Our constructed circuits are all parameterized ones, which may have applications in designing ansatzes for variational quantum circuits for quantum machine learning or quantum chemistry. 
\end{enumerate}


    
\bibliographystyle{alpha}
\bibliography{qsp_graph_constraint}

\newcommand{\etalchar}[1]{$^{#1}$}
\begin{thebibliography}{YGW{\etalchar{+}}19}

\bibitem[AAA{\etalchar{+}}23]{acharya2022suppressing}
Rajeev Acharya, Igor Aleiner, Richard Allen, et~al.
\newblock Suppressing quantum errors by scaling a surface code logical qubit.
\newblock {\em Nature}, 614(7949):676--681, 2023.

\bibitem[AAB{\etalchar{+}}19]{arute2019quantum}
Frank Arute, Kunal Arya, Ryan Babbush, Dave Bacon, Joseph~C Bardin, Rami
  Barends, Rupak Biswas, Sergio Boixo, Fernando~GSL Brandao, David~A Buell,
  et~al.
\newblock Quantum supremacy using a programmable superconducting processor.
\newblock {\em Nature}, 574(7779):505--510, 2019.

\bibitem[BAN11]{buluta2011natural}
Iulia Buluta, Sahel Ashhab, and Franco Nori.
\newblock Natural and artificial atoms for quantum computation.
\newblock {\em Reports on Progress in Physics}, 74(10):104401, 2011.

\bibitem[BCK15]{berry2015hamiltonian}
Dominic~W Berry, Andrew~M Childs, and Robin Kothari.
\newblock Hamiltonian simulation with nearly optimal dependence on all
  parameters.
\newblock In {\em 2015 IEEE 56th Annual Symposium on Foundations of Computer
  Science}, pages 792--809. IEEE, 2015.

\bibitem[Blo08]{bloch2008quantum}
Immanuel Bloch.
\newblock Quantum coherence and entanglement with ultracold atoms in optical
  lattices.
\newblock {\em Nature}, 453(7198):1016--1022, 2008.

\bibitem[BR12]{blatt2012quantum}
Rainer Blatt and Christian~F Roos.
\newblock Quantum simulations with trapped ions.
\newblock {\em Nature Physics}, 8(4):277--284, 2012.

\bibitem[BSK{\etalchar{+}}17]{bernien2017probing}
Hannes Bernien, Sylvain Schwartz, Alexander Keesling, Harry Levine, Ahmed
  Omran, Hannes Pichler, Soonwon Choi, Alexander~S Zibrov, Manuel Endres,
  Markus Greiner, et~al.
\newblock Probing many-body dynamics on a 51-atom quantum simulator.
\newblock {\em Nature}, 551(7682):579--584, 2017.

\bibitem[BVMS05]{bergholm2005quantum}
Ville Bergholm, Juha~J Vartiainen, Mikko M{\"o}tt{\"o}nen, and Martti~M
  Salomaa.
\newblock Quantum circuits with uniformly controlled one-qubit gates.
\newblock {\em Physical Review A}, 71(5):052330, 2005.

\bibitem[CSH{\etalchar{+}}00]{ciorga2000addition}
M~Ciorga, AS~Sachrajda, Pawel Hawrylak, C~Gould, Piotr Zawadzki, S~Jullian,
  Y~Feng, and Zbigniew Wasilewski.
\newblock Addition spectrum of a lateral dot from coulomb and spin-blockade
  spectroscopy.
\newblock {\em Physical Review B}, 61(24):R16315, 2000.

\bibitem[EHG{\etalchar{+}}03]{elzerman2003few}
JM~Elzerman, R~Hanson, JS~Greidanus, LH~Willems Van~Beveren, S~De~Franceschi,
  LMK Vandersypen, S~Tarucha, and LP~Kouwenhoven.
\newblock Few-electron quantum dot circuit with integrated charge read out.
\newblock {\em Physical Review B}, 67(16):161308, 2003.

\bibitem[Fra53]{frank1953pulse}
Gray Frank.
\newblock Pulse code communication, March~17 1953.
\newblock US Patent 2,632,058.

\bibitem[Gid15]{multi-controlled-gate}
Craig Gidney.
\newblock
  https://algassert.com/circuits/2015/06/22/{U}sing-{Q}uantum-{G}ates-instead-of-{A}ncilla-{B}its.html.
\newblock 2015.

\bibitem[Gil58]{gilbert1958gray}
Edgard~N Gilbert.
\newblock Gray codes and paths on the n-cube.
\newblock {\em The bell system technical journal}, 37(3):815--826, 1958.

\bibitem[GKG{\etalchar{+}}19]{graham2019rydberg}
TM~Graham, M~Kwon, B~Grinkemeyer, Z~Marra, X~Jiang, MT~Lichtman, Y~Sun,
  M~Ebert, and M~Saffman.
\newblock Rydberg-mediated entanglement in a two-dimensional neutral atom qubit
  array.
\newblock {\em Physical review letters}, 123(23):230501, 2019.

\bibitem[GR02]{grover2002creating}
Lov Grover and Terry Rudolph.
\newblock Creating superpositions that correspond to efficiently integrable
  probability distributions.
\newblock {\em arXiv preprint quant-ph/0208112}, 2002.

\bibitem[GWZ{\etalchar{+}}21]{gong2021quantum}
Ming Gong, Shiyu Wang, Chen Zha, Ming-Cheng Chen, He-Liang Huang, Yulin Wu,
  Qingling Zhu, Youwei Zhao, Shaowei Li, Shaojun Guo, et~al.
\newblock Quantum walks on a programmable two-dimensional 62-qubit
  superconducting processor.
\newblock {\em Science}, 372(6545):948--952, 2021.

\bibitem[Her20]{herbert2018depth}
Steven Herbert.
\newblock On the depth overhead incurred when running quantum algorithms on
  near-term quantum computers with limited qubit connectivity.
\newblock {\em Quantum Information \& Computation}, 20(9-10):787--806, 2020.

\bibitem[HHL09]{harrow2009quantum}
Aram~W Harrow, Avinatan Hassidim, and Seth Lloyd.
\newblock Quantum algorithm for linear systems of equations.
\newblock {\em Physical review letters}, 103(15):150502, 2009.

\bibitem[HLW06]{hoory2006expander}
Shlomo Hoory, Nathan Linial, and Avi Wigderson.
\newblock Expander graphs and their applications.
\newblock {\em Bulletin of the American Mathematical Society}, 43(4):439--561,
  2006.

\bibitem[IBM21]{IBMQ}
{IBM} quantum.
\newblock {\em https://quantum-computing.ibm.com/}, 2021.

\bibitem[JDM{\etalchar{+}}21]{johri2021nearest}
Sonika Johri, Shantanu Debnath, Avinash Mocherla, Alexandros Singk, Anupam
  Prakash, Jungsang Kim, and Iordanis Kerenidis.
\newblock Nearest centroid classification on a trapped ion quantum computer.
\newblock {\em npj Quantum Information}, 7(1):1--11, 2021.

\bibitem[Jor21]{quantumalgorithm}
Stephen Jordan.
\newblock Quantum algorithm zoo.
\newblock {\em https://quantumalgorithmzoo.org/}, 2021.

\bibitem[JWB05]{Janzing2005identity}
Dominik Janzing, Pawel Wocjan, and Thomas Beth.
\newblock Non-identity-check is {QMA}-complete.
\newblock {\em International Journal of Quantum Information}, 03(03):463--473,
  2005.

\bibitem[KBF{\etalchar{+}}15]{kelly2015state}
Julian Kelly, Rami Barends, Austin~G Fowler, Anthony Megrant, Evan Jeffrey,
  Theodore~C White, Daniel Sank, Josh~Y Mutus, Brooks Campbell, Yu~Chen, et~al.
\newblock State preservation by repetitive error detection in a superconducting
  quantum circuit.
\newblock {\em Nature}, 519(7541):66--69, 2015.

\bibitem[KL21]{kerenidis2020quantum}
Iordanis Kerenidis and Jonas Landman.
\newblock Quantum spectral clustering.
\newblock {\em Physical Review A}, 103(4):042415, 2021.

\bibitem[KLLP19]{kerenidis2019q}
Iordanis Kerenidis, Jonas Landman, Alessandro Luongo, and Anupam Prakash.
\newblock q-means: A quantum algorithm for unsupervised machine learning.
\newblock {\em Advances in Neural Information Processing Systems},
  32:4134--4144, 2019.

\bibitem[KP17]{kerenidis2017quantum}
Iordanis Kerenidis and Anupam Prakash.
\newblock Quantum recommendation systems.
\newblock In {\em 8th Innovations in Theoretical Computer Science Conference
  (ITCS 2017)}. Schloss Dagstuhl-Leibniz-Zentrum fuer Informatik, 2017.

\bibitem[LBMW03]{leibfried2003quantum}
Dietrich Leibfried, Rainer Blatt, Christopher Monroe, and David Wineland.
\newblock Quantum dynamics of single trapped ions.
\newblock {\em Reviews of Modern Physics}, 75(1):281, 2003.

\bibitem[LC17]{low2017optimal}
Guang~Hao Low and Isaac~L Chuang.
\newblock Optimal hamiltonian simulation by quantum signal processing.
\newblock {\em Physical review letters}, 118(1):010501, 2017.

\bibitem[LC19]{low2019hamiltonian}
Guang~Hao Low and Isaac~L Chuang.
\newblock Hamiltonian simulation by qubitization.
\newblock {\em Quantum}, 3:163, 2019.

\bibitem[LKS18]{low2018trading}
Guang~Hao Low, Vadym Kliuchnikov, and Luke Schaeffer.
\newblock Trading t-gates for dirty qubits in state preparation and unitary
  synthesis.
\newblock {\em arXiv preprint arXiv:1812.00954}, 2018.

\bibitem[LMR14]{lloyd2014quantum}
Seth Lloyd, Masoud Mohseni, and Patrick Rebentrost.
\newblock Quantum principal component analysis.
\newblock {\em Nature Physics}, 10(9):631--633, 2014.

\bibitem[MLA{\etalchar{+}}22]{madsen2022quantum}
Lars~S Madsen, Fabian Laudenbach, Mohsen~Falamarzi Askarani, Fabien Rortais,
  Trevor Vincent, Jacob~FF Bulmer, Filippo~M Miatto, Leonhard Neuhaus, Lukas~G
  Helt, Matthew~J Collins, et~al.
\newblock Quantum computational advantage with a programmable photonic
  processor.
\newblock {\em Nature}, 606(7912):75--81, 2022.

\bibitem[MV06]{mottonen2005decompositions}
Mikka M{\"o}tt{\"o}nen and Juha~J Vartiainen.
\newblock Decompositions of general quantum gates.
\newblock {\em Trends in Quantum Computing Research}, 2006.

\bibitem[PB11]{plesch2011quantum}
Martin Plesch and {\v{C}}aslav Brukner.
\newblock Quantum-state preparation with universal gate decompositions.
\newblock {\em Physical Review A}, 83(3):032302, 2011.

\bibitem[PFM{\etalchar{+}}21]{pogorelov2021compact}
Ivan Pogorelov, Thomas Feldker, Ch~D Marciniak, Lukas Postler, Georg Jacob,
  Oliver Krieglsteiner, Verena Podlesnic, Michael Meth, Vlad Negnevitsky,
  Martin Stadler, et~al.
\newblock Compact ion-trap quantum computing demonstrator.
\newblock {\em PRX Quantum}, 2(2):020343, 2021.

\bibitem[PJM{\etalchar{+}}04]{petta2004manipulation}
JR~Petta, AC~Johnson, CM~Marcus, MP~Hanson, and AC~Gossard.
\newblock Manipulation of a single charge in a double quantum dot.
\newblock {\em Physical review letters}, 93(18):186802, 2004.

\bibitem[RML14]{rebentrost2014quantum}
Patrick Rebentrost, Masoud Mohseni, and Seth Lloyd.
\newblock Quantum support vector machine for big data classification.
\newblock {\em Physical review letters}, 113(13):130503, 2014.

\bibitem[Ros13]{rosenbaum2013optimal}
David~J Rosenbaum.
\newblock Optimal quantum circuits for nearest-neighbor architectures.
\newblock In {\em 8th Conference on the Theory of Quantum Computation,
  Communication and Cryptography}, page 294, 2013.

\bibitem[Ros21]{rosenthal2021query}
Gregory Rosenthal.
\newblock Query and depth upper bounds for quantum unitaries via grover search.
\newblock {\em arXiv preprint arXiv:2111.07992}, 2021.

\bibitem[RSML18]{rebentrost2018quantum}
Patrick Rebentrost, Adrian Steffens, Iman Marvian, and Seth Lloyd.
\newblock Quantum singular-value decomposition of nonsparse low-rank matrices.
\newblock {\em Physical review A}, 97(1):012327, 2018.

\bibitem[Sav97]{savage1997survey}
Carla Savage.
\newblock A survey of combinatorial gray codes.
\newblock {\em SIAM review}, 39(4):605--629, 1997.

\bibitem[SGG{\etalchar{+}}07]{schroer2007electrostatically}
D~Schr{\"o}er, AD~Greentree, L~Gaudreau, K~Eberl, LCL Hollenberg, JP~Kotthaus,
  and S~Ludwig.
\newblock Electrostatically defined serial triple quantum dot charged with few
  electrons.
\newblock {\em Physical Review B}, 76(7):075306, 2007.

\bibitem[SMB04]{shende2004minimal}
Vivek~V Shende, Igor~L Markov, and Stephen~S Bullock.
\newblock Minimal universal two-qubit controlled-not-based circuits.
\newblock {\em Physical Review A}, 69(6):062321, 2004.

\bibitem[SNM{\etalchar{+}}13]{schindler2013quantum}
Philipp Schindler, Daniel Nigg, Thomas Monz, Julio~T Barreiro, Esteban
  Martinez, Shannon~X Wang, Stephan Quint, Matthias~F Brandl, Volckmar
  Nebendahl, Christian~F Roos, et~al.
\newblock A quantum information processor with trapped ions.
\newblock {\em New Journal of Physics}, 15(12):123012, 2013.

\bibitem[STY{\etalchar{+}}23]{sun2021asymptotically}
Xiaoming Sun, Guojing Tian, Shuai Yang, Pei Yuan, and Shengyu Zhang.
\newblock Asymptotically optimal circuit depth for quantum state preparation
  and general unitary synthesis.
\newblock {\em IEEE Transactions on Computer-Aided Design of Integrated
  Circuits and Systems}, 2023.

\bibitem[WHY{\etalchar{+}}19]{wu2019optimization}
Bujiao Wu, Xiaoyu He, Shuai Yang, Lifu Shou, Guojing Tian, Jialin Zhang, and
  Xiaoming Sun.
\newblock Optimization of cnot circuits on topological superconducting
  processors.
\newblock {\em arXiv preprint arXiv:1910.14478}, 2019.

\bibitem[WLH{\etalchar{+}}18]{wang201818}
Xi-Lin Wang, Yi-Han Luo, He-Liang Huang, Ming-Cheng Chen, Zu-En Su, Chang Liu,
  Chao Chen, Wei Li, Yu-Qiang Fang, Xiao Jiang, et~al.
\newblock 18-qubit entanglement with six photons’ three degrees of freedom.
\newblock {\em Physical review letters}, 120(26):260502, 2018.

\bibitem[WZP18]{wossnig2018quantum}
Leonard Wossnig, Zhikuan Zhao, and Anupam Prakash.
\newblock Quantum linear system algorithm for dense matrices.
\newblock {\em Physical review letters}, 120(5):050502, 2018.

\bibitem[YGW{\etalchar{+}}19]{ye2019propagation}
Yangsen Ye, Zi-Yong Ge, Yulin Wu, Shiyu Wang, Ming Gong, Yu-Ran Zhang, Qingling
  Zhu, Rui Yang, Shaowei Li, Futian Liang, et~al.
\newblock Propagation and localization of collective excitations on a 24-qubit
  superconducting processor.
\newblock {\em Physical review letters}, 123(5):050502, 2019.

\bibitem[YZ23]{yuan2022optimal}
Pei Yuan and Shengyu Zhang.
\newblock Optimal (controlled) quantum state preparation and improved unitary
  synthesis by quantum circuits with any number of ancillary qubits.
\newblock {\em Quantum}, 7:956, 2023.

\bibitem[ZHM{\etalchar{+}}16]{zajac2016scalable}
DM~Zajac, TM~Hazard, Xiao Mi, E~Nielsen, and Jason~R Petta.
\newblock Scalable gate architecture for a one-dimensional array of
  semiconductor spin qubits.
\newblock {\em Physical Review Applied}, 6(5):054013, 2016.

\bibitem[ZWD{\etalchar{+}}20]{zhong2020quantum}
Han-Sen Zhong, Hui Wang, Yu-Hao Deng, Ming-Cheng Chen, Li-Chao Peng, Yi-Han
  Luo, Jian Qin, Dian Wu, Xing Ding, Yi~Hu, et~al.
\newblock Quantum computational advantage using photons.
\newblock {\em Science}, 370(6523):1460--1463, 2020.

\bibitem[ZYY21]{zhang2021low}
Xiao-Ming Zhang, Man-Hong Yung, and Xiao Yuan.
\newblock Low-depth quantum state preparation.
\newblock {\em Physical Review Research}, 3(4):043200, 2021.

\end{thebibliography}

\appendix

\section{Basic quantum gates and circuits}
\label{append:basic_circuit}
In this section, we show definitions and implementations of some basic quantum gates and circuits which are used in the main text and other appendices.

\subsection{Proof of Lemma \ref{lem:cnot_path_constraint}}


\cnotpath*
\begin{proof}
Let the nodes along the shortest path in $G$ from $u$ to $v$ be $u_0, u_1, \cdots, u_{d}$ where $d=d(u,v)$, $u=u_0$ and $v=u_d$. $\Cnot_v^u$ can be implemented by the circuit in Fig. \ref{fig:cnot_path}, which has depth and size $O(d)$, and consists only of CNOT gates acting on adjacent qubits.
\begin{figure}[!hbt]
\centerline 
{
\Qcircuit @C=0.2em @R=0.5em {
\lstick{\scriptstyle u=u_0 ~~\ket{x}}&\ctrl{5}&\qw &\push{\scriptstyle \ket{x}}&&&\push{\scriptstyle \ket{x}}& \qw & \qw &\qw & \qw & \ctrl{1} & \qw & \qw & \qw & \qw & \qw &\qw &\qw &\qw &\ctrl{1} &\qw &\qw &\qw &\qw &\qw &\qw &\rstick{\scriptstyle \ket{x}}\\
\lstick{\scriptstyle u_1~~\ket{y_1}} &\qw &\qw &\push{\scriptstyle\ket{y_1}}&&&\push{\scriptstyle \ket{y_1}}&\qw &\qw &\qw & \ctrl{1} & \targ & \ctrl{1} &\qw & \qw & \qw & \qw & \qw& \qw &\ctrl{1} &\targ&\ctrl{1} &\qw &\qw &\qw &\qw &\qw &\rstick{\scriptstyle \ket{y_1}}\\
\lstick{\scriptstyle u_2~~\ket{y_2}} &\qw &\qw &\push{\scriptstyle\ket{y_2}}&&&\push{\scriptstyle \ket{y_2}}&\qw &\qw & \ctrl{1} & \targ & \qw & \targ & \ctrl{1} & \qw &\qw & \qw &\qw &\ctrl{1} &\targ &\qw &\targ& \qw &\ctrl{1}&\qw &\qw &\qw &\rstick{\scriptstyle \ket{y_2}}\\
\lstick{\vdots~~}&\qw &\qw &\vdots&=&&\push{\vdots~~}&\qw & \ctrl{1} &\targ & \qw & \qw &  \qw & \targ & \ctrl{1} & \qw &\qw &\ctrl{1}&\targ &\qw &\qw &\qw &\qw &\targ&\qw&\ctrl{1}&\qw &\rstick{\vdots}\\
\lstick{\scriptstyle u_{d-1}~~\ket{y_{d-1}}}&\qw &\qw &\push{\scriptstyle\ket{y_{d-1}}}&&&\push{\scriptstyle \ket{y_{d-1}}}&\ctrl{1} & \targ &\qw & \qw & \qw & \qw & \qw & \targ & \ctrl{1} &\qw &\targ &\qw &\qw &\qw &\qw &\qw &\qw& \qw &\targ &\qw &\rstick{\scriptstyle \ket{y_{d-1}}}\\
\lstick{\scriptstyle v=u_{d}~~\ket{y_d}}&\targ&\qw &\push{\scriptstyle\ket{x\oplus y_d}}&&&\push{\scriptstyle \ket{y_d}}&\targ & \qw &\qw & \qw & \qw & \qw & \qw & \qw &\targ &\qw &\qw &\qw &\qw &\qw &\qw &\qw &\qw &\qw &\qw &\qw &\rstick{\scriptstyle \ket{x\oplus y_d}}\\
}
}
\caption{Implementation of a $\Cnot_v^u$ gate under path constraint by $O(d)$ CNOT gates acting on adjacent qubits.}\label{fig:cnot_path}
\end{figure}
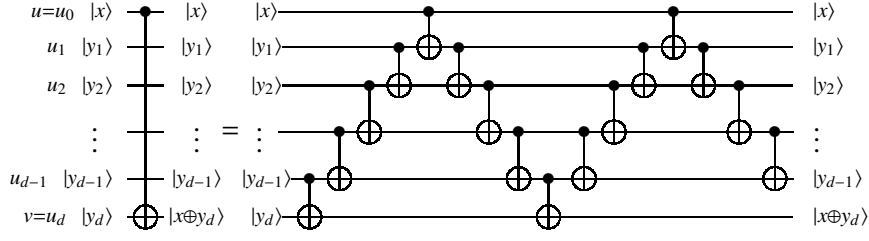
\end{proof}

\subsection{Basic quantum gates and circuits}

The SWAP gate $\textsf{SWAP}^{i}_{j}$ implements $\textsf{SWAP}^{i}_{j}\ket{x}_{i}\ket{y}_{j}= \ket{y}_{i}\ket{x}_{j}$ for any $x,y\in \B$, and can be realized by three CNOT gates, viz., $\textsf{SWAP}^{i}_{j}=\textsf{CNOT}^{i}_{j}\textsf{CNOT}^{j}_{i}\textsf{CNOT}^{i}_{j}$. The implementation of SWAP gates are used in Appendices \ref{append:diag_without_ancilla},\ref{append:diag_with_ancilla}, and \ref{append:binary_tree_improvement}. 

Two natural extensions of the CNOT gate are the Toffoli (multi-controlled not) gate and the multi-target CNOT gate. For qubit set $S$ and string $y\in \B^{|S|}$, the $(|S|+1)$-qubit Toffoli gate $\textsf{Tof}^S_i(y)$ is defined as 
\[\ket{x}_S\ket{b}_i\to\ket{x}_S\ket{1_{[x=y]}\oplus b}_i, \forall x\in \B^{|S|},\forall b\in \B,\]
where $1_{[x=y]} = 1$ if $x=y$, and 0 otherwise. Here $S$ is the control qubit set and $i$ is the target qubit. This extends the standard Toffoli gate in which $y = 11...1$.
\begin{lemma}[\cite{multi-controlled-gate}]\label{lem:tof}
An $n$-qubit Toffoli gate ${\sf Tof}^S_i(y)$ can be implemented by a quantum circuit of size and depth $O(n)$.
\end{lemma}
The definition and implementation of Toffoli gates is used in Appendix \ref{append:binary_tree_improvement}.

\begin{lemma}\label{lem:multicontrolcnot}
The $n+1$ qubit multi-target $\Cnot^i_{j_1, \ldots, j_n}$ gate can be implemented by a CNOT circuit consisting of $2n-1$ $\mathsf{CNOT}$ gates under $\Path_{n+1}$ constraint (see Fig.~\ref{fig:add-circuit}).
\end{lemma}
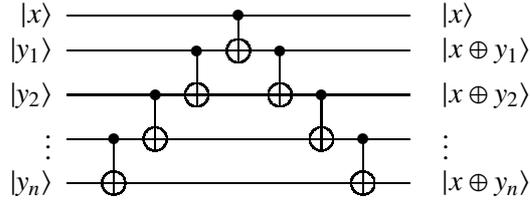
\begin{figure}[!hbt]
\centerline 
{
\Qcircuit @C=0.6em @R=0.7em {
\lstick{\ket{x}}&\qw & \qw &\qw & \qw & \ctrl{1} & \qw & \qw & \qw & \qw & \qw  &\rstick{\ket{x}} \\
\lstick{\ket{y_1}} &\qw &\qw &\qw & \ctrl{1} & \targ & \ctrl{1} &\qw & \qw & \qw & \qw  &\rstick{\ket{x\oplus y_1}}\\
\lstick{\ket{y_2}} &\qw &\qw & \ctrl{1} & \targ & \qw & \targ & \ctrl{1} & \qw &\qw & \qw   &\rstick{\ket{x\oplus y_2}}\\
\lstick{\vdots}&\qw & \ctrl{1} &\targ & \qw & \qw &  \qw & \targ & \ctrl{1} & \qw &\qw   &\rstick{\vdots}\\
\lstick{\ket{y_{n}}}&\qw & \targ &\qw & \qw & \qw & \qw & \qw & \targ & \qw &\qw   &\rstick{\ket{x\oplus y_{n}}}
}
}\caption{CNOT circuit for implementing the $n+1$ qubit multi-target $\mathsf{CNOT}$ gate.}\label{fig:add-circuit}
\end{figure}

\begin{lemma}[\cite{mottonen2005decompositions}]\label{lem:diag_size}
Any $n$-qubit diagonal unitary matrix $\Lambda_n$ can be implemented by a quantum circuit of size $O(2^n)$.
\end{lemma}

Lemmas \ref{lem:multicontrolcnot} and \ref{lem:diag_size} are used in Appendix \ref{append:diag_without_ancilla}.

\section{Circuit constructions for diagonal unitary matrices without ancillary qubits under qubit connectivity constraints}
\label{append:diag_without_ancilla}

\subsection{Circuit framework (Proofs omitted from Section~\ref{sec:diag_without_ancilla_framework_main})}
\label{sec:diag_without_ancilla_framework}

\Ck*
\begin{proof}

For each $p\in[2^{r_c}]$, define the set $F_k^{(p)}$ by
\begin{equation}\label{eq:Fkp}
 F_k^{(p)}=\left\{s:\ s\in F_k\text{~and~}s=c_{p}^{j_i }t_i^{(k)}\text{ for some }i\in[r_t]\right\} .
\end{equation}
By the definition of $F_k$ in Eq. \eqref{eq:F_k}, the collection of $F_k^{(p)}$'s satisfy
\begin{align}
    & F_k^{(i)}\cap F_{k}^{(j)}=\emptyset \text{~for~all~} i\neq j\in[2^{r_{c}}], \label{eq:Fkp_1}\\
    & F_k=\bigcup_{p\in[2^{r_c}]}F_k^{(p)}.\label{eq:Fkp_2}
\end{align}
We implement $C_k$ in two stages, via $U_{Gen}^{(k)}$ (Eq. \eqref{eq:Ugen_Graph}) and  $U^{(k)}_{Gray}$ (Eq. \eqref{eq:Ugray_Graph}). 

By Lemma \ref{lem:cnot_circuit}, $U_{Gen}^{(k)}$ can be implemented by a CNOT circuit of depth and size $O(n^2)$. 

$U^{(k)}_{Gray}$ can be realized by a quantum circuit consisting of the following steps $U_1,U_2,\ldots,U_{2^{r_c}+1}$.
\begin{align*}
     &\ket{x_{control}}_{\textsf{C}}\ket{y^{(k)}}_{\textsf{T}} &\\
     &= \ket{x_{control}}_{\textsf{C}}\ket{\langle 0^{r_c}t_{1}^{(k)},x\rangle,\langle 0^{r_c}t_{2}^{(k)},x\rangle,\cdots,\langle 0^{r_c}t_{r_t}^{(k)},x\rangle}_{\textsf{T}} & ( \text{by Eq. \eqref{eq:yk}})\\
     &= \ket{x_{control}}_{\textsf{C}}\ket{\langle c_1^{j_1 }t_{1}^{(k)},x\rangle,\langle c_1^{j_2 }t_{2}^{(k)},x\rangle,\cdots,\langle c_1^{j_{r_t} }t_{r_t}^{(k)},x\rangle}_{\textsf{T}} & (c_1^{j _i}=0^{r_c}, \forall i\in[r_t] ) \\
     & \xrightarrow{U_1} e^{i\sum_{s\in F_{k}^{(1)}}\langle s,x\rangle\alpha_s}\ket{x_{control}}_{\textsf{C}} \ket{\langle c_1^{j_1 }t_{1}^{(k)},x\rangle,\langle c_1^{j_2 }t_{2}^{(k)},x\rangle,\cdots,\langle c_1^{j_{r_t} }t_{r_t}^{(k)},x\rangle}_{\textsf{T}} & (\text{by~Eq. }     \eqref{eq:Fkp}) \\
     &\xrightarrow{U_2} e^{i\sum_{s\in F_{k}^{(1)}\cup F_{k}^{(2)}}\langle s,x\rangle\alpha_s}\ket{x_{control}}_{\textsf{C}} \ket{\langle c_2^{j_1 }t_{1}^{(k)},x\rangle,\langle c_2^{j_2 }t_{2}^{(k)},x\rangle,\cdots,\langle c_2^{j_{r_t} }t_{r_t}^{(k)},x\rangle}_{\textsf{T}} & (\text{by~Eq.\eqref{eq:Fkp},  \eqref{eq:Fkp_1}})\\
     & ~~~~\vdots & \\
     &\xrightarrow{U_{2^{r_c}}} e^{i\sum_{s\in \bigcup_{p\in[2^{r_c}]}F_k^{(p)}}\langle s,x\rangle\alpha_s}\ket{x_{control}}_{\textsf{C}} \ket{\langle c_{2^{r_c}}^{j_1 }t_{1}^{(k)},x\rangle,\langle c_{2^{r_c}}^{j_2 }t_{2}^{(k)},x\rangle,\cdots,\langle c_{2^{r_c}}^{j_{r_t} }t_{r_t}^{(k)},x\rangle}_{\textsf{T}} & (\text{by~Eq.\eqref{eq:Fkp},  \eqref{eq:Fkp_1}})\\
     & = e^{i\sum_{s\in F_k}\langle s,x\rangle\alpha_s}\ket{x_{control}}_{\textsf{C}} \ket{\langle c_{2^{r_c}}^{j_1 }t_{1}^{(k)},x\rangle,\langle c_{2^{r_c}}^{j_2 }t_{2}^{(k)},x\rangle,\cdots,\langle c_{2^{r_c}}^{j_{r_t} }t_{r_t}^{(k)},x\rangle}_{\textsf{T}} & (\text{by~Eq.}\eqref{eq:Fkp_2})\\
     &\xrightarrow{U_{2^{r_c}+1}} e^{i\sum_{s\in F_k}\langle s,x\rangle\alpha_s}\ket{x_{control}}_{\textsf{C}} \ket{\langle c_{1}^{j_1 }t_{1}^{(k)},x\rangle,\langle c_{1}^{j_2 }t_{2}^{(k)},x\rangle,\cdots,\langle c_{1}^{j_{r_t} }t_{r_t}^{(k)},x\rangle}_{\textsf{T}}\\
     &= e^{i\sum_{s\in F_k}\langle s,x\rangle\alpha_s}\ket{x_{control}}_{\textsf{C}} \ket{y^{(k)}}_{\textsf{T}} & ( \text{by Eq.\eqref{eq:yk}})
\end{align*}

%
For every $p\in\{2,3,\ldots,2^{r_c}\}$, $U_p$ itself consists of Phase 1 and a two-step Phase two (comprising steps $p$.1 and $p$.2, see main text for details). The depth and size of Step $p$.1 are $\mathcal{D}(C_{p.1})$ and $\mathcal{S}(C_{p.1})$ by definition. Step $p$.2 consist of $R(\theta)$ gates applied on different qubits in the target register, which can be implemented in depth 1 and size $O(r_t)$. 
For any $i\in[r_t]$, $c_{2^{r_c}}^{j_i}$ and $c_1^{j_i}$ differ in the $h_{j_i,1}$-th bit. $U_{2^{r_c}+1}$ can be implemented by adding $x_{h_{j_i,1}}$ to the $i$-th qubit of target register $\sf T$, using a CNOT circuit. 
Again, by Lemma \ref{lem:cnot_circuit}, step $2^{r_c}+1$ can be implemented by a circuit of depth and size $O(n^2)$.

Circuit $C_k$ thus has total depth $O(n^2)+O(1)+\sum_{p=2}^{2^{r_c}}(1+\mathcal{D}(C_{p.1}))+O(n^2)=O(n^2+2^{r_c}+\sum_{p=2}^{2^{r_c}}\mathcal{D}(C_{p.1}))$, and total size $O(n^2)+O(r_t)+\sum_{p=2}^{2^{r_c}}(r_t+\mathcal{S}(C_{p.1}))+O(n^2)=O(n^2+r_t2^{r_c}+\sum_{p=2}^{2^{r_c}}\mathcal{S}(C_{p.1}))$.
\end{proof}

\Lambdarc*
\begin{proof}
By Lemma \ref{lem:diag_size}, in the absence of any graph constraint, $\Lambda_{r_c}$ can be implemented by a quantum circuit of size (and thus also depth) $O(2^{r_c})$. 
Under arbitrary graph constraint, the distance between control and target qubits of any CNOT gate is at most $O(n)$, which can be realized by a circuit of size $O(n)$ (Lemma \ref{lem:cnot_path_constraint}). Therefore, the required circuit size and depth for $\Lambda_{r_c}$ is $O(n)\cdot O(2^{r_c})=O(n2^{r_c})$.
\end{proof}

The following lemma proves the correctness of the framework shown in Fig. \ref{fig:diag_without_ancilla_framwork}.
\begin{lemma}\label{lem:diag_without_ancilla_correctness}
Any diagonal unitary matrix $\Lambda_n$ can be realized by the quantum circuit  
\begin{equation}\label{eq:framework_withoutancilla}
   (\Lambda_{r_c}\otimes\mathcal{R})\Pi^{\dagger} C_{\ell}C_{\ell-1}\cdots C_{1}\Pi 
\end{equation}
shown in Fig. \ref{fig:diag_without_ancilla_framwork}, under arbitrary graph constraint.
\end{lemma}
\begin{proof}
For any input state $|x\rangle_{[n]}$, the quantum circuit $(\Lambda_{r_c}\otimes\mathcal{R})\Pi^{\dagger} C_{\ell}C_{\ell-1}\cdots C_{1}\Pi$ performs the following sequence of operations.
    \begin{align*}
     \ket{x}_{[n]}  & =\ket{x_1 x_2\cdots x_{r_c}}_{[r_c]}\ket{ x_{r_c+1}\cdots x_n}_{[n]-[r_c]}\\
     & \xrightarrow{\Pi} \ket{x_{control}}_{\textsf{C}}\ket{x_{target}}_{\textsf{T}} & (\text{by Eq. \eqref{eq:pi}}) \\
     &  = \ket{x_{control}}_{\textsf{C}} \ket{y^{(0)}}_{\textsf{T}} & (\text{by definition of } y^{(0)}, \text{Eq. \eqref{eq:yk}}) \\
      &  \xrightarrow{C_1} e^{i\sum_{s\in F_1}\langle s,x\rangle\alpha_s}\ket{x_{control}}_{\textsf{C}} \ket{y^{(1)}}_{\textsf{T}} & \text{(by Eq. \eqref{eq:Ck})}\\
  &  \xrightarrow{C_2} e^{i\sum_{s\in F_1\cup F_2 }\langle s,x\rangle\alpha_s}\ket{x_{control}}_{\textsf{C}} \ket{y^{(2)}}_{\textsf{T}} & (\text{by Eq. \eqref{eq:Ck} and } F_1\cap F_2 = \emptyset) \\
    & ~~~~\vdots \\
  &  \xrightarrow{C_\ell} e^{i\sum_{s\in\bigcup_{k\in[\ell]}F_k }\langle s,x\rangle\alpha_s}\ket{x_{control}}_{\textsf{C}} \ket{y^{(\ell)}}_{\textsf{T}} & (\text{by Eq.\eqref{eq:Ck} and } F_{j}\cap F_\ell = \emptyset,\ \forall j\in[\ell-1])\\
  & \xrightarrow{\Pi^{\dagger}} e^{i\sum_{s\in\bigcup_{k\in[\ell]}F_k }\langle s,x\rangle\alpha_s} \ket{x_{control}}_{[r_c]}\ket{y^{(\ell)}}_{[n]-[r_c]} & \text{(by Eq. \eqref{eq:pi})}\\
   &  \xrightarrow{\mathbb{I}_{r_c}\otimes \mathcal{R}} e^{i\sum_{s\in\bigcup_{k\in[\ell]}F_k }\langle s,x\rangle\alpha_s}\ket{x_{control}}_{[r_c]} \ket{y^{(0)}}_{[n]-[r_c]}& \text{(by Eq. \eqref{eq:reset})}\\
 &  \xrightarrow{\Lambda_{r_c}\otimes \mathbb{I}_{r_t}} e^{i\sum_{s\in\left(\bigcup_{k\in[\ell]}F_k\right)\cup \left(\left\{c0^{r_t}\right\}_{c\in\{0,1\}^{r_c}-\{0^{r_c}\}}\right) }\langle s,x\rangle\alpha_s}\ket{x_{control}}_{[r_c]} \ket{y^{(0)}}_{[n]-[r_c]}& \text{(by Eq. \eqref{eq:Lambda_rc})}\\
    & = e^{i\sum_{s\in\{0,1\}^n-\{0^{n}\} }\langle s,x\rangle\alpha_s}\ket{x_{control}}_{[r_c]} \ket{y^{(0)}}_{[n]-[r_c]} & \text{(by Eq. \eqref{eq:set_eq})}\\
    & = e^{i\theta(x)}\ket{x_{control}}_{[r_c]}\ket{y^{(0)}}_{[n]-[r_c]} & \text{(by Eq. \eqref{eq:alpha})}\\
    & = e^{i\theta(x)}\ket{x}_{[n]}
    \end{align*}
\end{proof}

\subsection{Circuit implementation under $\Path_{n}$ and $\Grid_{n}^{n_1,n_2,\ldots,n_d}$ constraints (Proofs omitted from Section \ref{sec:diag_without_ancilla_path_main})}
\label{sec:diag_without_ancilla_path}


\pipath*
\begin{proof}
The effect of $\Pi^{path}$ is shown in Fig. \ref{fig:pi_n^k}, and the transformation can be implemented by a sequence of $O(n^2)$ SWAP operations: each qubit $i\in \{ (n+\tau)/2+1, (n+\tau)/2+2,\ldots, n\}$ can be moved from its original to final position using $O(n)$ SWAP operations between adjacent qubits, and the SWAPs for different qubits can be implemented in parallel in a pipeline.
    \begin{figure}[!ht]
        \centering
    \begin{tikzpicture}
    \draw (-2,0) -- (1.2,0) (1.8,0) -- (5.2,0) (5.8,0) -- (7,0);
    
    \draw [fill=black] (-2,0) circle (0.05)
                       (-1,0) circle (0.05)
                       (0,0) circle (0.05)
                      (1,0) circle (0.05)
                      (2,0) circle (0.05)
                      (3,0) circle (0.05);
    \draw [fill=black] (4,0) circle (0.05)
                      (5,0) circle (0.05)
                      (6,0) circle (0.05)
                      (7,0) circle (0.05);
    \draw (-3,0) node{qubit};
    \draw (-3,-0.8) node{state};
    \draw (-3,-2) node{qubit};
    \draw (-3,-2.8) node{state};
                       
     \draw (1.2,0) node[anchor=west]{\scriptsize $\cdots$} (5.2,0) node[anchor=west]{\scriptsize $\cdots$};
          \draw (-2,-0.3) node{\scriptsize $1$} (-2,-0.8) node{\scriptsize $\ket{x_1}$};
         \draw (-1,-0.3) node{\scriptsize $2$} (-1,-0.8) node{\scriptsize $\ket{x_2}$};
     \draw (0,-0.3) node{\scriptsize $3$} (0,-0.8) node{\scriptsize $\ket{x_3}$}
          (1,-0.3) node{\scriptsize $4$} (1,-0.8) node{\scriptsize $\ket{x_4}$}
          (2,-0.3) node{\scriptsize ${r_c-1}$} (2,-0.8) node{\scriptsize $\ket{x_{r_c-1}}$}
          (3,-0.3) node{\scriptsize ${r_c}$} (3,-0.8) node{\scriptsize $\ket{x_{r_c}}$}
          (4,-0.3) node{\scriptsize ${r_c+1}$} (4,-0.8) node{\scriptsize \color{purple} $\ket{x_{r_c+1}}$}
          (5,-0.3) node{\scriptsize ${r_c+2}$} (5,-0.8) node{\scriptsize \color{purple} $\ket{x_{r_c+2}}$}
          (6,-0.3) node{\scriptsize ${n-1}$} (6,-0.8) node{\scriptsize \color{purple} $\ket{x_{n-1}}$}
          (7,-0.3) node{\scriptsize ${n}$} (7,-0.8) node{\scriptsize \color{purple} $\ket{x_n}$};
          
     \draw (1.2,-2) node[anchor=west]{\scriptsize $\cdots$} (5.2,-2) node[anchor=west]{\scriptsize $\cdots$};
        \draw (-2,-2) -- (1.2, -2) (1.8,-2) -- (5.2,-2) (5.8,-2) -- (7,-2);
    
    \draw [fill=black] (-2,-2) circle (0.05) (0,-2) circle (0.05)
                      (2,-2) circle (0.05) (4,-2) circle (0.05) (5,-2) circle (0.05) (6,-2) circle (0.05) (7,-2) circle (0.05);
    \draw [fill=purple, draw=purple] (-1,-2) circle (0.05) (1,-2) circle (0.05) (3,-2) circle (0.05);
            \draw (-2,-2.3) node{\scriptsize $1$} (-2,-2.8) node{\scriptsize $\ket{x_1}$};
        \draw (-1,-2.3) node{\scriptsize $2$} (-1,-2.8) node{\scriptsize \color{purple} $\ket{x_{r_c+1}}$};
     \draw (0,-2.3) node{\scriptsize $3$} (0,-2.8) node{\scriptsize $\ket{x_2}$}
          (1,-2.3) node{\scriptsize $4$} (1,-2.8) node{\scriptsize \color{purple} $\ket{x_{r_c+2}}$}
          (2,-2.3) node{\scriptsize ${n-\tau-1}$} (2,-2.8) node{\scriptsize $\ket{x_{r_t}}$}
          (3,-2.3) node{\scriptsize ${n-\tau}$} (3,-2.8) node{\scriptsize \color{purple} $\ket{x_{n}}$}
          (4,-2.3) node{\scriptsize ${n-\tau+1}$} (4,-2.8) node{\scriptsize  $\ket{x_{r_t+1}}$}
          (5,-2.3) node{\scriptsize ${n-\tau+2}$} (5,-2.8) node{\scriptsize  $\ket{x_{r_t+2}}$}
          (6,-2.3) node{\scriptsize ${n-1}$} (6,-2.8) node{\scriptsize $\ket{x_{r_c-1}}$}
          (7,-2.3) node{\scriptsize ${n}$} (7,-2.8) node{\scriptsize  $\ket{x_{r_c}}$};
       \draw (2.5,-1.5) node{$\downarrow~\Pi^{path}$};           
    
    \end{tikzpicture}
        \caption{$\Pi^{path}$. In the lower figure, the qubits in red form register $\sf T$, and those in black form register $\sf C$.}
        \label{fig:pi_n^k}
     \end{figure}
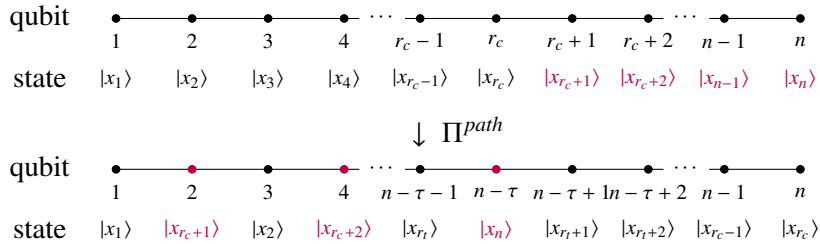
\end{proof}


\Ukwithoutancilla*

\begin{proof}
\textbf{Case 1 ($k\le \tau+1$).} If $k=1$, apply a CNOT circuit $\Pi_{j=1}^{r_t}\textsf{CNOT}^{2j-1}_{2j}$ of depth $1$ and size $O(r_t)$. 
Let us first show how to implement $y_1\oplus x_k$: (i) use a sequence of $2k-4$  SWAPs to move $y_1$ and $x_k$ adjacent to each other; (ii) apply a CNOT gate to change $\ket{y_1}$ to $\ket{y_1\oplus x_k}$; (iii) undo the first sequence of SWAPs to return the qubits to their original positions. The overall cost is $O(k)$ SWAP gates and one CNOT gate. Since each SWAP can be implemented by three CNOT gates, the cost to implement $y_1\oplus x_k$ is $O(k)$ CNOT gates. We can effect $y_2\oplus x_{k+1}$, $\ldots$, $y_{r_t}\oplus x_{r_t+k-1}$ similarly, and these can be implemented in parallel. The overall depth and size are $O(k)$ and $O(k r_t)$, respectively, as claimed.

\textbf{Case 2 ($k\ge \tau+2$) .} $U^{(k)}$ can be implemented by a CNOT circuit  
containing two parts: The first part adds $x_{i+k-1}$ to $y_i$ (for each $i=1,\ldots, r_c-k+1$), which are $O(k)$ apart. The second part adds $x_{i-r_c+k-1}$ to $y_i$ (for each $i=r_c-k+2,\ldots, r_t$), which are $O(r_c-k)$ apart. Since $k\ge \tau + 2$, we know that $r_c-k \le r_c-\tau -2 \le r_t - 2$.
Therefore, this CNOT circuit can be implemented in depth and size 
\begin{align*}
    & \ (r_c-k+1)\cdot O(k)+(r_t-(r_c-k+2)+1)\cdot O(r_c-k) & (\text{by Lemma \ref{lem:cnot_path_constraint}}) \\
    = & \ r_t \cdot O(k) + k \cdot O(r_t) & (r_t < r_c \text{, and }  r_c-k \le r_t-2) \\ 
    = & \ O(r_t k).
\end{align*}
\end{proof} 

\begin{lemma}\label{lem:diag_path_withoutancilla}
Any $n$-qubit diagonal unitary matrix $\Lambda_n$ can be realized by a quantum circuit of depth $O(2^n/n)$ and size $O(2^n)$, under $\Path_n$ constraint without ancillary qubits.
\end{lemma}
\begin{proof}

 By Lemma \ref{lem:diag_without_ancilla_correctness}, $\Lambda_n$ can be implemented by the circuit in Fig.~\ref{fig:diag_without_ancilla_framwork}. Recall that $r_c=\frac{n+\tau}{2}$, $r_t=n-r_c$, $2\lceil\log(n)\rceil\le \tau\le 2\lceil\log(n)\rceil+1$ and $\ell\le \frac{2^{r_t+2}}{r_t+1}-1$. Combining Lemmas \ref{lem:pi_path}, \ref{lem:Ck_path},  \ref{lem:reset} and \ref{lem:Lambda_rc}, the total depth and size of $\Lambda_n$ are
 \begin{align*}
     &\text{depth:}~2O(n)+\ell\cdot  O(2^{r_c})+O(n^2)+O(n2^{r_c})=O(2^n/n),\\
     &\text{size:}~2O(n^2)+\ell \cdot  O(r_t2^{r_c})+O(n^2)+O(n2^{r_c})=O(2^n),
 \end{align*}
under $\Path_n$ constraint.
\end{proof}

\subsection{Circuit implementation under $\Tree_n(d)$ constraints (Proof of Lemma~\ref{lem:diag_tree_withoutancilla})}
\label{sec:diag_without_ancilla_binarytree}




The \textit{depth} of a tree is the distance between the root and the furthest leaf. Note that a depth-$d$ tree has $d+1$ layers of nodes.


\paragraph{Choice of $\sf C$ and $\sf T$}
Label the qubits in the input register $[n]$ as follows.
For the binary tree $\Tree_n(2)$, label the root node with the empty string $\epsilon$. For a node with label $z$,  label its left and right children $z0$ and $z1$, respectively (see Fig.~\ref{fig:label_binarytree}). 

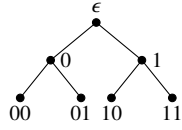
\begin{figure}[!hbt]
    \centering
    \begin{tikzpicture}
       \draw[fill=black]  (0,0) circle (0.05) (-0.6,-0.5) circle (0.05) (0.6,-0.5) circle (0.05) (-1,-1) circle (0.05) (-0.2,-1) circle (0.05) (1,-1) circle (0.05) (0.2,-1) circle (0.05);
       \draw (0,0)--(-0.6,-0.5)--(-1,-1) (-0.6,-0.5)--(-0.2,-1) (0,0)--(0.6,-0.5)--(1,-1) (0.6,-0.5)--(0.2,-1);
       \draw (0,0.2) node{\scriptsize $\epsilon$} (-0.4,-0.5) node{\scriptsize $0$} (0.8,-0.5) node{\scriptsize $1$} (-1,-1.2) node{\scriptsize $00$} (-0.2,-1.2) node{\scriptsize $01$} (1,-1.2) node{\scriptsize $11$} (0.2,-1.2) node{\scriptsize $10$} ;
    \end{tikzpicture}
    \caption{The labels of qubits in a depth-2 binary tree.}
    \label{fig:label_binarytree}
\end{figure}

Let $\epsilon$ denote an empty string. Define the set $\B^{\le k}:=\bigcup_{i=0}^{k}\B^{k}$, in which $\B^0:=\{\epsilon\}$.  Let $\kappa=\left\lceil\log(\frac{n+1}{2})\right\rceil$, $a=\left\lceil\log(2\log n)\right\rceil$. Let $\Tree_z^j=\{zy:y\in\{0,1\}^{\le j}\}$ denote the binary tree with root $z$ and depth $j$. A $\Tree_z^j$ consists of $j+1$ layers of qubits. The $n$ input qubits of $\Lambda_n$ are stored in a binary tree of depth $\kappa$, i.e. $\Tree_{\epsilon}^\kappa$. We divide these $n$ qubits into $O\left(\frac{n}{\log(n)}\right)$ binary subtrees, each of which has depth $a$ and $2^{a+1}-1=O(\log(n))$ vertices, except the `top' subtree, which may have fewer vertices and lower depth (see Fig.~\ref{fig:register_binarytree}). 

The target register $\sf T$ and control register $\sf C$ are defined as
\[\textsf{T}:=\bigcup_{j=1}^s\B^{\kappa-j(a+1)+1},\quad\textsf{C}:=\Tree_\epsilon^\kappa -\textsf{T} =\Big(\bigcup_{z\in \textsf{T}}(\Tree_z^a-\{z\})\Big)\cup \Tree_{\epsilon}^{\kappa-s(a+1)},\]
where $s+1$ is the total number of layers of binary subtrees, with $s=\left\lfloor\frac{\kappa+1}{a+1}\right\rfloor$.
In words, the target register consists of the root nodes of the binary subtrees (except the top subtree), while the control register consists of all other nodes. $\sf T$ and $\sf C$  have sizes $r_t=\sum_{j=1}^s2^{\kappa-j(a+1)+1}=O\left(\frac{n}{\log n}\right)$ and $r_c=n-r_t=O\left(n-\frac{n}{\log(n)}\right)$, respectively. 
 
 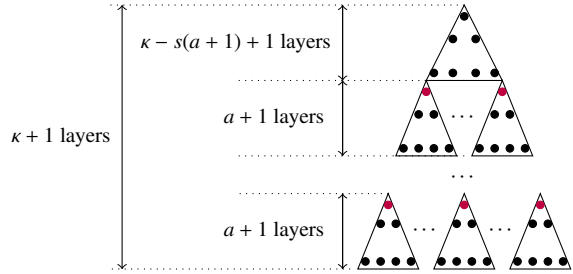
\begin{figure}[!hbt]
    \centering
    \begin{tikzpicture}
        \draw[dotted] (0,0)--(-4.7,0) (0,-1)--(-3,-1) (0,-2)--(-3,-2) (-1,-2.5)--(-3,-2.5) (-1.5,-3.5)--(-4.7,-3.5);
        \draw[->] (-4.5,-2)--(-4.5,-3.5);
        \draw[->] (-4.5,-2)--(-4.5,0);
        \draw[->] (-1.6,0) --(-1.6,-1);
         \draw[<-] (-1.6,0) --(-1.6,-1);
        \draw[->] (-1.6,-1) --(-1.6,-2);
         \draw[<-] (-1.6,-1) --(-1.6,-2);
        \draw[->] (-1.6,-2.5) --(-1.6,-3.5);
         \draw[<-] (-1.6,-2.5) --(-1.6,-3.5);
        \draw (0,0)--(0.5,-1)--(-0.5,-1)--cycle;
        \draw (-0.5,-1)--(-0.9,-2)--(-0.1,-2)--cycle (0.5,-1)--(0.1,-2)--(0.9,-2)--cycle;
        \draw (-1,-2.5)--(-1.4,-3.5)--(-0.6,-3.5)--cycle (1,-2.5)--(1.4,-3.5)--(0.6,-3.5)--cycle (0,-2.5)-- (-0.4,-3.5)--(0.4,-3.5)--cycle;
        \draw [fill=purple,draw=purple]  (-0.5,-1.15) circle (0.05) (0.5,-1.15) circle (0.05) (-1,-2.65) circle (0.05) (0,-2.65) circle (0.05) (1,-2.65) circle (0.05);
        \draw [fill=black] (0,-0.15) circle (0.05) (-0.15,-0.45) circle (0.05) (0.15,-0.45) circle (0.05) (-0.15,-0.9) circle (0.05) (0.15,-0.9) circle (0.05)(-0.4,-0.9) circle (0.05) (0.4,-0.9) circle (0.05);
        \draw [fill=black] (-0.6,-1.45) circle (0.05) (-0.4,-1.45) circle (0.05) (0.6,-1.45) circle (0.05) (0.4,-1.45) circle (0.05) (-0.6,-1.9) circle (0.05) (0.6,-1.9) circle (0.05)(-0.8,-1.9) circle (0.05) (0.8,-1.9) circle (0.05) (0.4,-1.9) circle (0.05) (0.2,-1.9) circle (0.05) (-0.4,-1.9) circle (0.05) (-0.2,-1.9) circle (0.05);
        \draw[fill=black] (-1.1,-2.9) circle (0.05) (-0.9,-2.9) circle (0.05) (-0.1,-2.9) circle (0.05) (0.1,-2.9) circle (0.05) (1.1,-2.9) circle (0.05) (0.9,-2.9) circle (0.05);
        \draw[fill=black] (-1.3,-3.4) circle (0.05) (-1.1,-3.4) circle (0.05) (-0.9,-3.4) circle (0.05) (-0.7,-3.4) circle (0.05)  (-0.3,-3.4) circle (0.05)(-0.1,-3.4) circle (0.05) (0.1,-3.4) circle (0.05) (0.3,-3.4) circle (0.05) (1.1,-3.4) circle (0.05) (0.9,-3.4) circle (0.05) (1.3,-3.4) circle (0.05) (0.7,-3.4) circle (0.05);
        \draw (0,-2.3) node{\scriptsize $\cdots$} (0,-1.5) node{\scriptsize $\cdots$} (-0.5,-3) node{\scriptsize $\cdots$} (0.5,-3) node{\scriptsize $\cdots$};
        \draw (-5.3,-1.75) node{\scriptsize $\kappa+1$ layers} (-3,-0.5) node{\scriptsize $\kappa-s(a+1)+1$ layers} (-2.5,-1.5) node{\scriptsize $a+1$ layers} (-2.5,-3) node{\scriptsize $a+1$ layers};
    \end{tikzpicture}
    \caption{Control ($\textsf{C}$) and target ($\textsf{T}$) registers for $\Tree_n(2)$. The tree is partitioned into $O\left(\frac{n}{\log (n)}\right)$ binary subtrees, each of size $O(\log(n))$ and depth $a$ ($a+1$ layers of qubits). $\textsf{T}$ consists of the root nodes of all subtrees except for the `top' subtree (red vertices), while $\textsf{C}$ consists of all other (black) vertices. }
    \label{fig:register_binarytree}
\end{figure}

\paragraph{Implementation of $\Pi$}
In this subsection, $\Pi$ (Eq.\eqref{eq:pi}) is denoted $\Pi^{binarytree}$. We wish to permute the qubit states in a way that groups consecutive qubit states together in binary subtrees. More precisely, we define $A(i):=(i-1)(2^{a+1}-2)$ and, for all $i\in[r_t]$, permute the $2^{a+1}-2$ states $x_{1+A(i)}, \ldots, x_{A(i+1)}$ to the binary subtree with root given by the $i$-th qubit in the target register (see Fig.~\ref{fig:subtree_states}). 

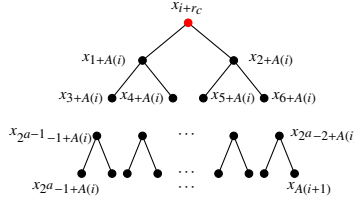
\begin{figure}[!hbt]
    \centering
    \begin{tikzpicture}
       \draw[fill=black]   (-0.6,-0.5) circle (0.05) (0.6,-0.5) circle (0.05) (-1,-1) circle (0.05) (-0.2,-1) circle (0.05) (1,-1) circle (0.05) (0.2,-1) circle (0.05) (-1.2,-1.5) circle (0.05) (-0.6,-1.5) circle (0.05) (1.2,-1.5) circle (0.05) (0.6,-1.5) circle (0.05) (-0.8,-2) circle (0.05) (-0.4,-2) circle (0.05) (-1.4,-2) circle (0.05) (-1,-2) circle (0.05) (0.8,-2) circle (0.05) (0.4,-2) circle (0.05) (1.4,-2) circle (0.05) (1,-2) circle (0.05);
       \draw (0,0)--(-0.6,-0.5)--(-1,-1) (-0.6,-0.5)--(-0.2,-1) (0,0)--(0.6,-0.5)--(1,-1) (0.6,-0.5)--(0.2,-1) (-1.2,-1.5)--(-1.4,-2) (-1.2,-1.5)--(-1,-2) (-0.6,-1.5)--(-0.8,-2) (-0.6,-1.5)--(-0.4,-2) (1.2,-1.5)--(1.4,-2) (1.2,-1.5)--(1,-2) (0.6,-1.5)--(0.8,-2) (0.6,-1.5)--(0.4,-2);
       \draw (0,0.2) node{\tiny $x_{i+r_c}$} (-1.1,-0.5) node{\tiny $x_{1+A(i)}$} (1.1,-0.5) node{\tiny $x_{2+A(i)}$} (-1.4,-1) node{\tiny $x_{3+A(i)}$} (-0.6,-1) node{\tiny $x_{4+A(i)}$} (1.4,-1) node{\tiny $x_{6+A(i)}$} (0.6,-1) node{\tiny $x_{5+A(i)}$} ;
       \draw (0,-1.5) node{\tiny$\cdots$} (0,-2) node{\tiny$\cdots$} (0,-2.2) node{\tiny$\cdots$} (-1.6,-2.2) node{\tiny $x_{2^a-1+A(i)}$}  (1.6,-2.2) node{\tiny $x_{A(i+1)}$}   (1.8,-1.5) node{\tiny $x_{2^{a}-2+A(i)}$}(-1.8,-1.5) node{\tiny $x_{2^{a-1}-1+A(i)}$};
       
       \draw[fill=red,draw=red]  (0,0) circle (0.05);
    \end{tikzpicture}
    \caption{The states of qubits in binary subtree $\Tree_{z_i}^a$ after applying $\Pi^{binarytree}$ (Lemma~\ref{lem:pi_binarytree}). $\Tree_{z_i}^a$ has $2^{a+1}-1$ vertices, with the root corresponding to the $i$-th qubit in the target register ${\sf T}$.}
    \label{fig:subtree_states}
\end{figure}

\begin{lemma}\label{lem:pi_binarytree}
The unitary transformation $\Pi^{binarytree}$, defined by
\begin{multline}\label{eq:pi_binarytree}
\ket{x_1x_2\cdots x_n}_{[n]}\xrightarrow{\Pi^{binarytree}} \ket{x_1x_2\cdots x_{r_c}}_{\sf{C}}\ket{x_{r_c+1}\cdots x_{n}}_{\sf{T}}\\
=\bigotimes_{z_i\in{\sf T}}\left(\ket{x_{r_c+i}}_{z_i}\ket{x_{1+A(i)} x_{2+A(i)} \cdots x_{A(i+1)}}_{\Tree_{z_i}^a-\{z_i\}}\right)\otimes\ket{x_{A(r_t+1)+1}\cdots x_{r_c}}_{\Tree_\epsilon^{\kappa-s(a+1)}},\forall x=x_1x_2\cdots x_n\in\Bn,
\end{multline}
where $z_i$ denotes the $i$-th element in target register $\sf T$ and $A(i):=(i-1)(2^{a+1}-2)$, can be implemented by a CNOT circuit of depth and size $O(n\log(n))$ under $\Tree_n(2)$ constraint.
\end{lemma} 
\begin{proof}
$\Pi^{binarytree}$ permutes the last $r_t$ qubits $x_{r_c+1},x_{r_c+2},\cdots,x_n$ to the target register, i.e., the root nodes of the binary subtrees, and the first $r_c$ qubits to the control register $\textsf{C}$. In the absence of graph constraints, $\Pi^{binarytree}$ can be implemented by at most $n$ SWAP gates. The result frollows from Lemma \ref{lem:cnot_path_constraint}, noting that the distance between control and target qubits for any CNOT gate in a binary tree of $n$ vertices is at most $O(\log(n))$, and every SWAP gate can be implemented by 3 CNOT gates.
\end{proof}

\paragraph{Implementation of $C_k$}

\begin{lemma}\label{lem:Ck_binarytree}
For all $k\in[\ell]$, operator $C_k$ (Eq.\eqref{eq:Ck}) can be implemented by a quantum circuit of depth $O(2^{r_c})$ under $\Tree_n(2)$ constraint.
\end{lemma}
\begin{proof}
First, we construct quantum circuits for $C_{p.1}$ (Eq.~\eqref{eq:step_p1}) for all $p\in\{2,3,\ldots,2^{r_c}\}$. For every $i\in[r_t]$, choose integers $j_i= 1+A(i)$, where recall $A(i)=(i-1)(2^{a+1}-2)$ is the index for the last node of the $(i-1)$-th subtree (Fig. \ref{fig:subtree_states}). Strings $c_{p-1}^{1+A(i)}$ and $c_{p}^{1+A(i)}$ in the $(r_c,1+A(i))$-Gray code differ in the $h_{1+A(i),p}$-th bit.


Let $z_i$ denote the $i$-th element in $\textsf{T}$. $C_{p.1}$ effects the transformation
\begin{align*}
&\bigotimes_{z_i\in\textsf{T}}\left(\ket{\langle c_{p-1}^{1+A(i)}t_i^{(k)},x\rangle}_{z_i}\ket{x_{1+A(i)} x_{2+A(i)} \cdots x_{A(i+1)}}_{\Tree_{z_i}^a-\{z_i\}}\right) \otimes\ket{x_{A(r_t+1)+1}x_{A(r_t+1)+2}\cdots x_{r_c}}_{\Tree_\epsilon^{\kappa-s(a+1)}}\\
\xrightarrow{C_{p.1}}&\bigotimes_{z_i\in\textsf{T}}\left(\ket{\langle c_{p-1}^{1+A(i)}t_i^{(k)},x\rangle \oplus h_{1+A(i),p}}_{z_i}\ket{x_{1+A(i)} x_{2+A(i)} \cdots x_{A(i+1)}}_{\Tree_{z_i}^a-\{z_i\}}\right) \otimes\ket{x_{A(r_t+1)+1}x_{A(r_t+1)+2}\cdots x_{r_c}}_{\Tree_\epsilon^{\kappa-s(a+1)}},
\end{align*}
The key operation is thus the mapping of $\ket{\langle c_{p-1}^{1+A(i)}t_i^{(k)},x\rangle }_{z_i}\rightarrow \ket{\langle c_{p-1}^{1+A(i)}t_i^{(k)},x\rangle \oplus h_{1+A(i),p}}_{z_i}$, for all $z_i\in {\sf T}$. To implement this,  
for each $i\in[r_t]$,  we apply a CNOT gate with target qubit $z_i$, and  control qubit $\ket{x_{h_{1+A(i),p}}}$.  By construction, $\ket{x_{h_{1+A(i),p}}}$ lies in subtree $\Tree_{z_i}^a-\{z_i\}$ if $h_{1+A(i),p}\in\{1+A(i),2+A(i),\ldots,A(i+1)\}$, and otherwise lies in subtree $\Tree_\epsilon^{\kappa}-\Tree_{z_i}^a$.


We now analyze the depth of $C_k$.  
\begin{enumerate}
    \item If $h_{1+A(i),p}:=k'+A(i)\in\{1+A(i),2+A(i),\ldots,A(i+1)\}$ for all $z_i\in{\sf T}$ and $k'\in[2^{a+2}-2]$, all CNOT gates in Step $p.1$ ($C_{p.1}$) can be implemented simultaneously because they are in disjoint binary subtrees $\Tree_{z_i}^a$. Since the distance between control and target qubits in each CNOT gate in Step $p.1$ is $O(\log(h_{1+A(i),p}-A(i)))=O(\log(k'))$, by Lemma \ref{lem:cnot_path_constraint}, $C_{p.1} $ can be realized in depth $O(\log(k'))$. 
    \item If $h_{1+A(i),p}\notin\{1+A(i),2+A(i)\ldots,A(i+1)\}$, Step $p.1$ is an $n$-qubit CNOT circuit under $\Tree_n(2)$ constraint. By Lemma \ref{lem:cnot_circuit} it can be implemented in depth $O(n^2)$. 
\end{enumerate}
By Lemma \ref{lem:GrayCode}, for every $k'\in[r_c]$, there are $2^{r_c-k'}$ many $p\in \{2,3,\ldots,2^{r_c}\}$ satisfying
\begin{equation*}
    h_{1+A(i),p} = \begin{cases}
    k'+A(i), & (\text{if }k'\le r_c-A(i)+1) \\
    k'+A(i)-r_c, & (\text{if }k'\ge r_c-A(i)+2)
    \end{cases}
\end{equation*}
Thus, there are $2^{r_c-k'}$ values of $p\in\{2,3\ldots,2^{r_c}\}$ such that $C_{p.1}$ has depth $\mathcal{D}(C_{p.1})=O(\log(k'))$, with $k'\in[2^{a+2}-2]$. The remaining $2^{r_c}-\sum_{k'=1}^{2^{a+2}-2}2^{r_c-k'}-1$ values of $p$ have corresponding circuits $C_{p.1}$ that can be realized in depth $\mathcal{D}(C_{p.1})=O(n^2)$, with $k'\ge 2^{a+2}-1$.  
By Lemma \ref{lem:Ck}, $C_k$ has circuit depth
\[O(n^2+2^{r_c}+\sum_{p=2}^{2^{r_c}}\mathcal{D}(C_{p.1}))=O(n^2+2^{r_c})+\sum_{k'=1}^{2^{a+1}-2} O(\log(k'))2^{r_c-k'} +O(n^2)\cdot (2^{r_c}-\sum_{k'=1}^{2^{a+1}-2}2^{r_c-k'}-1)= O(2^{r_c}),\]
where we use the fact that $a=\lceil\log(2\log n)\rceil$.
\end{proof}

\paragraph{Implementation of $\Lambda_n$}

We are now in a position to prove
\diagtreenoancilla*
The proof of this Lemma consists of the proofs of Lemmas \ref{lem:diag_bianrytree_withoutancilla}, \ref{lem:diag_d_tree_noancilla} and \ref{lem:diag_star_noancilla} below. 

\begin{lemma}\label{lem:diag_bianrytree_withoutancilla}
Any $n$-qubit diagonal unitary matrix $\Lambda_n$ can be realized by a quantum circuit of depth $O(\log(n)2^n/n)$ under  $\Tree_n(2)$ constraint, without ancillary qubits.
\end{lemma}
\begin{proof}

By Lemma \ref{lem:diag_without_ancilla_correctness}, $\Lambda_n$ can be implemented by the circuit in Fig.~\ref{fig:diag_without_ancilla_framwork}.  Recall that $r_t=O(n/\log(n))$, $r_c=n-r_t$, and $\ell\le \frac{2^{r_t+2}}{r_t+1}-1$. Combining Lemmas \ref{lem:pi_binarytree},  \ref{lem:Ck_binarytree}, \ref{lem:reset} and \ref{lem:Lambda_rc}, the total depth and size for $\Lambda_n$ are
\[O(n\log(n))+\ell\cdot O(2^{r_c})+O(n^2)+O(n2^{r_c})=O(\log(n)2^n/n),\]
under $\Tree_n(2)$ constraint.
\end{proof}

The circuit depth under general $d$-ary tree constraint is shown as follows.
\begin{lemma}[\cite{sun2021asymptotically}]\label{lem:diag_complete_graph_noancilla}
Any $n$-qubit diagonal unitary matrix $\Lambda_n$ can be implemented by a quantum circuit of depth $O\left( 2^n/n\right)$ and size $O(2^n)$, using no ancillary qubits, under no graph constraint.
\end{lemma}
\begin{lemma}\label{lem:diag_d_tree_noancilla}
Any $n$-qubit diagonal unitary matrix $\Lambda_n$ can be implemented by a quantum circuit of depth $O\left(\log_d(n)2^n\right)$, using no ancillary qubits, under $\Tree_n(d)$ constraint.
\end{lemma}
\begin{proof}
Follows from Lemma~\ref{lem:diag_complete_graph_noancilla}, and the fact that any $n$-qubit depth-1 CNOT circuit can be implemented by a CNOT circuit of depth $O(n\log_d(n))$ under $d$-ary tree constraint. The total depth required is thus $O\left(\frac{ 2^{n}}{n}\right)\cdot O(n\log_d(n))=O\left( \log_d(n) 2^n\right)$.
\end{proof}

\begin{lemma}\label{lem:diag_star_noancilla}
Any $n$-qubit diagonal unitary matrix $\Lambda_n$ can be implemented by a quantum circuit of depth $O(2^n)$, using no ancillary qubits, under $\Star_n$ constraint.
\end{lemma}
\begin{proof}
    Follows from Lemma \ref{lem:diag_d_tree_noancilla}, taking $n=d-1$.
\end{proof}

\subsection{Circuit implementation under $\Expander_n$ constraints (Proof of Lemma ~\ref{lem:diag_expander_withoutancilla})} 
\label{sec:diag_without_ancilla_expander}

\begin{lemma}[\cite{hoory2006expander}]\label{lem:distance}
Let $G=(V,E)$ be an expander. The distance between any two vertices in $G$ is $O(\log(|V|))$.
\end{lemma}
\begin{lemma}\label{lem:graph_property}
Let $G=(V,E)$ be a graph with vertex expansion $h_{out}(G) = h$. Let $S\subset V$ have size at most $|V|/2$.
Define a bipartite graph $B=(S\cup\partial_{out}(S),E')$, where $E'=\{(u,v)\in E: u\in S,v\in\partial_{out}(S)\}$. Then, the size of any maximal matching for $B$ is at least $\frac{h}{h+2} |S|$. In particular, if $G$ is an expander, then the size of any maximal matching in $B$ is $\Omega(|S|)$.
\end{lemma}
\begin{proof}
Let $M:=\left\{(u_i,w_i): u_i\in S, w_i\in\partial_{out}(S),~i\in[k]\right\}$ be a maximal matching in $B$, i.e, $M$ is not a proper subset of any other matching in $B$. Let $U = \{u_i: i\in [k]\}$ and $W = \{w_i: i\in [k]\}$. If $U = S$, then the matching $M$ has size $|S|$, and we have proved the claim. Now consider the case that $S - U$ is not empty. Since $M$ is maximal, there do not exist edges between $S-U$ and $\partial_{out}(S)-W$. 
The neighbors of $S-U$ must therefore lie in set $U\cup W$. This implies that the size of $\partial_{out}(S-U)$ is no more than $2k$. Since $0<|S-U|<|V|/2$, by definition of vertex expansion, we have 
 \[h \le \frac{|\partial_{out}(S-U)|}{|S-U|}\le \frac{2k}{|S|-k}.\]
Rearranging the terms gives $k\ge \frac{h}{h+2}|S|$, as claimed.
\end{proof}

Consider a graph $G$ with vertex expansion $h_{out}(G) = c$ for some constant $c$.

\paragraph{Choice of $\sf C$ and $\sf T$} 
Let $c' = \frac{c}{c+2}$. 
By Lemma \ref{lem:graph_property}, for any  $S\subset V$ of size at most $n/2$, we can find a matching  of size $c'|S|$ between $S$ and $V-S$.  Then, 
\[M_S=\{(u_j^S,w_j^S)\in E: u_j^S\in S, w_j^S\in\partial_{out}(S),\forall j\in[\lfloor c'|S|\rfloor]\}.\]
is a matching of size $\lfloor c'|S|\rfloor$ .
Define the corresponding vertex set of size $\lfloor c'|S|\rfloor$:
\[\Gamma(S):=\left\{w_j^S\in\partial_{out}(S): (u_j^S,w_j^S)\in M_S, u_j^S\in S, j\in[\lfloor c'|S|\rfloor]\right\}.\]
Let $d=\Big\lfloor\frac{\log(n)-1-\log(\lceil 1/c'\rceil+1)}{\log(1+c')}\Big\rfloor+2 = O(\log n)$  and define the sequence of vertex sets $S_1,\ldots,S_d$, where 
\begin{align}
    S_1 \subseteq V,  &\qquad\abs{S_1} = \lceil 1/c'\rceil+1, \text{vertices in }S_1 \text{ arbitrary}\label{eq:s1}\\
    S_{i+1}=S_i\cup \Gamma(S_i), &\qquad \forall i\in[d-1] \label{eq:siplus1}
\end{align}

By Eq. \eqref{eq:siplus1} and the definition of $\Gamma(S_i)$, $|S_{i+1}|=|S_i|+\lceil c'|S_i|\rceil$, which satisfies $(1+c')|S_i|-1\le |S_{i+1}|\le (1+c')|S_i|$. By reduction, we obtain $(|S_1|-\frac{1}{c'})(1+c')^{i-1}\le |S_i|\le |S_1|(1+c')^{i-1}$ for all $\forall i\in [d]$. We also have
\[\frac{n}{2(\lceil 1/c'\rceil+1)(1+c')}\le (|S_1|-\frac{1}{c'})(1+c')^{d-2}\le|S_{d-1}|\le |S_1|(1+c')^{d-2} \le n/2.\]

Define the control and target registers as
\[{\sf T}:=S_{d-1}\quad \text{and} \quad{\sf C}:=V-S_{d-1}.\]
By construction, $\sf T$ and $\sf C$ have sizes $r_t=|S_{d-1}|\in\big[\frac{n}{2(\lceil 1/c'\rceil+1)(1+c')}, n/2\big]$ and $r_c=|V-S_{d-1}|\in \big[n/2, n-\frac{n}{2(\lceil 1/c'\rceil+1)(1+c')}\big]$, respectively. 

\paragraph{Implementation of $\Pi$}
In this subsection, unitary transformation $\Pi$ (Eq. \eqref{eq:pi}) is denoted $\Pi^{expander}$. 

\begin{lemma}\label{lem:pi_expander}
The unitary transformation $\Pi^{expander}$, defined by
\begin{equation*}
    \ket{x_1x_2\cdots x_n}_{V}\xrightarrow{\Pi^{expander}}\ket{x_{control}}_{V-S_{d-1}}\ket{x_{target}}_{S_{d-1}}\defeq \ket{x_{control}}_{\sf{C}}\ket{x_{target}}_{\sf{T}},
\end{equation*}
can be realized by a CNOT circuit of depth and size $O(n\log(n))$ under $\Expander_n$ constraint.
\end{lemma} 
\begin{proof}
$\Pi^{expander}$ permutes $\ket{x_i}$ to a qubit in $\sf C$ for every $i\le r_c$ and to a qubit in $\sf T$ for every $i\ge r_c+1$. In the absence of any graph constraints, $\Pi^{expander}$ can be realized by $O(n)$ swap gates, each of which can be implemented by 3 CNOT gates. The distance between any two qubits in an expander is $O(\log(n))$. Thus, by Lemma \ref{lem:cnot_path_constraint}, the depth and size required is $O(n)\cdot O(\log(n))=O(n\log(n))$.
\end{proof}

\begin{lemma}\label{lem:Ck_expander}
For all $k\in[\ell]$, unitary transformation $C_k$ (Eq. \eqref{eq:Ck}) can be implemented by a quantum circuit of depth $O(\log(n)2^{r_c})$ under $\Expander_n$ constraint.
\end{lemma}
\begin{proof}
We first construct a quantum circuit for $C_{p.1}$ (Eq.~\eqref{eq:step_p1}) for all $p\in\{2,3,\ldots,2^{r_c}\}$.  For all $i\in [r_t]$, choose integers $j_i = 1$. The strings $c_{p-1}^{1}$ and $c_{p}^{1}$ in the  $(r_c,1)$-Gray code differ in the $h_{1p}$-th bit. $C_{p.1}$ effects the transformation
 \begin{align*}
     & \ket{x_1x_2\ldots x_{r_c}}_{V-S_{d-1}}\ket{\langle c_{p-1}^1t_1^{(k)},x\rangle,\ldots,\langle c_{p-1}^1t_{r_t}^{(k)},x\rangle}_{S_{d-1}}\\
     \to & \ket{x_1x_2\ldots x_{r_c}}_{V-S_{d-1}}\ket{\langle c_{p}^1t_1^{(k)},x\rangle,\ldots,\langle c_{p}^1t_{r_t}^{(k)},x\rangle}_{S_{d-1}} 
     \\
     = & \ket{x_1x_2\ldots x_{r_c}}_{V-S_{d-1}}\ket{\langle c_{p-1}^1t_1^{(k)},x\rangle\oplus x_{h_{1p}},\ldots,\langle c_{p-1}^1t_{r_t}^{(k)},x\rangle\oplus x_{h_{1p}}}_{S_{d-1}},\forall x\in\Bn.
 \end{align*}
 That is, it is equivalent to a multi-target CNOT gate (see Section.~\ref{subsec:qgates}), with control $\ket{x_{h_{1p}}}$ and targets being all qubits in $\sf T$. This multi-target CNOT gate can be implemented as follows. 
 For each set $S_i$ used in the construction of $\sf C$ and $\sf T$, there is an associated matching
\begin{equation}\label{eq:matching}
M_{S_i}=\left\{(u_j^{S_i},w_j^{S_i})\in E: u_j^{S_i}\in {S_i}, w_j^{S_i}\in\partial_{out}({S_i}),\text{~for~}\forall j\in[c'|{S_i}|]\right\}.
\end{equation}

\begin{figure}[]
\centering
 \begin{tikzpicture}
    \draw (0,0) ellipse (0.5 and 0.3);
    \draw (0,-0.2) ellipse (1 and 0.6);
    \draw (0,-0.5) ellipse (1.3 and 1);
    \draw (0,-0.8) ellipse (1.6 and 1.4);
    \draw (0,-1.2) ellipse (2 and 1.9);
    \draw (0,-1.6) ellipse (2.3 and 2.4);

    \draw [draw=yellow] (0,1)--(0,0);
    \draw [draw=teal] (0,0)--(0,-0.5);
    \draw [draw=violet] (0,-0.5)--(0,-1.15) (0,0)--(0.4,-1.15);
    \draw [draw=red] (0,-1.15)--(-0.2,-1.9) (0,-0.5)--(0.3,-1.9) (0.4,-1.15)--(0.6,-1.9);
    \draw [draw=blue] (-0.2,-1.9)--(-0.3,-2.7) (0,-1.15)--(-0.1,-2.7)  (0.3,-1.9)--(0.1,-2.7) (0.6,-1.9)--(0.7,-2.7);
    \draw [draw=orange] (-0.3,-2.7)--(-0.4,-3.5) (-0.1,-2.7)--(-0.2,-3.5) (0.1,-2.7)--(0,-3.5) (0.6,-1.9)--(0.6,-3.5) (0.7,-2.7)--(0.8,-3.5) ;
    
    \draw  (0.1,-1.9) node{\scriptsize $\cdots$} (0.4,-2.7) node{\scriptsize $\cdots$} (0.3,-3.5) node{\scriptsize $\cdots$};
    \draw (-0.4,-0.5) node{\tiny $\Gamma(S_1)$} (-0.4,-1.15) node{\tiny $\Gamma(S_2)$} (-0.6,-1.85) node{\tiny $\Gamma(S_{d-4})$} (-0.8,-2.7) node{\tiny $\Gamma(S_{d-3})$} (-0.85,-3.5) node{\tiny $\Gamma(S_{d-2})$};
    \draw (-0.5,0) node{\tiny $S_1$} (-1,-0.2) node{\tiny $S_2$} (-1.3,-0.4) node{\tiny $S_3$}  (-1.6,-0.7) node{\tiny $S_{d-3}$} (-1.9,-1) node{\tiny $S_{d-2}$} (-2.25,-1.4) node{\tiny $S_{d-1}$} (0.5,1) node{\tiny $\ket{x_{h_{1p}}}$};
     
    \draw [->] (4.3,-3)--(4.3,-2.7);
    \draw [->] (4.3,-2.2)--(4.3,-1.9);
    \draw [->] (4.3,-1.4)--(4.3,-1.15);
    \draw [->] (4.3,-0.68)-- (4.3,-0.5);
    \draw [->] (4.3,0)--(4.3,0.2);
    \draw [<-] (8.5,0)--(8.5,0.2);
    \draw [<-] (8.5,-3)--(8.5,-2.7);
    \draw [<-] (8.5,-2.2)--(8.5,-1.9);
    \draw [<-] (8.5,-1.4)--(8.5,-1.15);
    \draw [<-] (8.5,-0.68)-- (8.5,-0.5);
    \draw (4.3,-3.25) node{\tiny \fbox{Step 1: Apply CNOT gates on {\color{orange} $M_{S_{d-2}}$}~~}};
    \draw (4.3,-2.45) node{\tiny \fbox{Step 2: Apply CNOT gates on {\color{blue} $M_{S_{d-3}}$}}};
    \draw (4.3,-1.65) node{\tiny \fbox{Step 3: Apply CNOT gates on {\color{red} $M_{S_{d-4}}$}}};
    \draw (4.3,-0.9) node{\tiny \fbox{Step $d-3$: Apply CNOT gates on {\color{violet} $M_{S_{2}}$}}};
    \draw (4.3,-0.25) node{\tiny \fbox{Step $d-2$: Apply CNOT gates on {\color{teal} $M_{S_{1}}$}}};
    
    \draw (4.5,-1.3) node{\tiny $\cdots$} (8.7,-1.3) node{\tiny $\cdots$};
    
     \draw (6.4,0.5) node{\tiny \fbox{Step $d-1$: Apply $|S_1|$ CNOT gates, where the controls are $\ket{x_{h_{1p}}}$ and targets are in $S_1$.}};
    
    \draw (8.5,-3.25) node{\tiny \fbox{Step $2d-3$: Apply CNOT gates on {\color{orange} $M_{S_{d-2}}$}}};
    \draw (8.5,-2.45) node{\tiny \fbox{Step $2d-4$: Apply CNOT gates on {\color{blue} $M_{S_{d-3}}$}}};
    \draw (8.5,-1.65) node{\tiny \fbox{Step $2d-5$: Apply CNOT gates on {\color{red} $M_{S_{d-4}}$}}};
    \draw (8.5,-0.9) node{\tiny \fbox{Step $d+1$: Apply CNOT gates on {\color{violet} $M_{S_{2}}$}}};
    \draw (8.5,-0.25) node{\tiny \fbox{Step $d$: Apply CNOT gates on {\color{teal} $M_{S_{1}}$}}};
    
    \draw (1.2,-3.25)node{\tiny {\color{orange} $M_{S_{d-2}}$}} (1.1,-2.4)node{\tiny {\color{blue} $M_{S_{d-3}}$}} (0.85,-1.65)node{\tiny {\color{red} $M_{S_{d-4}}$}} (0.6,-1)node{\tiny {\color{violet} $M_{S_{2}}$}} (0.5,-0.4)node{\tiny {\color{teal} $M_{S_{1}}$}};
    
    \draw [fill=black] (0,1) circle (0.05);
    \draw [fill=black] (0,0) circle (0.05);
    \draw [fill=black] (0,-0.5) circle (0.05);
    \draw [fill=black] (0,-1.15) circle (0.05) (0.4,-1.15) circle (0.05);
    \draw [fill=black]  (-0.2,-1.9) circle (0.05) (0.3,-1.9) circle (0.05) (0.6,-1.9) circle (0.05);
    \draw [fill=black]  (-0.3,-2.7) circle (0.05) (-0.1,-2.7) circle (0.05) (0.1,-2.7) circle (0.05) (0.7,-2.7) circle (0.05);
    \draw [fill=black]  (-0.4,-3.5) circle (0.05) (-0.2,-3.5) circle (0.05) (0,-3.5) circle (0.05) (0.6,-3.5) circle (0.05) (0.8,-3.5) circle (0.05);
    \end{tikzpicture}
\caption{Circuit implementation of $C_{p.1}$ under $\Expander_n$ constraint. For all $i\in[d-2]$, applying CNOT gates on matching $M_{S_i}$ means applying CNOT gates on all edges in $M_{S_i}$, where the controls are in set $S_i$. }\label{fig:expander_cp1}
\end{figure}
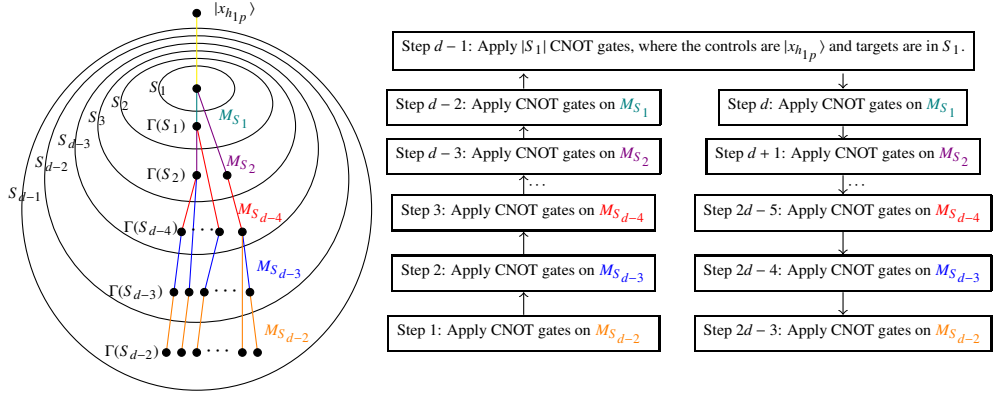
$C_{p.1}$ aims to XOR qubit $\ket{x_{h_{1p}}}$ to all qubits in $S_{d-1}$, and will be implemented in a way similar to that in Fig.~\ref{fig:add-circuit}. More precisely, this is constructed in $2d-3$ steps (see Fig.~\ref{fig:expander_cp1}).
\begin{itemize}
    \item Step $i$ for $i=1, 2, \ldots, d-2$: Apply CNOT gates to $M_{S_{d-i-1}}$;
    \item Step $d-1$: Apply $\lceil 1/c'\rceil+1$ CNOT gates, with each CNOT gate having a separate qubit in $S_1$ as target, and control qubit $\ket{x_{h_{1p}}}$;
    \item Step $j$ for $j=d, d+1, \ldots, 2d-3$: Apply CNOT gates to $M_{S_{j-(d-1)}}$;
\end{itemize}
Above, when we say ``apply CNOT gates to $M_{S_i}$'', we mean to apply CNOT gates to all qubit pairs $(u,v)$ corresponding to edges in the matching, with control qubits in set $S_i$. The correctness of this circuit can be seen by comparing Fig.~\ref{fig:expander_cp1} with the circuit in Fig.~\ref{fig:add-circuit}.

We now analyze the circuit depth of $C_{p.1}$. For each $i\in[d-2]$, all CNOT gates acting on $M_{S_i}$ can be implemented in depth 1 since $M_{S_i}$ is a matching. By Lemma~\ref{lem:distance}, the distance between $\ket{x_{h_{1p}}}$ and  any qubit in $S_1$ is at most $O(\log(n))$ and therefore, by Lemma \ref{lem:cnot_path_constraint},  Step $d-1$ can be implemented in depth  $O(\log(n))\cdot (\lceil 1/c'\rceil+1)=O(\log(n))$. The total depth of $C_{p.1}$ is thus $\mathcal{D}(C_{p.1})=2(d-2)+O(\log(n))=O(\log(n))$.

By Lemma \ref{lem:Ck}, the total depth of $C_k$ is
\[O(n^2+2^{r_c}+\sum_{p=2}^{2^{r_c}}\mathcal{D}(C_{p.1}))=O(n^2+2^{r_c}+(2^{r_c}-1)\cdot O(\log(n))=O(\log(n)2^{r_c}).\]
\end{proof}

\paragraph{Implementation of $\Lambda_n$}

\expandernoancilla*
\begin{proof}
By Lemma \ref{lem:diag_without_ancilla_correctness}, $\Lambda_n$ can be implemented by the circuit in Fig.~\ref{fig:diag_without_ancilla_framwork}. Recall that  both $r_t$ and $r_c=n-r_t$ are between $\Omega(n)$ and $n-\Omega(n)$, and $\ell\le \frac{2^{r_t+2}}{r_t+1}-1$. By Lemmas \ref{lem:reset} and \ref{lem:Lambda_rc}, \ref{lem:pi_expander} and \ref{lem:Ck_expander}, the total depth required is
\begin{equation*}
    2O(n\log(n))+\ell \cdot O(\log(n)2^{r_c})+O(n^2)+O(n2^{r_c})=O(\log(n)2^n/n).
\end{equation*}
\end{proof}

\subsection{Circuit implementation under arbitrary graph constraints (Proof of Lemma \ref{lem:diag_graph_withoutancilla})}
\label{sec:diag_without_ancilla_graph}

\paragraph{Choice of $\sf C$ and $\sf T$ }
Let $T$ be a spanning tree of connected graph $G=(E,V)$, with $\abs{V}=n$. We label all vertices as follows: we traverse $T$ by depth-first search (DFS), starting from the root, and label the vertices along the way in reverse order $n, n-1, \ldots, 2, 1$. 
 
Let $r_c=\lceil n/2\rceil$, $r_t=n-r_c$, and set $\textsf{C}=[r_c]$ and $\textsf{T}=[n]-\textsf{C}$. That is, $\textsf{T}$ contains the first $r_t=\lfloor n/2\rfloor$ vertices traversed in the DFS. By DFS, the vertices in register $\textsf{T}$ span a connected subgraph of graph $G=(V,E)$.

\begin{lemma}\label{lem:dfs_distance}
Let $d(i)$ denote the distance between qubits $i$ and qubit $i+1$ (as labelled by the DFS procedure above) in spanning tree $T$. Then, $\sum_{i=1}^{n-1}d(i)=O(n)$.
\end{lemma}
\begin{proof}
Note that when we traverse $T$ in DFS, we traverse qubits in the order $n,n-1,\ldots,1$. Since  $d(i) = dist_T(i,i+1)$ is the distance on the shortest path from $i$ to $i+1$ on $T$, $dist_T(i,i+1)$ is at most the distance we walk along the DFS traversal path from $i$ to $i+1$. Summing this up for all $i\in [n-1]$, we see that $\sum_{i=1}^{n-1} d(i) = \sum_{i=1}^{n-1} dist_T(i,i+1)$ is at most the total distance we travel in a DFS traveral, which is at most $2(n-1)$, as each edge is visited at most twice in DFS. 
\end{proof}

\paragraph{Implementation of $\Pi$}
In this subsection, unitary transformation $\Pi$ (Eq.\eqref{eq:pi}) is denoted by $\Pi^{graph}$. 

\begin{lemma}\label{lem:pi_graph}
The unitary transformation $\Pi^{graph}$, defined as
\begin{equation*}
    \ket{x_1x_2\cdots x_n}_{[n]}\xrightarrow{\Pi^{graph}}\bigotimes_{i=1}^n\ket{x_i}_i\defeq \ket{x_{control}}_{\sf{C}}\ket{x_{target}}_{\sf{T}},
\end{equation*}
can be realized by a CNOT circuit of depth and size $O(n^2)$ under arbitrary graph constraint.
\end{lemma} 
\begin{proof}
For all $i\in[n]$, $\Pi^{graph}$ permutes $\ket{x_i}$ to qubit $i$, and can be implemented by a SWAP gates, each of requires 3 CNOT gates. The result follows from Lemma \ref{lem:cnot_circuit}.
\end{proof}

\paragraph{Implementation of $C_k$}

\begin{lemma}\label{lem:Ck_graph}
For all $k\in[\ell]$, unitary transformation $C_k$ (Eq. \eqref{eq:Ck}) can be implemented by a standard quantum circuit of size $O(n2^{r_c})$ under arbitrary graph constraint. 
\end{lemma}
\begin{proof}
First, we construct quantum circuits for $C_{p.1}$ (Eq. \eqref{eq:step_p1}) for all $p\in\{2,3,\ldots,2^{r_c}\}$.
For every $i\in[r_t]$, choose integers $j_i=1$. Strings $c_{p-1}^{1}$ and $c_{p}^{1}$ in the $(r_c,1)$-Gray code differ in the $h_{1p}$-th bit.

For all $x\in \Bn$, $C_{p.1}$ effects the transformation
\begin{align*}\label{eq:implement_Cp1}
     \bigotimes_{j=1}^{r_c}\ket{x_j}_{j}\bigotimes_{i=1}^{r_t}\ket{\langle c_{p-1}^1t_i^{(k)},x\rangle}_{r_c+i} \to& \bigotimes_{j=1}^{r_c}\ket{x_j}_{j}\bigotimes_{i=1}^{r_t}\ket{\langle c_{p}^1t_i^{(k)},x\rangle}_{r_c+i} =\bigotimes_{j=1}^{r_c}\ket{x_j}_j\bigotimes_{i=1}^{r_t}\ket{\langle c_{p-1}^1t_i^{(k)},x\rangle\oplus x_{h_{1p}}}_{r_c+i},    
\end{align*}
and corresponds to a multi-target $\mathsf{CNOT}$ gate (see Appendix~\ref{append:basic_circuit}), with control being $\ket{x_{h_{1p}}}$ and targets being all qubits in $\sf T$. This can be implemented by the circuit in Fig.~\ref{fig:multi-cnot-general-g}, which is simply a relabelled version of Fig.~\ref{fig:add-circuit}.

\begin{figure}[!hbt]
\centerline 
{
\Qcircuit @C=0.6em @R=0.7em {
\lstick{\ket{x_{h_{1p}}}}&\qw & \qw &\qw & \qw & \ctrl{1} & \qw & \qw & \qw & \qw & \qw  &\rstick{\ket{x_{h_{1p}}}} \\
\lstick{\ket{x_{r_c+1}}} &\qw &\qw &\qw & \ctrl{1} & \targ & \ctrl{1} &\qw & \qw & \qw & \qw  &\rstick{\ket{x_{h_{1p}}\oplus x_{r_c+1}}}\\
\lstick{\ket{x_{r_c+2}}} &\qw &\qw & \ctrl{1} & \targ & \qw & \targ & \ctrl{1} & \qw &\qw & \qw   &\rstick{\ket{x_{h_{1p}}\oplus x_{r_c+2}}}\\
\lstick{\vdots~~~}&\qw & \ctrl{1} &\targ & \qw & \qw &  \qw & \targ & \ctrl{1} & \qw &\qw   &\rstick{~~~\vdots}\\
\lstick{\ket{x_{n}}}&\qw & \targ &\qw & \qw & \qw & \qw & \qw & \targ & \qw &\qw   &\rstick{\ket{x_{h_{1p}}\oplus x_{n}}}
}
}\caption{CNOT circuit construction of multi-target $\mathsf{CNOT}$ gate used to implement $C_{p.1}$.}\label{fig:multi-cnot-general-g}
\end{figure}
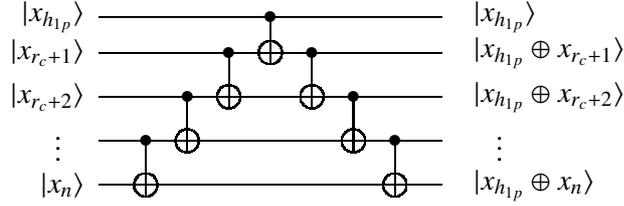



Let $d(i)$ denote the distance between qubits $i$ and $i+1$ in $G$. By Lemma~\ref{lem:dfs_distance}, $\sum_{i=r_c+1}^{n-1} d(i)=O(n)$, and $\sum_{i=h_{1p}}^{r_c+1} d(i) = O(n)$. By Lemma~\ref{lem:multicontrolcnot}, $C_{p.1}$ can be implemented in circuit size $\mathcal{S}(C_{p.1})=O(n)$.

By Lemma \ref{lem:Ck}, the total size of $C_k$ is
$O(n^2+r_t2^{r_c}+\sum_{p=2}^{2^{r_c}}\mathcal{S}(C_{p.1}))=O(n^2+r_t2^{r_c}+(2^{r_c}-1)\cdot O(n))=O(n2^{r_c})$.
\end{proof}

\graphnoancilla*
\begin{proof}
By Lemma \ref{lem:diag_without_ancilla_correctness}, $\Lambda_n$ can be implemented by the circuit in Fig.~\ref{fig:diag_without_ancilla_framwork}. Recall that $r_c=\lceil n/2\rceil$, $r_t=n-r_c$ and $\ell\le \frac{2^{r_t+2}}{r_t+1}-1$. By Lemmas \ref{lem:reset}, \ref{lem:Lambda_rc}, \ref{lem:pi_graph} and \ref{lem:Ck_graph},  the total size required is
\begin{equation*}
    O(n^2)+\ell\cdot  O(n2^{r_c})+O(n^2)+O(n2^{r_c})=O(2^n).
\end{equation*}
\end{proof}
Though this is not as good as the $\tilde O(2^n/n)$ upper bound obtained in the constructions for grids, trees, and expanders, in Appendix 
 \ref{append:QSP_US_lowerbound} we will see that this extra price of $O(n)$ is unavoidable for general graphs.

\section{Circuit constructions for diagonal unitary matrices with ancillary qubits under qubit connectivity constraints}
\label{append:diag_with_ancilla}

\subsection{Circuit framework}
\label{sec:diag_with_ancilla_framework}
Our circuit framework for $\Lambda_n$ using $m$ ancilla is shown in Fig.~\ref{fig:diag_with_ancilla_framwork}, which generalizes the ancilla-based framework of~\cite{sun2021asymptotically}. 
In our approach here, the $n+m$ qubits are divided into four registers:
\begin{itemize}
    \item ${\sf R}_{\rm inp}$: an $n$-qubit input register used to hold the input state $\ket{x}$, with $x\in\{0,1\}^n$ divided into an  $(n-p)$-bit prefix $x_{pre}=x_1x_2\ldots x_{n-p}$ and a $p$-bit suffix $x_{suf}=x_{n-p+1}\ldots x_n$. The first $\tau$ bits of $x_{pre}$ (with $\tau$ dependent on the constraint graph) are referred to as $x_{aux}$, i.e., $x_{aux}=x_1x_2\ldots x_{\tau}$ 
    {and hold frequently used content, to be copied close to the target qubits in order to reduce the circuit depth of the Gray cycle stage.} 
    \item The $m$ ancillary qubits are divided into three registers:
    \begin{itemize}
        \item ${\sf R}_{\rm copy}$: the copy register of size $\lambda_{copy} \ge n$
        \item ${\sf R}_{\rm targ}$: the target register of size $\lambda_{targ} = 2^p \ge n$
        \item ${\sf R}_{\rm aux}$: the auxiliary register  of size $\lambda_{aux} \ge n$ 
    \end{itemize}
\end{itemize}
{A few remarks on why we need the registers each have size at least $n$.} As our approach requires creating at least one copy for each of $\ket{x_{pre}}$ and $\ket{x_{suf}}$ (for a total of $n$ qubits), we require at least $n$ ancillary qubits for the copy register. If the size of the target register is $o(n)$, the circuit depths achievable by methods from this section will be larger than the circuit depths in
Appendix \ref{sec:diag_without_ancilla_path} and Appendix \ref{sec:diag_without_ancilla_binarytree}. We therefore must also allow $n$ ancillary qubits for the target register. While the auxiliary register may be smaller than $n$, for simplicity we also allow $n$ qubits here, and therefore in total we assume that $m\ge 3n$.

The circuit itself consists of $5$ stages. 
\begin{enumerate}
    \item Suffix Copy: makes $O\left(\lambda_{copy}/p\right)$ copies of $\ket{x_{suf}}$ in ${\sf R}_{\rm copy}$.
    
    \item Gray Initial: prepares the state $\ket{\langle c_1^{\ell_1}t_1, x\rangle}\otimes\cdots\otimes\ket{\langle c_1^{\ell_{2^p}}t_{2^p},x\rangle}=\ket{\langle t_1, x_{suf}\rangle}\otimes\cdots\otimes\ket{\langle t_{2^p},x_{suf}\rangle}$ in ${\sf R}_{\rm targ}$, where $\ell_k$ (for $k\in[2^p]$) are integers specifying $2^p$ $(n-p, \ell_k)$-Gray codes $\{c^{\ell_k}_1,c^{\ell_k}_2,\ldots c^{\ell_k}_{2^{n-p}}\}$, $\{t_1, \ldots, t_{2^p}\} = \{0,1\}^{p}$, and $c_1^i=0^{n-p}$ and $t_i$ are the prefix and suffix of $s$ (see Eq.~\eqref{eq:task1}). 
    
    \item Prefix Copy: makes $O\left(\lambda_{aux}/\tau\right)$ copies of $\ket{x_{aux}}$ in ${\sf R}_{\rm aux}$, and replaces the copies of $\ket{x_{suf}}$ in ${\sf R}_{\rm copy}$ with $O\left(\lambda_{copy}/(n-p)\right)$ copies of $\ket{x_{pre}}$.
    
    \item Gray Cycle: 
    This stage enumerates all $2^{n-p}$ prefixes of $s$ by going along a Gray code---each qubit $k$ uses $(n-p,\ell_k)$-Gray code, which consists of $2^{n-p}$ steps, with each step $j$ responsible for (i) updating 
    prefix, and (ii) implementing a phase shift (see further details below).
    
    \item Inverse: restores all ancillary qubits to zero.
\end{enumerate}
More precisely, if we define\footnote{There exist some qubits which are not utilized to store copies of suffixes and prefixes. We omit these qubits for simplicity.}
\begin{equation*}
     \ket{x_{SufCopy}}:=\underbrace{\ket{x_{suf}\cdots x_{suf}}}_{O(\frac{\lambda_{copy}}{p})\text{~copies~of~}x_{suf}},\quad 
     \ket{x_{PreCopy}}:= \underbrace{\ket{x_{pre}\cdots x_{pre}}}_{O(\frac{\lambda_{copy}}{n-p})~{\rm copies~of}~x_{pre}},\quad
     \ket{x_{AuxCopy}}:= \underbrace{\ket{x_{aux}\cdots x_{aux}}}_{O(\frac{\lambda_{aux}}{\tau})~{\rm copies~of}~x_{aux}},
\end{equation*}
as well as, for all $j\in[2^{n-p}]$ and all $k\in [2^p]$
\begin{equation}\label{eq:s,f}
    s(j,k):= c_j^{\ell_k} t_k, \qquad f_{j,k}:= \langle s(j,k), x\rangle, \qquad \ket{f_j}_{{\sf R}_{\rm targ}}:=\bigotimes_{k\in[2^p]}\ket{f_{j,k}}_{{\sf R}_{\rm targ,k}},
\end{equation}
where ${\sf R}_{{\rm targ},k}$ is the $k$-th qubit in ${\sf R}_{\rm targ}$, then unitary operators corresponding to each of the above 5 stages can be expressed as:
\begin{align}
U_{SufCopy}\ket{x}_{{\sf R}_{\rm inp}}\ket{0^{\lambda_{copy}}}_{{\sf R}_{\rm copy}}&=\ket{x}_{{\sf R}_{\rm inp}}\ket{x_{SufCopy}}_{{\sf R}_{\rm copy}}, \label{eq:sufcopy_graph} \\
U_{GrayInit}\ket{x_{SufCopy}}_{{\sf R}_{\rm copy}}\ket{0^{\lambda_{targ}}}_{{\sf R}_{\rm targ}}&=\ket{x_{SufCopy}}_{{\sf R}_{\rm copy}}\ket{f_1}_{{\sf R}_{\rm targ}},\label{eq:gray_initial_graph}\\
U_{PreCopy}\ket{x}_{{\sf R}_{\rm inp}}\ket{x_{SufCopy}}_{\sf{R}_{\rm copy}}\ket{0^{\lambda_{aux}}}_{{\sf R}_{\rm aux}}&=\ket{x}_{{\sf R}_{\rm inp}}\ket{x_{PreCopy}}_{\sf{R}_{\rm copy}}\ket{x_{AuxCopy}}_{{\sf R}_{\rm aux}},\label{eq:precopy_graph}\\
U_{GrayCycle}\ket{x_{PreCopy}}_{{\sf R}_{\rm copy}} \ket{f_1}_{{\sf R}_{\rm targ}}\ket{x_{AuxCopy}}_{{\sf R}_{{\rm aux}}} &=e^{i\theta(x)}\ket{x_{PreCopy}}_{{\sf R}_{\rm copy}} \ket{f_{1}}_{{\sf R}_{\rm targ}}\ket{x_{AuxCopy}}_{{\sf R}_{{\rm aux}}},\label{eq:gray_cycle_graph}\\
 U_{Inverse}\ket{x}_{{\sf R}_{\rm inp}}\ket{x_{PreCopy}}_{{\sf R}_{\rm copy}} \ket{f_1}_{{\sf R}_{\rm targ}}\ket{x_{AuxCopy}}_{{\sf R}_{\rm aux}} &=\ket{x}_{{\sf R}_{\rm inp}}\ket{0^{\lambda_{copy}}}_{{\sf R}_{\rm copy}}\ket{0^{\lambda_{targ}}}_{{\sf R}_{\rm targ}}\ket{0^{\lambda_{aux}}}_{{\sf R}_{\rm aux}}.\label{eq:inverse_graph}
\end{align}

It is straightforward to verify that the sequential application of these unitary operators implements
$\Lambda_n$, i.e. $\ket{x}\to e^{i\theta(x)}\ket{x}$ for all $x\in\Bn$, as in Eq. \eqref{eq:diag}. Note that $c_1^{\ell_i}:=0^{n-p}$ for all $i\in[2^p]$ and thus $\ket{f_1}= \ket{\langle t_1, x_{suf}\rangle}\otimes \cdots \otimes \ket{\langle t_{2^p}, x_{suf}\rangle}$.

\begin{figure}[hbt]
    \centerline{
    \begin{tabular}{c}
\begin{pgfpicture}{0em}{0em}{0em}{0em}
\color{llgray}
\pgfrect[fill]{\pgfpoint{4.5em}{0.5 em}}{\pgfpoint{31em}{-12em}}
\color{llblue}
\pgfrect[fill]{\pgfpoint{4.5em}{-12 em}}{\pgfpoint{31em}{-4.5em}}
\color{llgreen}
\pgfrect[fill]{\pgfpoint{4.5em}{-16.9 em}}{\pgfpoint{31em}{-4.6em}}
\color{llyellow}
\pgfrect[fill]{\pgfpoint{4.5em}{-22.0 em}}{\pgfpoint{31em}{-4.5em}}
\color{lgray}
\pgfrect[fill]{\pgfpoint{31em}{0.5 em}}{\pgfpoint{9em}{-12em}}
\color{lblue}
\pgfrect[fill]{\pgfpoint{31em}{-12 em}}{\pgfpoint{9em}{-4.5em}}
\color{lgreen}
\pgfrect[fill]{\pgfpoint{31em}{-16.9 em}}{\pgfpoint{9em}{-4.6em}}
\color{lyellow}
\pgfrect[fill]{\pgfpoint{31em}{-22.0 em}}{\pgfpoint{9em}{-4.5em}}
\end{pgfpicture}
    \Qcircuit  @C=0.6em @R=0.5em {
    &\lstick{\scriptstyle\ket{x_1}} & \qw & \qw & \qw &  \qw & \multigate{5}{\scriptstyle U_{PreCopy}}& \qw &\qw & \ustick{e^{i\theta(x)}}\qw &  \multigate{17}{\scriptstyle U_{Inverse}} & \qw &\rstick{\scriptstyle \ket{x_1}} \\
     &\lstick{\scriptstyle\vdots~~} & \qw & \qw &\qw &  \qw &  \ghost{\scriptstyle U_{PreCopy}}&\qw & \qw &\qw &  \ghost{\scriptstyle U_{Inverse}} & \qw & \rstick{\scriptstyle ~~\vdots} \inputgroupv{2}{2}{2.5 em}{0em}{\scriptstyle\ket{x_{aux}}~~~~~~}\\
     &\lstick{\scriptstyle\ket{x_\tau}} &\qw & \qw & \qw & \qw & \ghost{\scriptstyle U_{PreCopy}}&\qw & \qw &\qw &  \ghost{\scriptstyle U_{Inverse}} & \qw & \rstick{\scriptstyle \ket{x_\tau}} \\
    &\lstick{\scriptstyle\ket{x_{\tau+1}}} & \qw &\qw & \qw & \qw & \ghost{\scriptstyle U_{PreCopy}}&\qw & \qw &\qw &  \ghost{\scriptstyle U_{Inverse}} & \qw & \rstick{\scriptstyle \ket{x_{\tau+1}}}\\
    & \lstick{\scriptstyle\vdots~~} &\qw &\qw &\qw &\qw &  \ghost{\scriptstyle U_{PreCopy}}&\qw &\qw &\qw &  \ghost{\scriptstyle U_{Inverse}} & \qw & \rstick{{\scriptstyle ~~\vdots}~~~~{\sf R}_{\rm inp}} \\
     &\lstick{\scriptstyle\ket{x_{n-p}}} &\qw & \qw &\qw & \qw & \ghost{\scriptstyle U_{PreCopy}}&\qw & \qw &\qw &  \ghost{\scriptstyle U_{Inverse}} & \qw & \rstick{\scriptstyle\ket{x_{n-p}}}  \inputgroupv{3}{3}{7.5 em}{0em}{\scriptstyle \ket{x_{pre}}~~~~~~~~~~~~~~~~~~~~~~~~~~~}\\
    &\lstick{\scriptstyle\ket{x_{n-p+1}}} & \multigate{5}{\scriptstyle U_{SufCopy}} & \qw &\qw &\qw & \qw & \qw &\qw &\qw &  \ghost{\scriptstyle U_{Inverse}} & \qw & \rstick{\scriptstyle\ket{x_{n-p+1}}}\\
     &\lstick{\scriptstyle\vdots~~} & \ghost{\scriptstyle U_{SufCopy}} & \qw &\qw & \qw & \qw \qwx[-2]\qwx[2] &\qw &\qw &\qw & \ghost{\scriptstyle U_{Inverse}} & \qw & \rstick{\scriptstyle ~~\vdots} \inputgroupv{8}{8}{4.6 em}{0em}{\scriptstyle \ket{x_{suf}}~~~~~~~~~~~~~}\\
     &\lstick{\scriptstyle\ket{x_n}} & \ghost{\scriptstyle U_{SufCopy}} & \qw &\qw &\qw &  \qw & \qw &\qw &\qw &  \ghost{\scriptstyle U_{Inverse}} & \qw & \rstick{\scriptstyle\ket{x_n}} \\
    &\lstick{\scriptstyle\ket{0}} & \ghost{\scriptstyle U_{SufCopy}}  & \push{\scriptstyle\ket{x_{suf}}}\qw & \multigate{5}{\scriptstyle U_{GrayInit}}&\qw &  \multigate{2}{\scriptstyle U_{PreCopy}}& \push{\scriptstyle\ket{x_{pre}}}\qw & \multigate{8}{\scriptstyle U_{GrayCycle}} &\qw &  \ghost{\scriptstyle U_{Inverse}} & \qw & \rstick{\scriptstyle\ket{0}}\\
     &\lstick{\scriptstyle\vdots~~} & \ghost{\scriptstyle U_{SufCopy}} & \push{\scriptstyle\vdots}\qw &\ghost{\scriptstyle U_{GrayInit}}& \qw & \ghost{\scriptstyle U_{PreCopy}}& \push{\scriptstyle\vdots}\qw & \ghost{\scriptstyle U_{GrayCycle}} & \qw & \ghost{\scriptstyle U_{Inverse}} & \qw &\rstick{{\scriptstyle ~~\vdots}~~~~{\sf R}_{\rm copy}}\\
     &\lstick{\scriptstyle\ket{0}} & \ghost{\scriptstyle U_{SufCopy}} & \push{\scriptstyle\ket{x_{suf}}} \qw &\ghost{\scriptstyle U_{GrayInit}}&\qw &  \ghost{\scriptstyle U_{PreCopy}}&\push{\scriptstyle\ket{x_{pre}}} \qw &\ghost{\scriptstyle U_{GrayCycle}} &\qw &  \ghost{\scriptstyle U_{Inverse}} & \qw &\rstick{\scriptstyle\ket{0}} \inputgroupv{11}{11}{3 em}{0em}{\scriptstyle \lambda_{copy}~\text{qubits}~~~~~~~~~~~~~~~~}\\
    &\lstick{\scriptstyle\ket{0}} & \qw &\qw & \ghost{\scriptstyle U_{GrayInit}}& \push{\scriptstyle\ket{\langle c_1^{\ell_1}t_1,x\rangle}}\qw &  \qw &\qw &  \ghost{\scriptstyle U_{GrayCycle}} & \push{\scriptstyle\ket{\langle c_1^{\ell_1}t_1,x\rangle}}\qw &  \ghost{\scriptstyle U_{Inverse}} &\qw & \rstick{\scriptstyle\ket{0}}\\
     &\lstick{\scriptstyle\vdots~~} &\qw &\qw &\ghost{\scriptstyle U_{GrayInit}}&\push{\scriptstyle\vdots }\qw &  \qw \qwx[-2] \qwx[2] &\qw & \ghost{\scriptstyle U_{GrayCycle}} &\push{\scriptstyle \vdots}\qw & \ghost{\scriptstyle U_{Inverse}} &\qw &\rstick{{\scriptstyle ~~\vdots}~~~~{\sf R}_{\rm targ}}\\
     &\lstick{\scriptstyle\ket{0}} &\qw &\qw & \ghost{\scriptstyle U_{GrayInit}}& \push{\scriptstyle\ket{\langle c_1^{\ell_{2^p}}t_{2^p},x\rangle}}\qw &  \qw &\qw &  \ghost{\scriptstyle U_{GrayCycle}} &\push{\scriptstyle\ket{\langle c_1^{\ell_{2^p}}t_{2^p},x\rangle}}\qw &  \ghost{\scriptstyle U_{Inverse}} &\qw &\rstick{\scriptstyle\ket{0}} \inputgroupv{14}{14}{3 em}{0em}{\scriptstyle \lambda_{targ}=2^p~\text{qubits}~~~~~~~~~~~~~~~~~~~}\\
    &\lstick{\scriptstyle\ket{0}} &\qw & \qw &\qw & \qw &  \multigate{2}{\scriptstyle U_{PreCopy}}& \push{\scriptstyle\ket{x_{aux}}}\qw &\ghost{\scriptstyle U_{GrayCycle}} &\qw &  \ghost{\scriptstyle U_{Inverse}} & \qw &\rstick{\scriptstyle\ket{0}}\\
     &\lstick{\scriptstyle\vdots~~} &\qw &\qw &\qw &\qw &  \ghost{\scriptstyle U_{PreCopy}}& \push{\scriptstyle\vdots}\qw &\ghost{\scriptstyle U_{GrayCycle}} &\qw &  \ghost{\scriptstyle U_{Inverse}} &\qw &\rstick{{\scriptstyle ~~\vdots}~~~~{\sf R}_{\rm aux}}\\
     &\lstick{\scriptstyle\ket{0}} &\qw & \qw &\qw &\qw &  \ghost{\scriptstyle U_{PreCopy}}&  \push{\scriptstyle\ket{x_{aux}}}\qw &\ghost{\scriptstyle U_{GrayCycle}} &\qw &  \ghost{\scriptstyle U_{Inverse}} & \qw &\rstick{\scriptstyle\ket{0}}  \inputgroupv{17}{17}{3 em}{0em}{\scriptstyle \lambda_{aux}~\text{qubits}~~~~~~~~~~~~~~~~}\\
    }
    \end{tabular}
    }
    \caption{Circuit framework for implementing diagonal unitaries $\Lambda_n$ with $m$ ancillary qubits under graph constraints. The framework consists of 5 stages: prefix copy, Gray initial, prefix copy, Gray cycle and inverse. The input register $\textsf{R}_{\rm inp}$ (grey) corresponds to the $n$ input qubits of $\Lambda_n$. The $m$ ancillary qubits are partitioned into 3 registers, ${\sf R}_{\rm copy}$ (blue), ${\sf R}_{\rm targ}$ (green) and ${\sf R}_{\rm aux}$ (yellow). Darker shading indicates that the phase shift $e^{i\theta(x)}$ has been effected. Note that $c_1^{\ell_i}:=0^{n-p}$ for all $i\in[2^p]$, and thus $\langle c_1^{\ell_i}t_i, x\rangle = \langle t_i, x_{suf}\rangle$.
    }
    \label{fig:diag_with_ancilla_framwork}
\end{figure}
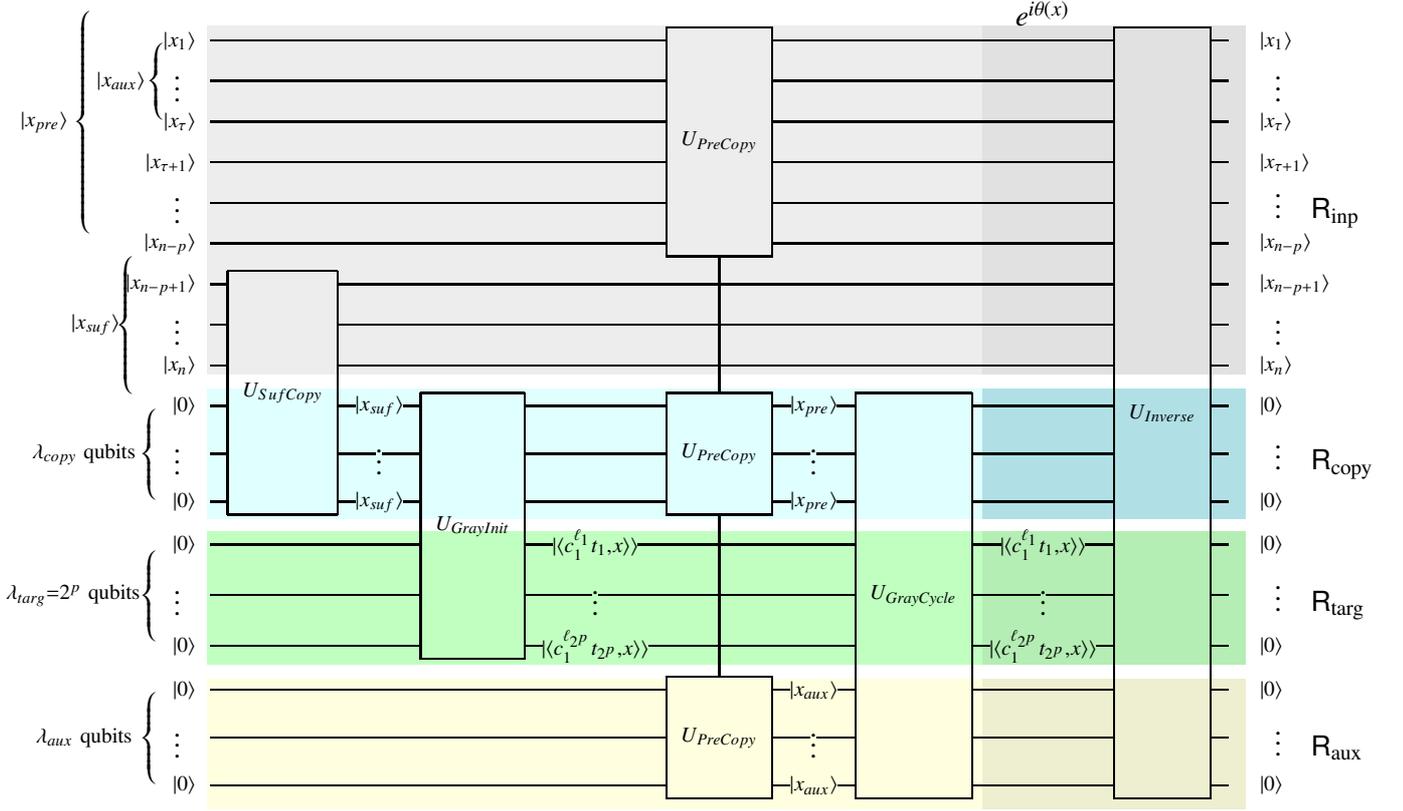

\begin{table}[h!]\small
    \centering
    \caption{Choice of integers $\ell_k$ ($k=1, \ldots, 2^p$) which specify the $2^{p}$ Gray codes used in the Gray Cycle stage, for various graph constraints. $K_{n+m}$ is the complete graph on $n+m$ vertices, and corresponds to no connectivity constraints. 
    }
    \begin{tabular}{c|c|cccc}\footnotesize
         &  $K_{n+m}$~\cite{sun2021asymptotically}  & $\Path_{n+m}$  & $\Grid^{n_1,\ldots,n_d}_{n+m}$& $\Tree_{n+m}(2)$ & $\Expander_{n+m}$ \\ \hline
    $\ell_k$    & $(k-1)\mod (n-p)+1$ & $(k-1)\mod (n-p)+1$ & $(k-1)\mod (n-p)+1$   & $1$ & $1$ 
    \end{tabular}
    \label{tab:choice-of-graycycle}
\end{table}

Next, we show circuit depth bounds for several of the stages under general graph constraint.  In what follows, we use $\mathcal{D}(U)$ to denote the circuit depth required to implement operator $U$.

\paragraph{Prefix Copy} It will be convenient to define the following two operators $U''_{PreCopy}$, $U'''_{PreCopy}$
\begin{align}
   \ket{x}_{{\sf R}_{\rm inp}}\ket{0^{\lambda_{copy}}}_{\sf{R}_{\rm copy}}\ket{0^{\lambda_{aux}}}_{{\sf R}_{\rm aux}} &\xrightarrow{U''_{PreCopy}}\ket{x}_{{\sf R}_{\rm inp}}\ket{x_{PreCopy}}_{\sf{R}_{\rm copy}}\ket{0^{\lambda_{aux}}}_{{\sf R}_{\rm aux}},\label{eq:precopy2_graph} \\
   \ket{x}_{{\sf R}_{\rm inp}}\ket{x_{PreCopy}}_{\sf{R}_{\rm copy}}\ket{0^{\lambda_{aux}}}_{{\sf R}_{\rm aux}}& \xrightarrow{U'''_{PreCopy}}\ket{x}_{{\sf R}_{\rm inp}}\ket{x_{PreCopy}}_{\sf{R}_{\rm copy}}\ket{x_{AuxCopy}}_{{\sf R}_{\rm aux}}.\label{eq:precopy3_graph}
\end{align}

\begin{lemma}[]\label{lem:precopy_graph}
$U_{PreCopy}$ can be implemented in circuit depth 
\begin{equation*}
    \mathcal{D}(U_{PreCopy})\le \mathcal{D}(U_{SufCopy})+\mathcal{D}(U''_{PreCopy})+\mathcal{D}(U'''_{PreCopy}).
\end{equation*} 
\end{lemma}

\begin{proof}
$U_{PreCopy}$ can be implemented in the following way:
\begin{align*}
    &\ket{x}_{{\sf R}_{\rm inp}}\ket{x_{SufCopy}}_{\sf{R}_{\rm copy}}\ket{0^{\lambda_{aux}}}_{{\sf R}_{\rm aux}}\\
    \xrightarrow{U^\dagger_{SufCopy}}&\ket{x}_{{\sf R}_{\rm inp}}\ket{0^{\lambda_{copy}}}_{\sf{R}_{\rm copy}}\ket{0^{\lambda_{aux}}}_{{\sf R}_{\rm aux}},\\
    \xrightarrow{U''_{PreCopy}}&\ket{x}_{{\sf R}_{\rm inp}}\ket{x_{PreCopy}}_{\sf{R}_{\rm copy}}\ket{0^{\lambda_{aux}}}_{{\sf R}_{\rm aux}},\\
    \xrightarrow{U'''_{PreCopy}}&\ket{x}_{{\sf R}_{\rm inp}}\ket{x_{PreCopy}}_{\sf{R}_{\rm copy}}\ket{x_{AuxCopy}}_{{\sf R}_{\rm aux}}.
\end{align*}
\end{proof}

\paragraph{Gray Cycle} Let $s(j,k)$ and $\ket{f_j}$ be as in Eq.~\eqref{eq:s,f}. $U_{GrayCycle}$ consists of $2^{n-p}$ phases.
For $j\le 2^{n-p}-1$, the $j$-th phase consists of two parts, $U^{(j)}_{Gen}$ and $R_j$:
\begin{align}
    &\ket{x_{PreCopy}}_{{\sf R}_{\rm copy}} \ket{f_j}_{{\sf R}_{\rm targ}}\ket{x_{AuxCopy}}_{{\sf R}_{\rm aux}}\nonumber\\
    \xrightarrow{U_{Gen}^{(j)}}&\ket{x_{PreCopy}}_{{\sf R}_{\rm copy}} \ket{f_{j+1}}_{{\sf R}_{\rm targ}}\ket{x_{AuxCopy}}_{{\sf R}_{\rm aux}},  \label{eq:Ugenj_graph}\\
     \xrightarrow{R_j}& e^{i(\sum_{k=1}^{2^p} f_{j+1,k},\alpha_{s(j+1,k)})}\ket{x_{PreCopy}}_{{\sf R}_{\rm copy}} \ket{f_{j+1}}_{{\sf R}_{\rm targ}}\ket{x_{AuxCopy}}_{{\sf R}_{\rm aux}},\label{eq:rotationj_graph}
\end{align}
Note that $R_j$ consists of $2^{p}$ single-qubit gates acting on target register, i.e., $R_j:=\bigotimes_{k=1}^{2^p}R(\alpha_{s(j+1,k)})$ of depth 1.
The $2^{n-p}$-th phase is
\begin{align}
    &\ket{x_{PreCopy}}_{{\sf R}_{\rm copy}} \ket{f_{2^{n-p}}}_{{\sf R}_{\rm targ}}\ket{x_{AuxCopy}}_{{\sf R}_{\rm aux}}\nonumber\\
    \xrightarrow{U_{Gen}^{(2^{n-p})}}&\ket{x_{PreCopy}}_{{\sf R}_{\rm copy}} \ket{f_{1}}_{{\sf R}_{\rm targ}}\ket{x_{AuxCopy}}_{{\sf R}_{\rm aux}},  \label{eq:Ugen2n-p_graph}\\
     \xrightarrow{R_{2^{n-p}}}& e^{i(\sum_{k=1}^{2^p} f_{1,k},\alpha_{s(1,k)})}\ket{x_{PreCopy}}_{{\sf R}_{\rm copy}} \ket{f_{1}}_{{\sf R}_{\rm targ}}\ket{x_{AuxCopy}}_{{\sf R}_{\rm aux}}.\label{eq:rotation2n-p_graph}
\end{align}
$R_{2^{n-p}}$ consists of $2^{p}$ single-qubit gates acting on target register, i.e., $R_{2^{n-p}}:=\bigotimes\limits_{k\in[2^p]}R(\alpha_{s(1,k)})$ of depth 1.

By applying these $2^{n-p}$ phases, the following transformation is implemented:
\[\ket{x_{PreCopy}}_{{\sf R}_{\rm copy}} \ket{f_1}_{{\sf R}_{\rm targ}}\ket{x_{AuxCopy}}_{{\sf R}_{{\rm aux}}} \to e^{i\left(\sum\limits_{j=1}^{2^{n-p}}\sum\limits_{k=1}^{2^p}f_{j,k}\alpha_{s(j,k)}\right)}\ket{x_{PreCopy}}_{{\sf R}_{\rm copy}} \ket{f_{1}}_{{\sf R}_{\rm targ}}\ket{x_{AuxCopy}}_{{\sf R}_{{\rm aux}}}.\]
From Eq.~\eqref{eq:alpha}, $\sum_{j=1}^{2^{n-p}}\sum_{k=1}^{2^p}f_{j,k}\alpha_{s(j,k)}=\theta(x)$ for all $x\in\Bn$, and the above procedure implements the desired $U_{GrayCycle}$ transformation of Eq.~\eqref{eq:gray_cycle_graph}.

\begin{lemma}[]\label{lem:graycycle_graph}
$U_{GrayCycle}$ can be implemented in circuit depth
\[\mathcal{D}(U_{GrayCycle})\le \sum_{j=1}^{2^{n-p}}\mathcal{D}(U_{Gen}^{(j)})+2^{n-p}.\]
\end{lemma}
\begin{proof}
For $j\in[2^{n-p}]$, the depth of the $j$-th phase is $\mathcal{D}(U_{Gen}^{(j)})+1$. The total depth is therefore $\sum_{j=1}^{2^{n-p}}(\mathcal{D}(U_{Gen}^{(j)})+1)=\sum_{j=1}^{2^{n-p}}\mathcal{D}(U_{Gen}^{(j)})+2^{n-p}.$
\end{proof}

\paragraph{Inverse}
The inverse stage can be implemented as follows.
\begin{align}
    &\ket{x}_{{\sf R}_{\rm inp}}\ket{x_{PreCopy}}_{{\sf R}_{\rm copy}} \ket{f_1}_{{\sf R}_{\rm targ}}\ket{x_{AuxCopy}}_{{\sf R}_{\rm aux}} \nonumber\\
  \xrightarrow{U^\dagger_{PreCopy}}  &\ket{x}_{{\sf R}_{\rm inp}}\ket{x_{SufCopy}}_{{\sf R}_{\rm copy}} \ket{f_1}_{{\sf R}_{\rm targ}}\ket{0^{\lambda_{aux}}}_{{\sf R}_{\rm aux}}, \label{eq:inverse1_graph}\\
  \xrightarrow{U^\dagger_{GrayInit}}  & \ket{x}_{{\sf R}_{\rm inp}}\ket{x_{SufCopy}}_{{\sf R}_{\rm copy}} \ket{0^{\lambda_{targ}}}_{{\sf R}_{\rm targ}}\ket{0^{\lambda_{aux}}}_{{\sf R}_{\rm aux}}, \label{eq:inverse2_graph}\\
  \xrightarrow{U_{SufCopy}^\dagger}  & \ket{x}_{{\sf R}_{\rm inp}}\ket{0^{\lambda_{copy}}}_{{\sf R}_{\rm copy}} \ket{0^{\lambda_{targ}}}_{{\sf R}_{\rm targ}}\ket{0^{\lambda_{aux}}}_{{\sf R}_{\rm aux}}.\label{eq:inverse3_graph}
\end{align}
It follows from Lemma~\ref{lem:precopy_graph} that:

\begin{lemma} []\label{lem:inverse_graph}
$U_{Inverse}$ can be implemented in depth
\[\mathcal{D}(U_{Inverse})\le 2\mathcal{D}(U_{SufCopy})+\mathcal{D}(U_{PreCopy}'')+\mathcal{D}(U_{PreCopy}''')+\mathcal{D}(U_{GrayInit}).\]
\end{lemma}

\subsection{Efficient circuits: general framework}

We use the framework of Fig.~\ref{fig:diag_with_ancilla_framwork} for implementing $\Lambda_n$ under path (Appendix \ref{sec:diag_with_ancilla_path}), grid (Appendix  \ref{sec:diag_with_ancilla_grid_d}) and complete binary tree (Appendix \ref{sec:diag_with_ancilla_binarytree}) constraints. The case for expander graph (Appendix \ref{sec:diag_with_ancilla_expander}) constraints differs slightly.

The constructions of~\cite{sun2021asymptotically} give $O\left(n+\frac{2^n}{n+m}\right)$-depth and $O(2^n)$-size upper bounds for implementing $\Lambda_n$ under no graph constraints, using $m$ ancillary qubits (see Table~\ref{tab:lambda-bounds_ancilla}).  Similar to the trivial upper bounds of Section~\ref{sec:general_framework_noancilla_main}, a trivial depth upper bound for $\Lambda_n$ of $O\left((n+m)\cdot \diam(G)\cdot (n+\frac{2^n}{n+m})\right)$ can be given under graph $G$ constraints.

\begin{table}[!ht]
    \centering
    \caption{Circuit depth upper (ub) and lower bounds (lb) required to implement $\Lambda_n$ in circuits under various graph constraints, using $m$ ancillary qubits.
    The trivial bounds are based on the unconstrained construction from~\cite{sun2021asymptotically} and Lemma~\ref{lem:cnot_path_constraint}, which implies that, under constraint graph $G$, the required circuit depth is $O((n+m)\cdot \diam(G)\cdot (n+\frac{2^n}{n+m}))$. Big O and $\Omega$ notation is suppressed. 
    }
    \resizebox{\textwidth}{!}{\begin{tabular}{c|cccc}
         &  $\Path_{n+m}$ & $\Grid^{n_1,\ldots,n_d}_{n+m}$ & $\Tree_{n+m}(2)$ & $\Expander_{n+m}$ \\ \hline
    $\diam(G)$  & $n+m$ & $\sum_{j=1}^d n_j$ & $\log (n+m)$ & $\log(n+m)$\\ \hline
    Depth (ub, trival)     & $(n+m)(n(n+m)+2^n)$ & $(\sum_{j=1}^d n_j)(n(n+m)+2^n)$ & $\log(n+m)(n(n+m)+2^n)$ & $\log(n+m)(n(n+m)+2^n)$ \\
    \multirow{2}{*}{Depth (ub)}    & $2^{n/2}+\frac{2^n}{n+m}$   &  $n^2+d2^{\frac{n}{d+1}}+\max\limits_{j\in\{2,\ldots,d\}}\left\{\frac{d2^{n/j}}{(\Pi_{i=j}^d n_i)^{1/j}}\right\}+\frac{2^n}{n+m}$  & $n^2\log(n)+\frac{\log(n) 2^n}{n+m}$ & $n^2+\frac{\log(m) 2^n}{n+m}$   \\
     & [Lem.~\ref{lem:diag_path_ancillary}] &[Lem.~\ref{lem:diag_grid_ancillary}] &[Lem.~\ref{lem:diag_binarytree_withancilla}]& [Lem.~\ref{lem:diag_expander_ancilla}]\\ \hline
    \multirow{2}{*}{Depth (lb)} &$2^{n/2}+\frac{2^n}{n+m}$  & $n+2^{\frac{n}{d+1}}+\max\limits_{j\in [d]}\big\{\frac{2^{n/j}}{(\Pi_{i=j}^d n_i)^{1/j}}\big\}$  &$n+\frac{2^n}{n+m}$  & $n+\frac{2^n}{n+m}$  \\
    &[Cor. \ref{coro:lower_bound_path}]&[Lem. \ref{lem:lower_bound_grid_k_Lambda}]&[Cor. \ref{coro:lower_bound_binary}]&[Cor. \ref{coro:lower_bound_expander}]
    \end{tabular}}
    \label{tab:lambda-bounds_ancilla}
\end{table}

To achieve the more efficient constructions summarized in the second last row of Table~\ref{tab:lambda-bounds_ancilla}, for each constraint graph type we must carefully choose:
\begin{enumerate}
    \item the size and locations for ${\sf R}_{\rm inp}$, ${\sf R}_{\rm copy}$, ${\sf R}_{\rm targ}$ and ${\sf R}_{\rm aux}$, and
    \item the particular Gray codes used, i.e., the integers $\ell_1, \ell_2, \ldots, \ell_{2^{p}}$ used to implement the Gray cycle stage.
\end{enumerate}
From the previous section (Lemmas~\ref{lem:precopy_graph}, \ref{lem:graycycle_graph}, \ref{lem:inverse_graph}), to bound the $\Lambda_n$ circuit depth complexity for each graph constraint type, it is sufficient to analyze the circuits implementing
$U_{SufCopy}$ (Eq.~\eqref{eq:sufcopy_graph}),
$U''_{PreCopy}$     (Eq.~\eqref{eq:precopy2_graph}), 
$U'''_{PreCopy}$    (Eq.~\eqref{eq:precopy3_graph}), 
$U_{GrayInit}$    (Eq.~\eqref{eq:gray_initial_graph}), and 
$U_{Gen}^{(j)}$   (Eq.~\eqref{eq:Ugen2n-p_graph}).

As in Appendix~\ref{append:diag_without_ancilla}, we aim to minimize circuit depth by arranging qubit registers and Gray codes such that control and target qubits for required CNOT gates are close, and constraint paths for different CNOT gates are disjoint (and hence implementable in parallel) where possible.


\subsection{Circuit implementation under $\Path_{n+m}$ constraints (Proof of Lemma \ref{lem:diag_d_grid_ancilla} (Case 1))}
\label{sec:diag_with_ancilla_path}
We assume that $m\ge 3n$ and $\frac{m}{3}$ is an integer. Without loss of generality, we also assume that $m\le 3\cdot 2^n$;  
if $m> 3\cdot 2^n$, we only use $3\cdot 2^n$ ancillary qubits. We take $p=\big\lfloor \log (\frac{m}{3}) \big\rfloor$, $\tau=2\lceil\log (n-p)\rceil$, $\lambda_{copy}=\lambda_{targ}=2^p$, and  $\lambda_{aux}=r\tau$ where $r=\frac{2^p}{n-p}$ \footnote{Here we assume $2^p$ is a multiple of $(n-p)$ for convenience. In the general case where this assumption does not hold, we can define $r=\big\lceil\frac{2^p}{n-p}\big\rceil$ with the last register $R_r$ holding  the leftover qubits. The details are tedious and technically uninteresting, thus omitted here.}.

\paragraph{Choice of registers}


We assign qubits to ${\sf R}_{\rm inp}$, ${\sf R}_{\rm copy}$, ${\sf R}_{\rm targ}$ and ${\sf R}_{\rm aux}$ as in Fig.~\ref{fig:register_in_path}.
\begin{itemize}
    \item ${\sf R}_{\rm inp}$ consists of the first $n$ qubits.
    
    \item The $2\cdot 2^p+r\tau$ ancillary qubits are divided into $r$ registers $\textsf{R}_1,\ldots,\textsf{R}_r$.
    
    \item Each $\textsf{R}_k$ for $k\in[r]$ contains $2(n-p) +\tau$ qubits, with the first $2(n-p)$ qubits alternately assigned to ${\sf R}_{\rm copy}$ and ${\sf R}_{\rm targ}$, and the final $\tau$ qubits assigned to ${\sf R}_{\rm aux}$.
    
\end{itemize}

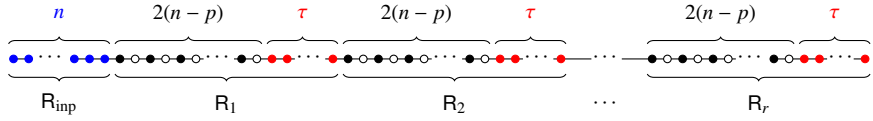
\begin{figure}[!hbt]
    \centering
    \begin{tikzpicture}
    \draw (-1.4,0)--(-1.2,0) (-0.6,0)--(1.1,0) (1.5,0)--(2.3,0) (2.7,0)--(4.1,0) (4.5,0)--(5.3,0) (5.7,0)--(6.2,0) (6.6,0)--(8.1,0) (8.6,0)--(9.3,0) (9.65,0)--(9.8,0);
       \draw [fill=blue,draw=blue] (-0.2,0) circle (0.05) (-0.4,0) circle (0.05) (-0.6,0) circle (0.05) (-1.2,0) circle (0.05) (-1.4,0) circle (0.05);
       \draw [fill=black] (0,0) circle (0.05) (0.4,0) circle (0.05) (0.8,0) circle (0.05) (1.6,0) circle (0.05) (3,0) circle (0.05) (3.4,0) circle (0.05) (3.8,0) circle (0.05) (4.6,0) circle (0.05) (7,0) circle (0.05) (7.4,0) circle (0.05) (7.8,0) circle (0.05) (8.6,0) circle (0.05) ;
       \draw [fill=white] (0.2,0) circle (0.05) (0.6,0) circle (0.05) (1.0,0) circle (0.05) (1.8,0) circle (0.05) (3.2,0) circle (0.05) (3.6,0) circle (0.05) (4.0,0) circle (0.05) (4.8,0) circle (0.05) (7.2,0) circle (0.05) (7.6,0) circle (0.05) (8.0,0) circle (0.05) (8.8,0) circle (0.05);
       \draw [draw=red, fill=red] (2,0) circle (0.05) (2.2,0) circle (0.05) (2.8,0) circle (0.05) (5,0) circle (0.05) (5.2,0) circle (0.05) (5.8,0) circle (0.05) (9,0) circle (0.05) (9.2,0) circle (0.05) (9.8,0) circle (0.05);
       \draw (1.3,0) node{\scriptsize $\cdots$} (2.5,0) node{\scriptsize $\cdots$} (4.3,0) node{\scriptsize $\cdots$} (5.5,0) node{\scriptsize $\cdots$} (8.3,0) node{\scriptsize $\cdots$} (9.5,0) node{\scriptsize $\cdots$} (6.4,0) node{\scriptsize $\cdots$} (-0.9,0) node{\scriptsize $\cdots$};
       \node (a) at (-0.2,-0.2) {};
       \node (b) at (3,-0.2) {};
       \draw [decorate,decoration={brace,mirror}] (a)--(b);
        \node (c) at (2.8,-0.2) {};
       \node (d) at (6,-0.2) {};
       \draw [decorate,decoration={brace,mirror}] (c)--(d);
       \node (e) at (6.8,-0.2) {};
       \node (f) at (10,-0.2) {};
       \draw [decorate,decoration={brace,mirror}] (e)--(f);
       \node (g) at (1.8,0.2) {};
       \node (h) at (3,0.2) {};
       \draw [decorate,decoration={brace}] (g)--(h);
       \node (g) at (4.8,0.2) {};
       \node (h) at (6,0.2) {};
       \draw [decorate,decoration={brace}] (g)--(h);
       \node (g) at (8.8,0.2) {};
       \node (h) at (10,0.2) {};
       \draw [decorate,decoration={brace}] (g)--(h);
        \node (g) at (-0.2,0.2) {};
       \node (h) at (2,0.2) {};
       \draw [decorate,decoration={brace}] (g)--(h);
        \node (g) at (2.8,0.2) {};
       \node (h) at (5,0.2) {};
       \draw [decorate,decoration={brace}] (g)--(h);
        \node (g) at (6.8,0.2) {};
       \node (h) at (9,0.2) {};
       \draw [decorate,decoration={brace}] (g)--(h);
       \node (a) at (-1.6,-0.2) {};
       \node (b) at (0,-0.2) {};
       \draw [decorate,decoration={brace,mirror}] (a)--(b);
       \node (a) at (-1.6,0.2) {};
       \node (b) at (0,0.2) {};
       \draw [decorate,decoration={brace}] (a)--(b);
       \draw (1.4,-0.6) node{\scriptsize $\textsf{R}_{1}$} (4.4,-0.6) node{\scriptsize $\textsf{R}_{2}$} (8.4,-0.6) node{\scriptsize $\textsf{R}_{r}$} (-0.8,-0.6) node{\scriptsize $\textsf{R}_{\rm inp}$};
       \draw (0.9,0.6) node {\scriptsize $2(n-p)$} (3.9,0.6) node {\scriptsize $2(n-p)$} (7.9,0.6) node {\scriptsize 
       $2(n-p)$};
       \draw (2.4,0.6) node{\color{red} \scriptsize $\tau$} (5.4,0.6) node{\color{red} \scriptsize $\tau$} (9.4,0.6) node{\color{red}\scriptsize $\tau$} (-0.8,0.6) node{\color{blue}\scriptsize $n$};
       \draw (6.4,-0.6) node{\scriptsize $\cdots$};
    \end{tikzpicture}
    \caption{${\sf R}_{\rm inp}$, ${\sf R}_{\rm copy}$, ${\sf R}_{\rm targ}$ ${\sf R}_{\rm aux}$ for quantum circuits under $\Path_{n+m}$ constraint. Colors correspond to input (blue), copy (black), target (white) and auxiliary (red) register qubits. The ancillary qubits are grouped into registers labelled $\textsf{R}_1,\textsf{R}_2,\cdots,\textsf{R}_r$. }
    \label{fig:register_in_path}
\end{figure}

It is easily verified that our construction uses $2\cdot 2^p+r\tau\le m $ of the total $m$ ancillary qubits available. For the integers specifying the Gray codes, we take $\ell_k=(k-1)\mod (n-p)+1$ for all $k\in[2^p]$.

\paragraph{Implementation of Suffix Copy and Prefix Copy stages}
\begin{lemma}\label{lem:copy_path}
The unitary transformation $U_{copy}^{path}$ making $t$ copies of an $n$-bit string $x$, defined by
\begin{equation}
\ket{x}\ket{0^{nt}}\xrightarrow{U^{path}_{copy}}\ket{x}\ket{\underbrace{xx\cdots xx}_{t {\rm~copies~of ~}x}},
\end{equation}
where the two registers are connected in the path graph, can be implemented by a circuit of depth $O(n^2+nt)$ and size $O(n^2t)$ under $\Path_{n(t+1)}$ constraint.
\end{lemma}
\begin{proof}
An explicit circuit, in the absence of any connectivity constraints, is given in Fig.~\ref{fig:my_copypath}, which consists of $t$ CNOT circuits arranged in a pipeline. The total circuit depth is $n + t-1$ and the size is $nt$. 
Now we consider the path constraint. In each layer, by Lemma~\ref{lem:cnot_path_constraint}, each CNOT gate in Fig.~\ref{fig:my_copypath} can be implemented in depth and size $O(n)$ under path constraint, since the distance between any pair of control and target qubits is $O(n)$. Also note that different CNOT gates in the same layer are on disjoint regions of the path graph, and can thus be implemented in parallel. The result follows.
\end{proof}

\begin{figure}[htb!]
    \centerline{
    \Qcircuit @C=0.6em @R=0.7em {
   \lstick{\scriptstyle\ket{x_1}} &\qw &\qw & \ctrl{3} &\qw & \qw &\\
   \lstick{\scriptstyle\vdots~~} & \qw & \ctrl{3} &\qw &\qw &\qw &\\
   \lstick{\scriptstyle\ket{x_n}} &\ctrl{3} &\qw&\qw &\qw &\qw &\\
    \lstick{\scriptstyle\ket{0}}& \qw & \qw &\targ & \ctrl{3} & \qw &\\
   \lstick{\scriptstyle\vdots~~} & \qw &\targ & \ctrl{3} &\qw &\qw &\\
    \lstick{\scriptstyle\ket{0}}&\targ&\ctrl{3} & \qw &\qw & \qw &\\
   \lstick{\scriptstyle\ket{0}} & \qw & \qw &\qw &\targ &\qw &\\
   \lstick{\scriptstyle\vdots~~} & \qw &\qw & \targ &\qw &\qw &\\
   \lstick{\scriptstyle\ket{0}} &\qw&\targ & \qw &\qw &\qw &\\
    & &\vdots & \vdots&\vdots &\vdots\\
    & & & & &\\
   \lstick{\scriptstyle\ket{0}} & \qw & \qw &\qw &\qw & \ctrl{3} &\\
   \lstick{\scriptstyle\vdots~~} & \qw &\qw & \qw &\ctrl{3} & \qw &\\
    \lstick{\scriptstyle\ket{0}}&\qw&\qw & \ctrl{3} &\qw & \qw & \\
   \lstick{\scriptstyle\ket{0}} & \qw & \qw &\qw &\qw & \targ&\\
   \lstick{\scriptstyle\vdots~~} & \qw &\qw & \qw &\targ & \qw &\\
   \lstick{\scriptstyle\ket{0}} &\qw&\qw & \targ &\qw & \qw & \\
    }
    }
    \caption{Implementation of $U_{copy}^{path}$ (Lemma \ref{lem:copy_path}) to create $t$ copies of $\ket{x_1 x_2, \ldots x_n}$. Under path constraint, each CNOT gate can be implemented in depth and size $O(n)$. }
    \label{fig:my_copypath}
\end{figure}
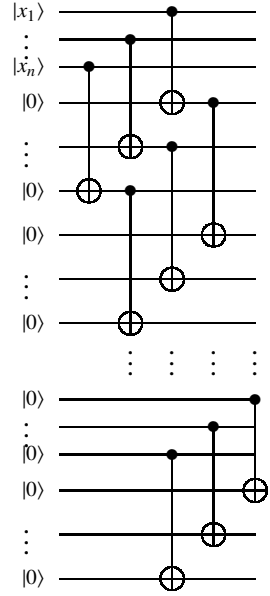
\begin{lemma}[]\label{lem:sufcopy_path}
$U_{SufCopy}$ and $U''_{PreCopy}$ can each be implemented by a quantum circuit of depth $O(m)$ under $\Path_{n+m}$ constraint.
\end{lemma}
\begin{proof}
$U_{SufCopy}$ creates $\gamma:=\big\lceil\lambda_{targ}/p\big\rceil$ copies of the $p$-qubit state $\ket{x_{suf}}$ in the copy register, i.e., 
\begin{equation*}
    \ket{x_{suf}}\ket{0^{p\gamma}}_{{\sf R}_{\rm copy}}\xrightarrow{U_{SufCopy}} \ket{x_{suf}}\ket{x_{suf}}^{\otimes \gamma}_{{\sf R}_{\rm copy}}.
\end{equation*}
As the $p$ qubits that comprise $\ket{x_{suf}}$ are located in a contiguous block in ${\sf R}_{\rm inp}$ that borders ${\sf R}_{\rm 1}$ (see Fig.~\ref{fig:register_in_path}), if the qubits in ${\sf R}_{\rm copy}$ were also located in a contiguous block bordering ${\sf R}_{\rm inp}$ then, by Lemma~\ref{lem:copy_path}, $U_{SufCopy}$ could be implemented by a CNOT circuit of depth $O(p^2+p\gamma)$, where each CNOT gate acts only on nearest neighbours in the path.  However, in each layer of the circuit in Fig.~\ref{fig:my_copypath}, each CNOT gate has its control and target qubits separated by either (i) $2p$ black and white qubits (in Fig.~\ref{fig:register_in_path}), or (ii) $2p+\tau$ black, white and red qubits. These two cases need depth $O(p)$ and $O(p+\tau)$, respectively. Putting the $p+\gamma-1$ layers in Fig.~\ref{fig:my_copypath} together,  
The total depth required to implement $U_{SufCopy}$ is thus $(p+\gamma-1)\cdot O(p+\tau) = O(m)$. 

The proof for $U_{PreCopy}''$ is similar, except in this case $r$ copies of the $(n-p)$-qubit state $\ket{x_{pre}}$ are made in the copy register.
\end{proof}

\begin{lemma}[]\label{lem:precopy3_path}
$U_{PreCopy}'''$ (Eq.~\eqref{eq:precopy3_graph}) can be implemented by a quantum circuit of depth  $O(n-p)$ under $\Path_{n+m}$ constraint.
\end{lemma}
\begin{proof}
$U_{PreCopy}'''$ makes $r=2^p/(n-p)$ 
copies of $\tau$-qubit state $\ket{x_{aux}}$ in ${\sf R}_{\rm aux}$. From Fig. \ref{fig:register_in_path}, the ancillary qubits are grouped into registers ${\sf R}_1,\ldots,{\sf R}_r$. Each ${\sf R}_{i}$ contains copy, target and auxiliary register qubits and we already have a copy of $\ket{x_{aux}}$ in the black qubits inside ${\sf R}_i$. Thus, within each ${\sf R}_{i}$ we can make a copy of $\ket{x_{aux}}$ from the the black qubits 
to red qubits.
This can be implemented in depth $O(\tau) + O(n-p)=O(n-p)$ for each ${\sf R}_i$, by Lemma \ref{lem:cnot_path_constraint} and noting that the $\tau$ qubits can be copied in a pipeline. Since paths in distinct ${\sf R}_i$ are disjoint, the $r$ copies of $\ket{x_{aux}}$ can be implemented in parallel. The result follows.
\end{proof}
\begin{lemma}[]\label{lem:precopy_path}
$U_{PreCopy}$ can be implemented by a quantum circuit of depth $O(m)$ under $\Path_{n+m}$ constraint.
\end{lemma}
\begin{proof}
Follows from Lemmas \ref{lem:precopy_graph}, \ref{lem:sufcopy_path} and \ref{lem:precopy3_path}.  
\end{proof}

\paragraph{Implementation of Gray Initial and Gray Cycle stages}


\begin{lemma}[] \label{lem:grayinitial_path}
$U_{GrayInit}$ (Eq. \eqref{eq:gray_initial_graph}) can be implemented by a CNOT circuit of depth $O(p^2)$ under $\Path_{n+m}$ constraint.
\end{lemma}

\begin{proof}
%
Recall that the Gray Initial stage aims to generate state $\ket{\langle t_1, x_{suf}\rangle}\otimes \cdots \otimes \ket{\langle t_{2^p}, x_{suf}\rangle}$ in ${\sf R}_{{\rm targ}}$.
{We consider 2 cases: }

{Case 1: $n-p\ge p$.}   Consider the first block of size $2p$ qubits in ${\sf R}_1$ in Fig.~\ref{fig:register_in_path}: The $p$ black qubits contain exactly $x_{suf}$ and the $p$ white qubits are to hold the state $\ket{\langle t_1, x_{suf}\rangle}\otimes \cdots \otimes \ket{\langle t_{p}, x_{suf}\rangle}$. By Lemma~\ref{lem:cnot_circuit}, this can be implemented by a $(2p)$-qubit CNOT circuit of depth and size $O(p^2)$ under $\Path_{2p}$ constraint. At the same time, we can also generate the state $\ket{\langle t_{p+1}, x_{suf}\rangle}\otimes \cdots \otimes \ket{\langle t_{2p}, x_{suf}\rangle}$ in the second block of $2p$ qubits in ${\sf R}_1$, and similarly for all the rest blocks of $2p$ qubits in all ${\sf R}_i$'s. All these blocks are on connected and disjoint regions on the path graph and can thus be implemented in parallel. The only possible exceptions are the end of each ${\sf R}_i$, where the leftover qubits may not form a complete $(2p)$-qubit block. But for each of these ``incomplete blocks'', there is still an $x_{suf}$ that is only $2p$-distance away, and thus these incomplete blocks can be handled in depth $O(p^2)$ as well. Putting them together, by handling all complete blocks first, and then handling all incomplete blocks afterwards, we achieve the desirable unitary with an overall depth of $2\cdot O(p^2) = O(p^2)$. 

{Case 2: $n-p<p$. We again divide the qubits into blocks, where each block has $p$ black qubits and $p$ white qubits. Different to Case 1, now there are red qubits in each block. But note from Fig. \ref{fig:register_in_path} that every $\tau = 2\lceil\log(n-p)\rceil$ red qubits appear after $2(n-p)$ black/white qubits, thus the total number of red qubits in each block is not more than that of black/white ones. Therefore, the length of each block is still $O(p)$ and any CNOT circuit on one block still has depth and size $O(p^2)$ (and different circuits on different blocks can be parallelized) as in the previous case. Thus the overall depth is $O(p^2)$ as claimed.} 
\end{proof}

The operator $U^{(k)}$, defined in the following lemma, is an important tool in the Gray cycle stage. In the lemma, the $\ket{x_i}$ and $\ket{x_j}$ are black qubits, the $\ket{y_i}$ and $\ket{y_j}$ are white qubits, and the $\ket{x_\ell}$ are red qubits. This lemma is where we use ${\sf R}_{\rm aux}$ to help the CNOT gates.  
\begin{lemma}\label{lem:U(k)}
Let $x,y\in\B^{n-p}$. For all $k\in[n-p]$, We desire a unitary transformation $U^{(k)}$ to satisfy
\begin{align*}
        &\bigotimes_{i=1}^{n-p-k+1}\ket{x_i}_{2i-1}\ket{y_i}_{2i}\bigotimes_{j=n-p-k+2}^{n-p}\ket{x_j}_{2j-1}\ket{y_j}_{2j}\bigotimes_{\ell=1}^{2\lceil\log (n-p)\rceil}\ket{x_\ell}_{2(n-p)+\ell}\\
        \xrightarrow{U^{(k)}}& \bigotimes_{i=1}^{n-p-k+1}\ket{x_i}_{2i-1}\ket{x_{i+k-1}\oplus y_i}_{2i}\bigotimes_{j=n-p-k+2}^{n-p}\ket{x_j}_{2j-1}\ket{x_{j-(n-p)+k-1}\oplus y_j}_{2j}\bigotimes_{\ell=1}^{2\lceil\log (n-p)\rceil}\ket{x_\ell}_{2(n-p)+\ell} &\forall x,y\in\B^{n-p}.
\end{align*}
Under $\Path_{2(n-p)+2\lceil\log (n-p)\rceil}$ constraint,  a $U^{(k)}$ can be implemented by a circuit of depth $O(k)$ and size $O(nk)$ if $k\in[2\lceil\log (n-p)\rceil+1]$; otherwise, $U^{(k)}$ can be implemented by a circuit of depth and size $O((n-p)k)$.
\end{lemma}
\begin{proof}
Case 1: $k\le 2\lceil\log(n-p)\rceil+1$. We use Lemma~\ref{lem:U(k)_without_ancilla} with parameters $r_t=n-p$, $r_c=n-p+2\lceil \log(n-p) \rceil$, $\tau=2\lceil \log(n-p) \rceil$, and  variables $x_{r_t+1},x_{r_t+2},\ldots,x_{r_c}$ there set as $x_{1},x_{2},\ldots,x_{2\lceil \log(n-p) \rceil}$ here. It is easily verified that $U^{(k)}$ in Lemma \ref{lem:U(k)_without_ancilla} satisfies the unitary requirement of this lemma. 

Case 2: $k\ge 2\lceil\log(n-p)\rceil+2$.  $U^{(k)}$ can be implemented in two parts: The first part adds $x_{i+k-1}$ to $y_i$, which are $O(k)$ apart, for each $i=1,2,\ldots,n-p-k+1$. The second part adds $x_{i-(n-p)+k-1}$ to $y_i$, which are $O(n-p-k)$ apart, for each $i=n-p-k+2,\cdots,n-p$. By Lemma \ref{lem:cnot_path_constraint}, this circuit can be implemented in depth and size
\[(n-p-k+1)\cdot O(k)+(n-p-(n-p-k+2)+1)\cdot O(n-p-k)=O((n-p)k).\]
\end{proof}
\begin{lemma}\label{lem:graycycle_path} $U_{GrayCycle}$ (Eq. \eqref{eq:gray_cycle_graph}) can be implemented by a quantum circuit of depth $O(2^{n-p})$ under $\Path_{n+m}$ constraint.
\end{lemma}

\begin{proof}
{Recall that each of the $2^p$ qubits in the target register corresponds to a suffix of $s$, and $U_{GrayCycle}$ enumerates all prefixes of $s$ in the order given in a Gray code---more precisely, qubit $k$ uses the $(n-p,\ell_k)$-Gray code where $\ell_{k} = (k-1)\bmod (n-p)+1$. $U_{GrayCycle}$ is given in Eqs. \eqref{eq:Ugenj_graph} to \eqref{eq:rotation2n-p_graph}, where the phase steps in Eqs. \eqref{eq:rotationj_graph} and \eqref{eq:rotation2n-p_graph} are straightforward, and let us} consider quantum circuits for $U_{Gen}^{(j)}$ for all $j\in[2^{n-p}]$ (Eqs. \eqref{eq:Ugenj_graph} and \eqref{eq:Ugen2n-p_graph}). Recall 
the decomposition of qubits in Fig.~\ref{fig:register_in_path} into $r$ registers ${\sf R_1}, {\sf R_2},\ldots, {\sf R_r}$. For every $\ell\in[r]$ and $i\in[n-p]$, if $k=i+(\ell-1)(n-p)$, then 
\[\ell_k = (k-1)\bmod (n-p)+1 = (i+(\ell-1)(n-p)-1)\bmod (n-p)+1 = i-1+1 = i.\]
For each register ${\sf R}_q$ where $q\in [r]$, since we have already copied the prefix by $U_{PreCopy}$, the white qubits are in state $\ket{x_1 x_2 \cdots x_{n-p}}$, and the red qubits are in state $\ket{x_1 x_2 \cdots x_{\tau}}$. For $j\in [2^{n-p}-1]$, before $U^{(j)}_{Gen}$, the black qubits are in state $\ket{f_{j,1+(q-1)(n-p)},f_{j,2+(q-1)(n-p)},\cdots, f_{j,q(n-p)}}$. Thus $U^{(j)}_{Gen}$ (Eq.\eqref{eq:Ugenj_graph}) can be represented as 
\begin{align*}
    &\ket{x_1,f_{j,1+(q-1)(n-p)},x_2,f_{j,2+(q-1)(n-p)},\cdots, x_{n-p},f_{j,q(n-p)}, x_1x_2\cdots x_{\tau}}_{{\sf R}_q}\\
    \to &\ket{x_1,f_{j+1,1+(q-1)(n-p)},x_2,f_{j+1,2+(q-1)(n-p)},\cdots, x_{n-p},f_{j+1,q(n-p)}, x_1x_2\cdots x_{\tau}}_{{\sf R}_q}\\
    = &\ket{x_1,f_{j,1+(q-1)(n-p)}\oplus x_{h_{1,j+1}},x_2,f_{j,2+(q-1)(n-p)}\oplus x_{h_{2,j+1}},\cdots, x_{n-p},f_{j,q(n-p)}\oplus x_{h_{n-p,j+1}}, x_1x_2\cdots x_{\tau}}_{{\sf R}_q},
\end{align*}
where $f_{j+1,k}=\langle c_{j+1}^{\ell_k}t_k,x\rangle=\langle c_{j}^{\ell_k}t_k,x\rangle\oplus x_{h_{\ell_k,j+1}}$ with $h_{ij}$ defined in Eq.~\eqref{eq:index}.

{Recall that $h_{1,j+1} = \zeta(j)$  (Eq.~\eqref{eq:index})}, and therefore
\begin{align*}
    h_{i,j+1} &= (\zeta(j)+i-2\mod (n-p))+1\\
    &=( h_{1,j+1}+i-2\mod (n-p))+1\\
    & = \begin{cases}
    h_{1,j+1}+i-1,  &\quad \text{if } 1\le  i\le n-p-{\zeta(j)}+1.\\
    h_{1,j+1}+i-1-(n-p), &\quad \text{if } n-p-\zeta(j)+2\le i\le n-p.
    \end{cases}
\end{align*}
Therefore, the integers $h_{1,j+1},h_{2,j+1},\ldots,h_{n-p,j+1}$ are equal to $h_{1,j+1}, h_{1,j+1}+1,\ldots, n-p, 1,2,\ldots, h_{1,j+1}-1$.

By Lemma~\ref{lem:U(k)} (with $k\leftarrow h_{1,j+1}$ and  $n\leftarrow n-p$),  the above transformation can be implemented by $U^{(h_{1,j+1})}$ acting on register ${\sf R}_q$, for every $q\in[r]$. Each $U^{(h_{1,j+1})}$ can be implemented in depth
\begin{equation*}
    \mathcal{D}(U^{(h_{1,j+1})})= \begin{cases} 
    O(h_{1,j+1}), &\quad \text{if } h_{1,j+1}\le \tau+1, \\
    O((n-p)h_{1,j+1}), & \quad  \text{otherwise.}
    \end{cases}
\end{equation*}
and, as paths in distinct ${\sf R_\ell}$ are disjoint, the $U^{(h_{1,j+1})}$ for all $r$ registers can be implemented in parallel.  

{In the final iteration, the unitary} $U_{Gen}^{(2^{n-p})}$ (Eq. \eqref{eq:Ugen2n-p_graph}) {moves from the last prefix back to the first one}, and it can be implemented in the same way {as $U^{(j)}_{Gen}$ for $j\le 2^{n-p}-1$}, with the same depth upper bound.

By Lemma~\ref{lem:GrayCode}, there are $2^{n-p-i}$ values of $j\in[2^{n-p}-1]$ such that $h_{1,j+1}=i$.
By Lemma~\ref{lem:graycycle_graph}, the circuit depth required to implement $U_{GrayCyle}$ is
\begin{equation}\label{eq:graycycle-depth}
\sum_{j=1}^{2^{n-p}}\mathcal{D}(U_{Gen}^{(j)})+2^{n-p}=\sum_{i=1}^{\tau+1}O(i)\cdot O(2^{n-p-i})+\sum_{i=\tau+2}^{n-p}O((n-p)i)O(2^{n-p-i})+O((n-p)^2)+2^{n-p}=O(2^{n-p}),
\end{equation}
where $\tau=2\lceil\log (n-p)\rceil$.
\end{proof}
\paragraph{Remark} {As in the proof of Lemma \ref{lem:Ck_path}, here the first term in Eq. \eqref{eq:graycycle-depth} consists of highly numerous CNOT operations, for which we use ${\sf R}_{\rm aux}$ to help to shrink the distance and cost. The second term in Eq. \eqref{eq:graycycle-depth} incurs more cost per operation but the number of operations is small. In general, the number of operations exponentially decays with the distance, thus we choose the cutoff point $\tau = 2\lceil\log(n-p)\rceil$ to make the overall cost small.}

\paragraph{Implementation of Inverse Stage} 
\begin{lemma}\label{lem:inverse_path}
$U_{Inverse}$
can be implemented by a CNOT circuit of depth $O(m)$ under $\Path_{n+m}$ constraint.
\end{lemma}
\begin{proof}
Follows from Lemmas \ref{lem:inverse_graph}, \ref{lem:sufcopy_path},  \ref{lem:precopy3_path} and \ref{lem:grayinitial_path}.
\end{proof}

\paragraph{Implementation of $\Lambda_n$}
\begin{lemma}[Lemma \ref{lem:diag_d_grid_ancilla} (Case 1)]
\label{lem:diag_path_ancillary}
Any $n$-qubit diagonal unitary matrix $\Lambda_n$ can be realized by a  quantum circuit of depth {$O\left(2^{n/2}+\frac{2^n}{m}\right)$}
and size $O(2^n)$ under $\Path_{n+m}$ constraint, using {$m \ge 3n$}
ancillary qubits. In particular, we can achieve circuit depth $O(2^{n/2})$ by using $m=\Theta(2^{n/2})$ ancillary qubits.
\end{lemma}
\begin{proof}
If $m\le 3\cdot 2^{n/2}$, the total depth of $\Lambda_n$ is
\[O(m)+O(m)+O(2^{n-p})+O(p^2)+O(m) = O(m+2^n/m) =O(2^n/m),\]
from Lemmas \ref{lem:sufcopy_path}, \ref{lem:graycycle_path},  \ref{lem:precopy_path} and \ref{lem:inverse_path}. Since there are at most $n+m$ gates in each circuit depth, the total size is $O(2^n/m)\cdot (n+m)=O(2^n)$.
If $m\ge 3\cdot 2^{n/2}$, we only use $3\cdot 2^{n/2}$ ancillary qubits. In this case, the total depth and size are $O(2^{n/2})$ and $O(2^n)$. Putting the two cases together gives the claimed result.
\end{proof}
As we will see from Appendix \ref{append:QSP_US_lowerbound} (Corollary \ref{coro:lower_bound_path}), this bound is optimal.


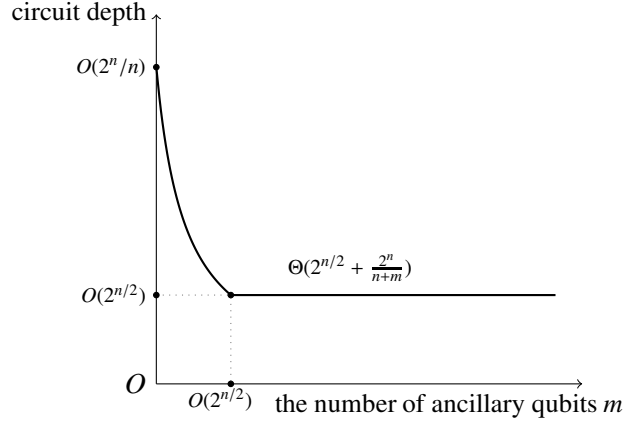
\begin{figure}
\centering
\begin{tikzpicture}[scale=0.7]
\draw[->] (0,0) -- (8,0);
\draw (5.5,0) node[anchor=north] {\small the number of ancillary qubits 
$m$};
		
\draw (0,0) node[anchor=east] {$O$};
\draw (0,6) node[anchor=east] {\scriptsize $O(2^{n}/n)$};
\draw (0,36/21.4) node[anchor=east] {\scriptsize $O(2^{n/2})$};
\draw (5,2.2) node[anchor=east] {\scriptsize $\Theta(2^{n/2}+\frac{2^n}{n+m})$};
\draw (2,-0.3) node[anchor=east]{\scriptsize $O(2^{n/2})$};
\draw[->] (0,0) -- (0,7) node[anchor=east] {\small circuit depth};

\draw[thick] plot[smooth, domain = 0:1.4] (\x, {(36/11)/(\x + (6/11))});
\draw[thick] (1.4,36/21.4) -- (7.5,36/21.4);

\draw[dotted, gray] (0,36/21.4) -- (1.4, 36/21.4) -- (1.4,0) ;
\draw (0,6) node[fill,black,draw=black,circle,scale = 0.2]{};
\draw (1.4, 36/21.4) node[fill,black,draw=black,circle,scale = 0.2]{};
\draw (0,36/21.4) node[fill,black,draw=black,circle,scale = 0.2]{};
\draw (1.4,0) node[fill,black,draw=black,circle,scale = 0.2]{};

\end{tikzpicture}
\caption{Circuit depth for $n$-qubit diagonal unitary matrix $\Lambda_n$ under $\Path_{n+m}$ constraint (Lemmas \ref{lem:diag_path_withoutancilla} and \ref{lem:diag_path_ancillary}).} 
\label{fig:depth_diag_path}
\end{figure}

\subsection{Circuit implementation under $\Grid_{n+m}^{n_1,n_2,\ldots, n_d}$ constraints (Proof of Lemma \ref{lem:diag_d_grid_ancilla} (Cases 2,3))}
\label{sec:diag_with_ancilla_grid_d}
We realize the suffix copy and prefix copy stages under $\Grid_{n+m}^{n_1,n_2,\ldots,n_d}$ constraint, and recall that, without loss of generality, we assume that $n_1 \ge n_2 \ge \cdots\ge n_d$. The Gray initial and Gray cycle stages are implemented by circuits in Appendix \ref{sec:diag_with_ancilla_path} under the Hamiltonian path constraint in $\Grid_{n+m}^{n_1,n_2,\ldots,n_d}$. We take the $n$ input qubits to be arranged in the corner of a $d$-dimensional grid. They can be permuted to any other locations in the grid without increasing the order of the circuit depth required. 

We assume that $m\ge 36n $. If $m< 36n $, diagonal unitary matrices are implemented in the way of Appendix \ref{sec:diag_without_ancilla_path}. We take $p= \log (\frac{m}{18})$, $\tau=2\lceil\log (n-p)\rceil$, $\lambda_{copy}=\lambda_{targ}=2^p$, and  $\lambda_{aux}=r\tau$ where   $r=\frac{2^p}{n-p}$. For the integers specifying the Gray codes, we take $\ell_k=(k-1)\mod (n-p)+1$ for all $k\in[2^p]$.

\paragraph{Choice of registers} We assign qubits to ${\sf R}_{\rm inp}$, ${\sf R}_{\rm copy}$, ${\sf R}_{\rm targ}$ and ${\sf R}_{\rm aux}$ as follows (see Fig.~\ref{fig:register_in_grid}). We divide $\Grid_{n+m}^{n_1,n_2,\ldots, n_d}$ into two grids: $\Grid_{n_1\cdots n_{d-1} \lfloor n_d/2\rfloor}^{n_1,n_2,\ldots, n_{d-1},  \lfloor n_d/2\rfloor}$ and $\Grid_{n_1\cdots n_{d-1}\lceil n_d/2\rceil}^{n_1,n_2,\ldots, n_{d-1}, \lceil n_d/2\rceil}$. We can verify that the sizes of these two grids are at least $m/3~(\ge 12n) $ and $m/2~(\ge 18 n)$ respectively. 
The input register is in $\Grid_{n_1\cdots n_{d-1} \lfloor n_d/2\rfloor}^{n_1,n_2,\ldots, n_{d-1}, \lfloor n_d/2\rfloor}$ and qubits in $\Grid_{n_1\cdots n_{d-1}\lceil n_d/2\rceil}^{n_1,n_2,\ldots, n_{d-1}, \lceil n_d/2\rceil}$ are utilized as ancillary qubits. 
\begin{itemize}
    \item We choose the lowest possible dimensional grid to store the input register. More specifically, let $k$ be the minimum integer satisfying $n_1\cdots n_k \ge n$, and $n_k'$ be the minimum integer satisfying $n_1\cdots n_{k-1} n_k' \ge n$. (When $k=1$, $n_1n_2\cdots n_{k-1}$ is defined to 1).
    ${\sf R}_{\rm inp}$ consists of the first $n$ qubits of sub-grid $\Grid^{n_1,n_2,\ldots,n_{k-1},n'_k,1,1,\ldots,1}_{n_1n_2\cdots n_{k-1}n'_k}$ in $\Grid_{n_1\cdots n_{d-1} \lfloor n_d/2\rfloor}^{n_1,n_2,\ldots, n_{d-1}, \lfloor n_d/2\rfloor}$. 
    
        
    
    \item We choose $2\cdot 2^p+r\tau$ ancillary qubits from $\Grid_{n_1\cdots n_{d-1}\lceil n_{d}/2 \rceil}^{n_1,\cdots, n_{d-1},\lceil n_d/2\rceil}$, 
    and utilize them to construct $r$ registers ${\sf R}_1,{\sf R}_2,\cdots, {\sf R}_r$. The size of each ${\sf R}_k$ is $2(n-p)+\tau$. 
    
    Now we construct register ${\sf R}_k$.  Let $j$ be the minimum integer satisfying $n_1\cdots n_j \ge 2(n-p)+\tau$, and $n_j'$ be the minimum integer satisfying $n_1\cdots n_{j-1} n_j' \ge 2(n-p)+\tau$. We divide $\Grid_{n_1\cdots n_{d-1}(n_{d}-1)}^{n_1,\cdots, n_{d-1},(n_{d}-1)}$ into sub-grids isomorphic to $\Grid^{n_1,n_2,\ldots,n_{j-1},n'_{j},1,1\ldots,1}_{n_1n_2\cdots n_{j-1} n'_{j}}$. (When $j=1$, $n_1\cdots n_{j-1}$ is defined to 1). Each sub-grid stores exactly one register ${\sf R}_k$, and note that ${\sf R}_k$ occupies at least half of this grid, so the number of qubits in the sub-grid not in ${\sf R}_k$ is at most $2(n-p)+\tau$, so at most $r (2(n-p)+\tau)\le \frac{m}{6}$ qubits are wasted (not used). We again choose a Hamiltonian path $P$ in each sub-grid and assign qubits to ${\sf R}_{\rm copy}$, ${\sf R}_{\rm targ}$ and ${\sf R}_{\rm copy}$ and ${\sf R}_{\rm targ}$ registers in the same way as our assignment for Path in Fig.~\ref{fig:register_in_path}. We choose the same Hamiltonian path for all ${\sf R}_k$ (i.e., same for each sub-grid). 

\end{itemize}


An example showing the registers for $\Grid_{n+m}^{n_1,n_2}$ ($n_1\ge n$ and $n_2\ge 2$) is shown in Fig.~\ref{fig:register_in_grid}.

Next we analyze the cost of prefix and suffix copy. Recall that the prefix consists of $n-p$ bits and the suffix has length $p$ bits. In the following lemma, we use a uniform parameter $n'$ to represent either of these two cases.
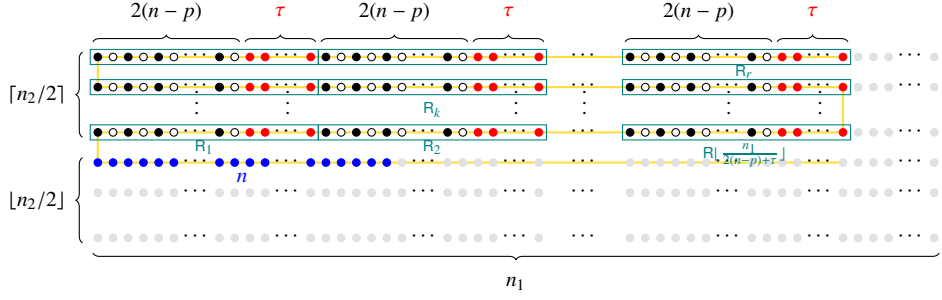
\begin{figure}[]
    \centering
    \begin{tikzpicture}
    \draw [draw=Goldenrod,thick] (9.8,-1.4)--(0,-1.4)--(0,-1)--(9.8,-1)--(9.8,-0.4)--(0,-0.4)--(0,0)--(9.8,0);
       \draw [fill=black] (3,0) circle (0.05) (3.4,0) circle (0.05) (3.8,0) circle (0.05) (4.6,0) circle (0.05) (7,0) circle (0.05) (7.4,0) circle (0.05) (7.8,0) circle (0.05) (8.6,0) circle (0.05) ;
       \draw [fill=white] (0.2,0) circle (0.05) (0.6,0) circle (0.05) (1.0,0) circle (0.05) (1.8,0) circle (0.05) (3.2,0) circle (0.05) (3.6,0) circle (0.05) (4.0,0) circle (0.05) (4.8,0) circle (0.05) (7.2,0) circle (0.05) (7.6,0) circle (0.05) (8.0,0) circle (0.05) (8.8,0) circle (0.05);
       \draw [draw=red, fill=red] (2,0) circle (0.05) (2.2,0) circle (0.05) (2.8,0) circle (0.05) (5,0) circle (0.05) (5.2,0) circle (0.05) (5.8,0) circle (0.05) (9,0) circle (0.05) (9.2,0) circle (0.05) (9.8,0) circle (0.05);
       \draw (1.3,0) node{\scriptsize $\cdots$} (2.5,0) node{\scriptsize $\cdots$} (4.3,0) node{\scriptsize $\cdots$} (5.5,0) node{\scriptsize $\cdots$} (8.3,0) node{\scriptsize $\cdots$} (9.5,0) node{\scriptsize $\cdots$} (6.4,0) node{\scriptsize $\cdots$};
       \draw [fill=black] (0,-0.4) circle (0.05) (0.4,-0.4) circle (0.05) (0.8,-0.4) circle (0.05) (1.6,-0.4) circle (0.05) (3,-0.4) circle (0.05) (3.4,-0.4) circle (0.05) (3.8,-0.4) circle (0.05) (4.6,-0.4) circle (0.05) (7,-0.4) circle (0.05) (7.4,-0.4) circle (0.05) (7.8,-0.4) circle (0.05) (8.6,-0.4) circle (0.05);
       \draw [fill=white] (0.2,-0.4) circle (0.05) (0.6,-0.4) circle (0.05) (1.0,-0.4) circle (0.05) (1.8,-0.4) circle (0.05) (3.2,-0.4) circle (0.05) (3.6,-0.4) circle (0.05) (4.0,-0.4) circle (0.05) (4.8,-0.4) circle (0.05) (7.2,-0.4) circle (0.05) (7.6,-0.4) circle (0.05) (8.0,-0.4) circle (0.05) (8.8,-0.4) circle (0.05);
       \draw [draw=red, fill=red] (2,-0.4) circle (0.05) (2.2,-0.4) circle (0.05) (2.8,-0.4) circle (0.05) (5,-0.4) circle (0.05) (5.2,-0.4) circle (0.05) (5.8,-0.4) circle (0.05) (9,-0.4) circle (0.05) (9.2,-0.4) circle (0.05) (9.8,-0.4) circle (0.05);
      \draw (1.3,-0.4) node{\scriptsize $\cdots$} (2.5,-0.4) node{\scriptsize $\cdots$} (4.3,-0.4) node{\scriptsize $\cdots$} (5.5,-0.4) node{\scriptsize $\cdots$} (8.3,-0.4) node{\scriptsize $\cdots$} (9.5,-0.4) node{\scriptsize $\cdots$} (6.4,-0.4) node{\scriptsize $\cdots$};
      
      \draw [fill=black] (0,0) circle (0.05) (0.4,0) circle (0.05) (0.8,0) circle (0.05) (1.6,0) circle (0.05) ;
      
       \draw [fill=black] (0,-1) circle (0.05) (0.4,-1) circle (0.05) (0.8,-1) circle (0.05) (1.6,-0.4) circle (0.05) (3,-1) circle (0.05) (3.4,-1) circle (0.05) (3.8,-1) circle (0.05) (4.6,-1) circle (0.05) (7,-1) circle (0.05) (7.4,-1) circle (0.05) (7.8,-1) circle (0.05) (8.6,-1) circle (0.05)(1.6,-1) circle (0.05) ;
       \draw [fill=white] (0.2,-1) circle (0.05) (0.6,-1) circle (0.05) (1.0,-1) circle (0.05) (1.8,-1) circle (0.05) (3.2,-1) circle (0.05) (3.6,-1) circle (0.05) (4.0,-1) circle (0.05) (4.8,-1) circle (0.05) (7.2,-1) circle (0.05) (7.6,-1) circle (0.05) (8.0,-1) circle (0.05) (8.8,-1) circle (0.05);
       \draw [draw=red, fill=red] (2,-1) circle (0.05) (2.2,-1) circle (0.05) (2.8,-1) circle (0.05) (5,-1) circle (0.05) (5.2,-1) circle (0.05) (5.8,-1) circle (0.05) (9,-1) circle (0.05) (9.2,-1) circle (0.05) (9.8,-1) circle (0.05);
      \draw (1.3,-1) node{\scriptsize $\cdots$} (2.5,-1) node{\scriptsize $\cdots$} (4.3,-1) node{\scriptsize $\cdots$} (5.5,-1) node{\scriptsize $\cdots$} (8.3,-1) node{\scriptsize $\cdots$} (9.5,-1) node{\scriptsize $\cdots$} (6.4,-1) node{\scriptsize $\cdots$};

       \draw [fill=blue,draw=blue] (0,-1.4) circle (0.05) (0.4,-1.4) circle (0.05) (0.8,-1.4) circle (0.05) (1.6,-1.4) circle (0.05) (3,-1.4) circle (0.05) (3.4,-1.4) circle (0.05)    (0.2,-1.4) circle (0.05) (0.6,-1.4) circle (0.05) (1.0,-1.4) circle (0.05) (1.8,-1.4) circle (0.05) (2,-1.4) circle (0.05) (2.2,-1.4) circle (0.05) (2.8,-1.4) circle (0.05) (3.8,-1.4) circle (0.05) (3.2,-1.4) circle (0.05) (3.6,-1.4) circle (0.05);
       \draw [fill=lgray,draw=lgray](8.6,-1.4) circle (0.05) (7.8,-1.4) circle (0.05)  (7,-1.4) circle (0.05) (7.4,-1.4) circle (0.05) (4.0,-1.4) circle (0.05) (4.8,-1.4) circle (0.05) (7.2,-1.4) circle (0.05) (7.6,-1.4) circle (0.05) (8.0,-1.4) circle (0.05) (8.8,-1.4) circle (0.05) (4.6,-1.4) circle (0.05);
       \draw [fill=lgray,draw=lgray] (5,-1.4) circle (0.05) (5.2,-1.4) circle (0.05) (5.8,-1.4) circle (0.05) (9,-1.4) circle (0.05) (9.2,-1.4) circle (0.05) (9.8,-1.4) circle (0.05);
      \draw (1.3,-1.4) node{\scriptsize $\cdots$} (2.5,-1.4) node{\scriptsize $\cdots$} (4.3,-1.4) node{\scriptsize $\cdots$} (5.5,-1.4) node{\scriptsize $\cdots$} (8.3,-1.4) node{\scriptsize $\cdots$} (9.5,-1.4) node{\scriptsize $\cdots$} (6.4,-1.4) node{\scriptsize $\cdots$};
      \draw[fill=lgray,draw=lgray] (10,0) circle (0.05) (10.2,0) circle (0.05) (10.4,0) circle (0.05) (11,0) circle (0.05);
      \draw (10.7,0) node{\scriptsize $\cdots$};
      \draw[fill=lgray,draw=lgray] (10,-0.4) circle (0.05) (10.2,-0.4) circle (0.05) (10.4,-0.4) circle (0.05) (11,-0.4) circle (0.05);
      \draw (10.7,-0.4) node{\scriptsize $\cdots$};
      \draw[fill=lgray,draw=lgray] (10,-1) circle (0.05) (10.2,-1) circle (0.05) (10.4,-1) circle (0.05) (11,-1) circle (0.05);
      \draw (10.7,-1) node{\scriptsize $\cdots$};
      \draw[fill=lgray,draw=lgray] (10,-1.4) circle (0.05) (10.2,-1.4) circle (0.05) (10.4,-1.4) circle (0.05) (11,-1.4) circle (0.05);
      \draw (10.7,-1.4) node{\scriptsize $\cdots$};
      
      \draw[fill=lgray,draw=lgray] (10,-1.8) circle (0.05) (10.2,-1.8) circle (0.05) (10.4,-1.8) circle (0.05) (11,-1.8) circle (0.05);
      \draw (10.7,-1.8) node{\scriptsize $\cdots$};
      \draw [fill=lgray,draw=lgray] (0,-1.8) circle (0.05) (0.4,-1.8) circle (0.05) (0.8,-1.8) circle (0.05) (1.6,-1.8) circle (0.05) (3,-1.8) circle (0.05) (3.4,-1.8) circle (0.05)    (0.2,-1.8) circle (0.05) (0.6,-1.8) circle (0.05) (1.0,-1.8) circle (0.05) (1.8,-1.8) circle (0.05) (2,-1.8) circle (0.05) (2.2,-1.8) circle (0.05) (2.8,-1.8) circle (0.05) (3.8,-1.8) circle (0.05) (3.2,-1.8) circle (0.05) (3.6,-1.8) circle (0.05);
       \draw [fill=lgray,draw=lgray](8.6,-1.8) circle (0.05) (7.8,-1.8) circle (0.05)  (7,-1.8) circle (0.05) (7.4,-1.8) circle (0.05) (4.0,-1.8) circle (0.05) (4.8,-1.8) circle (0.05) (7.2,-1.8) circle (0.05) (7.6,-1.8) circle (0.05) (8.0,-1.8) circle (0.05) (8.8,-1.8) circle (0.05) (4.6,-1.8) circle (0.05);
       \draw [fill=lgray,draw=lgray] (5,-1.8) circle (0.05) (5.2,-1.8) circle (0.05) (5.8,-1.8) circle (0.05) (9,-1.8) circle (0.05) (9.2,-1.8) circle (0.05) (9.8,-1.8) circle (0.05);
       
       \draw (1.3,-1.8) node{\scriptsize $\cdots$} (2.5,-1.8) node{\scriptsize $\cdots$} (4.3,-1.8) node{\scriptsize $\cdots$} (5.5,-1.8) node{\scriptsize $\cdots$} (8.3,-1.8) node{\scriptsize $\cdots$} (9.5,-1.8) node{\scriptsize $\cdots$} (6.4,-1.8) node{\scriptsize $\cdots$};
    \draw[fill=lgray,draw=lgray] (10,-2.4) circle (0.05) (10.2,-2.4) circle (0.05) (10.4,-2.4) circle (0.05) (11,-2.4) circle (0.05);
      \draw (10.7,-2.4) node{\scriptsize $\cdots$};
      \draw [fill=lgray,draw=lgray] (0,-2.4) circle (0.05) (0.4,-2.4) circle (0.05) (0.8,-2.4) circle (0.05) (1.6,-2.4) circle (0.05) (3,-2.4) circle (0.05) (3.4,-2.4) circle (0.05)    (0.2,-2.4) circle (0.05) (0.6,-2.4) circle (0.05) (1.0,-2.4) circle (0.05) (1.8,-2.4) circle (0.05) (2,-2.4) circle (0.05) (2.2,-2.4) circle (0.05) (2.8,-2.4) circle (0.05) (3.8,-2.4) circle (0.05) (3.2,-2.4) circle (0.05) (3.6,-2.4) circle (0.05);
       \draw [fill=lgray,draw=lgray](8.6,-2.4) circle (0.05) (7.8,-2.4) circle (0.05)  (7,-2.4) circle (0.05) (7.4,-2.4) circle (0.05) (4.0,-2.4) circle (0.05) (4.8,-2.4) circle (0.05) (7.2,-2.4) circle (0.05) (7.6,-2.4) circle (0.05) (8.0,-2.4) circle (0.05) (8.8,-2.4) circle (0.05) (4.6,-2.4) circle (0.05);
       \draw [fill=lgray,draw=lgray] (5,-2.4) circle (0.05) (5.2,-2.4) circle (0.05) (5.8,-2.4) circle (0.05) (9,-2.4) circle (0.05) (9.2,-2.4) circle (0.05) (9.8,-2.4) circle (0.05);
       
       \draw (1.3,-2.4) node{\scriptsize $\cdots$} (2.5,-2.4) node{\scriptsize $\cdots$} (4.3,-2.4) node{\scriptsize $\cdots$} (5.5,-2.4) node{\scriptsize $\cdots$} (8.3,-2.4) node{\scriptsize $\cdots$} (9.5,-2.4) node{\scriptsize $\cdots$} (6.4,-2.4) node{\scriptsize $\cdots$};
      
      \draw (1.3,-0.5) node{\scriptsize $\vdots$} (2.5,-0.5) node{\scriptsize $\vdots$}  (5.5,-0.5) node{\scriptsize $\vdots$} (8.3,-0.5) node{\scriptsize $\vdots$} (9.5,-0.5) node{\scriptsize $\vdots$} (6.4,-0.5) node{\scriptsize $\vdots$};
      
      \draw [draw=teal] (2.9,-0.5)--(5.9,-0.5)--(5.9,-0.3)--(2.9,-0.3)--cycle;
      \draw [draw=teal] (-0.1,-0.5)--(2.9,-0.5)--(2.9,-0.3)--(-0.1,-0.3)--cycle;
      \draw [draw=teal] (6.9,-0.5)--(9.9,-0.5)--(9.9,-0.3)--(6.9,-0.3)--cycle;
      \draw [draw=teal] (2.9,-0.1)--(5.9,-0.1)--(5.9,0.1)--(2.9,0.1)--cycle;
      \draw [draw=teal] (-0.1,-0.1)--(2.9,-0.1)--(2.9,0.1)--(-0.1,0.1)--cycle;
      \draw [draw=teal](6.9,-0.1)--(9.9,-0.1)--(9.9,0.1)--(6.9,0.1)--cycle;
      \draw [draw=teal] (2.9,-0.9)--(5.9,-0.9)--(5.9,-1.1)--(2.9,-1.1)--cycle;
      \draw [draw=teal] (-0.1,-0.9)--(2.9,-0.9)--(2.9,-1.1)--(-0.1,-1.1)--cycle;
      \draw [draw=teal](6.9,-0.9)--(9.9,-0.9)--(9.9,-1.1)--(6.9,-1.1)--cycle;
      \draw (4.4,-0.7) node{\tiny \color{teal} ${\sf R}_k$};
     \draw (1.4,-1.2) node{\tiny\color{teal} ${\sf R}_1$};
    \draw (4.4,-1.2) node{\tiny \color{teal} ${\sf R}_2$};
    \draw (8.5,-1.3) node{\tiny \color{teal} ${\sf R}{\lfloor\frac{n_1}{2(n-p)+\tau}\rfloor}$};
    \draw (8.5,-0.2) node{\tiny \color{teal} ${\sf R}_{r}$};
       \node (g) at (1.8,0.2) {};
       \node (h) at (3,0.2) {};
       \draw [decorate,decoration={brace}] (g)--(h);
       \node (g) at (4.8,0.2) {};
       \node (h) at (6,0.2) {};
       \draw [decorate,decoration={brace}] (g)--(h);
       \node (g) at (8.8,0.2) {};
       \node (h) at (10,0.2) {};
       \draw [decorate,decoration={brace}] (g)--(h);
        \node (g) at (-0.2,0.2) {};
       \node (h) at (2,0.2) {};
       \draw [decorate,decoration={brace}] (g)--(h);
        \node (g) at (2.8,0.2) {};
       \node (h) at (5,0.2) {};
       \draw [decorate,decoration={brace}] (g)--(h);
        \node (g) at (6.8,0.2) {};
       \node (h) at (9,0.2) {};
       \draw [decorate,decoration={brace}] (g)--(h);
       \node (g) at (-0.2,-1.2) {};
       \node (h) at (-0.2,0.2) {};
       \draw [decorate,decoration={brace}] (g)--(h);
       \node (g) at (-0.2,-2.6) {};
       \node (h) at (-0.2,-1.2) {};
       \draw [decorate,decoration={brace}] (g)--(h);
       \node (g) at (11.2,-2.6) {};
       \node (h) at (-0.2,-2.6) {};
       \draw [decorate,decoration={brace}] (g)--(h);
        \draw (5.5,-3) node{ \scriptsize $n_1$} (-0.8,-0.5) node{ \scriptsize $\lceil n_2/2\rceil$} (-0.8,-1.9) node{ \scriptsize $\lfloor n_2/2\rfloor$} (1.9,-1.6) node{ \color{blue} \scriptsize $n$};
       \draw (0.9,0.6) node {\scriptsize $2(n-p)$} (3.9,0.6) node {\scriptsize $2(n-p)$} (7.9,0.6) node {\scriptsize $2(n-p)$};
       \draw (2.4,0.6) node{\color{red} \scriptsize $\tau$} (5.4,0.6) node{\color{red} \scriptsize $\tau$} (9.4,0.6) node{\color{red}\scriptsize $\tau$};
    \end{tikzpicture}
    \caption{${\sf R}_{\rm inp}$,  ${\sf R}_{\rm copy}$, ${\sf R}_{\rm targ}$ and ${\sf R}_{\rm aux}$ for $\Grid_{n+m}^{n_1,n_2}$ constraint for $n_1\ge n$ and $n_2\ge 2$. Colors correspond to input (blue), copy (black), target (white) and auxiliary (red) register qubits. The grey qubits are not utilized in the circuit construction. 
    The ancillary qubits are grouped into registers labelled $\textsf{R}_1,\textsf{R}_2,\ldots,\textsf{R}_r$. }
    \label{fig:register_in_grid}
\end{figure}

\begin{lemma}\label{lem:copy_grid}
For any $n'\ge 1$, let $s\le d$ denote the minimum integer satisfying $n_1\cdots n_s\ge n'$, and $n_s'$ be the minimum integer satisfying $n_1\cdots n_{s-1}n'_s\ge n'$. For a general $y\in\{0,1\}^{n'}$, suppose  the state $\ket{y}$ is input in the first $n'$ qubits of sub-grid $\Grid^{n_1,n_2,\ldots,n_{s-1},n'_s,1,1,\ldots,1}_{n_1n_2\cdots n_{s-1}n'_s}$.
Then one can implement a unitary transformation  $U_{copy}^{grid_d}$ satisfying
\begin{equation}
\ket{y}\ket{0^{nt}}\xrightarrow{U^{grid_d}_{copy}}\ket{y}\underbrace{\ket{yy\cdots yy}}_{t:=O(\prod_{i=1}^dn_i/n') {\rm~copies~of ~}y},
\end{equation}
can be implemented by a circuit of depth $O((n')^2+\sum_{i=1}^d n_i)$ under $\Grid^{n_1,n_2,\ldots,n_d}_{n_1n_2\cdots n_d}$ constraint.
\end{lemma}
\begin{proof}
Label qubits in the grid by their integer coordinates $(i_1,i_2,\ldots,i_d)$ where $i_k\in[n_k]$, for all $k\in[d]$. 
    For $y\in \B^{n'}$, $\ket{y}$ stores in first $n'$ qubits of sub-grid $\Grid^{n_1,n_2,\ldots,n_{s-1},n'_s,1,1,\ldots,1}_{n_1n_2\cdots n_{s-1}n'_s}$. It can be verified that the size of $\Grid^{n_1,n_2,\ldots,n_{s-1},n'_s,1,1,\ldots,1}_{n_1n_2\cdots n_{s-1}n'_s}$ is less than $2n'$. For simplicity of presentation assume that $n_s$ is a multiple of $n'_s$.
    \begin{enumerate}
        \item   First, we make $\frac{n_s}{n_s'}-1$ copies of qubits of $\Grid^{n_1,n_2,\ldots,n_{s-1},n'_s,1,1,\ldots,1}_{n_1n_2\cdots n_{s-1}n'_s}$ in $\Grid^{n_1,n_2,\ldots,n_{s-1},n_s,1,1,\ldots,1}_{n_1n_2\cdots n_{s-1}n_s}$. 
        For every $(i_1,\cdots,i_{s-1})\in[n_1]\times [n_2]\times \ldots \times [n_{s-1}]$, define path $P_{(i_1,i_2,\cdots,i_{s-1})}$ of length $n_s$:
        \begin{align*}
        P_{(i_1,i_2,\cdots,i_{s-1})} =  \{(i_1,i_2,\ldots,i_{s-1},v_s,1,\ldots,1): v_s\in [n_s]\}.
        \end{align*}
        For every $(n_1,\ldots,n_{s-1})$, we make $\frac{n_s}{n_s'}-1$ copies of the first $n'_s$ qubit of $P_{(i_1,i_2,\cdots,i_{s-1})}$ in this path under path constraint .
         By Lemma~\ref{lem:copy_path}, this requires depth $O((n'_s)^2+n'_s(n_s/n'_s-1))=O((n'_s)^2+n_s)$. 
        \item Second, we make $n_{s+1}n_{s+2}\cdots n_d$ copies of qubits of $\Grid^{n_1,n_2,\ldots,n_{s-1},n_s,1,1,\ldots,1}_{n_1n_2\cdots n_s}$ in $\Grid^{n_1,n_2,\ldots,n_{d}}_{n_1\cdots n_d}$. This can be implemented in $d-s$ steps. For every $k\in[d-s]$, in the $k$-th step, we make $n_{s+k}-1$ copies of qubits of $\Grid_{n_1n_2\cdots n_{s+k-1}}^{n_1,n_2,\cdots,n_{s+k-1},1,\cdots,1}$ in  $\Grid_{n_1n_2\cdots n_{s+k},1,\cdots,1}^{n_1,n_2,\cdots,n_{s+k},1,\cdots,1}$. Similar to the discussion above, the $k$-th step requires $O(1^2+n_{s+k})=O(n_{s+k})$ depth. The total depth of these $d-k$ steps is $\sum_{k=1}^{d-s}O(n_{s+k})=O(\sum_{i=s+1}^dn_i)$.
      
    \end{enumerate}
    Then we have made $\frac{n_s}{n'_s}n_{s+1}n_{s+2}\cdots n_d=O((\prod_{i=1}^d n_i)/n')$ copies of $y$ in total, since $n'\le (\prod_{i=1}^{s-1}n_i)n'_s<2n'$. The total depth is 
      \[O \big((n'_s)^2 + n_s+\sum_{i=s+1}^d n_i\big)\le O \big((n')^2 + \sum_{i=1}^d n_i\big).\]
 

 
\end{proof}

\begin{lemma}[Lemma \ref{lem:diag_d_grid_ancilla} (Case 2, 3)]\label{lem:diag_grid_ancillary}
Any $n$-qubit diagonal unitary matrix $\Lambda_n$ can be realized by a  quantum circuit of depth \[O\Big(n^2+d2^{\frac{n}{d+1}}+\max_{k\in\{2,\ldots,d\}}\Big\{\frac{d2^{n/k}}{(\Pi_{i=k}^d n_i)^{1/k}}\Big\}+\frac{2^n}{n+m}\Big)\] 
under $\Grid^{n_1,n_2,\ldots,n_d}_{n+m}$ constraint, using $m \ge {36n}$ ancillary qubits. If $n_1=n_2=\cdots=n_d$, the circuit depth is $O\left(n^2+d2^{\frac{n}{d+1}}+\frac{2^n}{n+m}\right)$.
\end{lemma}
\begin{proof}

%

From Lemma~\ref{lem:copy_grid}, it follows from a similar argument to that in the proof of Lemma~\ref{lem:sufcopy_path} that $U_{SufCopy}$ and  $U_{PreCopy}''$ can be realized by circuits of depth $O(n^2+\sum_{j=1}^d n_j)$. By Lemma \ref{lem:precopy3_path}, $U_{PreCopy}'''$ can be realized in depth $O(n-p)$. By Lemmas \ref{lem:precopy_graph}, \ref{lem:grayinitial_path}, \ref{lem:graycycle_path} and \ref{lem:inverse_graph}, we can see that the depth for prefix copy, Gray initial, Gray cycle and inverse stages are $O\big(n^2+\sum_{j=1}^d n_j\big)$, $O(p^2)$, $O(2^{n-p})$, $O\big(n^2+\sum_{j=1}^d n_j\big)$ and  $O\big(n^2+\sum_{j=1}^d n_j\big)$ respectively. 

The total circuit depth for $\Lambda_{n}$ is thus $O(n^2+\sum_{j=1}^d n_j+\frac{2^n}{n+m})$ under $\Grid_{n+m}^{n_1,n_2,\ldots,n_d}$ constraint. This bound is good when all the $n_j$'s are of similar magnitude: Indeed, when $n_1 = n_2 = \cdots = n_d$, the bound becomes $O(n^2+ dm^{1/d}+\frac{2^n}{n+m})$. If $m \le O(2^{\frac{d}{d+1}n}/d)$, then the third term dominates and the bound is $O(2^n/m)$. If $m \ge \Omega(2^{\frac{d}{d+1}n}/d)$, we choose to only use $2^{\frac{d}{d+1}n}/d$ many ancilla, yielding a depth bound of $O(d2^{\frac{n}{d+1}})$. This completes the proof for the special case of $n_1 = n_2 = \cdots = n_d$. 

In the general case, where some $n_j$'s are much larger than others, we need some further treatment. 
Actually, we can use only a sub-grid $\Grid_{n+m'}^{n_1',n_2',\ldots,n_d'}$ of $\Grid_{n+m}^{n_1,n_2,\ldots,n_d}$ for the construction of $\Lambda_n$, where $n_i'\le n_i$ for all $i\in[d]$, and $m'=\prod_{i=1}^dn'_i-n \ge 36n$.
We consider  $d+1$ cases: 
\begin{itemize}
    \item Case 1: $n_d\ge 2^{\frac{n}{d+1}}$. In this case we choose $n_i'=2^{\frac{n}{d+1}}$ for all $i\in[d]$, which gives $m'=\prod_{i=1}^dn'_i-n \ge \omega(n)$. The total depth is
    \[O\Big(n^2+\sum_{j=1}^d n'_j+\frac{2^n}{n+m'}\Big) = O\left(n^2+d2^{\frac{n}{d+1}}+\frac{2^n}{2^{dn/(d+1)}}\right) = O\left(n^2+d2^{\frac{n}{d+1}}\right).
    \]
    
    \item Case $j$ ($2\le j\le d$): $n_d,n_{d-1},\ldots,n_{d-j+1}$ satisfy
     \begin{equation}\label{eq:range_casej}
         n_d < 2^{n/(d+1)},\quad n_{d-i}<\frac{2^{\frac{n}{d-i+1}}}{(n_{d-i+1}\cdots n_{d})^{\frac{1}{d-i+1}}}, \forall i\in[j-2], \quad n_{d-j+1}\ge  \frac{2^{\frac{n}{d-j+2}}}{(n_{d-j+2}\cdots n_d)^{\frac{1}{d-j+2}}}.
         \end{equation}
   We set 
   \begin{equation}\label{eq:ni'}
       n_i' =\begin{cases}
       \frac{2^{\frac{n}{d-j+2}}}{(n_{d-j+2}\cdots n_d)^{\frac{1}{d-j+2}}} &\quad i\in[d-j+1]\\
       n_i &\quad  i\in\{d-j+2,d-j+3,\ldots,d\}
       \end{cases}
   \end{equation}
   The number of ancillary qubits satisfies
 \begin{align}\label{eq:m'}
 m'&= \prod_{i=1}^{d}n'_i-n
 =\Big( \frac{2^{\frac{n}{d-j+2}}}{(n_{d-j+2}\cdots n_d)^{\frac{1}{d-j+2}}} \Big)^{d-j+1}(n_{d-j+2}\cdots n_d)-n= 2^{\frac{(d-j+1)n}{d-j+2}}(n_{d-j+2}\cdots n_d)^{\frac{1}{d-j+2}}-n\\
 \ge& 2^{\frac{(d-j+1)n}{d-j+2}}-n\ge 2^{\frac{n}{2}}-n=\omega(n).\nonumber
 \end{align}
    Now we have the following, where the first inequality holds because $n_{k-1}\le \frac{2^{n/k}}{(n_k\cdots n_d)^{1/k}}$ (Eq. \eqref{eq:range_casej}), and the second inequality holds because $n_{d-j+2}\ge n_{d-j+3}\ge \cdots \ge n_{k-1}$ for $ k\in \{d-j+3,\ldots, d\}$.
 \begin{align*}
     &\frac{2^{\frac{n}{d-j+2}}}{(n_{d-j+2}\cdots n_d)^{\frac{1}{d-j+2}}}/ \frac{2^{\frac{n}{k}}}{(n_k\cdots n_d)^{\frac{1}{k}}}= \frac{2^{\frac{n(k-(d-j+2))}{(d-j+2)k}}}{(n_{d-j+2}\cdots n_{k-1})^{\frac{1}{d-j+2}}(n_{k}\cdots n_d)^{\frac{k-(d-j+2)}{(d-j+2)k}}}\\
     \ge& \frac{(n_{k-1})^{\frac{k-(d-j+2)}{d-j+2}}}{(n_{d-j+2}\cdots n_{k-1})^{\frac{1}{d-j+2}}}\ge \frac{( n_{k-1})^{\frac{k-(d-j+2)}{d-j+2}}}{(n_{d-j+2})^{\frac{k-(d-j+2)}{d-j+2}}}\ge 1,\quad \forall k\in \{d-j+3,\ldots, d\}.
 \end{align*}

   Therefore, we have 
   \begin{equation}\label{eq:casej_ineq}
       \frac{2^{\frac{n}{d-j+2}}}{(n_{d-j+2}\cdots n_d)^{\frac{1}{d-j+2}}}\ge \frac{2^{\frac{n}{k}}}{(n_k\cdots n_d)^{\frac{1}{k}}},\quad \forall k\in \{d-j+3,\ldots, d\}.
   \end{equation} 
   The total circuit depth in this case is
\begin{align*}
  & O\Big(n^2+\sum_{i=1}^d n'_i+\frac{2^n}{n+m'}\Big)\\
  =& O\Big(n^2+\frac{(d-j+1)2^{\frac{n}{d-j+2}}}{(n_{d-j+2}\cdots n_d)^{\frac{1}{d-j+2}}}+\sum_{k=d-j+2}^d n_k+\frac{2^n}{n+m'}\Big) & (\text{by~Eq. \eqref{eq:ni'}})\\
   \le &O\Big(n^2+\frac{(d-j+1)2^{\frac{n}{d-j+2}}}{(n_{d-j+2}\cdots n_d)^{\frac{1}{d-j+2}}}+\sum_{k=d-j+2}^{d-1} \frac{2^{\frac{n}{k+1}}}{(n_{k+1}\cdots n_d)^{\frac{1}{k+1}}}+2^{\frac{n}{d+1}}+\frac{2^n}{n+m'}\Big) & (\text{by~Eq~}\eqref{eq:range_casej})\\
   \le &  O\Big(n^2+\frac{(d-j+1)2^{\frac{n}{d-j+2}}}{(n_{d-j+2}\cdots n_d)^{\frac{1}{d-j+2}}}+\sum_{k=d-j+2}^{d-1} \frac{2^{\frac{n}{k+1}}}{(n_{k+1}\cdots n_d)^{\frac{1}{k+1}}}+2^{\frac{n}{d+1}}+\frac{2^{\frac{n}{d-j+2}}}{(n_{d-j+2}\cdots n_d)^{\frac{1}{d-j+2}}}\Big) &(\text{by~Eq.~}\eqref{eq:m'})\\
   \le & O\Big(n^2+2^{\frac{n}{d+1}}+\frac{d2^{\frac{n}{d-j+2}}}{(n_{d-j+2}\cdots n_d)^{\frac{1}{d-j+2}}}\Big)& (\text{by~Eq.~}\eqref{eq:casej_ineq})\\
   \le & O\Big(n^2+2^{\frac{n}{d+1}}+\max_{k\in\{2,\ldots,d\}}\Big\{\frac{d2^{n/k}}{(\Pi_{i=k}^d n_i)^{1/k}}\Big\}+\frac{2^n}{n+m}\Big).
\end{align*}
\item Case $d+1$: $n_d,n_{d-1},\ldots,n_1$ satisfy
\begin{equation*}
         n_d < 2^{n/(d+1)},\quad  n_{d-i}<\frac{2^{\frac{n}{d-i+1}}}{(n_{d-i+1}\cdots n_{d})^{\frac{1}{d-i+1}}},\forall i\in[d-1].
         \end{equation*}
In this case we set $n'_i=n_i$ for all $i\in[d]$.  Thus, $m'=\prod_{i=1}^{d}n'_i-n=\prod_{i=1}^{d}n_i-n=m\ge 36n$. The total depth is 
\[O\Big(n^2+\sum_{j=1}^d n_j+\frac{2^n}{n+m}\Big)\le O\Big(n^2+2^{\frac{n}{d+1}}+\sum_{j=2}^d\frac{2^{n/j}}{(\Pi_{i=j}^d n_i)^{1/j}}+\frac{2^n}{n+m}\Big)\le O\Big(n^2+2^{\frac{n}{d+1}}+\max_{k\in\{2,\ldots,d\}}\Big\{\frac{d2^{n/k}}{(\Pi_{i=k}^d n_i)^{1/k}}\Big\}+\frac{2^n}{n+m}\Big).\]
\end{itemize}
Combining the above $d+1$ cases, the depth upper bound can be summarized as
\[O\Big(n^2+d2^{\frac{n}{d+1}}+\max_{k\in\{2,\ldots,d\}}\Big\{\frac{d2^{n/k}}{(\Pi_{i=k}^d n_i)^{1/k}}\Big\}+\frac{2^n}{n+m}\Big).\]

\end{proof}

\paragraph{Remark}
Under path and $d$-dimensional grid constraints, we prove later (Lemma \ref{lem:lower_bound_grid_k_Lambda}) that the depth lower bound for $\Lambda_n$ is $\Omega\left(\max\limits_{j\in [d]}\left\{n,2^{\frac{n}{d+1}},\frac{2^{n/j}}{(\prod_{i=j}^d n_i)^{1/j}}\right\}\right)$, using $m$ ancillary qubits. 
If $d$ is a constant, the depth upper and lower bounds match. If the number of ancillary qubits $m\le O\Big(\frac{2^n}{n^2+d2^{\frac{n}{d+1}}+\max\limits_{j\in\{2,\ldots,d\}}\frac{d2^{n/j}}{(\prod_{i=j}^d n_i)^{1/j}}}\Big)$, the depth upper bound is $O\left(\frac{2^n}{n+m}\right)$, which matches the corresponding lower bound $\Omega(\frac{2^n}{n+m})$, and both upper and lower bounds equal those under no graph constraints (\cite{sun2021asymptotically}). For example, if $m\le O(2^{n/2})$, the depth upper bound is $O\left(\frac{2^n}{n+m}\right)$ under path constraint.
It is somewhat surprising that the path and grid constraints do not asymptotically increase the circuit depth of diagonal unitary matrix $\Lambda_n$ if the size of grid (or the number of ancillary qubits) is not too large.
Moreover, our circuit depth is optimal if $d=O(1)$.

\subsection{Circuit implementation under $\Tree_{n+m}(d)$ constraints (Proof of Lemma \ref{lem:diag_tree_ancilla})}
\label{sec:diag_with_ancilla_binarytree}

We first consider the case of $\Tree_{n+m}(2)$ constraint. Without loss of generality, we assume that $m\le O(2^n)$; if $m\ge \omega(2^n)$, we only use $O(2^n)$ ancillary qubits. Our choice of Gray codes is given by setting $\ell_k = 1$, for all $k\in[2^p]$.

\paragraph{Choice of registers} We label the qubits in a binary tree as in Appendix \ref{sec:diag_without_ancilla_binarytree}. Recall that $\Tree_z^k:=\{zy:y\in\B^k\}$ denotes a depth $k$ binary tree where the root node is labelled $z$. The allocation of qubits to registers is shown schematically in Fig.~\ref{fig:register_ancillar_binary}, where the parameters $d$, $\kappa$ and $b$ are taken to be 
\begin{equation*}
    d=\left\lceil\log \left(n+m+1\right)\right\rceil-1, \qquad \kappa=\left\lceil\log \left(n+1\right)\right\rceil-1, \qquad b=\left\lceil\log (2\log n)\right\rceil.
\end{equation*}

\begin{enumerate}
    \item The $n+m$ qubits/nodes are in a depth-$d$ complete-binary tree. 
    
    \item The input register is located in the top sub-tree of $\kappa+1$ layers of nodes (the green part), namely 
    ${\sf R}_{\rm inp}\defeq \Tree_\epsilon^{\kappa}.$ 
    The input corresponds to the first $n$ qubits of ${\sf R}_{\rm inp}$.
    
    \item The remaining $(d+1)-(\kappa+1)=d-\kappa$ layers of nodes are divided into $\lfloor\frac{d-\kappa}{\kappa+b+1}\rfloor$ layers of subtrees, each of $\kappa+b+1$ layers of nodes. The roots of these subtrees collectively form the root register ${\sf R}_{\rm roots}$ (the black vertices in Fig.~\ref{fig:register_ancillar_binary}). The number of these subtrees, i.e. the size of ${\sf R}_{\rm roots}$, is 
    \[|{\sf R}_{\rm roots}|=\sum_{j=1}^{\lfloor\frac{d-\kappa}{\kappa+b+1}\rfloor}2^{d-j(\kappa+b+1)+1} = \Theta\left(\frac{m}{n\log(n)}\right).\]

    \item For each $z\in {\sf R}_{\rm roots}$, the first $\kappa+1$ layers (including the root $z$ itself) in subtree $\Tree_z^{\kappa+b}$ are assigned to the copy register ${\sf R}_{\rm copy}$ (north-east lines part in Fig.~\ref{fig:register_ancillar_binary}), the next one layer is assigned to the target register ${\sf R}_{\rm targ}$ (red part in Fig. \ref{fig:register_ancillar_binary}) and the last $b-1$ layers are assigned to the auxiliary register ${\sf R}_{\rm aux}$ (white part in Fig. \ref{fig:register_ancillar_binary}). The sizes of these three parts in each subtree are 
    \begin{equation}\label{eq:size_threepart}
        \sum_{i=0}^\kappa 2^i = \left(2^{\kappa+1}-1\right) = \Theta(n),\quad  2^{\kappa+1} = \Theta(n),\quad \text{and}\quad \sum_{i=\kappa+2}^{\kappa+b} 2^i = 2^{\kappa+1}(2^{b}-2) = \Theta(n\log(n)),
    \end{equation}
respectively. Putting all subtrees together, we multiply these sizes by $|{\sf R}_{\rm roots}|$, and get the sizes of the registers ${\sf R}_{\rm copy}$, ${\sf R}_{\rm targ}$, and ${\sf R}_{\rm aux}$ 
\begin{align*}
   \lambda_{copy} = \Theta\left(\frac{m}{\log (n)}\right),\quad \lambda_{targ} =  \Theta\left(\frac{m}{\log (n)}\right), \quad
   \lambda_{aux} = \Theta\left(m\right),
\end{align*}
respectively.
\end{enumerate} 

\begin{figure}[]
    \centering
    \begin{tikzpicture}
       \filldraw[fill=green!20] (0,0)--(0.3,-0.75)--(-0.3,-0.75)--cycle;
       \draw[dotted] (-3.6,0)--(0,0) (-2,-0.75)--(0,-0.75) (-3.6,-3.5)--(0,-3.5) (-2,-1)--(-0.5,-1) (-2,-2)--(-0.5,-2) (-2,-2.5)--(0,-2.5);
       \draw (-4.3, -1.9)node{\scriptsize $d+1$ layers} (-2.5,-0.4) node{\scriptsize $\kappa+1$ layers} (-2.5,-1.5) node{\scriptsize $\kappa+b+1$ layers} (-2.5,-3) node{\scriptsize $\kappa+b+1$ layers} (-1.5,-2.1) node{\scriptsize $\vdots$};
       \draw[->] (-3.5, -1.75)--(-3.5, 0); 
       \draw [->] (-3.5, -1.75)--(-3.5, -3.5);
       \draw [->](-1.5,-0.375)--(-1.5,0);
       \draw [->] (-1.5,-0.375)--(-1.5,-0.75);
       
      \draw [->](-1.5,-1.375)--(-1.5,-1);
       \draw [->] (-1.5,-1.375)--(-1.5,-2);
       \draw [->](-1.5,-2.875)--(-1.5,-2.5);
       \draw [->] (-1.5,-2.875)--(-1.5,-3.5);

       \filldraw[fill=red!20] (-0.7,-1.5)--(-0.3,-1.5)--(-0.2,-1.75)--(-0.8,-1.75) -- cycle (0.7,-1.5)--(0.3,-1.5)--(0.2,-1.75)--(0.8,-1.75) -- cycle;
        \filldraw[fill=blue!20] (-0.5,-1)--(-0.75,-1.625)--(-0.25,-1.625)--cycle (0.5,-1)--(0.25,-1.625)--(0.75,-1.625)--cycle;
        \draw (-0.5,-1)--(-0.9,-2)--(-0.1,-2)--cycle (0.5,-1)--(0.1,-2)--(0.9,-2)--cycle;
        \filldraw[fill=red!20] (-1.2,-3)--(-0.8,-3)--(-0.7,-3.25)--(-1.3,-3.25)--cycle  (1.2,-3)--(0.8,-3)--(0.7,-3.25)--(1.3,-3.25)--cycle  (-0.2,-3)--(0.2,-3) --(0.3,-3.25)--(-0.3,-3.25) --cycle;
        \filldraw[fill=blue!20] (-1,-2.5)--(-1.25,-3.125)--(-0.75,-3.125)--cycle (1,-2.5)--(1.25,-3.125)--(0.75,-3.125)--cycle (0,-2.5)-- (-0.25,-3.125)--(0.25,-3.125)--cycle;
        
        \draw (-1,-2.5)--(-1.4,-3.5)--(-0.6,-3.5)--cycle (1,-2.5)--(1.4,-3.5)--(0.6,-3.5)--cycle (0,-2.5)-- (-0.4,-3.5)--(0.4,-3.5)--cycle;
        \draw (0,-2.3) node{\scriptsize $\cdots$} (0,-1.5) node{\scriptsize $\cdots$} (-0.5,-3) node{\scriptsize $\cdots$} (0.5,-3) node{\scriptsize $\cdots$};
        \draw [->] (1,-2.75)--(1.9,-2.75);
        \draw [->] (1,-3.17)--(1.9,-3.17);
        \draw [->] (1,-3.45)--(1.9,-3.45);
        \draw [fill=green!20] (1.5,0)--(1.5,-0.3)--(1.8,-0.3)--(1.8,-0)--cycle;
        \draw[fill=blue!20] (1.5,-0.5)--(1.5,-0.8)--(1.8,-0.8)--(1.8,-0.5)--cycle;
        \draw[fill=red!20] (1.5,-1)--(1.5,-1.3)--(1.8,-1.3)--(1.8,-1)--cycle;
        \draw (1.5,-1.5)--(1.5,-1.8)--(1.8,-1.8)--(1.8,-1.5)--cycle;
        \draw (4,-2.75) node{\scriptsize $2^{\kappa+1}-1=O(n)$ ~qubits, $\kappa+1$ layers} (3.5,-3.14) node{\scriptsize $2^{\kappa+1}=O(n)$ ~qubits, $1$ layer} (4.6,-3.5)node{\scriptsize $2^{\kappa+1}(2^b-2)=O(n\log(n))$ ~qubits, $b-1$ layers} (2.7,-0.15) node{\scriptsize input register}  (2.7,-0.65) node{\scriptsize copy register} (2.75,-1.15) node{\scriptsize target register}(2.95,-1.65) node{\scriptsize auxiliary register};
        \draw (1.3,0.2)--(4,0.2)--(4,-2)--(1.3,-2)--cycle;
        \draw [fill=black,draw=black] (-0.5,-1.13) circle (0.03) (0.5,-1.13) circle (0.03) (-1,-2.63) circle (0.03) (0,-2.63) circle (0.03) (1,-2.63) circle (0.03);
    \end{tikzpicture}
    \caption{Input, copy, target and auxiliary registers in a binary tree with $d+1$ layers of qubits. The $n$ input qubits are assigned to a sub-tree with $\kappa+1$ layers of qubits at the top of the tree (green). The $m$ ancillary qubits are divided into $O\left(\frac{m}{n\log(n)}\right)$ binary sub-trees with $\kappa+b+1$ layers of qubits, each further divided into three parts: (i) the first $\kappa+1$ layers are the copy register (north-east lines), and have size $O(n)$, (ii) a single layer of target register (red) qubits, of size $O(n)$, and (iii) $b-1$ layers of the auxiliary register (white), of size $O(n\log(n))$. The values of $d$, $\kappa$ and $b$ are given in the main text. Note that every target register qubit has $\tau:=2^b-2$ auxiliary register descendants. ${\sf R}_{\rm roots}$ consists of the root nodes of all blue sub-trees (black vertices).}
    \label{fig:register_ancillar_binary}
\end{figure}
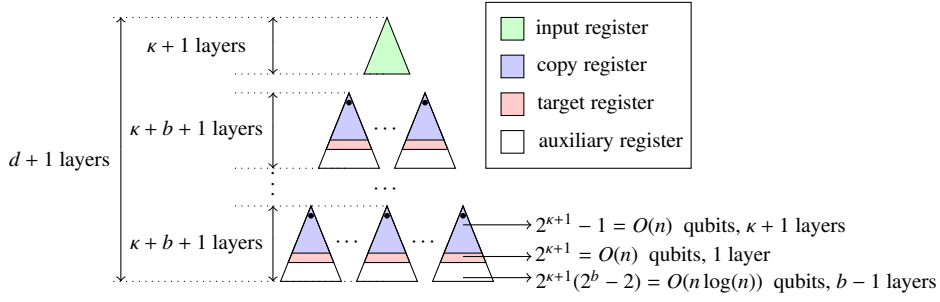

The formal definitions of these registers are given as follows, using the label notation as in Appendix \ref{sec:diag_without_ancilla_binarytree}. 
\begin{align}
\label{eq:subtree-root-nodes}
 & {\sf R}_{\rm roots}\defeq\bigcup_{j=1}^{\left\lfloor\frac{d-\kappa}{\kappa+b+1}\right\rfloor}\B^{d-j(\kappa+b+1)+1},\\
    &{\sf R}_{\rm copy}\defeq \bigcup_{\scriptsize z\in {\sf R}_{\rm roots}}\Tree_z^\kappa,\nonumber\\
    & {\sf R}_{\rm targ}\defeq \bigcup_{\scriptsize z\in {\sf R}_{\rm roots}}\left\{zy: y\in \B^{\kappa+1}\right\}, \nonumber\\
    &{\sf R}_{\rm aux}\defeq \bigcup_{\scriptsize z\in {\sf R}_{\rm roots}}\Big\{zy: y\in \bigcup_{j=\kappa+2}^{\kappa+b}\B^{j}\Big\}= \bigcup_{z\in {\sf R}_{\rm targ}}(\Tree_z^{b-1}-\{z\}).\nonumber
\end{align}

We take
\begin{align*}
    p&=\log(\lambda_{targ}) = \log m - \log \log n \pm O(1), \qquad\tau=2^b-2 = \Theta(\log(n)),
\end{align*}
in specifying $x_{pre} = x_1\ldots x_{n-p}$, $x_{suf}=x_{n-p+1}\ldots x_n$, and $x_{aux}=x_1\ldots x_\tau$.

\paragraph{Implementation of Suffix Copy and Prefix Copy Stages}
\begin{lemma}\label{lem:copy_binary_tree}
A unitary operation  realizing the following transformation
\begin{equation}\label{eq:copy_binary_tree}
   \ket{x'}\ket{0^{n't}}\xrightarrow{U_{copy}^{binarytree}}\ket{x'}\underbrace{\ket{x'\cdots x'}}_{t\text{~copies~of~} x'},\forall x'\in\Bn, {\rm where }~ n'\le n, t=|{\sf R}_{\rm roots}|
\end{equation}
can be implemented by a circuit of depth $O(n'\kappa^2+n'\log(m)\kappa)$ under $\Tree_{n+m}(2)$ constraint, where input $x'$ is in ${\sf R}_{\rm inp}$ and every copy of $x'$ is in a subtree in Fig. \ref{fig:register_ancillar_binary}.
\end{lemma}
\begin{proof}
Let $\kappa':=\kappa+b+1$; note that $\kappa < \kappa' < 2\kappa$. We first implement the following unitary transformation
\begin{equation}\label{eq:copy_tree_2^{k+1}}
    \ket{x'0^{2^{\kappa+1}-n'-1}}_{\Tree_\epsilon^\kappa}\bigotimes_{\scriptsize 
    z\in\B^{\kappa+1}}\ket{0^{2^{\kappa'+1}-1}}_{\Tree_z^{\kappa'}}\to\ket{x'0^{2^{\kappa +1}-n'-1}}_{\Tree_\epsilon^\kappa}\bigotimes_{\scriptsize 
    z\in\B^{\kappa+1}}\ket{x'0^{2^{\kappa'+1}-n'-1}}_{\Tree_z^{\kappa'}},\forall x'\in\B^{n'},
\end{equation}
which makes $2^{\kappa+1}$ copies of $x$ in all $2^{\kappa+1}$ subtrees directly below the input register. It can be implemented in $\kappa+2$ steps.
\begin{enumerate}
    \item Step 0: Make 1 copy of $x'$ from the top subtree $\Tree_{\epsilon }^\kappa$ to the first (i.e. leftmost) subtree under it (i.e. $\Tree_{0^{\kappa+1}}^{\kappa'}$), by applying $n$ CNOT gates with control and target qubits $O(\kappa)$-away. By Lemma~\ref{lem:cnot_path_constraint}, this step can be implemented in depth $O(n')\cdot O(\kappa)=O(n'\kappa)$.
    
    \item Step 1: Make 1 copy of $x'$ from $\Tree_{\epsilon }^\kappa$ to $\Tree_{10^{\kappa}}^\kappa$. This step can similarly be implemented in depth $O(n'\kappa)$.
    
    \item Step $j$ for $j = 2, 3, \cdots, \kappa+1$: For all $z\in\B^{j-1}$, make 1 copy of $x'$ from $\Tree_{z00^{\kappa-j+1}}^{\kappa}$ to $\Tree_{z10^{\kappa-j+1}}^{\kappa}$. Such a copy can be realized in depth $O(n'\kappa)$ by Lemma \ref{lem:cnot_path_constraint}. Note that for different $z$, the copying processes are on disjoint connected components of the binary tree, thus these $2^{j-1}$ copies can be implemented in parallel. Therefore for each $j$, this step can be implemented in depth $O(n'\kappa)$.
\end{enumerate}
The total circuit depth required to implement Eq. \eqref{eq:copy_tree_2^{k+1}} is $O(n'\kappa) + \kappa\cdot O(n'\kappa) =  O(n'\kappa^2)$.

Eq.~\eqref{eq:copy_binary_tree} can be implemented by using  Eq.~\eqref{eq:copy_tree_2^{k+1}}  repeatedly. 
For every newly copied $x'$ in a $\kappa'$-depth binary sub-tree, we repeat the construction to make further copies of $x'$ in its adjacent binary sub-trees, which requires depth $O(n'(\kappa')^2)$, and so on.
We repeat this process $s:=\lfloor\frac{d-\kappa}{\kappa'+1}\rfloor=O(\log(m)/\kappa)$ times and make $t$ copies of $x'$. The total circuit depth is  $O(n'\kappa^2)+(s-1)O(n'(\kappa')^2)=O(n'\kappa^2+n'\log(m)\kappa)$.
\end{proof}

\begin{lemma}[]\label{lem:sufcopy_tree}
$U_{SufCopy}$ and $U_{PreCopy}''$ can each be implemented by a circuit of depth $O(n\log(n)\log(m))$ under  $\Tree_{n+m}(2)$ constraint.
\end{lemma}
\begin{proof}
$U_{SufCopy}$ makes $|{\sf R}_{\rm roots}|
$ copies of $x_{suf}$, and can be represented as
\begin{equation*}
      \ket{x0^{2^{\kappa+1}-n-1}}_{\Tree_\epsilon^\kappa}\bigotimes_{\scriptsize z\in {\sf R}_{\rm roots}}\ket{0^{2^{\kappa+1}-1}}_{\Tree_z^{\kappa}}
   \to   \ket{x0^{2^{\kappa+1}-n-1}}_{\Tree_\epsilon^\kappa}\bigotimes_{\scriptsize z\in {\sf R}_{\rm roots}}\ket{x_{suf}0^{2^{\kappa+1}-p-1}}_{\Tree_z^{\kappa}},\forall x\in\Bn.
\end{equation*}
By Lemma~\ref{lem:copy_binary_tree} (with $n' =p$ and $\kappa = O(\log(n))$), it can be implemented by a circuit of depth 
\[O\Big(p\log^2(n)+p\log(m)\log(n)\Big) = O(\log^2(m)\log(n)) = O(\log(m) n \log(n)),\] 
where we used $p \le \log(m)$ and the assumption $m=O(2^n)$. The argument for $U_{PreCopy}''$ is similar though now the parameter $n'$ in Lemma~\ref{lem:copy_binary_tree} is $n-p$, and thus the depth upper bound is $O(n \log(m) \log(n))$. 
\end{proof}


\begin{lemma}[]\label{lem:precopy3_tree}
$U_{PreCopy}'''$ (Eq. \eqref{eq:precopy3_graph}) can be implemented by a quantum circuit of depth $O(n\log^2(n))$ under $\Tree_{n+m}(2)$ constraint.
\end{lemma}
\begin{proof}

For each $z\in  {\sf R}_{\rm roots}$ (Eq.~\eqref{eq:subtree-root-nodes}), $U_{PreCopy}'''$ makes $2^{\kappa+1}$ copies of $x_{aux}=x_1\ldots x_\tau$ (with $\tau=2^b-2)$ from the ${\sf R}_{\rm copy}$ to ${\sf R}_{\rm aux}$ parts of $\Tree_z^{\kappa+b}$ (i.e., from blue to white portions of each sub-tree in Fig.~\ref{fig:register_ancillar_binary}). As the distance between any two qubits in $\Tree_z^{\kappa+b}$ is $O(\log(n))$, by Lemma~\ref{lem:cnot_path_constraint}, the $2^{\kappa+1}$ copies
can be implemented in depth $O(2^{b}-2)\cdot O(\log(n))\cdot 2^{\kappa+1}=O(n\log^2 (n))$. Since all the binary sub-trees are disjoint, they can be implemented in parallel, and thus $U_{PreCopy}'''$ has circuit depth $O(n\log^2(n))$.
\end{proof}

\begin{lemma}[]\label{lem:precopy_tree}
$U_{PreCopy}$ can be implemented by a quantum circuit of depth $O(n\log(m)\log(n))$ under $\Tree_{n+m}(2)$ constraint.
\end{lemma}
\begin{proof}
Follows from Lemmas~\ref{lem:precopy_graph}, \ref{lem:sufcopy_tree} and \ref{lem:precopy3_tree}.
\end{proof}

\paragraph{Implementation of Gray Initial Stage}


\begin{lemma}[]\label{lem:grayinitial_tree}
$U_{GrayInit}$ can be implemented by a CNOT circuit of depth $O(n^2)$ under $\Tree_{n+m}(2)$ constraint.
\end{lemma}
\begin{proof}
Recall $U_{GrayInit}$ defined in Eq.~\eqref{eq:gray_initial_graph}, which can be represented as follows. For all $z_{q}\in {\sf R}_{\rm roots}$, 
\begin{equation}\label{eq:grayinitial_tree_version2}
\begin{array}{ll}
     & \ket{x_{suf}0^{2^{\kappa+1}-p-1}}_{\Tree_{z_{q}}^{\kappa}}\ket{0^{2^{\kappa+1}}}_{\scriptsize \left\{z_{q}y:y\in \B^{\kappa+1}\right\}} \\
   \to  & \ket{x_{suf}0^{2^{\kappa+1}-p-1}}_{\Tree_{z_{q}}^{\kappa}}\ket{f_{1,1+(q-1)2^{\kappa+1}} f_{1,2+(q-1)2^{\kappa+1}} \cdots f_{1,q2^{\kappa+1}}}_{\scriptsize \left\{z_{q}y:y\in \B^{\kappa+1}\right\}},
\end{array}
\end{equation}
where $z_q$ is the $q$-th element in $ {\sf R}_{\rm roots}$. 
Eq. \eqref{eq:grayinitial_tree_version2} acts on qubits in 
{$\Tree_{z_q}^{\kappa+1}$ of size $O(n)$} and, by Lemma~\ref{lem:cnot_circuit}, can be realized by a CNOT circuit of depth $O(n^2)$ under binary tree constraint. 
All trees $\Tree_{z_q}^{\kappa+1}$ are disjoint and $U_{GrayInit}$ can therefore be implemented in parallel in depth $O(n^2)$. 
\end{proof}

\paragraph{Implementation of Gray Cycle Stage}
\begin{lemma}[]\label{lem:graycycle_tree}
$U_{GrayCycle}$ (Eq. \eqref{eq:gray_cycle_graph}) can be implemented by a quantum circuit of depth $O(2^{n-p})$ under $\Tree_{n+m}(2)$ constraint.
\end{lemma}
\begin{proof}
First, we construct circuits for $U_{Gen}^{(j)}$ for all $j\in[2^{n-p}]$ in Eqs. \eqref{eq:Ugenj_graph} and \eqref{eq:Ugen2n-p_graph}.
Let $z_q$ be the $q$-th element in $ {\sf R}_{\rm roots}$. For every $i\in [2^{\kappa+1}]$, let $y_i$ denote the $i$-th element in $\B^{\kappa+1}$ in lexicographical order.
$U_{Gen}^{(r)}$ (Eq. \eqref{eq:Ugenj_graph}) can be represented as acting on qubits in $\Tree_{z_q}^{\kappa+b} = \Tree_{z_q}^\kappa \cup  (\bigcup_{y_i\in\{0,1\}^{\kappa+1}} \Tree_{z_qy_i}^{b-1})$ in the following way:
\begin{align}
    &\ket{x_{pre}0^{2^{\kappa+1}-(n-p)-1}}_{\Tree_{z_q}^{\kappa}}\bigotimes_{\scriptsize y_i\in \B^{\kappa+1}}\ket{f_{r,i+(q-1)2^{\kappa+1}}x_{aux}}_{\Tree_{z_q y_i}^{b-1}} \nonumber\\
    \to &\ket{x_{pre}0^{2^{\kappa+1}-(n-p)-1}}_{\Tree_{z_q}^{\kappa}}\bigotimes_{\scriptsize y_i\in \B^{\kappa+1}}\ket{f_{r+1,i+(q-1)2^{\kappa+1}} x_{aux}}_{\Tree_{z_q y_i}^{b-1}},\quad \forall z_q \in  {\sf R}_{\rm roots}. \label{eq:Ugen_small}
\end{align}
For all $y_i\in\B^{\kappa+1}$, Eq.~\eqref{eq:Ugen_small} transforms $\ket{f_{r,i+(q-1)2^{\kappa+1}}}_{z_{q}y_i}$ to $\ket{f_{r+1,i+(q-1)2^{\kappa+1}}}_{z_{q}y_i}=\ket{f_{r,i+(q-1)2^{\kappa+1}}\oplus x_{h_{1,r+1}}}_{z_{q}y_i}$, and can be implemented by a CNOT gate with target qubit $z_{q}y_i$, and control qubit in state $\ket{x_{h_{1,r+1}}}$. We consider two cases:
\begin{enumerate}
    \item Case 1: If $h_{1,r+1}\le 2^b-2$, we use the control qubit  $\ket{x_{h_{1,r+1}}}$ in ${\sf R}_{\rm aux}$ in $\Tree_{z_q y_i}^{b-1}-\{z_q y_i\}$, the subtree under the current target qubit $z_{q}y_i$. The distance between control and target qubits is $O(\log(h_{1,r+1}))$ in binary tree $\Tree_{z_{q}y_i}^{b-1}$.  By Lemma \ref{lem:cnot_path_constraint}, it can be implemented in depth $O(\log(h_{1,r+1}))$ under an $O(\log(h_{1,r+1}))$-long path in binary tree. For all $y_{i}\in\B^{\kappa+1}$, trees $\Tree_{z_{q}y_i}^{b-1}$ are disjoint. Eq.~\eqref{eq:Ugen_small} can thus be implemented in depth $O(\log(h_{1,r+1}))$.
    
    \item Case 2: If $h_{1,r+1}> 2^b-2$, we use the control qubit $\ket{x_{h_{1,r+1}}}$ in ${\sf R}_{\rm copy}$ in $\Tree_{z_{q}}^{\kappa+1}$. Then Eq. \eqref{eq:Ugen_small} can be implemented by a CNOT circuit acting on qubits purely within $\Tree_{z_{q}}^{\kappa+1}$.  By Lemma~\ref{lem:cnot_circuit}, Eq. \eqref{eq:Ugen_small} can be implemented in depth $O(n^2)$. 
\end{enumerate}

For all $z_q\in {\sf R}_{\rm roots}$, $\Tree_{z_{q}y_i}^{b-1}$ are disjoint. Therefore, the circuit depth of Eq. \eqref{eq:Ugenj_graph} is $O(\log(h_{1,r+1}))$ if $h_{1,r+1}\le 2^b-2$, and $O(n^2)$ if $h_{1,r+1}>2^b-2$ under $\Tree_{n+m}(2)$ constraint.
Similar to the circuit of Eq. \eqref{eq:Ugenj_graph}, $U_{Gen}^{(2^{n-p})}$ (Eq. \eqref{eq:Ugen2n-p_graph}) can be implemented by a CNOT circuit of depth $O(n^2)$ according to Lemma \ref{lem:cnot_circuit}. 

We now bound the circuit depth required to implement $U_{GrayCycle}$. From Lemma~\ref{lem:GrayCode}, there are $2^{n-p-i}$ values of $r$ in $[2^{n-p}-1]$ such that $h_{1,r+1}=i$.
Recall that $b=\lceil\log(2\log(n))\rceil$. By Lemma~\ref{lem:graycycle_graph}, the depth of the Gray cycle stage is 
\[\sum_{j=1}^{2^{n-p}}\mathcal{D}(U_{Gen}^{(j)})+2^{n-p}=\sum_{i=1}^{\tau}O(\log (i))2^{n-p-i}+\sum_{i=\tau+1}^{n-p}O(n^2)2^{n-p-i}+2^{n-p}=O(2^{n-p}),\]
where $\tau=2^b-2\ge 2\log(n)-2$.
\end{proof}
\paragraph{Implementation of Inverse Stage} 
\begin{lemma} []\label{lem:inverse_tree}
$U_{Inverse}$ (Eq. \eqref{eq:inverse_graph})
can be implemented by a CNOT circuit of depth $O(\log(m)n\log (n))$ under $\Tree_{n+m}(2)$ constraint.
\end{lemma}
\begin{proof}
Follows from Lemmas \ref{lem:inverse_graph}, \ref{lem:grayinitial_path}, \ref{lem:sufcopy_tree} and \ref{lem:precopy3_tree}.
\end{proof}

\paragraph{Implementation of $\Lambda_n$} 
\begin{lemma}[Lemma \ref{lem:diag_tree_ancilla} (Case 1)]\label{lem:diag_binarytree_withancilla}
Any $n$-qubit unitary diagonal matrix $\Lambda_n$ can be implemented by a quantum circuit of depth  \[O\left(n^2\log n+\frac{\log(n)2^n}{m}\right)\]
under $\Tree_{n+m}(2)$ constraint, using $m\ge 3n$ ancillary qubits.
\end{lemma}
\begin{proof}
We assume that $m\le O(2^n)$. If $m=\omega(2^n)$, we only use $O(2^n)$ ancillary qubits.
From Lemmas \ref{lem:sufcopy_tree}, \ref{lem:grayinitial_tree}, \ref{lem:precopy_tree}, \ref{lem:graycycle_tree} and  \ref{lem:inverse_tree} the total depth is 
\[3O(n^2\log n)+O(n^2)+O(2^{n-p})=O\left(n^2\log (n)+\frac{\log(n)2^n}{m}\right).\]
where we used $p = \log m - \log \log n \pm O(1)$.
\end{proof}

Similarly, we can construct a circuit for $\Lambda_n$ under $d$-ary trees constraint. 
\begin{lemma}[Lemma \ref{lem:diag_tree_ancilla} (Case 2)]\label{lem:diag_d_tree_ancilla}
Any $n$-qubit diagonal unitary matrix can be implemented by a quantum circuit of depth $O\left(nd\log_d (n+m)\log_d(n+d)+\frac{(n+d)\log_d(n+d) 2^{n}}{n+m}\right)$ under $\Tree_{n+m}(d)$ constraint, using $m\ge 3n$ ancillary qubits.
\end{lemma}
\begin{proof}
The proof is similar to the proofs in Appendices \ref{append:diag_without_ancilla} and \ref{append:diag_with_ancilla}, so we only sketch the approach. Let $x=x_{pre}x_{suf}\in \Bn$, where $x_{pre}=x_1x_2\cdots x_{n-p}$, $x_{suf}=x_{n-p+1}\cdots x_{n}$, with $p$ specified below. The implementation of $\Lambda_n$ under $d$-ary tree constraint is discussed in two cases. We label qubits in a $d$-ary tree as follows. The root node is labelled with the empty string $\epsilon$. For a node with label $z$, for all $j\in[d]$, the $j$-th child of $z$ is labelled $z(j-1)$. Let $\Tree_z^k(d)=\{zy:y\in\{0,1,\ldots,d-1\}^{\le k}\}$ denote a $d$-ary tree of depth $k$, where $\{0,1,\ldots,d-1\}^{\le k}=\bigcup_{j=0}^k\{0,1,\ldots,d-1\}^j$. The depth of $\Tree_{n+m}(d)$ is $h=\lceil\log_d((m+n)(d-1)+1)\rceil-1=O(\log_d(n+m))$.
\begin{enumerate}
    \item Case 1: $d\ge 2n$. Assume that the $n$ input qubits $\ket{x_1\cdots x_n}$ are the first $n$ qubits at depth one of $\Tree_\epsilon^1(d)$ (i.e., the layer below the root). Take $p=O(\log((d-n)m/d))$, and divide the remaining part of $d$-ary tree into $O(m/d)$ $d$-ary trees, each of depth $1$. The set of root nodes of these trees is ${\sf R}_{\rm root}=\bigcup_{j=1}^{\lceil h/2\rceil} \{0,1,\ldots,d-1\}^{h-1-2(j-1)}$ of size $O(m/d)$. For all $z\in {\sf R}_{\rm root}$, the first $n$ qubits in the first depth of $\Tree_z^1(d)$ form the copy register, and the remaining $d-n$ qubits in the first depth form the target register.
     The circuit implementation of $\Lambda_n$ consists of 5 stages. 
    \begin{enumerate}
        \item Suffix Copy. We make $O(m/d)$ copies of $x_{suf}$, where each copy is made on the first $p$ qubits of one of the $\Tree_z^1(d)$ for all $z\in{\sf R}_{\rm root}$. The circuit depth required is $O(\log_d(n+m)nd)$.
        \item Gray Initial. For all $z\in{\sf R}_{\rm root}$, we generate linear combinations of $x_{suf}$, i.e., all $\ket{\langle x, 0^{n-p}s\rangle}$ for $s\in\B^{p}$ on qubits $zn,z(n+1),z(n+2),\ldots,zd$ of $\Tree_z^k(d)$.  
        In each  $\Tree_z^1(d)$, there exists a copy of $x_{suf}$ and all these trees are disjoint. Therefore, the circuit depth required is $O(dn)$. 
         \item Prefix Copy. We invert the suffix copy stage to restore the copy register to zero states. Then we make $O(m/d)$ copies of $x_{pre}$, with each copy made on the first $n-p$ qubits $\{z0,z1,\ldots,z(n-p-1)\}$ of $\Tree_z^1(d)$ for all $z\in{\sf R}_{\rm root}$. The total depth required is $O(\log_d(n+m)nd)$.
         \item Gray Cycle. We generate $\ket{\langle x,s\rangle}$ for all $x\in\Bn$. As in Appendix \ref{append:diag_with_ancilla}, this stage consists of $2^{n-p}$ steps. We generate $(n-p,1)$-Gray codes in every qubit of the target register, and introduce the corresponding phases. Each step can be implemented by a circuit of depth $O(d)$ in $\Tree_z^1(d)$, for every $z\in{\sf R}_{\rm root}$. Since all $\Tree_z^1(d)$s are disjoint, the total depth required is $O(d2^{n-p})$.
         \item Inverse. We invert the Prefix Copy, Gray Initial and Suffix Copy stages.The total depth required is $O(\log_d(n+m)nd)$.
    \end{enumerate}
    In total, $\Lambda_n$ can be implemented in depth $O(\log_d(n+m)nd+d2^{n-p})=O(\log_d(n+m)nd+d2^n/m)$.
    \item Case 2: $d<2n$. 
    The $n$ input qubits are stored in the first $n$ qubit of $\Tree_z^{\kappa}(d)$, where $\kappa=\lceil\log_d(n(d-1)+1)\rceil-1=O(\log_d(n))$. Let
    \begin{equation*}
        k=\begin{cases}
        \lceil\log_d(n(d-1)+1)-1\rceil&\qquad \text{if~}  \frac{d^{\lceil\log_d(n(d-1)+1)-1\rceil}-1}{d-1}-n\ge \Omega(n),\\
        \lceil\log_d(n(d-1)+1)\rceil &\qquad\text{otherwise}.
        \end{cases}
    \end{equation*}
    Apart from the qubits in $\Tree_z^{\kappa}(d)$, the rest qubits in $\Tree_{n+m}(d)$ are divided into $O(m/t)$ subtrees of depth $k$, each of which consists of $t:=\frac{d^{k+1}-1}{d-1}$ vertices . The set of root nodes for these trees are ${\sf R}_{\rm root}=\bigcup_{j=1}^{\frac{h-\kappa}{k+1}} \{0,1,\ldots,d-1\}^{h+1-j(k+1)}$.
    For all $z\in{\sf R}_{\rm root}$, the first $n$ qubits of $\Tree_z^k(d)$ form the copy register and the rest $t-n$ qubits form the target register.
    Take $p=\log\left[\left(t-n\right)\cdot O(m/t)\right]=\log(m)$.
The circuit implementation of $\Lambda_n$ consists of 5 stages. 
    
    \begin{enumerate}
        \item Suffix Copy. We make $r$ copies of $x_{suf}$, and every copy of $x_{suf}$ is made on the first $p$ qubits of each $\Tree^k_z(d)$ for all $z\in{\sf R}_{\rm root}$. The depth required is $O(\log_d(n+m)nd)$.
        \item Gray Initial. We generate all linear combinations of $x_{suf}$, i.e, all $\ket{\langle x, 0^{n-p}s\rangle}$ for $s\in\B^{p}$, on the last $\frac{d^{k+1}-1}{d-1}-n$ qubits of each $\Tree^k_z(d)$. The depth required is $O(nd\log_d(m)\log_d(n))$.
         \item Prefix Copy. We invert the Suffix Copy stage and then make $r$ copies of $x_{pre}$, with each made on the first $n-p$ qubits of the $\Tree^k_z(d)$ for all $z\in{\sf R}_{\rm root}$. The depth required is $O(\log_d(n+m)nd)$.
         \item Gray Cycle. We generate all $\ket{\langle x,s\rangle}$ for all $x\in\Bn$. This stage consists of $2^{n-p}$ steps. We generate $(n-p,1)$-Gray code in every qubit of the target register. Every step can be implemented by a circuit of depth $O(n\log_d(n))$, and the total depth required is $O(n\log_d(n)2^{n-p})$.
         \item Inverse. We apply the inverse circuits of the prefix copy, Gray initial and suffix copy stages, in total depth $O(\log_d(n+m)nd)$.
    \end{enumerate}
        In total, $\Lambda_n$ can be implemented by a circuit of depth $O(nd\log_d(m)\log_d(n)+\log_d(n+m)nd+n\log_d(n)2^{n-p})=O(nd\log_d(m)\log_d(n)+n\log_d(n)2^n/m)$.
\end{enumerate}
\end{proof}

\begin{lemma}[Lemma \ref{lem:diag_tree_ancilla} (Case 3)]\label{lem:diag_star_ancilla}
Any $n$-qubit diagonal unitary matrix can be implemented by a quantum circuit of depth $O(2^n)$ with $m>0$ ancillary qubits, under $\Star_{n+m}$ constraint.
\end{lemma}
\begin{proof}
    Do not use ancillary qubits. The result follows from Lemma \ref{lem:diag_star_noancilla}.
\end{proof}

\subsection{Circuit implementation under $\Expander_{n+m}$ constraints (Proof of Lemma \ref{lem:diag_expander_ancilla})}
\label{sec:diag_with_ancilla_expander}
For this case, we use a different circuit framework to that shown in Fig.~\ref{fig:diag_with_ancilla_framwork}: 
\begin{enumerate}
    \item Here, the ancillary qubits are divided into only two registers, ${\sf R}_{\rm copy}$ and ${\sf R}_{\rm targ}$, and there is no auxiliary register ${\sf R}_{\rm aux}$. 
    \item In Fig.~\ref{fig:diag_with_ancilla_framwork}, the suffix copy and prefix 
    copy stages make copies of $x_{suf}$ and $x_{pre}$, in order to reduce the depth of the Gray initial and Gray cycle stages which follow them, respectively. Here, the suffix copy and prefix copy stages are not used, and the circuit consists only of the other three stages, i.e. the Gray initial, Gray cycle and inverse stages. 
    The precise definition of these three steps are given in Eq. \eqref{eq:grayinitial_expander}, Eq. \eqref{eq:graycycle_expander} and Eq. \eqref{eq:inverse_expander}, respectively, from which it is easily verified that the diagonal unitary $\Lambda_n$ is realized.
\end{enumerate}

\paragraph{Choice of registers}
Consider an expander graph $G$ with vertex expansion $h_{out}(G) = c$ for some constant $c>0$.
Let $c' = \frac{c}{c+2}$. 
Let $\ell=\Big\lfloor\frac{\log(m)-1-\log(\lceil 1/c'\rceil+1)}{\log(1+c')}\Big\rfloor+2$ and define a sequence of sets $S_1, S_2, \ldots S_\ell$ as in Appendix \ref{sec:diag_without_ancilla_expander} (Eqs.~\eqref{eq:s1} and~\eqref{eq:siplus1}), i.e. 
\begin{enumerate}
    \item For some constant $c'>0$, choose arbitrary set $S_1$ of size $\lceil 1/c'\rceil+1$;
    \item For every $2\le i\le \ell$, $S_i=S_{i-1}\cup \Gamma(S_{i-1})$, where $\Gamma(S_{i-1})\subset V-{S_{i-1}}$ consists of $\lfloor c'|S_{i-1}|\rfloor$ vertices. The size of a maximum matching $M_{S_{i-1}} $between $S_{i-1}$ and $\Gamma(S_{i-1})$ is $\lfloor c'|S_{i-1}|\rfloor$.
\end{enumerate}
By construction, $|S_{\ell-1}|\le m/2$, $|S_{\ell-1}|=\Theta(m)$, $|\Gamma(S_{\ell-1})|=\Theta(m)$.  We take
\begin{itemize}
    \item ${\sf R}_{\rm copy}:= S_{\ell-1}$;
    \item ${\sf R}_{\rm targ}:= \Gamma(S_{\ell-1})$;
    \item ${\sf R}_{\rm inp} \subseteq V-({\sf R}_{\rm copy}\cup {\sf R}_{\rm targ})=V-S_{\ell}$.
\end{itemize}
The copy and target registers have sizes $\lambda_{copy}=\Theta(m)$ and $\lambda_{targ}=\Theta(m)$, respectively, while ${\sf R}_{\rm inp}$ consists of $n$ qubits in $V-S_{\ell}$. We define $p=\log(\lambda_{targ})$.

Our choice of Gray codes is given by setting $\ell_k = 1$, for all $k\in[2^p]$.

\paragraph{Remark} Note that, once $S_1, \ldots, S_\ell$ have been constructed, it may not be the case that the $n$ input qubits (which have been loaded with non-zero inputs $\ket{x}$) lie entirely within $V-S_\ell$. However, by using at most $n$ SWAP operations (that may across some distance under $G$ constraint), we can permute the input qubits so that they do lie within $V-S_\ell$, and we can then take the locations of those qubits to define ${\sf R}_{\rm inp}$.  By Lemma~\ref{lem:distance} the distance between the two qubits in any of these SWAP gates is $O(\log(n+m))$, and each SWAP can be implemented by three CNOT gates. By Lemma~\ref{lem:cnot_path_constraint}, permuting all input qubits into $V-S_\ell$ can be implemented in circuit depth $n\cdot O(\log(n+m))=O(n\log(n+m))$. We shall see that this does not impact the final circuit depth complexity required to implement $\Lambda_n$.
 

\paragraph{Implementation of $\Lambda_n$} 

We assume $m\ge \Omega(n)$. If $m\le o(n)$, the circuit depth in this section is larger than the depth in Lemma \ref{lem:diag_expander_withoutancilla}, which does not use ancillary qubits.
\diagexpanderancilla*
\begin{proof}
We assume that the number of ancillary qubits $m\le O(2^n)$. Let ${\sf R}_{{\rm inp},k}$ denote the $k$-th qubit of the input register,  $\ket{x}_{{\sf R}_{\rm inp}}=\bigotimes_{k=1}^n\ket{x_k}_{{\sf R}_{{\rm inp},k}}$, and define $s(j,k)$, $f_{j,k}$ and $\ket{f_j}$ as in Eq.~\eqref{eq:s,f}.

Let $U^k_{copy}$ be a transformation which makes copies of $\ket{x_k}_{{\sf R}_{{\rm inp},k}}$ in ${\sf R}_{\rm copy}=S_{\ell-1}$, i.e.,
\begin{equation}\label{eq:copy_expander_graph}
    \ket{x_k}_{{\sf R}_{{\rm inp},k}}\ket{0^{|S_{\ell-1}|}}_{S_{\ell-1}}\xrightarrow{ U_{copy}^k } \ket{x_k}_{{\sf R}_{{\rm inp},k}}\underbrace{\ket{x_k\cdots x_k}_{S_{\ell-1}}}_{|S_{\ell-1}|\text{~copies~of~}x_k} 
\end{equation}
which can be realized in $\ell-1$ steps: 
\begin{enumerate}
    \item Step 1: make $|S_1|$ copies of $x_k$ from ${\sf R}_{{\rm inp},k}$ to $S_1$, i.e.,
    \begin{equation*}
        \ket{x_k}_{{\sf R}_{{\rm inp},k}}\ket{0^{|S_1|}}_{S_1}\to \ket{x_k}_{{\sf R}_{{\rm inp},k}}\underbrace{\ket{x_k\cdots x_k}_{S_{1}}}_{|S_{1}|\text{~copies~of~}x_k}
    \end{equation*}
    This can be implemented by applying $\lceil 1/c'\rceil+1$ CNOT gates, with each CNOT gate having a separate qubit in $S_1$ as target, and control qubit ${\sf R}_{{\rm inp},k}$. By Lemmas~\ref{lem:distance} and~\ref{lem:cnot_path_constraint}, Step 1 can be realized in depth $|S_1|\cdot O(\log(n+m))=O(\log(n+m))$.
    
    \item Step $2\le i\le \ell-1$: make copies of $x_k$ from $S_{i-1}$ to $\Gamma(S_{i-1})$, i.e.,
    \begin{equation}\label{eq:copy_expander_graph_i}
        \underbrace{\ket{x_k\cdots x_k}_{S_{i-1}}}_{|S_{i-1}|\text{~copies~of~}x_k}\ket{0^{\lfloor c'|S_{i-1}|\rfloor}}_{\Gamma(S_{i-1})}\to \ket{x_k\cdots x_k}_{S_{i-1}}\underbrace{\ket{x_k\cdots x_k}_{\Gamma(S_{i-1})}}_{\lfloor c'|S_{i-1}|\rfloor\text{~copies~of~}x_k}=\underbrace{\ket{x_k\cdots x_k}_{S_{i}}}_{|S_{i}|\text{~copies~of~}x_k}, \forall x_k\in \B.
    \end{equation}
    By the construction of $S_1,S_2,\ldots,S_\ell$, there exists a maximum matching $M_{S_{i-1}}$  between $S_{i-1}$ and $\Gamma(S_{i-1})$ of size $\lfloor c'|S_{i-1}|\rfloor$.
    Eq. \eqref{eq:copy_expander_graph_i} can be implemented by applying CNOT gates to all pairs of qubits $(u,v)$ corresponding to edges in $M_{S_{i-1}}$. Each of these can be implemented in parallel, and thus the total depth required is $1$. 
\end{enumerate}
The total depth required to implement Eq.~\eqref{eq:copy_expander_graph} is therefore $O(\log(n+m))+\ell-2=O(\log(m+n))$.

We now consider the circuit construction for $\Lambda_n$, which we implement in 3 stages:
%

\begin{enumerate}
    \item Gray initial stage:
    \begin{equation}\label{eq:grayinitial_expander}
        \ket{x}_{{\sf R}_{\rm inp}}\ket{0^{|S_{\ell-1}|}}_{S_{\ell-1}}\ket{0^{\lfloor c'|S_{\ell-1}|\rfloor}}_{\Gamma(S_{\ell-1})}\to \ket{x}_{{\sf R}_{\rm inp}}\ket{0^{|S_{\ell-1}|}}_{S_{\ell-1}}\ket{f_1}_{\Gamma(S_{\ell-1})}.
    \end{equation}
      This can be realized in $p$ steps by handling the $p$ suffix bits one by one. For all $j\in[p]$, the $j$-th step is implemented as follows:
      \begin{enumerate}
          \item First, we make $|S_{\ell-1}|$ copies of $x_{n-p+j}$ in copy register $S_{\ell-1}$ by the implementation of Eq. \eqref{eq:copy_expander_graph}, i.e.,
    \begin{equation}\label{eq:copy_expander_phase1}
     \ket{x}_{{\sf R}_{\rm inp}}\ket{0^{|S_{\ell-1}|}}_{S_{\ell-1}}\to \ket{x}_{{\sf R}_{\rm inp}}\underbrace{\ket{x_{n-p+j}x_{n-p+j}\cdots x_{n-p+j}}_{S_{\ell-1}}}_{|S_{\ell-1}|~\text{~copies~of~}x_{n-p+j}}
     \end{equation}
      This requires depth $O(\log(n+m))$.
     \item Second, for all $k\in[2^p]$, if $f_{j,k}=\langle s(1,k),x\rangle$ (viewed as a linear function of the variables $x_i$) contains $x_{n-p+j}$, we add $x_{n-p+j}$ to the $k$-th qubit of target register $\Gamma(S_{\ell-1})$. This can be implemented by applying a CNOT gate of which the control qubit is $\ket{x_{n-p+j}}$ in $S_{\ell-1}$ and the target qubit is the $k$-th qubit of $\Gamma(S_{\ell-1})$. Since there exists a $\lfloor c'|S_{\ell-1}|\rfloor$-size matching between $S_{\ell-1}$ and $\Gamma(S_{\ell-1})$ and each qubit in $S_{\ell-1}$ contains a copy of $x_{n-p+j}$, all the CNOT gates can be applied in parallel, and the required circuit depth is $1$. 
     \item Third, we apply the inverse circuit of Eq. \eqref{eq:copy_expander_phase1} of depth $O(\log(n+m))$ to restore the copy register.
      \end{enumerate}
    In total, the circuit depth for the Gray initial stage is $O(p\log(n+m))$.
    \item Gray cycle stage:
    \begin{equation}\label{eq:graycycle_expander}
        \ket{x}_{{\sf R}_{\rm inp}}\ket{0^{|S_{\ell-1}|}}\ket{f_1}_{\Gamma(S_{\ell-1})} \to e^{i\theta(x)}\ket{x}_{{\sf R}_{\rm inp}}\ket{0^{|S_{\ell-1}|}}\ket{f_1}_{\Gamma(S_{\ell-1})}.
    \end{equation}
    We implement this in $2^{n-p}$ steps. For $j\le 2^{n-p}-1$, the $j$-th step is defined as
    \begin{equation}
     \ket{x}_{{\sf R}_{\rm inp}}\ket{0^{|S_{\ell-1}|}}_{S_{\ell-1}}\ket{f_j}_{\Gamma(S_{\ell-1})}\to e^{i\sum\limits_{k\in[2^p]}f(j+1,k)\alpha_{s(j+1,k)}}\ket{x}_{{\sf R}_{\rm inp}}\ket{0^{|S_{\ell-1}|}}_{S_{\ell-1}}\ket{f_{j+1}}_{\Gamma(S_{\ell-1})},\forall x\in\Bn.
    \end{equation}
    Recall that $s(j,k)$ and $s(j+1,k)$ differ in the $h_{1,j+1}$-th bit.
    \begin{enumerate}
        \item First, we make $|S_{\ell-1}|$ copies of $x_{h_{1,j+1}}$ in $S_{\ell-1}$. This can be done in depth $O(\log(n+m))$ using $U_{copy}^{h_{1,j+1}}$ (Eq. \eqref{eq:copy_expander_graph}).
        
        \item Second, we add $x_{h_{1,j+1}}$ to every qubit of $\Gamma(S_{\ell-1})$, by applying CNOT gates to all qubit pairs $(u,v)$ corresponding to edges in $M_{S_{\ell-1}}$. This can be done in depth 1. 
        
        \item Third, for all $k\in[2^p]$, we apply $R(\alpha_{s(j+1,k)})$ on the $k$-th qubit of $\Gamma(S_{\ell-1})$, where $\alpha_{s(j+1,k)}\in\mathbb{R}$ is defined in Eq. \eqref{eq:alpha}. This can be done in depth 1. 
        
        \item Finally, we restore the copy register using the inverse of $U_{copy}^{h_{1,j+1}}$.
    \end{enumerate} 
    The total depth of the $j$-th step is $O(\log(n+m))$. The last step, the $2^{n-p}$-th step, is defined as
    \begin{equation}
      \ket{x}_{{\sf R}_{\rm inp}}\ket{0^{|S_{\ell-1}|}}_{S_{\ell-1}}\ket{f_{2^{n-p}}}_{\Gamma(S_{\ell-1})}\to e^{i\sum\limits_{k\in[2^p]}f(1,k)\alpha_{s(1,k)}}\ket{x}_{{\sf R}_{\rm inp}}\ket{0^{|S_{\ell-1}|}}_{S_{\ell-1}}\ket{f_{1}}_{\Gamma(S_{\ell-1})},\forall x\in\Bn.
    \end{equation}
     By a similar discussion, this can be implemented in depth $O(\log(n+m))$. 
    In summary, the total depth of the Gray Cycle stage is $O(2^{n-p}\log(n+m))$.
    \item Inverse stage:
    \begin{equation}\label{eq:inverse_expander}
        \ket{x}_{{\sf R}_{\rm inp}}\ket{0^{|S_{\ell-1}|}}_{S_{\ell-1}}\ket{f_1}_{\Gamma(S_{\ell-1})}\to\ket{x}_{{\sf R}_{\rm inp}}\ket{0^{|S_{\ell-1}|}}_{S_{\ell-1}}\ket{0^{\lfloor c'|S_{\ell-1}|\rfloor}}_{\Gamma(S_{\ell-1})}.
    \end{equation}
    This can be implemented by the inverse of Eq. \eqref{eq:grayinitial_expander}.
\end{enumerate}
The total depth required to implement $\Lambda_n$ is thus $2O(p\log(n+m))+O(2^{n-p}\log(n+m))=O(\log(n+m)\log(m)+\frac{\log(m)2^n}{m})=O(n^2+\frac{\log(m)2^n}{m})$ for $m\le O(2^n)$. If $m\ge\omega(2^n)$, we only use $O(2^n)$ of the ancillary qubits, and the circuit depth is $O(n^2)$. Therefore, the total depth is $O\Big(n^2+\frac{\log(m)2^n}{m}\Big)$.
\end{proof}

\section{Circuit constructions for QSP and GUS under qubit connectivity constraints}
\label{append:QSP_US_graph}
In this section, we bound the circuit size and depth for quantum state preparation (QSP) and general unitary synthesis (GUS) under different graph constraints, based on the circuit constructions for diagonal unitary matrices in Appendix \ref{append:diag_without_ancilla} and Appendix \ref{append:diag_with_ancilla}. In Appendix \ref{sec:QSP_graph} and Appendix \ref{sec:US_graph}, we present QSP and GUS circuits under path, $d$-dimensional grid, binary tree, expander graph and general graph constraints. In Appendix \ref{sec:circuit_transformation}, we present a transformation between circuits under different graph constraints, which we use to upper bound the circuit depth for QSP and GUS under brick-wall constraint. 


\subsection{Circuit complexity for QSP under graph constraints (Proofs of Theorems \ref{thm:QSP_grid_main} -  \ref{thm:QSP_graph})}
\label{sec:QSP_graph}

\subsubsection{QSP under $\Path_{n+m}$ and $\Grid_{n+m}^{n_1,n_2,\ldots,n_d}$ constraints}

\label{sec:QSP_path_grid}

The results of this section are based on the fact that (i) every $j$-qubit uniformly controlled gate (UCG) $V_j$ can be decomposed into 3 $j$-qubit diagonal unitary matrices and 4 single-qubit gates (Lemma~\ref{lem:UCG_decomposition}); and (ii) any QSP circuit can be decomposed into a sequence of UCGs $V_1,V_2,\ldots, V_n$: 

\begin{lemma}[\cite{grover2002creating,kerenidis2017quantum}]\label{lem:QSP_framework_UCG}
The QSP problem can be solved by $n$ UCGs acting on $1,2,\ldots, n$ qubits, respectively, 
\[V_n (V_{n-1}\otimes \mathbb{I}_1) \cdots (V_2\otimes \mathbb{I}_{n-2})(V_1\otimes\mathbb{I}_{n-1}),\]
by the circuit in Fig.~\ref{fig:QSP_circuit}.
\end{lemma}

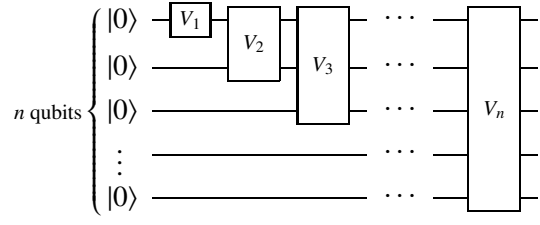
\begin{figure}[]
\centerline{
\Qcircuit @C=0.6em @R=0.6em {
\lstick{\ket{0}} & \gate{\scriptstyle V_1} & \multigate{1}{\scriptstyle V_2} & \multigate{2}{\scriptstyle V_3} & \qw & \push{\cdots} & &  \qw & \multigate{4}{\scriptstyle V_n} & \qw\\
\lstick{\ket{0}} & \qw & \ghost{\scriptstyle V_2} & \ghost{\scriptstyle V_3} & \qw &\push{\cdots} & & \qw & \ghost{\scriptstyle V_n} & \qw\\
\lstick{\ket{0}} & \qw & \qw & \ghost{\scriptstyle V_3} & \qw & \push{\cdots} &  & \qw & \ghost{\scriptstyle V_n} & \qw\\
\vdots~~~~~~~~~ & \qw & \qw  & \qw  & \qw & \push{\cdots} &  & \qw & \ghost{\scriptstyle V_n}& \qw\\
\lstick{\ket{0}} & \qw & \qw & \qw & \qw & \push{\cdots} &  & \qw & \ghost{\scriptstyle V_n} & \qw \inputgroupv{2}{4}{4em}{1.6em}{ \scriptstyle n \text{~qubits}~~~~~~~~~~~~~~~}\\
}
}
    \caption{A QSP circuit to prepare an $n$-qubit state. Every $V_j$ is a $j$-qubit uniformly controlled gate (UCG) for $j\in[n]$, where the first $j-1$ qubits are control qubits and the last qubit is the target qubit. }
    \label{fig:QSP_circuit}
\end{figure}
\begin{lemma}\label{lem:UCG_grid}
Any $n$-qubit UCG $V_n$ can be realized by a quantum circuit of depth
\begin{equation*}
    O\Big(n^2+d2^{\frac{n}{d+1}}+\max_{j\in\{2,\ldots,d\}}\Big\{\frac{d2^{n/j}}{(\Pi_{i=j}^d n_i)^{1/j}}\Big\}+\frac{2^n}{n+m}\Big),
\end{equation*}
under $\Grid_{n+m}^{n_1,n_2,\ldots,n_d}$ constraint, using $m\ge 0$ ancillary qubits. 
If $n_1=n_2=\cdots=n_d$, the depth is $O\left(n^2+d2^{\frac{n}{d+1}}+\frac{2^n}{n+m}\right)$.
\end{lemma}
\begin{proof}
Any $n$-qubit  $\Lambda_n$ can be implemented in depth
\[
\begin{cases}
    O(2^n/n), & \text{ if $m=0$, \qquad (\text{Lemma}~\ref{lem:diag_grid_withoutancilla_main})}\\
    O\Big(n^2+d2^{\frac{n}{d+1}}+\max\limits_{j\in\{2,\ldots,d\}}\Big\{\frac{d2^{n/j}}{(\Pi_{i=j}^d n_i)^{1/j}}\Big\}+\frac{2^n}{n+m}\Big) & \text{ if $m\ge 3n$. \qquad (\text{Lemma}~\ref{lem:diag_grid_ancillary})}
\end{cases}
\]
under $\Grid_{n+m}^{n_1,n_2,\ldots,n_d}$ constraint. If $0 < m < 3n$ we do not use the ancillary qubits.
The result follows from Lemma~\ref{lem:UCG_decomposition}. The case where $n_i=(n+m)^{1/d}$ is also dealt with in Lemma~\ref{lem:diag_grid_ancillary}. 
 
\end{proof}

\begin{theorem}[Theorem \ref{thm:QSP_grid_main} (Case 3)]\label{thm:QSP_grid}
Any $n$-qubit quantum state can be prepared by a quantum circuit of depth
\begin{equation*}
   O\Big(n^3+d2^{\frac{n}{d+1}}+\max\limits_{j\in\{2,\ldots,d\}}\Big\{\frac{d2^{n/j}}{(\Pi_{i=j}^d n_i)^{1/j}}\Big\}+\frac{2^n}{n+m}\Big),
\end{equation*} 
under $\Grid_{n+m}^{n_1,n_2,\ldots,n_d}$ constraint, using $m \ge 0$ ancillary qubits. 
If $n_1=n_2=\cdots=n_d$, the depth is $O\left(n^3+d2^{\frac{n}{d+1}}+\frac{2^n}{n+m}\right)$.
\end{theorem}
\begin{proof}
By Lemma~\ref{lem:QSP_framework_UCG}, any $n$-qubit QSP circuit can be decomposed into $n$ UCGs, $V_1,V_2,\ldots,V_n$ of growing size. Combined with Lemma~\ref{lem:UCG_grid}, this gives a circuit depth upper bound of
\[\sum_{k=1}^{n}O\Big(k^2+d2^{\frac{k}{d+1}}+\max_{j\in\{2,\ldots,d\}}\Big\{\frac{d2^{k/j}}{(\Pi_{i=j}^d n_i)^{1/j}}\Big\}+\frac{2^k}{k+m}\Big)=O\Big(n^3+d2^{\frac{n}{d+1}}+\max_{j\in\{2,\ldots,d\}}\Big\{\frac{d2^{n/j}}{(\Pi_{i=j}^d n_i)^{1/j}}\Big\}+\frac{2^n}{n+m}\Big).\]
\end{proof}

\begin{corollary}[Theorem \ref{thm:QSP_grid_main} (Case 1, 2)]\label{coro:QSP_path_grid}
Any $n$-qubit quantum state can be prepared by a circuit with $m\ge 0$ ancillary qubits, of depth 
\begin{enumerate}
    \item $O\left(2^{n/2}+\frac{2^n}{n+m}\right)$ under $\Path_{n+m}$ constraint.
    \item $O\left(2^{n/3}+\frac{2^{n/2}}{(n_2)^{1/2}}+\frac{2^n}{n+m}\right)$ under $\Grid^{n_1,n_2}_{n+m}$ constraint.
    \item $O\left(2^{n/4}+\frac{2^{n/2}}{(n_2n_3)^{1/2}}+\frac{2^{n/3}}{(n_3)^{1/3}}+\frac{2^n}{n+m}\right)$ under $\Grid^{n_1,n_2,n_3}_{n+m}$ constraint.
\end{enumerate}
\end{corollary}
Note that the $\Path_{n+m}$ result holds by setting $d=1$ in Theorem \ref{thm:QSP_grid}. 

\subsubsection{QSP under $\Expander_{n+m}$ constraints}
\label{sec:QSP_expander}
\begin{lemma}\label{lem:UCG_expander}
Any $n$-qubit UCG $V_n$ can be realized by a quantum circuit of depth
\begin{equation*}
    O\left(n^2+\frac{\log(n+m)2^n}{n+m}\right)
\end{equation*}
under $\Expander_{n+m}$ constraint, using $m\ge 0$ ancillary qubits.
\end{lemma}
\begin{proof}
Any $n$-qubit diagonal unitary $\Lambda_n$ can be implemented in depth
\[
\begin{cases}
    O(\log(n)2^n/n), & \text{if $m=0$, \qquad (\text{Lemma}~\ref{lem:diag_expander_withoutancilla})}\\
    O\left(n^2+\frac{\log(m)2^n}{m}\right), &\text{ if  $m\ge \Omega(n)$. \quad (\text{Lemma}~\ref{lem:diag_expander_ancilla})}
\end{cases}
\]
under $\Expander_{n+m}$ constraint. If $0 < m < O(n)$ we do not use the ancillary qubits. This result follows from Lemma \ref{lem:UCG_decomposition}.
\end{proof}

\qspexpander*
\begin{proof}
By Lemma~\ref{lem:QSP_framework_UCG}, any $n$-qubit QSP circuit can be decomposed into $n$ UCGs, $V_1,V_2,\ldots,V_n$ of growing size. Combined with Lemma~\ref{lem:UCG_expander}, this gives a circuit depth upper bound of
\[\sum_{k=1}^{n}O\left(k^2+\frac{\log(k+m)2^k}{k+m}\right)=O\left(n^3+\frac{\log(n+m)2^n}{n+m}\right).\]
\end{proof}

\subsubsection{QSP under general graph $G$ constraints}
\label{sec:QSP_graph_general}
\begin{lemma}\label{lem:UCG_graph}
Any $n$-qubit UCG $V_n$ can be implemented by a quantum circuit of size and depth $O(2^n)$ under arbitrary graph constraint, using no ancillary qubits.
\end{lemma}
\begin{proof}
Follows directly from Lemmas~\ref{lem:UCG_decomposition} and \ref{lem:diag_graph_withoutancilla}.
\end{proof}

\qspgraph*
\begin{proof}
For every $i\in[n]$, UCG $V_i$ acts on $i$ qubits in $G$. We first swap the locations of these $i$ qubits such that they lie in a connected subgraph of $G$ with $i$ vertices. $V_i$ can then be implemented in depth $O(2^i)$ by Lemma \ref{lem:UCG_graph}, and the qubits then swapped back to their original positions. The process of swapping and unswapping the qubits can be realized by a CNOT circuit of size and depth $O(n^2)$ by Lemma \ref{lem:cnot_circuit}. The total depth and size of to implement the QSP circuit is $\sum_{i=1}^n (O(2^i)+O(n^2))=O(2^n)$. 
\end{proof}

\subsubsection{QSP under $\Tree_{n+m}(2)$ and $\Tree_{n+m}(d)$ constraints}
\label{sec:QSP_binarytree}
\paragraph{QSP under $\Tree_{n+m}(2)$ constraints}
\begin{lemma}\label{lem:UCG_binarytree}
Any $n$-qubit UCG $V_n$ can be realized by a quantum circuit of depth
\begin{equation*}
O\left(n^2\log(n)+\frac{\log(n)2^n}{n+m}\right)    
\end{equation*} 
under $\Tree_{n+m}(2)$ constraint, using $m\ge 0$ ancillary qubits.
\end{lemma}
\begin{proof}
Follows directly from Lemmas~\ref{lem:UCG_decomposition}, \ref{lem:diag_bianrytree_withoutancilla} and \ref{lem:diag_binarytree_withancilla}.
\end{proof}

\begin{theorem}[Theorem \ref{thm:QSP_tree} (Case 1)]\label{thm:QSP_binarytree}
Any $n$-qubit quantum state can be realized by a  quantum circuit of depth 
\begin{equation*}
O\left(n^3\log(n)+\frac{\log(n)2^n}{n+m}\right)    
\end{equation*}
under $\Tree_{n+m}(2)$ constraint, using $m\ge 0$ ancillary qubits.
\end{theorem}
\begin{proof}
By Lemma~\ref{lem:QSP_framework_UCG}, any $n$-qubit QSP circuit can be decomposed into $n$ UCGs, $V_1,V_2,\ldots,V_n$ of growing size. Combined with Lemma~\ref{lem:UCG_binarytree}, this gives a circuit depth upper bound of
\[\sum_{k=1}^{n}O\left(k^2\log(k)+\frac{\log(k)2^k}{k+m}\right)=O\left(n^3\log(n)+\frac{\log(n)2^n}{n+m}\right).\]
\end{proof}

In Appendix~\ref{append:binary_tree_improvement} we show that using a unary encoding for the QSP circuit and a different circuit framework, the circuit depth in Theorem \ref{thm:QSP_binarytree} can be improved to $O\left(n^2\log^2(n)+\frac{\log(n)2^n}{n+m}\right)$ if $m\le o(2^n)$, and $O(n^2\log(n))$ if $m\ge \Omega(2^n)$.

\paragraph{QSP under $\Tree_{n+m}(d)$ and $\Star_{n+m}$ constraints}

\begin{theorem}[Theorem \ref{thm:QSP_tree} (Case 2)]\label{thm:QSP_darytree}
Any $n$-qubit quantum state can be realized by a  quantum circuit of depth \[O\left(n^2d\log_d (n+m)\log_d(n+d)+\frac{(n+d)\log_d(n+d) 2^{n}}{n+m}\right),\]
under $\Tree_{n+m}(d)$ constraint, using $m\ge 0$ ancillary qubits, for $d< n+m$.
\end{theorem}
\begin{proof}
Lemmas \ref{lem:diag_d_tree_noancilla} and \ref{lem:diag_d_tree_ancilla} show that any $k$-qubit diagonal matrix can be implemented by a quantum circuit of depth $O\left(kd\log_d (k+m)\log_d(k+d)+\frac{(k+d)\log_d(k+d) 2^{k}}{k+m}\right)$ under $\Tree_{n+m}(d)$ constraint. By the frameworks of Lemma ~\ref{lem:UCG_decomposition} and Lemma~\ref{lem:QSP_framework_UCG}, any $n$-qubit quantum state can be prepared by a circuit of depth 
\begin{align*}
    & \sum_{k=1}^n O\left(kd\log_d (k+m)\log_d(k+d)+\frac{(k+d)\log_d(k+d) 2^{k}}{k+m}\right)\\
    = & O\left(n^2d\log_d (n+m)\log_d(n+d)+\frac{(n+d)\log_d(n+d) 2^{n}}{n+m}\right).
\end{align*}
\end{proof}
\begin{corollary}[Theorem \ref{thm:QSP_tree} (Case 3)]\label{coro:QSP_star}
Any $n$-qubit quantum state can be realized by a  quantum circuit of depth $O\left(2^n\right)$
under $\Star_{n+m}$ constraint, using $m\ge 0$ ancillary qubits.
\end{corollary}
\begin{proof}
Do not use the ancillary qubits. The result follows Theorem~\ref{thm:QSP_graph}. 
\end{proof}

\subsection{Circuit complexity for GUS under graph constraints (Proof of Theorems \ref{thm:US_path_grid} - \ref{thm:US_graph})}
\label{sec:US_graph}

\begin{lemma}[\cite{mottonen2005decompositions}]\label{lem:unitary_decomposition}
Any $n$-qubit unitary matrix $U\in\mathbb{C}^{2^{n}\times 2^n}$ can be decomposed into $2^{n}-1$ $n$-qubit UCGs.
\end{lemma}
Note that the target qubit of the UCGs in Lemma~\ref{lem:unitary_decomposition} may be arbitrary, which generalizes the UCGs in Eq. \eqref{eq:UCG} for which the target is always the $n$-th qubit. 

\begin{theorem}[Theorem \ref{thm:US_path_grid} (Case 3)]\label{thm:US_grid}
Any $n$-qubit unitary can be realized by a quantum circuit of depth \[O\Big(n^22^n+d4^{\frac{(d+2)n}{2(d+1)}}+\max_{j\in\{2,\ldots,d\}}\Big\{\frac{d4^{(j+1)n/(2j)}}{(\Pi_{i=j}^d n_i)^{1/j}}\Big\}+\frac{4^n}{n+m}\Big)\] 
under $\Grid^{n_1,n_2,\ldots,n_d}_{n+m}$ constraint, using $m\ge 0$ ancillary qubits.  
When $n_1=n_2=\cdots=n_d$, the depth is $O\left(n^22^n+d4^{\frac{(d+2)n}{2(d+1)}}+\frac{4^n}{n+m}\right)$.
\end{theorem}
\begin{proof}
Follows from Lemmas~\ref{lem:UCG_grid} and \ref{lem:unitary_decomposition}.
\end{proof}

\begin{corollary}[Theorem \ref{thm:US_path_grid} (Case 1, 2)]\label{coro:US_path_grid}
Any $n$-qubit unitary can be realized by a quantum circuit with $m\ge 0$ ancillary qubits, of depth
\begin{enumerate}
    \item $O\left(4^{3n/4}+\frac{4^n}{n+m}\right)$ under $\Path_{n+m}$ constraint.
    \item $O\left(4^{2n/3}+\frac{4^{3n/4}}{(n_2)^{1/2}}+\frac{4^n}{n+m}\right)$ under $\Grid^{n_1,n_2}_{n+m}$ constraint.
    \item $O\left(4^{5n/8}+\frac{4^{3n/4}}{(n_2n_3)^{1/2}}+\frac{4^{2n/3}}{(n_3)^{1/3}}+\frac{4^n}{n+m}\right)$ under $\Grid^{n_1,n_2,n_3}_{n+m}$ constraint.
\end{enumerate}
\end{corollary}

Note that the case for $\Path_{n+m}$ follows from choosing $d=1$ in  Theorem \ref{thm:US_grid}. 

\begin{theorem}[Theorem \ref{thm:US_tree} (Case 1)]\label{thm:US_binarytree}
Any $n$-qubit unitary can be realized by a quantum circuit of depth $O\left(n^2\log(n)2^n+\frac{\log(n)4^n}{n+m}\right)$ under $\Tree_{n+m}(2)$ constraint, using $m\ge 0$ ancillary qubits.
\end{theorem}
\begin{proof}
Follows from Lemmas~\ref{lem:UCG_binarytree} and \ref{lem:unitary_decomposition}.
\end{proof}

\begin{theorem}[Theorem \ref{thm:US_tree} (Case 2)]\label{thm:US_darytree}
Any $n$-qubit unitary can be realized by a quantum circuit of depth \[O\left(n2^nd\log_d (n+m)\log_d(n+d)+\frac{(n+d)\log_d(n+d) 4^{n}}{n+m}\right)\]
under $\Tree_{n+m}(d)$ constraint, using $m\ge 0$ ancillary qubits.
\end{theorem}
\begin{proof}
Lemmas \ref{lem:diag_d_tree_noancilla} and \ref{lem:diag_d_tree_ancilla} show that any $n$-qubit diagonal unitary matrix can be implemented by a quantum circuit of depth $O\left(nd\log_d (n+m)\log_d(n+d)+\frac{(n+d)\log_d(n+d) 2^{n}}{n+m}\right)$ under $\Tree_{n+m}(d)$ constraint, using $m$ ancillary qubits. By Lemma \ref{lem:UCG_decomposition} and Lemma \ref{lem:unitary_decomposition}, the total circuit depth required is therefore 
\begin{align*}
    & O\left(nd\log_d (n+m)\log_d(n+d)+\frac{(n+d)\log_d(n+d) 2^{n}}{n+m}\right)\cdot O(2^n)\\
    = & O\left(n2^nd\log_d (n+m)\log_d(n+d)+\frac{(n+d)\log_d(n+d) 4^{n}}{n+m}\right).
\end{align*}
\end{proof}

\usexpander*
\begin{proof}
Follows from Lemmas~\ref{lem:UCG_expander} and \ref{lem:unitary_decomposition}.
\end{proof}

\usgraph*
\begin{proof}
Follows from Lemmas~\ref{lem:UCG_graph} and \ref{lem:unitary_decomposition}.
\end{proof}

\begin{corollary}[Theorem \ref{thm:US_tree} (Case 3)]\label{coro:US_star}
Any $n$-qubit unitary matrix can be realized by a quantum circuit of depth $O\left(4^n\right)$
under $\Star_{n+m}$ constraint, using $m\ge 0$ ancillary qubits.
\end{corollary}
\begin{proof}
Do not use the ancilla. The result follows from Theorem \ref{thm:US_graph}.
\end{proof}

\subsection{Circuit transformation between different graph constraints }
\label{sec:circuit_transformation}
We first show a transformation between circuits under different graph constraints. 
\circuittrans*

\begin{proof}
We say that a CNOT gate acts on $e=(v_s,v_t)$, if it acts on qubit $v_s$ and $v_t$. In $\mathcal{C}'$ there are $d$ layers of gates, each layer $C_k'$ ($k\in [d]$) can be represented as $C'_k = C'_k(E)\otimes (\otimes_{i=1}^c C'_k(E_i)) $, where $C'_k(E)$ consists of single-qubit gates and CNOT gates in $C'$ acting on edges in $E$, and $C'_k(E_i)$ consists of CNOT gates in $C'$ acting on edges in $E_i$. 

Since $E$ is the edge set of $G$, the circuit $C_k'(E)$ can be realized by a circuit of depth $1$ under $G$ constraint. By assumption, for each $e=(v_s,v_t)\in E_i$, there exists a path from $v_s$ to $v_t$ of length at most $c'$ in $G$, and these paths for different edges $e\in E_i$ are disjoint. Thus, all CNOT gates in $C_k'(E_i)$ can be implemented in parallel, in depth and size $O(c')$. 
A depth-$1$ CNOT circuit under $G'$ constraint can thus be realized by a CNOT circuit of depth $1+O(c')\cdot c = O(cc')$ under $G$ constraint.  As every CNOT gate in $\mathcal{C}'$ can be realized in size $O(c')$ under path constraint in $G$, the total size of $C$ is $O(c')\times s=O(c's)$.


\end{proof}

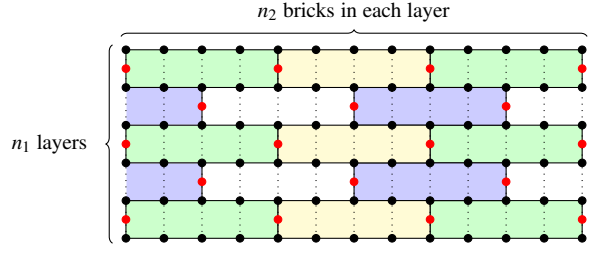
\begin{figure}[]
    \centering
    \begin{tikzpicture}
       \filldraw[fill=green!20] (0,0)--(2,0)--(2,0.5)--(0,0.5)--cycle (4,0)--(6,0)--(6,0.5)--(4,0.5)--cycle (0,1)--(2,1)--(2,1.5)--(0,1.5)--cycle (4,1)--(6,1)--(6,1.5)--(4,1.5)--cycle (0,2)--(2,2)--(2,2.5)--(0,2.5)--cycle (4,2)--(6,2)--(6,2.5)--(4,2.5)--cycle;
       \filldraw[fill=blue!20,draw=white] (0,0.5)--(1,0.5)--(1,1)--(0,1)--cycle  (0,1.5)--(1,1.5)--(1,2)--(0,2)--cycle ;
       \filldraw[fill=blue!20] (3,0.5)--(5,0.5)--(5,1)--(3,1)--cycle  (3,1.5)--(5,1.5)--(5,2)--(3,2)--cycle;
       \draw (0,0.5)--(1,0.5)--(1,1)--(0,1)  (0,1.5)--(1,1.5)--(1,2)--(0,2);
       \filldraw[fill=yellow!20] (2,0)--(4,0)--(4,0.5)--(2,0.5)--cycle (2,1)--(4,1)--(4,1.5)--(2,1.5)--cycle (2,2)--(4,2)--(4,2.5)--(2,2.5)--cycle;
       
       \draw[dotted] (0,0)--(0,2.5) (2,0)--(2,2.5) (4,0)--(4,2.5) (0.5,0)--(0.5,2.5) (2.5,0)--(2.5,2.5) (4.5,0)--(4.5,2.5) (1,0)--(1,2.5) (3,0)--(3,2.5) (5,0)--(5,2.5) (1.5,0)--(1.5,2.5) (3.5,0)--(3.5,2.5) (5.5,0)--(5.5,2.5) (6,0)--(6,2.5);
       \draw [fill=black] (0,0) circle (0.05) (0.5,0) circle (0.05) (1,0) circle (0.05) (1.5,0) circle (0.05) (2,0) circle (0.05) (2.5,0) circle (0.05) (3,0) circle (0.05) (3.5,0) circle (0.05) (4,0) circle (0.05) (4.5,0) circle (0.05) (5,0) circle (0.05) (5.5,0) circle (0.05) (6,0) circle (0.05);
        \draw [fill=black] (0,0.5) circle (0.05) (0.5,0.5) circle (0.05) (1,0.5) circle (0.05) (1.5,0.5) circle (0.05) (2,0.5) circle (0.05) (2.5,0.5) circle (0.05) (3,0.5) circle (0.05) (3.5,0.5) circle (0.05) (4,0.5) circle (0.05) (4.5,0.5) circle (0.05) (5,0.5) circle (0.05) (5.5,0.5) circle (0.05) (6,0.5) circle (0.05);
     \draw [fill=black] (0,1) circle (0.05) (0.5,1) circle (0.05) (1,1) circle (0.05) (1.5,1) circle (0.05) (2,1) circle (0.05) (2.5,1) circle (0.05) (3,1) circle (0.05) (3.5,1) circle (0.05) (4,1) circle (0.05) (4.5,1) circle (0.05) (5,1) circle (0.05) (5.5,1) circle (0.05) (6,1) circle (0.05);
     \draw [fill=black] (0,1.5) circle (0.05) (0.5,1.5) circle (0.05) (1,1.5) circle (0.05) (1.5,1.5) circle (0.05) (2,1.5) circle (0.05) (2.5,1.5) circle (0.05) (3,1.5) circle (0.05) (3.5,1.5) circle (0.05) (4,1.5) circle (0.05) (4.5,1.5) circle (0.05) (5,1.5) circle (0.05) (5.5,1.5) circle (0.05) (6,1.5) circle (0.05);
    \draw [fill=red,draw=red]   (1,1.75) circle (0.05) (3,1.75) circle (0.05) (5,1.75) circle (0.05);
     \draw [fill=black] (0,2) circle (0.05) (0.5,2) circle (0.05) (1,2) circle (0.05) (1.5,2) circle (0.05) (2,2) circle (0.05) (2.5,2) circle (0.05) (3,2) circle (0.05) (3.5,2) circle (0.05) (4,2) circle (0.05) (4.5,2) circle (0.05) (5,2) circle (0.05) (5.5,2) circle (0.05) (6,2) circle (0.05);
     \draw [fill=black] (0,2.5) circle (0.05) (0.5,2.5) circle (0.05) (1,2.5) circle (0.05) (1.5,2.5) circle (0.05) (2,2.5) circle (0.05) (2.5,2.5) circle (0.05) (3,2.5) circle (0.05) (3.5,2.5) circle (0.05) (4,2.5) circle (0.05) (4.5,2.5) circle (0.05) (5,2.5) circle (0.05) (5.5,2.5) circle (0.05) (6,2.5) circle (0.05);
     
     \node (a) at (-0.2,2.5) {};
     \node (b) at (6.2,2.5) {};
     \draw[decorate,decoration={brace,raise=5pt}] (a) -- (b);
      \draw (3,3) node{\scriptsize $n_2$ bricks in each layer};
    \node (a) at (0,-0.2) {};
     \node (b) at (0,2.7) {};
     \draw[decorate,decoration={brace,raise=5pt}] (a) -- (b);
      \draw (-1,1.25) node{\scriptsize $n_1$ layers};
      
      \draw [fill=red,draw=red]  (1,0.75) circle (0.05) (3,0.75) circle (0.05) (5,0.75) circle (0.05);
    \draw [fill=red,draw=red]   (0,1.25) circle (0.05) (2,1.25) circle (0.05)  (4,1.25) circle (0.05) (6,1.25) circle (0.05);
     \draw [fill=red,draw=red]  (0,2.25) circle (0.05) (2,2.25) circle (0.05)  (4,2.25) circle (0.05) (6,2.25) circle (0.05);
      \draw [fill=red,draw=red] (0,0.25) circle (0.05) (2,0.25) circle (0.05)  (4,0.25) circle (0.05) (6,0.25) circle (0.05);
    \end{tikzpicture}
   \caption{A new graph $G'$ constructed by adding edges to $\brickwall_{n+m}^{n_1,n_2,b_1,b_2}$. The red nodes are removed, and new dotted edges are added. Bricks are divided into $4$ groups, indicated by the green, white, yellow and blue colors. 
   }
   \label{fig:brickwall_grid}
\end{figure}

We use this lemma to obtain QSP and GUS circuits under brick-wall constraint, by reducing that to our 2D grid results. 
\qspbrickwall*
\begin{proof}
In $\brickwall_{n+m}^{n_1,n_2,b_1,b_2}$, each brick has a rectangle containing $b_1$ vertices on each vertical edge. We wish to apply Lemma~\ref{lem:circuit_trans}, by defining the brick-wall as a subgraph of a fully connected grid.  To do so, conceptually we must first `remove' the $b_1-2$ vertices (indicated by the red nodes in Fig.~\ref{fig:brickwall_grid}) in the middle of each vertical edge and add an edge between the remaining two vertices. This is possible because, at the cost of an $O(b_1)$ overhead, we can implement a CNOT between the two remaining vertices along the path between them.  Thus, we can view the brick-wall as a new graph $G=(V,E)$ where the red nodes have been removed, and CNOT gates across the newly added edges cost $O(b_1)$.

From $G=(V,E)$ we construct yet another new graph $G'=(V, E\cup E')$, by adding vertical edges across layers as in Fig.~\ref{fig:brickwall_grid}. We color the bricks in even and odd layers with alternating colors (using four colors total: two colors for each of the even and odd layers) and set $E' = \cup_{i=1}^4 E_i'$ where the four new edge sets $E_i'$ correspond to which brick color the edge lies in. Note that all bricks of the same color are vertex disjoint. We further decompose each $E_i'$ into at most $b_2 -2$ disjoint subsets $E_i$, with each $E_i$ formed by selecting at most $1$ edge from every brick in $E_i'$.  Thus, $E' = \cup_{i=1}^c E_i$ for $c\le 4(b_2-2) = O(1)$.  As each $(u,v)\in E_i$ lies in a separate brick, there exists a path from $u$ to $v$ in $G$ of length at most $c'=b_1+b_2 = O(1)$, and all such paths are disjoint. 

We can now invoke our results for the 2D grid to obtain circuits for the brick-wall. We consider two cases.
\begin{enumerate}
    \item Case 1: $m\ge 2n$. Note that $G'(V, E\cup E')$ is a 2-dimensional grid $\Grid^{n_1+1,n_2b_2-n_2+1}_{(n_1+1)(n_2b_2-n_2+1)}$. By Corollary~\ref{coro:QSP_path_grid} (result 2), the circuit depth required for $n$-qubit QSP is  
\[
O\left(2^{n/3}+\frac{2^{n/2}}{\sqrt{\min\{n_1+1, n_2 b_2-n_2+1\}}}+\frac{2^n}{n+m}\right)
=  O\left(2^{n/3}+\frac{2^{n/2}}{\sqrt{\min\{n_1,n_2\}}}+\frac{2^n}{n+m}\right)
\] 
under $\Grid^{n_1+1,n_2b_2-n_2+1}_{(n_1+1)(n_2b_2-n_2+1)}$ constraint, which translates, via Lemma \ref{lem:circuit_trans} into a circuit depth bound of $O\left(2^{n/3}+\frac{2^{n/2}}{\sqrt{\min\{n_1,n_2\}}}+\frac{2^n}{n+m}\right)$ under graph $G$ constraint. 

By Lemma~\ref{lem:cnot_path_constraint}, in one layer of the above circuit, all CNOT gates acting on the vertical edges in the brick-wall can be implemented in depth $O(b_1)$ simultaneously. Therefore, the circuit depth of $n$-qubit QSP is \[O(b_1)\cdot O\left(2^{n/3}+\frac{2^{n/2}}{\sqrt{\min\{n_1,n_2\}}}+\frac{2^n}{n+m}\right)=O\left(2^{n/3}+\frac{2^{n/2}}{\sqrt{\min\{n_1,n_2\}}}+\frac{2^n}{n+m}\right).\]

    \item Case 2: $m\le 2n$. 
    There exists a Hamiltonian path in $G'$ and thus, by Corollary \ref{coro:QSP_path_grid} (result 1), the circuit depth for $n$-qubit QSP is $O(\frac{2^n}{n+m})$. By Lemma \ref{lem:circuit_trans}, the circuit depth for $n$-qubit QSP is $O(\frac{2^n}{n+m})$ under $\brickwall_{n+m}^{n_1,n_2,b_1,b_2}$ constraint.
\end{enumerate}

\end{proof}
In the general case where $b_1$ and $b_2$ are not necessarily constant, the result still holds, although with an additional factor of $O(b_1b_2(b_1+b_2))$.

By the same argument, we obtain the following depth of GUS circuits under brick-wall constraint.
\usbrickwall*

\section{Circuit size and depth lower bounds under graph constraints}
\label{append:QSP_US_lowerbound}
In this section, we show circuit depth and size lower bounds for QSP, diagonal unitary matrix preparation and GUS under graph constraints. 

\subsection{Circuit lower bounds under general graph constraints}
\label{sec:QSP_US_lowerbound_general}

\paragraph{Size lower bounds}

\begin{lemma}[\cite{shende2004minimal,plesch2011quantum}]\label{lem:lowerbound_size_previous}
There exist $n$-qubit quantum states and $n$-qubit unitaries which can only be implemented by quantum circuits under no graph constraints, of size at least $\Omega\left(2^n\right)$ and $\Omega\left(4^n\right)$ respectively.
\end{lemma}
Since a connectivity graph constraint only adds difficulty, the same lower bounds also hold for any constraint graph $G$.
\begin{proposition}[Theorems \ref{thm:size_lowerbound_QSP_graph} and \ref{thm:size_lowerbound_US_graph}]\label{prop:lowerbound_size_graph}
For any connected graph $G$, there exist $n$-qubit quantum states and $n$-qubit unitaries which require quantum circuits of size at least $\Omega\left(2^n\right)$ and $\Omega\left(4^n\right)$, respectively, under $G$ constraint.
\end{proposition}
\begin{proof}
Let $G=K_n$ be the complete graph on $n$ vertices. The ability to implement any $n$-qubit quantum state and unitary matrix by circuits of size $o(2^n)$ and $o(4^n)$ would contradict Lemma~\ref{lem:lowerbound_size_previous}. 
\end{proof}
The proof of Lemma \ref{lem:lowerbound_size_previous} is by parameter counting, which also applies to diagonal unitaries.
\begin{proposition}\label{prop:size_lowerbound_Lambda}
For any graph $G$, there exist $n$-qubit diagonal unitary matrices which can be implemented by quantum circuits under $G$ constraint, of size at least $\Omega(2^n)$.
\end{proposition}
\begin{proof}
A circuit of size $o(2^n)$ consisting of arbitrary $2$-qubit gates introduces $o(2^n)$ real parameters. On the other hand, a diagonal unitary matrix $\Lambda_n$ is determined by at least $2^n-1$ free real parameters.
\end{proof}

\paragraph{Depth lower bounds}
\begin{lemma}[\cite{sun2021asymptotically}]\label{lem:lowerbound_previous}
There exist $n$-qubit quantum states and $n$-qubit unitaries which can only be implemented by quantum circuits under no graph constraints, of depth  at least $\Omega\left(n+\frac{2^n}{n+m}\right)$ and $\Omega\left(n+\frac{4^n}{n+m}\right)$ respectively, using $m\ge 0$ ancillary qubits.
\end{lemma}
By the same method used to prove Lemma~\ref{lem:lowerbound_previous}, it can be shown that:
\begin{proposition}\label{prop:lowerbound_Lambda_noconstraint}
There exist $n$-qubit diagonal unitary matrices which can only be implemented by quantum circuits of depth at least $\Omega\left(n+\frac{2^n}{n+m}\right)$, under no graph constraints, using $m\ge 0$ ancillary qubits.
\end{proposition}


And again these lower bounds hold under any graph constraint.
\begin{proposition}\label{prop:depth_lowerbound_graph}
Let $G=(V,E)$ denote an arbitrary connected graph with $n+m$ vertices for any $m\ge 0$. There exist $n$-qubit quantum states, diagonal unitary matrices and unitaries which can only be implemented by circuits under $G$ constraint of depth at least $\Omega\left(n+\frac{2^n}{n+m}\right)$, $\Omega\left(n+\frac{2^n}{n+m}\right)$ and $\Omega\left(n+\frac{4^n}{n+m}\right)$ , respectively, using $m$ ancillary qubits.
\end{proposition}
\begin{proof}
These results follow from Lemma \ref{lem:lowerbound_previous} and Proposition \ref{prop:lowerbound_Lambda_noconstraint}.
\end{proof}

To give depth lower bounds under general graph constraints, we first associate a quantum circuit with a directed graph.

\begin{definition}[Directed graphs for quantum circuits]\label{def:circuit-digraph}
Let $C$ be a quantum circuit on $n$ input and $m$ ancillary qubits consisting of $d$ depth-1 layers, with odd layers consisting only of single-qubit gates, even layers consisting only of CNOT gates, and any two (non-identity) single-qubit gates acting on the same qubit must be separated by at least one CNOT gate acting on that qubit (either as control or target).  Let $L_1,L_2,\cdots,L_d$ denote the $d$ layers of this circuit, i.e., $C=L_dL_{d-1}\cdots L_1$.  Define the directed graph $H=(V_C,E_C)$ associated with $C$ as follows. 
\begin{enumerate}
    \item Vertex set $V_C$: For each $i\in[d+1]$, define $S_{i}:=\{v_i^j:j\in[n+m]\}$, where $v_i^j$ is a label corresponding to the $j$-th qubit. Then, $V_C:=\bigcup_{i=1}^{d+1} S_i$. 
    \item Edge set $E_C$: For all $i\in [d]$:
    \begin{enumerate}
        \item 
        If there is a single-qubit gate acting on the $j$-th qubit in layer $L_i$ then, for all $i \le i' \le d$ there exists a directed edge $(v_{i'+1}^j,v_{i'}^j)$.
        
        \item If there is a CNOT gate acting on qubits $j_1$ and $j_2$ in layer $L_i$, then there exist $4$ directed edges $(v_{i+1}^{j_1},v_i^{j_1})$, $(v_{i+1}^{j_2}, v_i^{j_1})$, $(v_{i+1}^{j_1},v_i^{j_2})$ and $(v_{i+1}^{j_2},v_i^{j_2})$.
    \end{enumerate}
  Note that edges are directed from $S_{i+1}$ to $S_i$. 
\end{enumerate}
\end{definition}
\paragraph{Remark}
The circuits $C$ in Def.~\ref{def:circuit-digraph} assume a particular structure of alternating layers of single qubit gates and CNOT gates. However, an arbitrary circuit can be brought into this form with at most a constant factor overhead in depth: consecutive single qubit gates acting on a single qubit can be combined into one single-qubit gate, and consecutive circuit layers containing CNOT gates can be separated by a layer of identity gates. Thus, without loss of generality, for the remainder of this section, we will assume (as in~\cite{shende2004minimal}) that circuits have this alternating layer structure.    

\begin{definition}[Reachable subsets]\label{def:reachable} Let $H=(V_C,E_C)$ be the directed graph associated with quantum circuit $C$ of depth $d$, with vertex set $V_C = \bigcup_{i=1}^{d+1} S_i$.  For each $i\in[d+1]$ define the reachable subsets $S'_{i}$ of $H$ as follows:
\begin{itemize}
    \item $S'_{d+1} = \{v^j_{d+1} : j\in[n]\}$, i.e., the subset of $n$ vertices in $S_{d+1}$ corresponding to the $n$ input qubits.
    \item For $i\in[d]$, $S'_{i}\subseteq S_i$ is the subset of vertices $v^j_i$ in $S_i$ which are (i) reachable by a directed path from vertices in $S'_{d+1}$, and (ii) there is a quantum gate acting on qubit $j$ in circuit layer $L_{i}$. 
\end{itemize}
\end{definition}


\reachablesetbounds*
\begin{proof}
Let $V$ denote the set of $n+m$ qubits, and $V_{d+1}$ denote the $n$-input qubit set. For $i\in[d]$, $V_{i}=V_{i+1}\cup V'_i$, where $V'_i$ denotes the qubit set corresponding to $S_i'$.
The $i$-th layer $L_i$ can be represented as $L_i=L_{V_i}\otimes L_{\overline{V}_i}$, where $L_{V_i}$ consists of gates acting on $V_i$, and $L_{\overline{V}_i}$ consists of gates acting on $V-V_i$. Then, $C$ can be expressed as
\begin{align*}
    C&=L_dL_{d-1}\cdots L_1 = (L_{V_d}\otimes L_{\overline{V}_d}) (L_{V_{d-1}}\otimes L_{\overline{V}_{d-1}})\cdots (L_{V_1}\otimes L_{\overline{V}_1})
\end{align*}

Since $C$ is a QSP circuit acting on $n$ input and $m$ ancillary qubits, we have
\begin{equation}
    \ket{\psi}_I\ket{0^m}_A = C \ket{0^n}_I\ket{0^m}_A = (L_{V_d}\otimes L_{\overline{V}_d}) (L_{V_{d-1}}\otimes L_{\overline{V}_{d-1}})\cdots (L_{V_1}\otimes L_{\overline{V}_1}) \ket{0^n}_I\ket{0^m}_A
\end{equation}
where $I$ and $A$ are registers for holding the input and ancilla, respectively. Now we can cancel the gates in $L_{\overline{V}_i}$ without affecting $\ket{\psi}$. More precisely, we multiply the two sides of the above equation by $L_{\overline{V}_d}, \ldots, L_{\overline{V}_1}$ in that order, and get
\begin{equation}\label{eq:lightcone}
    (\mathbb{I}_{V_1}\otimes L_{\overline{V}_1}^\dagger) (\mathbb{I}_{V_{2}}\otimes L_{\overline{V}_{2}}^\dagger)\cdots (\mathbb{I}_{V_d}\otimes L_{\overline{V}_d}^\dagger) \ket{\psi}_I\ket{0^m}_A = (L_{V_d}\otimes \mathbb{I}_{\overline{V}_d}) (L_{V_{d-1}}\otimes \mathbb{I}_{\overline{V}_{d-1}})\cdots (L_{V_1}\otimes \mathbb{I}_{\overline{V}_1}) \ket{0^n}_I\ket{0^m}_A.
\end{equation}
Note that $V_{d+1}\subseteq V_{d} \subseteq \cdots \subseteq V_1$ by definition, thus $\overline{V}_1\subseteq \overline{V}_{2} \subseteq \cdots \subseteq \overline{V}_d \subseteq \overline{V}_{d+1} = [n+m]-[n]$. Therefore, all the operators $L_{\overline{V}_i}^\dagger$ at the LHS of the above equation act on the second register $A$ only, thus 
\[ (L_{V_d}\otimes \mathbb{I}_{\overline{V}_d}) (L_{V_{d-1}}\otimes \mathbb{I}_{\overline{V}_{d-1}})\cdots (L_{V_1}\otimes \mathbb{I}_{\overline{V}_1}) \ket{0^n}_I\ket{0^m}_A = \ket{\psi}_I\ket{\phi}_A\]
for some $m$-qubit state $\ket{\phi}_A$. That is, by removing all gates outside the lightcone, we get another circuit $C' = (L_{V_d}\otimes \mathbb{I}_{\overline{V}_d}) (L_{V_{d-1}}\otimes \mathbb{I}_{\overline{V}_{d-1}})\cdots (L_{V_1}\otimes \mathbb{I}_{\overline{V}_1})$ that also generates state $\ket{\psi}_I$, though with a garbage state $\ket{\phi}_A$ unentangled with $\ket{\psi}_I$.

Now we analyze the number of parameters in $C'$ to see how many different $\ket{\psi}_I$ it can generate. For any $i\in[d]$, according to the definition of $L_{V_i}$, there are $O(|S_i'|)=O(|V'_i|)$ gates in $L_{V_i}$. Therefore, $C'$ consists of $O(\sum_{i=1}^d|S'_i|)$ gates. As each gate can be fully specified by $O(1)$ free real parameters, $C'$ can be specified by $O(\sum_{i=1}^d|S'_i|)$ free real parameters. Thus the output of circuit $C'$ is a manifold of  dimension at most $O(\sum_{i=1}^d|S'_i|)$. Since the set of all $n$-qubit states $\ket{\psi}$ is a sphere of dimension $2^n-1$, we have that $O(\sum_{i=1}^d|S'_i|)\ge 2^n-1$.

The results for $n$-qubit diagonal unitary matrices and $n$-qubit general unitary matrices follow similarly, noting that they are specified by at least $2^n-1$ and $4^n-1$ free parameters, respectively. 
\end{proof}
\paragraph{Remark} In the above proof, we assume that $C$ is a parameterized circuit, i.e., the architecture is fixed and only the parameters vary. But note that even if we allow flexible architecture for depth-$d$ circuits, that only multiplies the measure of the output by a finite number, and in particular cannot increase its dimension.
\begin{theorem}[Theorems \ref{thm:depth_lowerbound_QSP_graph} and \ref{thm:depth_lowerbound_US_graph}]\label{thm:depth_lower_bound_graph}
Let $G_\nu=(V,E)$ be a connected graph with $n+m$ vertices, with $\nu$ the size of a maximum matching in $G$. 
There exist $n$-qubit quantum states, diagonal unitary matrices and general unitaries which require quantum circuits under $G_\nu$ constraint of depth at least $\Omega\left(\max\{n,2^n/\nu\}\right)$, $\Omega\left(\max\{n,2^n/\nu\}\right)$ and $\Omega\left(\max\{n,4^n/\nu\}\right)$, respectively, to be implemented, using $m\ge 0$ ancillary qubits.
\end{theorem}
\begin{proof}


We consider the directed graph $H=(V_C, E_C)$ and reachable sets $S'_1$, $S'_2$, $\ldots$, $S'_d$, $S'_{d+1}$ for a QSP circuit $C$ of depth $d$. For $i\in [d]$, for every vertex in $S'_{i+1}$, there are at most 2 neighbors in $S'_{i}$ and thus  $\abs{S'_{i}} \le 2|S'_{i+1}|$. Since $|S'_{d+1}|=n$, we have $\abs{S'_{i}} \le 2^{d-i+1}n$ for all $i\in[d]$. 
Since the maximum matching size of $G$ is $\nu$, $\abs{S'_{i}}\le 2\nu$ for all $i\in[d]$. This is true for even $i$ as there are at most $2\nu$ CNOT gates in layer $i$. This actually also holds for odd $i$, as there are at most $2\nu$ CNOT gates in layer $i-1$. For qubits $j$ that these CNOT gates do not touch in layer $i-1$, there is no single-qubit gate in layer $i$ on them as well, because if there are, they should have been absorbed into the single-qubit gates in layer $i-2$ or earlier.

Combining the above two cases, we obtain $\abs{S'_i} \le \min\{2^{d-i+1}n,2\nu\}$, for all $i\in[d]$. 
Note that $2\nu \le n$, thus based on Theorem \ref{thm:reachable-set-bounds}, we have
\begin{align*}
   &2^n-1 \le O\big(\sum_{i=1}^d |S'_i|\big)= O\big( \sum_{i=1}^d \min\{2^{d-i+1}n,2\nu\}\big)\\
   = &O\Big(\sum\limits_{i=1}^{ d- \lfloor\log(\frac{\nu}{n})\rfloor}2\nu+\sum\limits_{i= d-\lfloor \log(\frac{\nu}{n})\rfloor+1}^d2^{d-i+1}n\Big)=O\left((d-\log(\frac{\nu}{n}))\nu\right)\le O(d\nu),  
\end{align*}
which implies $d=\Omega(2^n/\nu)$.
Lemma \ref{lem:lowerbound_previous} gives a depth lower bound $\Omega(n)$ for QSP. Combined with this result, we obtain a depth lower bound for QSP of $\Omega(\max\{n,2^n/\nu\})$.

The results for diagonal unitaries and arbitrary $n$-qubit unitaries follow by the same argument.

\end{proof}


\subsection{Circuit size and depth lower bounds under specific graph constraints}
We prove lower bounds for specific graph constraints, starting with grid graphs. The generated state is in ${\sf R}_{\rm inp}$ as defined in Appendix \ref{sec:diag_with_ancilla_grid_d}.

\lowerboundgrid*
\begin{proof}
Recall that $n_1\ge n_2\ge \cdots\ge n_d$. Proposition \ref{prop:depth_lowerbound_graph} gives a depth lower bound $\Omega\Big(\max\Big\{n,\frac{2^n}{\Pi_{i=1}^d n_i}\Big\}\Big)$. Let $D$ denote the depth of the $n$-qubit QSP circuit implementing the quantum state, and $H=(V_C, E_C)$ the associated directed graph with reachable sets $S'_1, \ldots, S'_{D+1}$. Recall the arrangement of input register ${\sf R}_{\rm inp}$ in Appendix \ref{sec:diag_with_ancilla_grid_d}: Let $k$ be the minimum integer satisfying $n_1\cdots n_k \ge n$, and $n_k'$ be the minimum integer satisfying $n_1\cdots n_{k-1} n_k' \ge n$. (When $k=1$, $n_1\cdots n_{k-1}$ is defined to be 1.) Register ${\sf R}_{\rm inp}$ consists of the first $n$ qubits of sub-grid $\Grid_{n_1n_2\cdots n_{k-1}n'_k}^{n_1,n_2,\cdots, n_{k-1},n'_k,1,1,\cdots,1}$.
Note that, for $\Grid_{n+m}^{n_1, n_2, \ldots, n_d}$, $S'_{D+1} \subseteq [n_1]\times \cdots [n_{k-1}] \times [n_k'] \times \{1\}\times \cdots \times \{1\}$, $S'_{D}\subseteq [n_1]\times \cdots [n_{k-1}] \times [n_k'+1] \times [2]\times \cdots \times [2]$, $S'_{D-1}\subseteq [n_1]\times \cdots [n_{k-1}] \times [n_k'+2] \times [3]\times \cdots \times [3]$, and so on. Since $n_d\le \cdots \le n_1$, the last dimensions $[n_d]$, $[n_{d-1}]$ ... may be saturated as $i$ (in $S_i'$) decreases. In general, 
we have the following bounds for $|S_i'|$, where, in the middle line, the last $\ell \in[d-k]$ dimensions are saturated. 
     \begin{equation}\label{eq:grid-Si-bound}
    \abs{S'_{i}} \le \begin{cases}
    O\left(n_1n_2\cdots n_{k-1}(n'_k+D-i+1)(D-i+2)^{d-k}\right) &\quad \text{if }D-i+2\le n_d,\\
    O\left(n_1 n_2\cdots n_{k-1}(n'_k+D-i+1)(D-i+2)^{d-\ell-k}n_{d-\ell+1}\cdots n_d \right) &\quad \text{if } n_{d-\ell+1}< D-i+2\le n_{d-\ell},\\
    n_1 n_2\cdots n_d & \quad \text{if } D-i+2>n_k.
    \end{cases}
\end{equation}

We consider $d+1$ cases.
\begin{itemize}
    \item Case 1: If $n_d\ge \Omega(2^{\frac{n}{d+1}})$, assume for the sake of contradiction that $D=o\left(2^{\frac{n}{d+1}}\right)$. Then $D=o(n_d)$ in this case, and for all $i\in [D]$,
    \begin{align*}
    |S_i'|=O\left(n_1n_2\cdots n_{k-1}(n'_k+D-i+1)(D-i+2)^{d-k}\right)=O(n(n'_k+D-i+1)(D-i+2)^{d-k}) 
    \end{align*}
    By Theorem~\ref{thm:reachable-set-bounds}, 
    \begin{align*}
        2^n-1\le O\left(\sum_{i=1}^D|S'_i|\right)= \sum_{i=1}^{D} O\left(n(n'_k+D-i+1)(D-i+2)^{d-k}\right)=O(n(2D)^{d-k+1})\le O((2D)^{d+1}), 
    \end{align*}
    as $D\ge \Omega(n)$. This implies $D\ge \Omega(2^{n/(d+1)})$, which contradicts with our assumption. Therefore, $D$ must satisfy $D=\Omega(2^{n/(d+1)})$. In this case, it is not hard to verify that $\frac{2^{n/j}}{(n_j\cdots n_d)^{1/j}} \le O(2^{n/(d+1)})$ for all $j\in [d]$, and thus $D=\Omega(2^{n/(d+1)})=\Omega\Big(n+2^{\frac{n}{d+1}}+ \max\limits_{j\in [d]}\Big\{\frac{2^{n/j}}{(\Pi_{i=j}^d n_i)^{1/j}}\Big\}\Big)$.
    
    \item Case $j$ ($2\le j\le d-k+1$): $n_d,n_{d-1},\ldots,n_{d-j+2}$ satisfy
     \begin{equation}\label{eq:grid-lb-case-j}
         n_d \le o( 2^{n/(d+1)}),\quad n_{d-i}\le o\Big(\frac{2^{\frac{n}{d-i+1}}}{(n_{d-i+1}\cdots n_{d})^{\frac{1}{d-i+1}}}\Big),  \quad\forall i\in[j-2].
     \end{equation}
     and $n_{d-j+1}$ satisfies $n_{d-j+1}\ge\Omega\Big( \frac{2^{\frac{n}{d-j+2}}}{(n_{d-j+2}\cdots n_d)^{\frac{1}{d-j+2}}}\Big)$.
     Assume
     for the sake of contradiction that $D=o\Big(\frac{2^{\frac{n}{d-j+2}}}{(n_{d-j+2}\cdots n_{d})^{\frac{1}{d-j+2}}}\Big)$. Then we have $D=o(n_{d-j+1})$. We claim that $D\ge n_d$. Suppose that it does not hold, i.e, $D< n_d$. Based on Eq. \eqref{eq:grid-Si-bound} (the first case), 
         $|S'_i|\le O(n_1\cdots n_{k-1}(n'_k+D-i+1)(D-i+2)^{d-k})$ for all $i\in[D]$. Since $D<n_d$, $|S_i|$ satisfies
         \begin{align*}
             |S_i|&\le O((n_1\cdots n_{k-1})(n'_k+n_d)(n_d)^{d-k})
         \end{align*}
         Recall $n_d = o(2^{n/(d+1)})$, $n_1\cdots n_{k-1}n'_k=O(n)$, $n_1\cdots n_{k-1}$ is defined to 1 when $k=1$, and the assumption $\Omega(n)\le D<n_d$. We can obtain that the above bound is at most $O((n_d)^{d})$ both in case $k=1$ and $k\ge 2$. Thus 
         \[\sum_{i                    =1}^D|S'_i|\le D\cdot O((n_d)^{d}) = O((n_d)^{d+1}) = o(2^{n}).
         \]
         But according to Theorem \ref{thm:reachable-set-bounds}, $\sum_{i=1}^D|S'_i|\ge 2^n-1$, which contradicts the above equation. Therefore, we have $D\ge n_d$. 
         
     
        Recall that we assumed $D=o(n_{d-j+1})$, so $D$ falls in an interval $[n_{d-j+\tau+1},  n_{d-j+\tau})$ for some $1\le \tau\le j-1$. Now we upper bound $|S_i'|$ for different $i$. First consider those $i$ with  $D-i+2\le n_d$: we have $D-i+2\le n_d\le n_{d-1} \le \cdots \le n_{d-j+\tau+1}$, thus by Eq.\eqref{eq:grid-Si-bound} (the first case)
         \begin{align}\label{eq:reachable_set_size1}
              |S_i'| &\le O(n_1\cdots n_{k-1}(n'_k+D-i+1)(D-i+2)^{d-k})\nonumber\\
            &\le O(n_1\cdots n_{k-1}(n'_k+D-i+1)(D-i+2)^{d-j+\tau-k}n_{d-j+\tau+1}\cdots n_{d}).
        \end{align}
         Next consider those $i$ with  $n_{d-\ell+1} < D-i+2 \le n_{d-\ell}$ for some $\ell\in[j-\tau-1]$: we have $D-i+2 \le n_{d-\ell} \le \cdots \le n_{d-j+\tau+1}$, thus by Eq.\eqref{eq:grid-Si-bound} (second case)
         \begin{align}\label{eq:reachable_set_size2}
             |S_i'|&\le O(n_1\cdots n_{k-1}(n'_k+D-i+1)(D-i+2)^{d-\ell-k}n_{d-\ell+1}\cdots n_d)\nonumber\\
              &\le  O(n_1\cdots n_{k-1}(n'_k+D-i+1)(D-i+2)^{d-j+\tau-k}n_{d-j+\tau+1}\cdots n_{d})
         \end{align}
         For $i$ with $1\le i<D-n_{d-j+\tau+1}+2$, we have $n_{d-j+\tau+1}\le D-i+2\le D+1\le n_{d-j+\tau}$, thus by Eq. \eqref{eq:grid-Si-bound} (the second case)
            \begin{align}\label{eq:reachable_set_size3}
             |S_i'|
              &\le  O(n_1\cdots n_{k-1}(n'_k+D-i+1)(D-i+2)^{d-j+\tau-k}n_{d-j+\tau+1}\cdots n_{d}).
         \end{align}


     By Theorem~\ref{thm:reachable-set-bounds}, we have
    \begin{align*}
          2^n-1\le O\left(\sum_{i=1}^D|S'_i|\right)
          =O\left(\sum_{i=D-n_d+2}^D|S_i'|+\sum_{
          \ell=1}^{j-\tau-1}\sum_{i=D-{n_{d-\ell}+2}}^{D-n_{d-\ell+1}+1}|S_i'|+\sum_{i=1}^{D-n_{d-j+\tau+1}+1}|S_i'|\right)
    \end{align*}
    
    Now we use Eq. \eqref{eq:reachable_set_size1}, Eq. \eqref{eq:reachable_set_size2}, and Eq. \eqref{eq:reachable_set_size3} to bound the first, second and third term, respectively, and obtain the upper bound
    \[2^n-1 \le \sum_{i=1}^{D} O(n_1\cdots n_{k-1}(n'_k+D-i+1)(D-i+2)^{d-j+\tau-k} n_{d-j+\tau+1} \cdots n_{d}).
    \]
    Note that $D-i+1\le D$ for all $i\in D$, $n'_k\le n$ and $D\ge \Omega(n)$, therefore
    \begin{align*}
         \sum_{i=1}^D (n'_k+D-i+1)(D-i+2)^{d-j+\tau-k} \le  \sum_{i=1}^D (n'_k+D)(D+1)^{d-j+\tau-k} 
        \le  (2D)^{d-j+\tau+2-k},
    \end{align*}
    and the upper bound becomes 
    \[2^n-1 \le O(n_1\cdots n_{k-1}(2D)^{d-j+\tau+2-k}n_{d-j+\tau+1}\cdots n_d).\]
    Recall that $n_1\cdots n_{k-1}=1$ if $k=1$ and $n_1\cdots n_{k-1} = O(n) = O(D)$ if $k\ge 2$. In either case, we have $2^n-1 \le O((2D)^{d-j+\tau+1} n_{d-j+\tau+1}\cdots n_d)$.
    Thus $D\ge\Omega\left(\frac{2^{\frac{n}{d-j+\tau+1}}}{(n_{d-j+\tau+1}\cdots n_{d})^{\frac{1}{d-j+\tau+1}}}\right)$. 
    
    If $2\le \tau\le j-1$, i.e. $j-\tau\in [j-2]$, Eq.\eqref{eq:grid-lb-case-j} with $i$ set to be $j-\tau$ gives 
    $D\le o\left(\frac{2^{\frac{n}{d-j+\tau+1}}}{(n_{d-j+\tau+1}\cdots n_d)^{\frac{1}{d-j+\tau+1}}}\right)$, contradicting the above lower bound of $D$. Thus $\tau = 1$ and $D=\Omega\Big(\frac{2^{\frac{n}{d-j+2}}}{(n_{d-j+2}\cdots n_{d})^{\frac{1}{d-j+2}}}\Big)$. 
    Now we shall show that 
    \begin{align}\label{eq:grid-lb-D}
        D=\Omega\Big(\max\Big\{n,\frac{2^{\frac{n}{d-j+2}}}{(n_{d-j+2}\cdots n_{d})^{\frac{1}{d-j+2}}}\Big\}\Big)=\Omega\Big( n+2^{\frac{n}{d+1}}+\max\limits_{j\in [d]}\Big\{\frac{2^{n/j}}{(\Pi_{i=j}^d n_i)^{1/j}}\Big\}\Big).
    \end{align}

    We first show the following facts
        \begin{align}
            &\frac{2^{\frac{n}{d-i+1}}}{(n_{d-i+1}\cdots n_{d})^{\frac{1}{d-i+1}}}\ge \frac{2^{\frac{n}{d-i+2}}}{(n_{d-i+2}\cdots n_{d})^{\frac{1}{d-i+2}}}, \quad \text{if~} 2\le i\le j-1. \label{eq:compare1}\\
            &\frac{2^{\frac{n}{d-j+2}}}{(n_{d-j+2}\cdots n_{d})^{\frac{1}{d-j+2}}}\ge \Omega\left(\frac{2^{\frac{n}{k'}}}{(n_{k'}\cdots n_{d})^{\frac{1}{k'}}}\right), \quad \text{if~} 1\le k'\le d-j+1.\label{eq:compare2}
        \end{align}
        Eq. \eqref{eq:grid-lb-case-j} implies that $n_{d-i+1}\le \frac{2^{\frac{n}{d-i+2}}}{(n_{d-i+2}\cdots n_d)^{\frac{1}{d-i+2}}}$ for all $i\in \{2,\ldots,j-1\}$. Then we have
        \begin{align*}
            &\frac{2^{\frac{n}{d-i+1}}}{(n_{d-i+1}\cdots n_{d})^{\frac{1}{d-i+1}}}/ \frac{2^{\frac{n}{d-i+2}}}{(n_{d-i+2}\cdots n_{d})^{\frac{1}{d-i+2}}}
            =\frac{2^{\frac{n}{(d-i+1)(d-i+2)}}}{(n_{d-i+1})^{\frac{1}{d-i+1}}(n_{d-i+2}\cdots n_{d})^{\frac{1}{(d-i+1)(d-i+2)}}}
            \ge (n_{d-i+1})^{\frac{1}{d-i+1}}/(n_{d-i+1})^{\frac{1}{d-i+1}}=1.
        \end{align*}
        Eq.~\eqref{eq:compare1} thus holds. Recall that $n_1\ge n_2\ge \cdots \ge n_{d-j+1}\ge \Omega\Big( \frac{2^{\frac{n}{d-j+2}}}{(n_{d-j+2}\cdots n_d)^{\frac{1}{d-j+2}}}\Big)$. Then we have
        \begin{align*}
            &\frac{2^{\frac{n}{d-j+2}}}{(n_{d-j+2}\cdots n_{d})^{\frac{1}{d-j+2}}}/ \frac{2^{\frac{n}{k'}}}{(n_{k'}\cdots n_{d})^{\frac{1}{k'}}}=\frac{(n_{d-j+2}\cdots n_d)^{\frac{d-j+2-k'}{(d-j+2)k'}}(n_{k'}\cdots n_{d-j+1})^{1/k'}}{2^{\frac{(d-j+2-k')n}{(d-j+2)k'}}}\\
            &\ge \Omega\left(\frac{(n_{k'}\cdots n_{d-j+1})^{1/k'}}{(n_{d-j+1})^{\frac{d-j+2-k'}{k'}}}\right)\ge \Omega\left(\frac{(n_{d-j+1})^{\frac{d-j+1-k'+1}{k'}}}{(n_{d-j+1})^{\frac{d-j+2-k'}{k'}}}\right)=\Omega(1),
        \end{align*}
        and Eq.~\eqref{eq:compare2} holds.

        Combining Eq. \eqref{eq:compare1} and Eq. \eqref{eq:compare2}, we see that
        $\frac{2^{\frac{n}{d-j+2}}}{(n_{d-j+2}\cdots n_{d})^{\frac{1}{d-j+2}}}\ge \Omega\left(\frac{2^{\frac{n}{k'}}}{(n_{k'}\cdots n_d)^{\frac{1}{k}}}\right)$ for $k'\in[d]$. For Eq. \eqref{eq:grid-lb-D}, it remains to prove         $\frac{2^{\frac{n}{d-j+2}}}{(n_{d-j+2}\cdots n_{d})^{\frac{1}{d-j+2}}}\ge 2^{\frac{n}{d+1}}$. Since $n_d\le 2^{n/(d+1)}$ (Eq. \eqref{eq:grid-lb-case-j}), it follows that $\frac{2^{n/d}}{(n_d)^{1/d}}\ge \frac{2^{n/d}}{(2^{n/(d+1)})^{1/d}}=2^{n/(d+1)}$. According to Eq. \eqref{eq:compare1}, we have 
            \[\frac{2^{\frac{n}{d-j+2}}}{(n_{d-j+2}\cdots n_{d})^{\frac{1}{d-j+2}}}\ge \frac{2^{\frac{n}{d-j+3}}}{(n_{d-j+3}\cdots n_{d})^{\frac{1}{d-j+3}}}\ge \cdots \ge \frac{2^{\frac{n}{d}}}{( n_{d})^{\frac{1}{d}}}\ge  2^{n/(d+1)}.\]
        This completes the proof of Eq. \eqref{eq:grid-lb-D}.
       
     \item   Case $j$ $(d-k+2\le j\le d)$: Same as Eq. \eqref{eq:grid-lb-case-j},  $n_d,n_{d-1},\ldots,n_{d-j+1}$ satisfy
     \begin{equation*}
         n_d \le o( 2^{n/(d+1)}), \quad n_{d-i}\le o\Big(\frac{2^{\frac{n}{d-i+1}}}{(n_{d-i+1}\cdots n_{d})^{\frac{1}{d-i+1}}}\Big),~ \forall i\in[j-2], \quad n_{d-j+1}\ge\Omega\Big( \frac{2^{\frac{n}{d-j+2}}}{(n_{d-j+2}\cdots n_d)^{\frac{1}{d-j+2}}}\Big).
     \end{equation*}
     Assume for the sake of contradiction that  $D=o\Big( \frac{2^{\frac{n}{d-j+2}}}{(n_{d-j+2}\cdots n_d)^{\frac{1}{d-j+2}}}\Big)$. Then we have $n_{d-j+2}\cdots n_d=o\Big( \frac{2^{n}}{D^{d-j+2}}\Big)$. For all $i\in[D]$, we use a trivial size bound of $|S'_i|$ is $|S_i'|\le n_1n_2\cdots n_d$. Since in the current case $j$ we have $d-k+2\le j$, $n_{d-j+2}\cdots n_{k-1}$ is well defined and at least 1. Thus by Theorem~\ref{thm:reachable-set-bounds}, we have
     \begin{equation}\label{eq:d-k+2<=j<=d_bound}
          2^n-1\le O\left(\sum_{i=1}^D|S'_i|\right)\le O(Dn_1\cdots n_d)
          \le O\left (Dn_1\cdots n_{k-1} n_{d-j+2}\cdots n_{k-1} n_{k}\cdots n_d\right).
     \end{equation}

     Since $n_{d-j+2}\cdots n_d=o\Big( \frac{2^{n}}{D^{d-j+2}}\Big)$,  $n_1n_2\cdots n_{k-1}n'_k=O(n)$ and $D\ge \Omega (n)$, we have
     \[O\left (Dn_1\cdots n_{k-1} n_{d-j+2}\cdots n_{k-1} n_{k}\cdots n_d\right)\le 
        o\left(\frac{n2^n}{D^{d-j+1}}\right)\le o\left(\frac{n2^n}{D}\right)= o(2^n).\]
    This contradicts Eq.~\eqref{eq:d-k+2<=j<=d_bound}.
     Therefore, the assumption that $D=o\Big( \frac{2^{\frac{n}{d-j+2}}}{(n_{d-j+2}\cdots n_d)^{\frac{1}{d-j+2}}}\Big)$ does not hold and $D$ satisfies $D=\Omega\Big(\frac{2^{\frac{n}{d-j+2}}}{(n_{d-j+2}\cdots n_{d})^{\frac{1}{d-j+2}}}\Big)$. By the same discussion in Case $j~(2\le j\le d-k+1)$, we have $\Omega\Big(\max\Big\{n,\frac{2^{\frac{n}{d-j+2}}}{(n_{d-j+2}\cdots n_{d})^{\frac{1}{d-j+2}}}\Big\}\Big)=\Omega\Big( n+2^{\frac{n}{d+1}}+\max\limits_{j\in [d]}\Big\{\frac{2^{n/j}}{(\Pi_{i=j}^d n_i)^{1/j}}\Big\}\Big)$.
    
    %
    \item  Case $d+1$: $n_d,n_{d-1},\ldots, n_1$ satisfy 
         \begin{equation*}
         n_d \le o( 2^{n/(d+1)}), \qquad n_{d-i}\le o\Big(\frac{2^{\frac{n}{d-i+1}}}{(n_{d-i+1}\cdots n_{d})^{\frac{1}{d-i+1}}}\Big),~ \forall i\in[d-1],
     \end{equation*}
     Proposition \ref{prop:depth_lowerbound_graph} gives a depth lower bound of $\Omega\Big(\max\Big\{n,\frac{2^n}{\Pi_{i=1}^d n_i}\Big\}\Big)$. For any $k'\ge 2$, the above inequality is rephrased as $n_{k'-1}\le o\left(\frac{2^{\frac{n}{k'}}}{(n_{k'}\cdots n_d)^{\frac{1}{k'}}}\right)$. We have 
     \begin{align*}
         &\frac{2^{\frac{n}{k'-1}}}{(n_{k'-1}\cdots n_d)^{\frac{1}{k'-1}}}\Big/\frac{2^{n/k'}}{(n_{k'}\cdots n_d)^{1/k'}}=\frac{2^{\frac{n}{k'(k'-1)}}}{(n_{k'-1})^{\frac{1}{k'-1}}(n_{k'}\cdots n_d)^{\frac{1}{k'(k'-1)}}}
        \ge  \frac{(n_{k'-1})^{\frac{1}{k'-1}}}{(n_{k'-1})^{\frac{1}{k'-1}}}=1,\quad \forall k'\ge 2.
      \end{align*}
      Therefore, \[\frac{2^n}{n_1\cdots n_d}\ge \frac{2^{\frac{n}{2}}}{(n_2\cdots n_d)^{1/2}}\ge \cdots \ge\frac{2^{\frac{n}{d}}}{(n_d)^{1/d}}\ge \frac{2^{\frac{n}{d}}}{(2^{n/(d+1)})^{1/d}}=O(2^{n/(d+1)})\] 
      where the last inequality used $n_d\le o(2^{n/(d+1)})$. This implies $\Omega\Big(\max\Big\{n,\frac{2^n}{\Pi_{i=1}^d n_i}\Big\}\Big)=\Omega\Big( n+2^{\frac{n}{d+1}}+\max\limits_{j\in [d]}\Big\{\frac{2^{n/j}}{(\Pi_{i=j}^d n_i)^{1/j}}\Big\}\Big)$.
\end{itemize}
Let $d=1$ and $d=2$, the first and second cases are obtained.
\end{proof}

Using the same argument, we can also show the following.
\begin{lemma}\label{lem:lower_bound_grid_k_Lambda}
There exists an $n$-qubit diagonal unitary matrix that requires a quantum circuit of depth \[\Omega\Big( n+2^{\frac{n}{d+1}}+\max\limits_{j\in [d]}\Big\{\frac{2^{n/j}}{(\Pi_{i=j}^d n_i)^{1/j}}\Big\}\Big)\] under $\Grid_{n+m}^{n_1,n_2,\cdots,n_d}$ constraint to be implemented, using $m\ge0$ ancillary qubits.
\end{lemma}

\lowerboundgridus*
\begin{proof}
Lemma \ref{lem:lowerbound_previous} and Theorem \ref{thm:depth_lower_bound_graph} give a depth lower bound of $\Omega\Big(\max\Big\{n,\frac{4^n}{\Pi_{i=1}^d n_i}\Big\}\Big)$. By the same argument as in the proof of Theorem \ref{thm:lower_bound_grid_k_QSP}, the following $d+1$ cases:
\begin{itemize}
    \item Case 1:  $n_d\ge \Omega(4^{\frac{n}{d+1}})$.
    \item Case $j$ ($2\le j\le d$):  $n_d,n_{d-1},\ldots,n_{d-j+1}$ satisfy
     \begin{equation*}
         n_d \le o( 4^{n/(d+1)}), n_{d-i}\le o\Big(\frac{4^{\frac{n}{d-i+1}}}{(n_{d-i+1}\cdots n_{d})^{\frac{1}{d-i+1}}}\Big), n_{d-j+1}\ge\Omega\Big( \frac{4^{\frac{n}{d-j+2}}}{(n_{d-j+2}\cdots n_d)^{\frac{1}{d-j+2}}}\Big), \quad\forall i\in[j-2].
     \end{equation*}
     \item Case $d+1$: $n_d,n_{d-1},\ldots,n_{1}$ satisfy
     \begin{equation*}
         n_d \le o( 4^{n/(d+1)}),\quad n_{d-i}\le o\Big(\frac{4^{\frac{n}{d-i+1}}}{(n_{d-i+1}\cdots n_{d})^{\frac{1}{d-i+1}}}\Big),\quad\forall i\in[d-1].
     \end{equation*}
\end{itemize}
have depth lower bounds of 
\begin{align*}
    \Omega(4^{\frac{n}{d+1}})&=\Omega\Big(n+4^{\frac{n}{d+1}}+\max\limits_{j\in [d]}\Big\{\frac{4^{n/j}}{(\Pi_{i=j}^d n_i)^{1/j}}\Big\}\Big), \qquad \text{Case }1;\\
    \Omega\left(\frac{4^{\frac{n}{d-j+2}}}{(n_{d-j+2}\cdots n_d)^{\frac{1}{d-j+2}}}\right)&=\Omega\Big(n+4^{\frac{n}{d+1}}+\max\limits_{j\in [d]}\Big\{\frac{4^{n/j}}{(\Pi_{i=j}^d n_i)^{1/j}}\Big\}\Big),\qquad\text{Cases }2\le j \le d;\\
    \Omega\left(\max\left\{n,\frac{4^n}{\prod_{i=1}^d n_i}\right\}\right)&=\Omega\Big(n+4^{\frac{n}{d+1}}+\max\limits_{j\in [d]}\Big\{\frac{4^{n/j}}{(\Pi_{i=j}^d n_i)^{1/j}}\Big\}\Big),\qquad\text{Case }d+1.
\end{align*}
\end{proof}

\begin{corollary}\label{coro:lower_bound_path}
Under $\Path_{n+m}$ constraint, using $m$ ancillary qubits,
\begin{enumerate}
    \item the depth and size lower bounds for $n$-qubit QSP are $\Omega\left(\max\left\{2^{n/2},\frac{2^n}{n+m}\right\}\right)$ and  $\Omega(2^n)$.
    \item the depth and size lower bounds for $n$-qubit diagonal unitary matrices are $\Omega\left(\max\left\{2^{n/2},\frac{2^n}{n+m}\right\}\right)$ and  $\Omega(2^n)$.    
    \item the depth and size lower bounds for $n$-qubit GUS are $\Omega\left(\max\left\{4^{n/2},\frac{4^n}{n+m}\right\}\right)$ and $\Omega(4^n)$.
\end{enumerate}
\end{corollary}
\begin{proof}
 Let $d=1$, then $\Grid_{n+m}^{n_1,n_2,\ldots,n_d}$ is $\Path_{n+m}$. The results follow from  Proposition \ref{prop:size_lowerbound_Lambda}, Theorem \ref{thm:lower_bound_grid_k_QSP}, Lemma \ref{lem:lower_bound_grid_k_Lambda}, Theorem \ref{thm:lower_bound_grid_k_US} and Proposition \ref{prop:lowerbound_size_graph}.
\end{proof}

\begin{corollary}\label{coro:lower_bound_binary}
Under $\Tree_{n+m}(2)$ constraint, using $m$ ancillary qubits, 
\begin{enumerate}
    \item $n$-qubit QSP needs quantum circuits of depth  $\Omega\left(\max\left\{n,\frac{2^n}{n+m}\right\}\right)$ and size $\Omega(2^n)$,
    
    \item $n$-qubit diagonal unitary matrix needs quantum circuits of depth  $\Omega\left(\max\left\{n,\frac{2^n}{n+m}\right\}\right)$ and size $\Omega(2^n)$,
    
    \item $n$-qubit GUS needs quantum circuits of depth  $\Omega\left(\max\left\{n,\frac{4^n}{n+m}\right\}\right)$ and size $\Omega(4^n)$.
\end{enumerate}
\end{corollary}
\begin{proof}
Follows from Propositions \ref{prop:lowerbound_size_graph}, \ref{prop:size_lowerbound_Lambda} and  \ref{prop:depth_lowerbound_graph}.
\end{proof}

\begin{corollary}[Theorems \ref{thm:lower_exapnder_QSP} and \ref{thm:lower_exapnder_US}]\label{coro:lower_bound_expander}
Under $\Expander_{n+m}$ constraint, using $m\ge 0$ ancillary qubits,
\begin{enumerate}
    \item $n$-qubit QSP needs quantum circuits of depth  $\Omega\left(\max\left\{n,\frac{2^n}{n+m}\right\}\right)$ and size $\Omega(2^n)$.
    
    \item $n$-qubit diagonal unitary matrix needs quantum circuits of depth  $\Omega\left(\max\left\{n,\frac{2^n}{n+m}\right\}\right)$ and size $\Omega(2^n)$.
    
    \item $n$-qubit GUS needs quantum circuits of depth  $\Omega\left(\max\left\{n,\frac{4^n}{n+m}\right\}\right)$ and size $\Omega(4^n)$.
\end{enumerate}
\end{corollary}
\begin{proof}
Follows from Theorem \ref{prop:lowerbound_size_graph}, and Propositions \ref{prop:size_lowerbound_Lambda} and  \ref{prop:depth_lowerbound_graph}.
\end{proof}

\begin{lemma}[Theorems \ref{thm:lower_bound_tree_QSP} and \ref{thm:lower_bound_tree_US}]\label{lem:depth_lower_dary}
There exist $n$-qubit quantum states, $n$-qubit diagonal unitary matrices and $n$-qubit unitary matrices which require quantum circuits under $\Tree_{n+m}(d)$ constraint of depth at least $\Omega\left(\max\left\{n,\frac{d2^n}{n+m}\right\}\right)$, $\Omega\left(\max\left\{n,\frac{d2^n}{n+m}\right\}\right)$ and  $\Omega\left(\max\left\{n,\frac{d4^n}{n+m}\right\}\right)$ respectively, to implement, using $m\ge 0$ ancillary qubits.
\end{lemma}
\begin{proof}
The size of a maximum matching in $\Tree_{n+m}(d)$ is $O\left(\frac{n+m}{d}\right)$. The result follows from Theorem \ref{thm:depth_lower_bound_graph}. 
\end{proof}

\begin{corollary}\label{coro:depth_lower_star}
There exist $n$-qubit quantum states and $n$-qubit unitary matrices which require quantum circuits under $\Star_{n+m}$ constraint of depth at least $\Omega (2^n)$ and  $\Omega(4^n)$, respectively, to implement, using $m\ge 0$ ancillary qubits.
\end{corollary}

\qspbrickwalllowerbound*

\usbrickwalllowerbound*
    The proofs of Theorems \ref{thm:lower_bound_brickwall_QSP} and \ref{thm:lower_bound_brickwall_US} are essentially the same as those of Theorems \ref{thm:lower_bound_grid_k_QSP} and \ref{thm:lower_bound_grid_k_US}, so we omit them here.

\section{An improvement of QSP circuit under $\Tree_{n+m}(2)$ constraint}
\label{append:binary_tree_improvement}

We first present a method of preparing quantum states using a unary encoding of the basis states, which requires $O(2^n)$ ancillary qubits. We then combine this with a way of converting from unary to binary encoded bases to bound the QSP circuit depth under binary tree constraints.

\paragraph{Unary-encoded state preparation}
\label{sec:unary_encoding}

Let $e_x$ denote the vector where the $x$-th bit is 1 and all other bits are 0. Let $\ket{e_x}$ denote the corresponding $2^n$-qubit quantum state. 
The unary-encoded quantum state preparation problem is: Given $v=(v_x)_{x\in\{0,1\}^n}\in\mathbb{C}^{2^n}$ with $\sqrt{\sum_{x\in\{0,1\}^n}|v_x|^2}=1$, prepare the $2^n$-qubit state 
    \[\ket{\psi'_v}=\sum_{x\in\{0,1\}^n}v_x\ket{e_x},\]
starting from the initial state $\ket{0^{2^n}}$. We call a circuit that implements this task a \textit{unary-encoded QSP circuit}.
Ref. \cite{johri2021nearest} constructed an $O(n)$-depth unary-encoded QSP circuit under no graph constraints. Our construction below can be applied to binary tree constraints.

We label qubits in a binary tree as in in Section~\ref{sec:diag_without_ancilla_binarytree} (see Fig.~\ref{fig:label_binarytree}).
Given parameter $\alpha\in\mathbb{R}$, define the 2-qubit rotation gate $R^{i,j}_y(\alpha)$ by its action on qubits $i,j$:
\begin{align*}
    &\ket{0}_i\ket{0}_j\to \ket{0}_i\ket{0}_j, \\
    &\ket{0}_i\ket{1}_j\to \cos(\alpha/2)\ket{0}_i\ket{1}_j+\sin(\alpha/2)\ket{1}_i\ket{0}_j,\\
    &\ket{1}_i\ket{0}_j\to -\sin(\alpha/2)\ket{0}_i\ket{1}_j+\cos(\alpha/2)\ket{1}_i\ket{0}_j,\\
    &\ket{1}_i\ket{1}_j\to \ket{1}_i\ket{1}_j .
\end{align*}

\begin{lemma}\label{lem:unary_qsp_tree}
Any $2^n$-qubit unary-encoded state $\ket{\psi'_v}=\sum_{\scriptsize x\in\Bn}v_x\ket{e_x}$ can be implemented by a quantum circuit of depth $O(n)$ and size $O(2^n)$ under $\Tree_{2^{n+1}-1}(2)$ constraint using $2^{n}-1$ ancillary qubits, where $\ket{\psi_v'}$ is in the $(n+1)$-th layer of the binary tree.
\end{lemma}
\begin{proof}
For all $x\in\Bn$, express the coefficients $v_x$ as $v_x=e^{i\theta_x}u_x$, where $\theta_x\in\mathbb{R}$ and $u_x\in\mathbb{R}_+$. For $0\le k\le n-1$ and all $z\in\B^k$, define $u_z=\sqrt{(u_{z0})^2+(u_{z1})^2}$ and $\cos(\alpha_z/2)=u_{z0}/u_z$. By construction, $z_\epsilon = \sqrt{\sum_{x\in\{0,1\}^n} u_x^2} = 1$. The creation of $\ket{\psi'_v}$ is carried out in $n+1$ steps: 

\begin{enumerate}
    \item Step $0$: Apply a single-qubit gate $X$ on the root qubit $\epsilon$.
    \item Step $k$ $(1\le k\le n-1)$: For all $z\in\{0,1\}^{k-1}$,
    \begin{enumerate}  
          \item Apply $\textsf{SWAP}^{z}_{z0}$.
        \item Apply $R_y^{z0,z1}(\alpha_z)$ under path $z0-z-z1$ constraint.
    \end{enumerate}
    \item Step $n$: Apply $R(\theta_x)$ on qubit $x$, for all $x\in\{0,1\}^{n}$.
\end{enumerate}


The effect of this procedure is illustrated in Fig.~\ref{fig:unary_QSP_example}, and can easily be verified to carry out the transformation
\begin{equation*}
    \ket{0^{2^{n+1}-1}}_{\Tree_\epsilon^n}\rightarrow  \ket{0^{2^{n+1}-1-2^n}}_{\Tree_\epsilon^n-\{0,1\}^{n}}\big(\sum_{x\in\{0,1\}^n}v_x\ket{e_x}_{\{0,1\}^n}\big)
\end{equation*} 
which is the desired state creation, where the leaf nodes of $\Tree_\epsilon^n$ store the basis states $\ket{e_x}$, and all other nodes are ancilla.

Step 0 involves only a single $X$ gate. For each $k\in[n-1]$ and  $z\in\{0,1\}^{k-1}$, $\textsf{SWAP}^{z}_{z0}$ and $R_y^{z0,z1}(\alpha_z)$ can each be implemented in size and depth $O(1)$ and, as all paths $z0-z-z1$ are disjoint in $\Tree_{\epsilon}^n$, they can be executed in parallel. Step $n$ involves $2^n$ single qubit gates, which requires depth $1$. The total size and depth required for the procedure are therefore $O(n)$ and $O(2^n)$, respectively.

\end{proof}

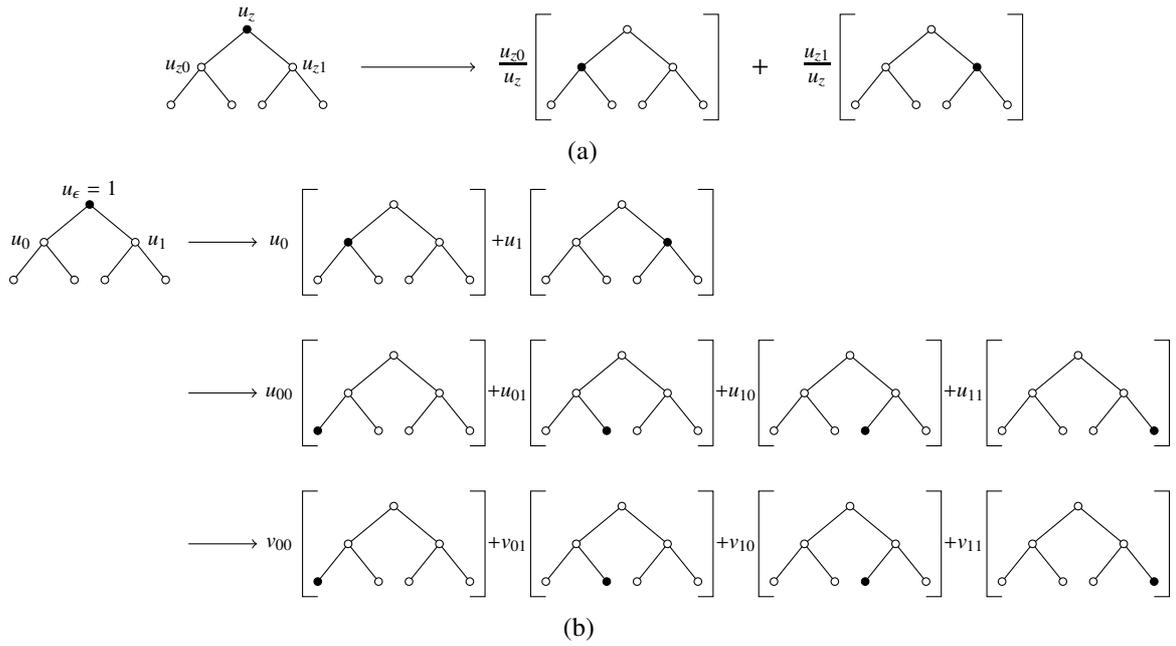
\begin{figure}[!t]
\centering
\subfloat[]{
\begin{minipage}[]{1\textwidth}
    \centering
    \begin{tikzpicture}
     \draw (0,0)--(-0.6,-0.5)--(-1,-1) (-0.6,-0.5)--(-0.2,-1) (0,0)--(0.6,-0.5)--(1,-1) (0.6,-0.5)--(0.2,-1);
     \draw [fill=black] (0,0) circle (0.05);
     \draw[fill=white]  (-0.6,-0.5) circle (0.05) (0.6,-0.5) circle (0.05) (-1,-1) circle (0.05) (-0.2,-1) circle (0.05) (1,-1) circle (0.05) (0.2,-1) circle (0.05);
     \draw (0,0.2) node{\scriptsize $u_z$} (-0.9,-0.5) node{\scriptsize $u_{z0}$} (0.9,-0.5) node{\scriptsize $u_{z1}$};
     \draw[->] (1.5,-0.5)--(3,-0.5);
     \draw (4,0.2)--(3.8,0.2)--(3.8,-1.2)--(4,-1.2);
     \draw (6,0.2)--(6.2,0.2)--(6.2,-1.2)--(6,-1.2);
    \draw (5,0)--(4.4,-0.5)--(4,-1) (4.4,-0.5)--(4.8,-1) (5,0)--(5.6,-0.5)--(6,-1) (5.6,-0.5)--(5.2,-1);
    \draw[fill=white]  (5,0) circle (0.05) (5.6,-0.5) circle (0.05)  (4,-1) circle (0.05) (4.8,-1) circle (0.05) (6,-1) circle (0.05) (5.2,-1) circle (0.05);
    \draw [fill=black] (4.4,-0.5) circle (0.05);
     \draw (3.5,-0.5) node{$\frac{u_{z0}}{u_z}$};
     \draw (8,0.2)--(7.8,0.2)--(7.8,-1.2)--(8,-1.2);
     \draw (10,0.2)--(10.2,0.2)--(10.2,-1.2)--(10,-1.2);
     \draw (9,0)--(8.4,-0.5)--(8,-1) (8.4,-0.5)--(8.8,-1) (9,0)--(9.6,-0.5)--(10,-1) (9.6,-0.5)--(9.2,-1);
    \draw[fill=white]  (9,0) circle (0.05) (8.4,-0.5) circle (0.05) (8,-1) circle (0.05) (8.8,-1) circle (0.05) (10,-1) circle (0.05) (9.2,-1) circle (0.05);
    \draw [fill=black] (9.6,-0.5) circle (0.05);
     \draw (7.5,-0.5) node{$\frac{u_{z1}}{u_z}$};
     \draw (6.75,-0.5) node{$+$};
     \end{tikzpicture}
      \end{minipage}
}
\hfil
\subfloat[
]{ \begin{minipage}[]{1\textwidth}

      \centering
     \begin{tikzpicture}
     \draw (0,-2)--(-0.6,-2.5)--(-1,-3) (-0.6,-2.5)--(-0.2,-3) (0,-2)--(0.6,-2.5)--(1,-3) (0.6,-2.5)--(0.2,-3);
     \draw [fill=black] (0,-2) circle (0.05);
     \draw[fill=white]  (-0.6,-2.5) circle (0.05) (0.6,-2.5) circle (0.05) (-1,-3) circle (0.05) (-0.2,-3) circle (0.05) (1,-3) circle (0.05) (0.2,-3) circle (0.05);
     \draw (0,-1.8) node{\scriptsize $u_\epsilon=1$} (-0.9,-2.5) node{\scriptsize $u_{0}$} (0.9,-2.5) node{\scriptsize $u_{1}$};
     \draw[->] (1.3,-2.5)--(2.2,-2.5);
     \draw (2.5,-2.5) node{\scriptsize $u_0$} (5.5,-2.5) node{\scriptsize $+u_1$};
     \draw (3,-1.8)--(2.8,-1.8)--(2.8,-3.2)--(3,-3.2) (5,-1.8)--(5.2,-1.8)--(5.2,-3.2)--(5,-3.2);
     \draw (4,-2)--(3.4,-2.5)--(3,-3) (3.4,-2.5)--(3.8,-3) (4,-2)--(4.6,-2.5)--(5,-3) (4.6,-2.5)--(4.2,-3);
     \draw [fill=black] (3.4,-2.5) circle (0.05);
     \draw[fill=white]  (4,-2) circle (0.05) (4.6,-2.5) circle (0.05) (3,-3) circle (0.05) (3.8,-3) circle (0.05) (5,-3) circle (0.05) (4.2,-3) circle (0.05);
    \draw (6,-1.8)--(5.8,-1.8)--(5.8,-3.2)--(6,-3.2) (8,-1.8)--(8.2,-1.8)--(8.2,-3.2)--(8,-3.2);
    \draw (7,-2)--(6.4,-2.5)--(6,-3) (6.4,-2.5)--(6.8,-3) (7,-2)--(7.6,-2.5)--(8,-3) (7.6,-2.5)--(7.2,-3);
     \draw [fill=black] (7.6,-2.5) circle (0.05);
     \draw[fill=white]  (7,-2) circle (0.05)(6.4,-2.5) circle (0.05)  (6,-3) circle (0.05) (6.8,-3) circle (0.05) (8,-3) circle (0.05) (7.2,-3) circle (0.05);
     \draw[->] (1.3,-4.5)--(2.2,-4.5);
     \draw (2.5,-4.5) node{\scriptsize $u_{00}$} (5.5,-4.5) node{\scriptsize $+u_{01}$} (8.5,-4.5) node{\scriptsize $+u_{10}$} (11.5,-4.5) node{\scriptsize $+u_{11}$};
     \draw (3,-3.8)--(2.8,-3.8)--(2.8,-5.2)--(3,-5.2) (5,-3.8)--(5.2,-3.8)--(5.2,-5.2)--(5,-5.2);
     \draw (4,-4)--(3.4,-4.5)--(3,-5) (3.4,-4.5)--(3.8,-5) (4,-4)--(4.6,-4.5)--(5,-5) (4.6,-4.5)--(4.2,-5);
     \draw [fill=black] (3,-5) circle (0.05);
     \draw[fill=white]  (4,-4) circle (0.05)(3.4,-4.5) circle (0.05) (4.6,-4.5) circle (0.05)  (3.8,-5) circle (0.05) (5,-5) circle (0.05) (4.2,-5) circle (0.05);
        \draw (6,-3.8)--(5.8,-3.8)--(5.8,-5.2)--(6,-5.2) (8,-3.8)--(8.2,-3.8)--(8.2,-5.2)--(8,-5.2);
    \draw (7,-4)--(6.4,-4.5)--(6,-5) (6.4,-4.5)--(6.8,-5) (7,-4)--(7.6,-4.5)--(8,-5) (7.6,-4.5)--(7.2,-5);
     \draw [fill=black] (6.8,-5) circle (0.05);
     \draw[fill=white]  (7,-4) circle (0.05)(6.4,-4.5) circle (0.05)(7.6,-4.5) circle (0.05)  (6,-5) circle (0.05)  (8,-5) circle (0.05) (7.2,-5) circle (0.05);
        \draw (9,-3.8)--(8.8,-3.8)--(8.8,-5.2)--(9,-5.2) (11,-3.8)--(11.2,-3.8)--(11.2,-5.2)--(11,-5.2);
    \draw (10,-4)--(9.4,-4.5)--(9,-5) (9.4,-4.5)--(9.8,-5) (10,-4)--(10.6,-4.5)--(11,-5) (10.6,-4.5)--(10.2,-5);
     \draw [fill=black](10.2,-5) circle (0.05);
     \draw[fill=white]  (10,-4) circle (0.05) (9.4,-4.5) circle (0.05) (10.6,-4.5) circle (0.05) (9,-5) circle (0.05) (9.8,-5) circle (0.05)  (11,-5) circle (0.05);
    \draw (12,-3.8)--(11.8,-3.8)--(11.8,-5.2)--(12,-5.2) (14,-3.8)--(14.2,-3.8)--(14.2,-5.2)--(14,-5.2);
    \draw (13,-4)--(12.4,-4.5)--(12,-5) (12.4,-4.5)--(12.8,-5) (13,-4)--(13.6,-4.5)--(14,-5) (13.6,-4.5)--(13.2,-5);
     \draw [fill=black] (14,-5) circle (0.05);
     \draw[fill=white]  (13,-4) circle (0.05)(12.4,-4.5) circle (0.05) (13.6,-4.5) circle (0.05) (12,-5) circle (0.05) (12.8,-5) circle (0.05) (13.2,-5) circle (0.05) ;
     \draw[->] (1.3,-6.5)--(2.2,-6.5);
     \draw (2.5,-6.5) node{\scriptsize $v_{00}$} (5.5,-6.5) node{\scriptsize $+v_{01}$} (8.5,-6.5) node{\scriptsize $+v_{10}$} (11.5,-6.5) node{\scriptsize $+v_{11}$};
     \draw (3,-5.8)--(2.8,-5.8)--(2.8,-7.2)--(3,-7.2) (5,-5.8)--(5.2,-5.8)--(5.2,-7.2)--(5,-7.2);
     \draw (4,-6)--(3.4,-6.5)--(3,-7) (3.4,-6.5)--(3.8,-7) (4,-6)--(4.6,-6.5)--(5,-7) (4.6,-6.5)--(4.2,-7);
     \draw [fill=black] (3,-7) circle (0.05);
     \draw[fill=white]  (4,-6) circle (0.05)(3.4,-6.5) circle (0.05) (4.6,-6.5) circle (0.05)  (3.8,-7) circle (0.05) (5,-7) circle (0.05) (4.2,-7) circle (0.05);
        \draw (6,-5.8)--(5.8,-5.8)--(5.8,-7.2)--(6,-7.2) (8,-5.8)--(8.2,-5.8)--(8.2,-7.2)--(8,-7.2);
    \draw (7,-6)--(6.4,-6.5)--(6,-7) (6.4,-6.5)--(6.8,-7) (7,-6)--(7.6,-6.5)--(8,-7) (7.6,-6.5)--(7.2,-7);
     \draw [fill=black] (6.8,-7) circle (0.05);
     \draw[fill=white]  (7,-6) circle (0.05)(6.4,-6.5) circle (0.05) (7.6,-6.5) circle (0.05) (6,-7) circle (0.05)  (8,-7) circle (0.05) (7.2,-7) circle (0.05);
        \draw (9,-5.8)--(8.8,-5.8)--(8.8,-7.2)--(9,-7.2) (11,-5.8)--(11.2,-5.8)--(11.2,-7.2)--(11,-7.2);
    \draw (10,-6)--(9.4,-6.5)--(9,-7) (9.4,-6.5)--(9.8,-7) (10,-6)--(10.6,-6.5)--(11,-7) (10.6,-6.5)--(10.2,-7);
     \draw [fill=black](10.2,-7) circle (0.05) ;
     \draw[fill=white]  (10,-6) circle (0.05) (9.4,-6.5) circle (0.05) (10.6,-6.5) circle (0.05) (9,-7) circle (0.05) (9.8,-7) circle (0.05)  (11,-7) circle (0.05);
     \draw (12,-5.8)--(11.8,-5.8)--(11.8,-7.2)--(12,-7.2) (14,-5.8)--(14.2,-5.8)--(14.2,-7.2)--(14,-7.2);
    \draw (13,-6)--(12.4,-6.5)--(12,-7) (12.4,-6.5)--(12.8,-7) (13,-6)--(13.6,-6.5)--(14,-7) (13.6,-6.5)--(13.2,-7);
     \draw [fill=black](14,-7) circle (0.05) ;
     \draw[fill=white] (13,-6) circle (0.05) (12.4,-6.5) circle (0.05) (13.6,-6.5) circle (0.05) (12,-7) circle (0.05) (12.8,-7) circle (0.05)  (13.2,-7) circle (0.05);
     \end{tikzpicture}
     \end{minipage}}
    \caption{Schematic of the unary-encoded state preparation procedure of Lemma~\ref{lem:unary_qsp_tree}. Black (white) circles indicated qubits in the $\ket{1}$ ($\ket{0}$) state. Values outside parentheses are amplitudes for each basis function. (a) In step $1\le k \le n$, for all $z\in\{0,1\}^k$, a SWAP gate is applied between qubits $z$ and $z0$, followed by a two qubit rotation between qubits $z0$ and $z1$ along the path $z0-z-z1$. (b) Full procedure for $n=2$. Step 0 (creation of the initial $\ket{1}$ state at the root node) is not shown. Steps $1$ and $2$ create the state $\sum_z u_z \ket{e_z}$.  The final step consists of single qubit rotations which add the phase to each amplitude, i.e. $u_x\rightarrow v_x = u_xe^{i\theta_x}$.}
    \label{fig:unary_QSP_example}
\end{figure}

\begin{lemma}[Unary to binary basis encodings]\label{lem:uary_binary}
The $(2^n+n)$-qubit unitary transformation 
\begin{equation*}
    \ket{0^n}\ket{e_x}\to\ket{x}\ket{0^{2^n}},\quad \forall x\in\Bn,
\end{equation*}
can be implemented by a quantum circuit of depth $O(n^2\log (n))$ and size $O(n^22^n)$, using $2^n-1$ ancillary qubits under binary tree constraint, where $x$ is in the first $\lceil \log(n+1)\rceil$ layers and $e_x$ is in the $(n+1)$-th layer of the binary tree.
\end{lemma}
\begin{proof}
Let $\kappa=\lceil\log(\frac{n+1}{2})\rceil$, and label qubits in the binary tree as in Section~\ref{sec:diag_with_ancilla_path} (Fig.~\ref{fig:label_binarytree}), i.e., with the root labelled with the empty string $\epsilon$, and with left and right children of qubit $z$ labelled as $z0$ and $z1$, respectively. This unitary transformation can be implemented in 2 steps.
\begin{enumerate}
    \item Step 1: 
    $\ket{0^{2^n-1}}_{\Tree_\epsilon^{n-1}}\ket{e_x}_{\scriptsize\Bn}\to\ket{x0^{2^{\kappa+1}-n-1}}_{\Tree_\epsilon^\kappa}\ket{0^{2^n-2^{\kappa+1}}}_{\Tree_\epsilon^{n-1}-\Tree_\epsilon^\kappa}\ket{e_x}_{\scriptsize\Bn}$.
   
   First, we implement unitary transformation $\ket{0}_\epsilon\ket{e_x}_{\scriptsize\Bn}\to \ket{x_i}_\epsilon\ket{e_x}_{\scriptsize \Bn}$ for all $i\in[n]$ by circuit $C_i$ under the binary tree constraint. 
    Unitary $C_i$ consists of $2n-1$ steps. In the first step, we apply CNOT gates where the controls are the subset of leaf nodes $x\in\Bn$ where $x_i=1$, and the targets are their respective parent nodes. In the $k$-th step ($2\le k\le n$), we apply CNOT gates of which the control qubits are in the $(n-k+1)$-th depth of the binary tree and the target qubits are their parents. In the $k$-th step $n+1\le k\le 2n-2$, we apply CNOT gates where the control qubits are in the $(k-n+1)$-th depth of the binary tree and the target qubits are their parents. The last step is the same as the first step. See an example in Fig. \ref{fig:C1_example}.
    
    Second, we use $C_i$ to implement step 1. Let $\{q_1,q_2,\ldots,q_n\}$ denote a set consisting of the first $n$ qubits in the $(n+1)$-th depth ($\B^{n+1}$) of the binary tree. Step 1 is realized step by step as follows:
    \begin{align*}
        &\ket{0^{2^{n}-1}}_{\Tree^{n-1}_\epsilon}\ket{e_x}_{\scriptsize\Bn}\ket{0^{2^{n+1}}}_{\scriptsize\B^{n+1}}\\
    \xrightarrow{C_1} &\ket{x_1}_{\epsilon}\ket{0^{2^{n}-2}}_{\Tree^{n-1}_\epsilon-\{\epsilon\}}\ket{e_x}_{\scriptsize\Bn} \ket{0^{2^{n+1}}}_{\scriptsize\B^{n+1}}\\
    \xrightarrow{\textsf{SWAP}^{\epsilon}_{q_1}}&\ket{0^{2^n-1}}_{\Tree_\epsilon^{n-1}}\ket{e_x}_{\scriptsize\Bn}\ket{x_10^{2^{n+1}-1}}_{\scriptsize\B^{n+1}}\\
    &\vdots\\
    \xrightarrow{C_j} &\ket{x_j}_{\epsilon}\ket{0^{2^{n}-2}}_{\Tree_\epsilon^{n-1}-\{\epsilon\}}\ket{e_x}_{\scriptsize\Bn}\ket{x_1x_2\ldots x_{j-1}0^{2^{n+1}-j+1}}_{\scriptsize\B^{n+1}}\\
    \xrightarrow{\textsf{SWAP}^\epsilon_{q_j}} &\ket{0^{2^{n}-1}}_{\Tree_\epsilon^{n-1}}\ket{e_x}_{\scriptsize\Bn}\ket{x_1x_2\ldots x_{j}0^{2^{n+1}-j}}_{\scriptsize\B^{n+1}}\\
    &\vdots\\
    \xrightarrow{C_n}&\ket{x_n}_{\epsilon}\ket{0^{2^{n}-2}}_{\Tree_\epsilon^{n-1}-\{\epsilon\}}\ket{e_x}_{\scriptsize\Bn}\ket{x_1x_2\ldots x_{n-1}0^{2^{n+1}-n+1}}_{\scriptsize\B^{n+1}}\\
    \xrightarrow{\textsf{SWAP}^\epsilon_{q_n}} &\ket{0^{2^{n}-1}}_{\Tree_\epsilon^{n-1}}\ket{e_x}_{\scriptsize\Bn}\ket{x_1x_2\ldots x_{n}0^{2^{n+1}-n}}_{\scriptsize\B^{n+1}}\\
   \xrightarrow{C}&\ket{x0^{2^{\kappa+1}-n-1}}_{\Tree_\epsilon^{\kappa}}\ket{0^{2^n-2^{\kappa+1}}}_{\Tree_\epsilon^{n-1}-\Tree_\epsilon^{\kappa}}\ket{e_x}_{\scriptsize\B^{n}}\ket{0^{2^{n+1}}}_{\scriptstyle\B^{n+1}}
    \end{align*}
    
   As discussed above, $C_j$ can be realized by a CNOT circuit of depth $O(n)$ for all $j\in[n]$. There exists an $O(\log(n))$-path between node $\epsilon$ and $q_j$ in a binary tree and thus, by Lemma \ref{lem:cnot_path_constraint}, $\textsf{SWAP}^{\epsilon}_{q_j}=\textsf{CNOT}^{\epsilon}_{q_j}\textsf{CNOT}_{\epsilon}^{q_j}\textsf{CNOT}^{\epsilon}_{q_j}$ can be implemented by a circuit of depth $O(n)$ under an $O(n)$-path.
   The functionality of $C$ is to swap the first $n$ qubits of the binary tree with qubits $\{q_1,q_2,\ldots,q_n\}$. Therefore, $C$ is an invertible linear transformation consisting of swap gates (each of which can be implemented with $3$ CNOT gates). By~Lemma \ref{lem:cnot_circuit}, it can be implemented by a CNOT circuit of depth $O(n^2)$. In summary, the total depth  step 1 is $n(O(n)+O(n)+O(n^2)=O(n^2)$.
    \item Step 2: 
    $\ket{x0^{2^{\kappa+1}-n-1}}_{\Tree_\epsilon^{\kappa}}\ket{0^{2^n-2^{\kappa+1}}}_{\Tree_\epsilon^{n-1}-\Tree_\epsilon^{\kappa}}\ket{e_x}_{\scriptsize \Bn}\to \ket{x0^{2^{\kappa+1}-n-1}}_{\Tree_\epsilon^{\kappa}}\ket{0^{2^{n+1}-2^{\kappa+1}}}_{\Tree_\epsilon^{n}-\Tree_\epsilon^{\kappa}} ,\forall x\in \Bn.$
    \begin{enumerate}
        \item Step 2.1: For simplicity, we assume that $\frac{n}{\kappa+1}$ is an integer. Let ${\sf R}_{\rm root}=\bigcup_{j=1}^{\frac{n}{\kappa+1}}\B^{(j-1)(\kappa+1)}$ denote the root nodes of all sub-binary trees of depth $\kappa$. The size of ${\sf R}_{\rm root}$ is $O(2^n/n)$. The following unitary transformation makes $O(2^n/n)$ copies of $x$.
        \begin{align*}
           &\ket{x0^{2^{\kappa+1}-n-1}}_{\Tree_\epsilon^{\kappa}}\bigotimes_{z\in{\sf R}_{\rm root}-\{\epsilon\}}\ket{0^{2^{\kappa+1}-1}}_{\Tree_z^{\kappa}}\\
         \to &\ket{x0^{2^{\kappa+1}-n-1}}_{\Tree_\epsilon^{\kappa}}\bigotimes_{z\in{\sf R}_{\rm root}-\{\epsilon\}}\ket{x0^{2^{\kappa+1}-n-1}}_{\Tree_z^{\kappa}}=\bigotimes_{z\in{\sf R}_{\rm root}}\ket{x0^{2^{\kappa+1}-n-1}}_{\Tree_z^{\kappa}},\forall x\in\Bn.
        \end{align*}
         By Lemma~\ref{lem:copy_binary_tree}, step 2.1 can be implemented by a circuit of depth $O(\log(2^n/n)n\log n)=O(n^2\log n)$ .
        \item Step 2.2: 
        \begin{align*}
            &\bigotimes_{z\in{\sf R}_{\rm root}}\ket{x0^{2^{\kappa+1}-n-1}}_{\Tree_z^{\kappa}}\ket{e_x}_{\scriptsize\Bn}
        \to \bigotimes_{z\in{\sf R}_{\rm root}}\ket{x0^{2^{\kappa+1}-n-1}}_{\Tree_z^{\kappa}}\ket{0^{2^n}}_{\scriptsize\Bn}.
        \end{align*}
        Lemma \ref{lem:tof} shows that the $|S|+1$-qubit $\textsf{Tof}^S_i(y)$ can be implemented by a circuit of size $O(n)$ for $|S|=n$. Let $S_z$ denote the set comprising the first $n$ qubits in $\Tree_z^{\kappa}$ for $z\in\B^{(\frac{n}{\kappa+1}-1)(\kappa+1)}\subset {\sf R}_{\rm root}$. The size of $\B^{(\frac{n}{\kappa+1}-1)(\kappa+1)}$ is $O(2^n/n)$. 
        If there exists one copy of $x$, we can apply $2^n$ toffoli gates $\textsf{Tof}^S_i(y)$ to implement step 2.2, where $S$ is the qubit set of $x$, $i$ is the target qubit and all $y\in \B^n$. Since we have $O(2^n/n)$ copy of $x$, we can apply $O(2^n/n)$ Toffoli gates in parallel, whose control and target set are disjoint. Under binary tree constraint, $\textsf{Tof}^S_i(y)$ can be implemented in depth and size $O(\log(n))\cdot O(n)=O(n\log(n))$. Therefore, step 2.2 can be implemented in depth $O(n\log(n))\cdot O(2^n/(2^n/n))=O(n^2\log(n))$.
    \item Step 2.3: applying the inverse circuit of step 2.1, we can implement the following circuit of depth $O(n^2\log n)$. 
            \begin{align*}
           &\ket{x0^{2^{\kappa+1}-n-1}}_{\Tree_\epsilon^{\kappa}}\bigotimes_{z\in{\sf R}_{\rm root}-\{\epsilon\}}\ket{x0^{2^{\kappa+1}-n-1}}_{\Tree_z^{\kappa}}\ket{0^{2^n}}_{\scriptsize \Bn}\\
         \to &\ket{x0^{2^{\kappa+1}-n-1}}_{\Tree_\epsilon^{\kappa}}\bigotimes_{z\in{\sf R}_{\rm root}-\{\epsilon\}}\ket{0^{2^{\kappa+1}-1}}_{\Tree_z^{\kappa}}\ket{0^{2^n}}_{\scriptsize \Bn}\\
         =&\ket{x0^{2^{\kappa+1}-n-1}}_{\Tree_\epsilon^{\kappa}}\ket{0^{2^{n+1}-2^{\kappa+1}}}_{\Tree_\epsilon^{n}-\Tree_\epsilon^{\kappa}}.
        \end{align*}
    \end{enumerate}
In summary the total depth of unitary transformation is $O(n^2)+3\cdot O(n^2\log(n))=O(n^2\log(n))$.
\end{enumerate}
\end{proof}
\begin{figure}[]
    \centering
    \begin{tikzpicture}[scale=0.8]
     \draw (5.5,-4.5) node{\scriptsize $+$} (8.5,-4.5) node{\scriptsize $+$} (11.5,-4.5) node{\scriptsize $+$};
     \draw (3,-3.8)--(2.8,-3.8)--(2.8,-5.2)--(3,-5.2) (5,-3.8)--(5.2,-3.8)--(5.2,-5.2)--(5,-5.2);
     \draw (4,-4)--(3.4,-4.5)--(3,-5) (3.4,-4.5)--(3.8,-5) (4,-4)--(4.6,-4.5)--(5,-5) (4.6,-4.5)--(4.2,-5);
     \draw [fill=black] (3,-5) circle (0.05) ;
     \draw [draw=red](4.2,-5)--(4.6,-4.5)--(5,-5);
     \draw[fill=white]  (4,-4) circle (0.05)(3.4,-4.5) circle (0.05) (4.6,-4.5) circle (0.05)  (3.8,-5) circle (0.05) (5,-5) circle (0.05) (4.2,-5) circle (0.05);
        \draw (6,-3.8)--(5.8,-3.8)--(5.8,-5.2)--(6,-5.2) (8,-3.8)--(8.2,-3.8)--(8.2,-5.2)--(8,-5.2);
    \draw (7,-4)--(6.4,-4.5)--(6,-5) (6.4,-4.5)--(6.8,-5) (7,-4)--(7.6,-4.5)--(8,-5) (7.6,-4.5)--(7.2,-5);
     \draw [fill=black] (6.8,-5) circle (0.05);
     \draw [draw=red] (7.2,-5)--(7.6,-4.5)--(8,-5) ;
     \draw[fill=white]  (7,-4) circle (0.05)(6.4,-4.5) circle (0.05)(7.6,-4.5) circle (0.05)  (6,-5) circle (0.05)  (8,-5) circle (0.05) (7.2,-5) circle (0.05);
        \draw (9,-3.8)--(8.8,-3.8)--(8.8,-5.2)--(9,-5.2) (11,-3.8)--(11.2,-3.8)--(11.2,-5.2)--(11,-5.2);
    \draw (10,-4)--(9.4,-4.5)--(9,-5) (9.4,-4.5)--(9.8,-5) (10,-4)--(10.6,-4.5)--(11,-5) (10.6,-4.5)--(10.2,-5);
     \draw [fill=black](10.2,-5) circle (0.05);
     \draw [draw=red] (10.2,-5)--(10.6,-4.5)--(11,-5);
     \draw[fill=white]  (10,-4) circle (0.05) (9.4,-4.5) circle (0.05) (10.6,-4.5) circle (0.05) (9,-5) circle (0.05) (9.8,-5) circle (0.05)  (11,-5) circle (0.05);
    \draw (12,-3.8)--(11.8,-3.8)--(11.8,-5.2)--(12,-5.2) (14,-3.8)--(14.2,-3.8)--(14.2,-5.2)--(14,-5.2);
    \draw (13,-4)--(12.4,-4.5)--(12,-5) (12.4,-4.5)--(12.8,-5) (13,-4)--(13.6,-4.5)--(14,-5) (13.6,-4.5)--(13.2,-5);
     \draw [fill=black] (14,-5) circle (0.05);
     \draw [draw=red] (13.2,-5)--(13.6,-4.5)--(14,-5);
     \draw[fill=white]  (13,-4) circle (0.05)(12.4,-4.5) circle (0.05) (13.6,-4.5) circle (0.05) (12,-5) circle (0.05) (12.8,-5) circle (0.05) (13.2,-5) circle (0.05) ;
     \draw[->] (1.3,-6.5)--(2.2,-6.5);
     \draw (5.5,-6.5) node{\scriptsize $+$} (8.5,-6.5) node{\scriptsize $+$} (11.5,-6.5) node{\scriptsize $+$};
     \draw (3,-5.8)--(2.8,-5.8)--(2.8,-7.2)--(3,-7.2) (5,-5.8)--(5.2,-5.8)--(5.2,-7.2)--(5,-7.2);
     \draw (4,-6)--(3.4,-6.5)--(3,-7) (3.4,-6.5)--(3.8,-7) (4,-6)--(4.6,-6.5)--(5,-7) (4.6,-6.5)--(4.2,-7);
     \draw [fill=black] (3,-7) circle (0.05);
       \draw[draw=red] (4.6,-6.5)--(4,-6)--(3.4,-6.5);
     \draw[fill=white]  (4,-6) circle (0.05)(3.4,-6.5) circle (0.05) (4.6,-6.5) circle (0.05)  (3.8,-7) circle (0.05) (5,-7) circle (0.05) (4.2,-7) circle (0.05);
        \draw (6,-5.8)--(5.8,-5.8)--(5.8,-7.2)--(6,-7.2) (8,-5.8)--(8.2,-5.8)--(8.2,-7.2)--(8,-7.2);
    \draw (7,-6)--(6.4,-6.5)--(6,-7) (6.4,-6.5)--(6.8,-7) (7,-6)--(7.6,-6.5)--(8,-7) (7.6,-6.5)--(7.2,-7);
     \draw [fill=black] (6.8,-7) circle (0.05);
        \draw[draw=red] (7.6,-6.5)--(7,-6)--(6.4,-6.5);
     \draw[fill=white]  (7.6,-6.5) circle (0.05)(6.4,-6.5) circle (0.05)(7,-6) circle (0.05)  (6,-7) circle (0.05)  (8,-7) circle (0.05) (7.2,-7) circle (0.05);
        \draw (9,-5.8)--(8.8,-5.8)--(8.8,-7.2)--(9,-7.2) (11,-5.8)--(11.2,-5.8)--(11.2,-7.2)--(11,-7.2);
    \draw (10,-6)--(9.4,-6.5)--(9,-7) (9.4,-6.5)--(9.8,-7) (10,-6)--(10.6,-6.5)--(11,-7) (10.6,-6.5)--(10.2,-7);
     \draw [fill=black](10.2,-7) circle (0.05) (10.6,-6.5) circle (0.05);
          \draw[draw=red] (10.6,-6.5)--(10,-6)--(9.4,-6.5);
     \draw[fill=white]  (10,-6) circle (0.05) (9.4,-6.5) circle (0.05)  (9,-7) circle (0.05) (9.8,-7) circle (0.05)  (11,-7) circle (0.05);
     \draw (12,-5.8)--(11.8,-5.8)--(11.8,-7.2)--(12,-7.2) (14,-5.8)--(14.2,-5.8)--(14.2,-7.2)--(14,-7.2);
    \draw (13,-6)--(12.4,-6.5)--(12,-7) (12.4,-6.5)--(12.8,-7) (13,-6)--(13.6,-6.5)--(14,-7) (13.6,-6.5)--(13.2,-7);
     \draw [fill=black](14,-7) circle (0.05)  (13.6,-6.5) circle (0.05);
     \draw[draw=red] (13.6,-6.5)--(13,-6)--(12.4,-6.5);
     \draw[fill=white] (13,-6) circle (0.05)(12.4,-6.5) circle (0.05)  (12,-7) circle (0.05) (12.8,-7) circle (0.05)  (13.2,-7) circle (0.05);
     \draw[->] (1.3,-8.5)--(2.2,-8.5);
     \draw (5.5,-8.5) node{\scriptsize $+$} (8.5,-8.5) node{\scriptsize $+$} (11.5,-8.5) node{\scriptsize $+$};
     \draw (3,-7.8)--(2.8,-7.8)--(2.8,-9.2)--(3,-9.2) (5,-7.8)--(5.2,-7.8)--(5.2,-9.2)--(5,-9.2);
     \draw (4,-8)--(3.4,-8.5)--(3,-9) (3.4,-8.5)--(3.8,-9) (4,-8)--(4.6,-8.5)--(5,-9) (4.6,-8.5)--(4.2,-9);
     \draw [fill=black] (3,-9) circle (0.05);
    \draw [draw=red](4.2,-9)--(4.6,-8.5)--(5,-9);
     \draw[fill=white]  (4,-8) circle (0.05)(3.4,-8.5) circle (0.05) (4.6,-8.5) circle (0.05)  (3.8,-9) circle (0.05) (5,-9) circle (0.05) (4.2,-9) circle (0.05);
        \draw (6,-7.8)--(5.8,-7.8)--(5.8,-9.2)--(6,-9.2) (8,-7.8)--(8.2,-7.8)--(8.2,-9.2)--(8,-9.2);
    \draw (7,-8)--(6.4,-8.5)--(6,-9) (6.4,-8.5)--(6.8,-9) (7,-8)--(7.6,-8.5)--(8,-9) (7.6,-8.5)--(7.2,-9);
     \draw [fill=black] (6.8,-9) circle (0.05);
     \draw [draw=red](7.2,-9)--(7.6,-8.5)--(8,-9);
     \draw[fill=white] (6.4,-8.5) circle (0.05)(7,-8) circle (0.05)  (7.6,-8.5) circle (0.05) (6,-9) circle (0.05)  (8,-9) circle (0.05) (7.2,-9) circle (0.05);
        \draw (9,-7.8)--(8.8,-7.8)--(8.8,-9.2)--(9,-9.2) (11,-7.8)--(11.2,-7.8)--(11.2,-9.2)--(11,-9.2);
    \draw (10,-8)--(9.4,-8.5)--(9,-9) (9.4,-8.5)--(9.8,-9) (10,-8)--(10.6,-8.5)--(11,-9) (10.6,-8.5)--(10.2,-9);
     \draw [fill=black](10,-8) circle (0.05)(10.2,-9) circle (0.05) (10.6,-8.5) circle (0.05);
    \draw [draw=red](10.2,-9)--(10.6,-8.5)--(11,-9);
     \draw[fill=white]   (9.4,-8.5) circle (0.05)  (9,-9) circle (0.05) (9.8,-9) circle (0.05)  (11,-9) circle (0.05);
     \draw (12,-7.8)--(11.8,-7.8)--(11.8,-9.2)--(12,-9.2) (14,-7.8)--(14.2,-7.8)--(14.2,-9.2)--(14,-9.2);
    \draw (13,-8)--(12.4,-8.5)--(12,-9) (12.4,-8.5)--(12.8,-9) (13,-8)--(13.6,-8.5)--(14,-9) (13.6,-8.5)--(13.2,-9);
     \draw [fill=black](14,-9) circle (0.05) (13,-8) circle (0.05)(13.6,-8.5) circle (0.05) ;
     \draw [draw=red](13.2,-9)--(13.6,-8.5)--(14,-9);
     \draw[fill=white](12.4,-8.5) circle (0.05)   (12,-9) circle (0.05) (12.8,-9) circle (0.05)  (13.2,-9) circle (0.05);
     
     \draw[->] (1.3,-10.5)--(2.2,-10.5);
     \draw (5.5,-10.5) node{\scriptsize $+$} (8.5,-10.5) node{\scriptsize $+$} (11.5,-10.5) node{\scriptsize $+$};
     \draw (3,-9.8)--(2.8,-9.8)--(2.8,-11.2)--(3,-11.2) (5,-9.8)--(5.2,-9.8)--(5.2,-11.2)--(5,-11.2);
     \draw (4,-10)--(3.4,-10.5)--(3,-11) (3.4,-10.5)--(3.8,-11) (4,-10)--(4.6,-10.5)--(5,-11) (4.6,-10.5)--(4.2,-11);
     \draw [fill=black] (3,-11) circle (0.05);
     \draw[fill=white]  (4,-10) circle (0.05)(3.4,-10.5) circle (0.05) (4.6,-10.5) circle (0.05)  (3.8,-11) circle (0.05) (5,-11) circle (0.05) (4.2,-11) circle (0.05);
        \draw (6,-9.8)--(5.8,-9.8)--(5.8,-11.2)--(6,-11.2) (8,-9.8)--(8.2,-9.8)--(8.2,-11.2)--(8,-11.2);
    \draw (7,-10)--(6.4,-10.5)--(6,-11) (6.4,-10.5)--(6.8,-11) (7,-10)--(7.6,-10.5)--(8,-11) (7.6,-10.5)--(7.2,-11);
     \draw [fill=black]  (6.8,-11) circle (0.05);
     \draw[fill=white]  (7,-10) circle (0.05)(6,-11) circle (0.05)(6.4,-10.5) circle (0.05) (7.6,-10.5) circle (0.05)  (8,-11) circle (0.05) (7.2,-11) circle (0.05);
        \draw (9,-9.8)--(8.8,-9.8)--(8.8,-11.2)--(9,-11.2) (11,-9.8)--(11.2,-9.8)--(11.2,-11.2)--(11,-11.2);
    \draw (10,-10)--(9.4,-10.5)--(9,-11) (9.4,-10.5)--(9.8,-11) (10,-10)--(10.6,-10.5)--(11,-11) (10.6,-10.5)--(10.2,-11);
     \draw [fill=black](10,-10) circle (0.05)(10.2,-11) circle (0.05) ;
     \draw[fill=white]   (9.4,-10.5) circle (0.05) (10.6,-10.5) circle (0.05) (9,-11) circle (0.05) (9.8,-11) circle (0.05)  (11,-11) circle (0.05);
     \draw (12,-9.8)--(11.8,-9.8)--(11.8,-11.2)--(12,-11.2) (14,-9.8)--(14.2,-9.8)--(14.2,-11.2)--(14,-11.2);
    \draw (13,-10)--(12.4,-10.5)--(12,-11) (12.4,-10.5)--(12.8,-11) (13,-10)--(13.6,-10.5)--(14,-11) (13.6,-10.5)--(13.2,-11);
     \draw [fill=black](14,-11) circle (0.05)  (13,-10) circle (0.05) ;
     \draw[fill=white](12.4,-10.5) circle (0.05) (13.6,-10.5) circle (0.05) (12,-11) circle (0.05) (12.8,-11) circle (0.05)  (13.2,-11) circle (0.05);
     \draw (1.5,-11.6) node{\scriptsize The state of 4 leaf nodes:} (4,-11.6) node{\scriptsize $\ket{e_{00}}$} (7,-11.6) node{\scriptsize $\ket{e_{01}}$}(10,-11.6) node{\scriptsize $\ket{e_{10}}$} (13,-11.6) node{\scriptsize $\ket{e_{11}}$};
     \end{tikzpicture}
    \caption{The operator $C_1$ when $n=2$. Black (white) circle indicate qubits in $\ket{1}$ ($\ket{0}$) states. Red line denotes a CNOT gate, for which the control qubit is in the lower depth of the tree. Step 1: CNOT gates applied to leaf nodes $S= \{x\in\{0,1\}^n, x_1 = 1\} = \{10,11\}$. Step 2: CNOT gates applied to all nodes in the layer above the leaves.  Step 3 is the same as Step 1.}
    \label{fig:C1_example}
\end{figure}
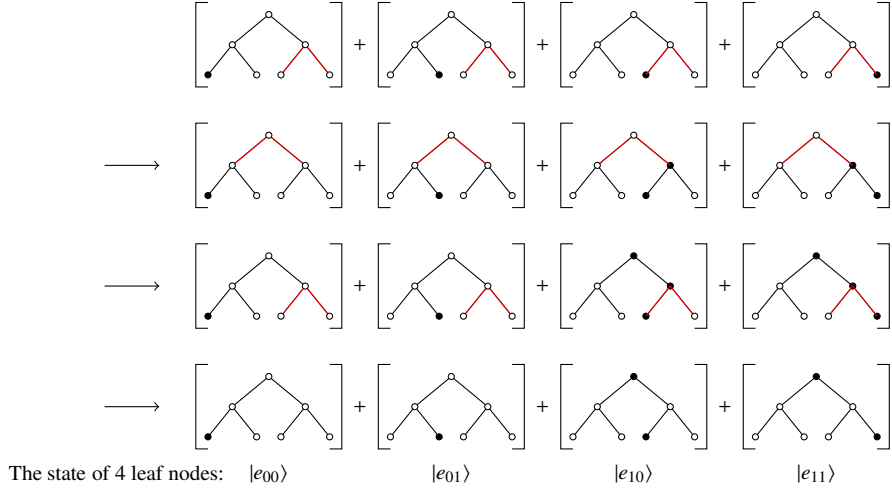

By combining the previous two lemmas, we can create a unary-encoded state, and then transform it to a standard (binary-encoded) state, to give the following result.
\begin{lemma}\label{lem:qsp_unary_binary_tree}
Any $n$-qubit quantum state $\ket{\psi_v}=\sum_{\scriptsize x\in\Bn}v_x\ket{x}$ can be prepared by a quantum circuit of depth $O(n^2\log n)$ and size $O(n^22^n)$ under $\Tree_{n+m}(2)$ constraint using $m=2^{n+1}-1$ ancillary qubits.
\end{lemma}

\paragraph{Circuit implementation for QSP under $\Tree_{n+m}(2)$ constraint}
\label{sec:QSP_mix_binarytree}
~\\

The QSP circuit in Lemma~\ref{lem:qsp_unary_binary_tree} requires an exponential number of ancilla. To allow for arbitrary numbers of ancilla, we use the QSP circuit framework of~\cite{sun2021asymptotically}, see Fig.~\ref{fig:new_framework}.
\begin{figure}[ht]
\centerline{
\Qcircuit @C=0.6em @R=0.7em {
\lstick{\ket{0}} & \qw  & \multigate{4}{\qsp_t} &\multigate{5}{\scriptstyle V_{t+1}} &\qw& \push{\cdots} & &\multigate{7}{\scriptstyle V_n} &  \qw\\
\lstick{\ket{0}} &  \qw  & \ghost{\qsp_t} &\ghost{\scriptstyle V_{t+1}} &\qw& \cdots & &\ghost{\scriptstyle V_n}&  \qw\\
\lstick{\ket{0}} & \qw   & \ghost{\qsp_t} &\ghost{\scriptstyle V_{t+1}} &\qw& \cdots & &\ghost{\scriptstyle V_n}&  \qw\\
\vdots~~~~~~~~~ &  \qw   & \ghost{\qsp_t}& \ghost{\scriptstyle V_{t+1}} &\qw& \cdots & &\ghost{\scriptstyle V_n}& \qw\\
\lstick{\ket{0}} &  \qw   & \ghost{\qsp_t} &\ghost{\scriptstyle V_{t+1}} &\qw& \cdots & &\ghost{\scriptstyle V_n}&  \qw \\
\lstick{\ket{0}} &  \qw & \qw  & \ghost{\scriptstyle V_{t+1}} &\qw& \cdots & &\ghost{\scriptstyle V_n}&  \qw \\
\vdots~~~~~~~~~ &  \qw & \qw  & \qw &\qw& \cdots & &\ghost{\scriptstyle V_n}&  \qw \inputgroupv{2}{7}{4em}{4 em}{ \scriptstyle n \text{~qubits}~~~~~~~~~~~~~~~}\\
\lstick{\ket{0}} &  \qw & \qw & \qw &\qw& \cdots& &\ghost{\scriptstyle V_n}&  \qw \\
}
}
\caption{The QSP circuit framework of~\cite{sun2021asymptotically}. For all $t\in[n]$, the unitary transformation $\qsp_t$ is a $t$-qubit QSP circuit (which we take here to be implemented by the method of Lemma~\ref{lem:qsp_unary_binary_tree}), and $V_{t+1},V_{t+2},\ldots,V_n$ are UCGs of size $t+1,t+2,\ldots,n$, respectively.}
\label{fig:new_framework}
\end{figure}
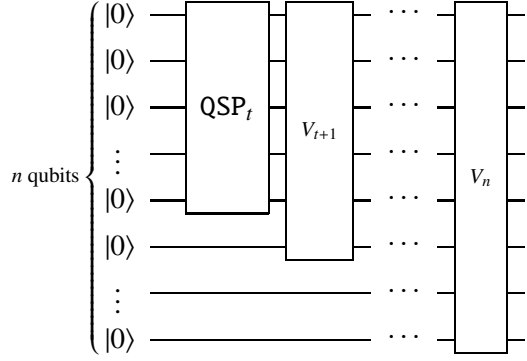

\begin{theorem}\label{thm:QSP_binarytree_improvement}
Any $n$-qubit quantum state can be realized by a  quantum circuit of depth 
\[
\begin{cases}
    O\big(n^2\log^2(n)+\frac{\log(n)2^n}{n+m}\big),   &\quad \text{if~} m\le o(2^n),\\
    O(n^2\log(n)),   &\quad \text{if~} m\ge \Omega(2^n),        
\end{cases}
\]
under $\Tree_{n+m}(2)$ constraint, using $m\ge 0 $ ancillary qubits.
\end{theorem}
\begin{proof} We consider the following cases. 

\begin{enumerate}
    \item $0\le m\le O(2^n/n^3)$. Use the QSP framework in Fig.~\ref{fig:QSP_circuit}. By Lemmas~\ref{lem:QSP_framework_UCG} and \ref{lem:UCG_binarytree}, the total depth required is $\sum_{k=1}^{n} O\left(k^2\log(k)+\frac{\log(k)2^k}{k+m}\right)=O\left(\frac{\log(n)2^n}{n+m}\right)$.
    \item $\omega(2^n/n^3)\le m\le o(2^n)$.
Use the QSP framework in Fig.~\ref{fig:new_framework}, with $t=n-3\log(n)$. By Lemma~\ref{lem:qsp_unary_binary_tree}, $\qsp_t$ can be implemented in depth $O(t^2\log(t))=O(n^2\log(n))$. The total circuit depth required is therefore 
\begin{equation*}
O(n^2\log(n))+\sum_{k=t+1}^n O\left(k^2\log(k)+\frac{\log(k)2^k}{k+m}\right)=O\left(n^2\log^2(n)+\frac{\log(n)2^n}{n+m}\right).    
\end{equation*}
\item $m\ge \Omega(2^n)$. If $m\ge 2^{n+1}-1$, use only $2^{n+1}-1$ of the ancilla. 
Use the QSP framework in Fig.~\ref{fig:new_framework}, with $t=n-\log(m+1)+2$. In this case, the total depth required is
\begin{equation*}
O(n^2\log(n))+\sum_{k=t+1}^n O\left(k^2\log(k)+\frac{\log(k)2^k}{k+m}\right)=O\left(n^2\log^2(n)\right).    
\end{equation*}
\end{enumerate}
\end{proof}

Note that, for general $m$, applying Lemma~\ref{lem:UCG_binarytree} to the QSP framework of Fig.~\ref{fig:QSP_circuit} (as in the first case in the proof above) leads to a circuit depth of $\sum_{i=k}^nO\left(k^2\log(k)+\frac{\log(k)2^k}{k+m}\right)=O\left(n^3\log(n)+\frac{\log(n)2^n}{n+m}\right)$ under binary tree constraint. For $m\ge \Omega(2^n/n^3)$ this gives an $O(n^3 \log(n))$ circuit depth bound, which is weaker than what we are able to achieve using the QSP framework of Fig.~\ref{fig:new_framework}. 

\end{document}